\documentclass[11pt, oneside]{article}   	

\usepackage{jheppub}

\usepackage{amsmath}
\usepackage{amssymb}
\usepackage{amsfonts}
\usepackage{amsthm}
\usepackage{amsbsy}
\usepackage{array}

\usepackage{mathtools}
\usepackage[usenames,dvipsnames]{xcolor}

\usepackage{subcaption}

\usepackage{dsdshorthand}

\usepackage[vcentermath]{youngtab}

\usepackage{tikz}
\usetikzlibrary{trees}
\usetikzlibrary{decorations.pathmorphing}
\usetikzlibrary{decorations.markings}
\usetikzlibrary{shapes.misc}
 
\tikzset{
  threept/.style={
    circle,
    draw,
    inner sep=2pt,
  },
  twopt/.style={
    circle,
    draw,
    fill=black,
    inner sep=1pt,
    minimum size=1pt
  },
  cross/.style={
    cross out,
    draw=black, 
    minimum size=7pt, 
    inner sep=0pt,
    outer sep=0pt
  },
  scalar/.style={
    thick,
    dashed,
    postaction={
      decorate,
      decoration={
        markings,
        mark=at position 0.5 with {\arrow{>}}
      }
    }
  },
  spinning/.style={
    thick,
    postaction={
      decorate,
      decoration={
        markings,
        mark=at position 0.5 with {\arrow{>}}
      }
    }
  },
  spinning no arrow/.style={
    thick,
  },
  finite with arrow/.style={
    decoration={
      snake,
      amplitude=1pt,
      segment length=6pt,
      post length=2pt
    },
    decorate,
    thick,->
  },
  finite/.style={
    decoration={
      snake,
      amplitude=1pt,
      segment length=6pt,
    },
    decorate,
    thick
  }
}

\newcommand\wL{\mathbf{L}}
\newcommand\wS{\mathbf{S}}
\newcommand\wSD{\mathbf{S}_\De}
\newcommand\wSE{\mathbf{S}_E}
\newcommand\wSJ{\mathbf{S}_J}
\newcommand\wF{\mathbf{F}}
\newcommand\wR{\mathbf{R}}
\newcommand\wRb{\bar{\mathbf{R}}}

\newcommand{\tsym}{\cT}

\makeatletter
\def\@fpheader{\ }
\makeatother

\title{Light-ray operators in conformal field theory}
\author{Petr Kravchuk and David Simmons-Duffin}
\affiliation{Walter Burke Institute for Theoretical Physics, Caltech, Pasadena, California 91125, USA }
\date{}
\abstract{We argue that every CFT contains light-ray operators labeled by a continuous spin $J$. When $J$ is a positive integer, light-ray operators become integrals of local operators over a null line. However for non-integer $J$, light-ray operators are genuinely nonlocal and give the analytic continuation of CFT data in spin described by Caron-Huot. A key role in our construction is played by a novel set of intrinsically Lorentzian integral transforms that generalize the shadow transform.  Matrix elements of light-ray operators can be computed via the integral of a double-commutator against a conformal block. This gives a simple derivation of Caron-Huot's Lorentzian OPE inversion formula and lets us generalize it to arbitrary four-point functions. Furthermore, we show that light-ray operators enter the Regge limit of CFT correlators, and generalize conformal Regge theory to arbitrary four-point functions. The average null energy operator is an important example of a light-ray operator. Using our construction, we find a new proof of the average null energy condition (ANEC), and furthermore generalize the ANEC to continuous spin.}

\preprint{CALT-TH 2018-018}

\begin{document}

\maketitle

\newpage

\section{Introduction}

Singularities of Euclidean correlators in conformal field theory (CFT) are described by the operator product expansion (OPE). However, in Lorentzian signature there exist singularities that cannot be described in a simple way using the OPE. One of the most important is the Regge limit of a time-ordered four-point function  (figure~\ref{fig:reggelimit}) \cite{Cornalba:2006xk,Cornalba:2006xm,Cornalba:2007fs,Cornalba:2008qf,Cornalba:2009ax,Banks:2009bj}.\footnote{In perturbation theory, Lorentzian singularities correspond to Landau diagrams \cite{Maldacena:2015iua}. It is possible that this is also true nonperturbatively.} The Regge limit is the CFT version of a high-energy scattering process: operators $\cO_1(x_1)$ and $\cO_3(x_3)$ create excitations that move along nearly lightlike trajectories, interact, and then are measured by operators $\cO_2(x_2)$ and $\cO_4(x_4)$. In holographic theories, the Regge limit is dual to high-energy forward scattering in the bulk \cite{Brower:2006ea}.

\begin{figure}[ht!]
	\centering
		\begin{tikzpicture}
		
		\draw (-3,0) -- (0,3) -- (3,0) -- (0,-3) -- cycle;
				
		\draw[fill=black] (-1.4,-1) circle (0.07); 
		\draw[fill=black] (1.4,1) circle (0.07);  
		\draw[fill=black] (-1.4,1) circle (0.07); 
		\draw[fill=black] (1.4,-1) circle (0.07);  
				
		\draw[opacity=.3] (-1.5,-1.5) -- (1.5,1.5);
		\draw[opacity=.3] (-1.5,1.5) -- (1.5,-1.5);

		\draw[->,opacity=0.7] (-1.42,1.12) -- (-1.48,1.37);
		\draw[->,opacity=0.7] (-1.42,-1.12) -- (-1.48,-1.37);
		\draw[->,opacity=0.7] (1.42,1.12) -- (1.48,1.37);
		\draw[->,opacity=0.7] (1.42,-1.12) -- (1.48,-1.37);				
		
		\node[left] at (-1.5,-1) {$1$};
		\node[right] at (1.5,1) {$2$};	
		\node[left] at (-1.5,1) {$4$};
		\node[right] at (1.5,-1) {$3$};	
		
		\end{tikzpicture}
		\caption{The Regge limit of a four-point function: the points $x_1,\dots,x_4$ approach null infinity, with the pairs $x_1,x_2$ and $x_3,x_4$ becoming nearly lightlike separated.
		}
		\label{fig:reggelimit}
\end{figure}
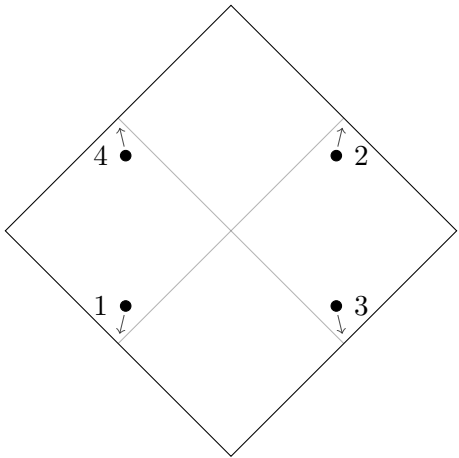

In Lorentzian signature, the OPE $\cO_i\x\cO_j$ converges if the product $\cO_i \cO_j$ acts on the vacuum (either past or future) \cite{mack1977}. That is, we have an equality of states
\be
\label{eq:opeeq}
\cO_i \cO_j |\Omega\> &= \sum_k f_{ijk} \cO_k|\Omega\>,
\ee
where $k$ runs over local operators of the theory (we suppress position dependence, for brevity). 
Thus, in figure~\ref{fig:reggelimit} the OPE $\cO_1\x\cO_3$ converges because it acts on the past vacuum, the OPE $\cO_2\x\cO_4$ converges because it acts on the future vacuum, and the OPEs $\cO_1\x\cO_4$ and $\cO_2\x\cO_3$ converge because they act on either the past or future vacuum. (Here we use the fact that spacelike-separated operators commute to rearrange the operators in the time-ordered correlator to apply (\ref{eq:opeeq}).) However, each of these OPEs is converging very slowly in the Regge limit. They can be used to prove results like analyticity and boundedness in the Regge limit \cite{Maldacena:2015waa,Hartman:2015lfa}, but they are less useful for computations (unless one has good control over the theory). Meanwhile, the OPEs $\cO_1 \x \cO_2$ and $\cO_3\x\cO_4$ are invalid in the Regge regime.

The problem of describing four-point functions in the Regge regime was partially solved in \cite{Brower:2006ea,Cornalba:2007fs,Costa:2012cb}. The behavior of the correlator is controlled by the analytic continuation of data in the $\cO_1\x\cO_2$ and $\cO_3 \x\cO_4$ OPEs to non-integer spin. For example, in a planar theory, the Regge correlator behaves (very) schematically as
\be
\frac{\<\cO_1\cO_2\cO_3\cO_4\>}{\<\cO_1\cO_2\>\<\cO_3\cO_4\>} &\sim 1-f_{12\cO}(J_0)f_{34\cO}(J_0) e^{t(J_0-1)} + \dots.
\label{eq:reggecartoonplanar}
\ee
Here, $f_{12\cO}(J)$ and $f_{34\cO}(J)$ are OPE coefficients that have been analytically continued in the spin $J$ of $\cO$.
The parameter $t$ measures the boost of $\cO_1,\cO_2$ relative to $\cO_3,\cO_4$. $J_0\in \R$ is the Regge/Pomeron intercept, and is determined by the analytic continuation of the dimension $\De_\cO$ to non-integer $J$.\footnote{In $d=2$, the Regge regime is the same as the chaos regime. In $d\geq 3$, it is related to chaos in hyperbolic space. See \cite{Roberts:2014ifa,Murugan:2017eto} for discussions. Note that $J_0-1$ plays the role of a Lyapunov exponent, and it is constrained by the chaos bound to be less than $1$ \cite{Shenker:2014cwa,Maldacena:2015waa}.} The ``$\dots$" in (\ref{eq:reggecartoonplanar}) represent higher-order corrections in $1/N^2$ and also terms that grow slower than $e^{t(J_0-1)}$ in the Regge limit $t\to \oo$.

A missing link in this story was provided recently by Caron-Huot, who proved that OPE coefficients and dimensions have a natural analytic continuation in spin in any CFT \cite{Caron-Huot:2017vep}. The analytic continuation of OPE data in a scalar four-point function $\<\f_1\f_2\f_3\f_4\>$ can be computed by a ``Lorentzian inversion formula," given by the integral of a double-commutator $\<[\f_4,\f_1][\f_2,\f_3]\>$ times a conformal block $G_{J+d-1,\De-d+1}$ with unusual quantum numbers. Specifically, $\De,J$ are replaced with
\be
\label{eq:chtransformation}
(\De,J)\to(J+d-1,\De-d+1)
\ee
relative to a conventional conformal block.
Caron-Huot's Lorentzian inversion formula has many other useful applications, for example to large-spin perturbation theory and the lightcone bootstrap \cite{Fitzpatrick:2012yx,Komargodski:2012ek,Fitzpatrick:2015qma,Li:2015itl,Li:2015rfa,Alday:2015eya,Alday:2015ota,Alday:2015ewa,Simmons-Duffin:2016wlq,Alday:2016njk}, and to the SYK model \cite{Sachdev:1992fk,kitaevfirsttalk,Maldacena:2016hyu,Polchinski:2016xgd}.\footnote{In the 1-dimensional SYK model, the analog of analytic continuation in spin is analytic continuation in the weight of discrete states in the conformal partial wave expansion \cite{Maldacena:2016hyu,Simmons-Duffin:2017nub}.}

However, Caron-Huot's result raises some obvious questions:
\begin{itemize}
\item Can operators themselves (not just their OPE data) be analytically continued in spin?
\item What is the space of continuous spin operators in a given CFT?
\item Do continuous-spin operators have a Hilbert space interpretation (similar to how integer-spin operators correspond to CFT states on $S^{d-1}$)?
\item What is the meaning of the funny block in the Lorentzian inversion formula, and how do we generalize it?
\end{itemize}
Answering these questions is important for making sense of the Regge limit, and more generally for understanding how to write a convergent OPE in non-vacuum states.

It is easy to describe continuous-spin operators mathematically. Consider first a primary operator $\cO^{\mu_1\cdots\mu_J}(x)$ with integer spin $J$. Let us introduce a null polarization vector $z_\mu$ and contract it with the indices of $\cO$ to form a function of $(x,z)$:
\be
\label{eq:integerspincontraction}
\cO(x,z) &\equiv \cO^{\mu_1\cdots\mu_J}(x) z_{\mu_1}\cdots z_{\mu_J},
\qquad (z^2 = 0).
\ee
The tensor $\cO^{\mu_1\cdots\mu_J}(x)$ can be recovered from the function $\cO(x,z)$ by stripping off the $z$'s and subtracting traces. Thus, $\cO(x,z)$ is a valid alternative description of a traceless symmetric tensor.  Note that $\cO(x,z)$ is a homogeneous polynomial of degree $J$ in $z$. The generalization to a continuous spin operator $\mathbb{O}$ is now straightforward: we simply drop the requirement that $\mathbb{O}(x,z)$ be polynomial in $z$ and allow it to have non-integer homogeneity,
\be
\mathbb{O}(x,\l z) &= \l^J \mathbb{O}(x,z),\qquad \l>0,\ J\in \C.
\ee

Continuous-spin operators are necessarily nonlocal. This follows from Mack's classification of positive-energy representations of the Lorentzian conformal group $\tl \SO(d,2)$ \cite{mack19772}, which only includes nonnegative integer spin representations.\footnote{For non traceless-symmetric tensor operators, we define spin as the length of the first row of the Young diagram for their $\SO(d)$ representation. For fermionic representations spin is a half-integer and for simplicity of language we include this case into the notion of ``integer spin'' operators.} CFT states have positive energy, so by the state-operator correspondence, local operators must have nonnegative integer spin, and conversely continuous-spin operators must be nonlocal.
Mack's classification also shows that continuous-spin operators must annihilate the vacuum:
\be
\mathbb{O}(x,z)|\Omega\> &= 0\qquad (J\notin \Z_{\geq 0}),
\ee
otherwise $\mathbb{O}(x,z)|\Omega\>$ would transform in a nontrivial continuous-spin representation, which would include a state with negative energy.

If continuous-spin operators annihilate the vacuum, how can we analytically continue the local operators of a CFT, which certainly do not annihilate the vacuum? The answer is that we must first turn local operators into something nonlocal that annihilates the vacuum, and then analytically continue that. The correct object turns out to be the integral of a local operator along a null line,
\be
\label{eq:nulllineintegral}
\int_{-\oo}^\oo d\a\, \cO(\a z,z) &= \int_{-\oo}^\oo d\a\,\cO^{\mu_1\cdots \mu_J}(\a z)z_{\mu_1}\cdots z_{\mu_J}.
\ee
This can be written more covariantly by performing a conformal transformation to bring the beginning of the null line to a generic point $x$:\footnote{As $\a\to 0^-$, the point $x-z/\a$ diverges to future null infinity, and the integration contour should be understood as extending into the next Poincare patch on the Lorentzian cylinder. We give more detail in section~\ref{sec:lighttransformintro}.}
\be
\wL[\cO](x,z) &\equiv \int_{-\oo}^\oo d\a (-\a)^{-\De-J} \cO\p{x-\frac z \a, z}.
\ee
This defines an integral transform $\wL$ that we call the ``light transform." The expression (\ref{eq:nulllineintegral}) corresponds to $\wL[\cO](-\oo z,z)$, where $x=-\oo z$ is a point at past null infinity.

After reviewing some representation theory in sections~\ref{sec:lorentziancylinder} and \ref{sec:reptheoryreview}, we show in section~\ref{sec:weylandintegral} that if $\cO_{\De,J}$ has dimension $\De$ and spin $J$, then $\wL[\cO_{\De,J}](x,z)$ transforms like a primary operator with dimension $1-J$ and spin $1-\De$:
\be
\label{eq:wLquantumnumbers}
\wL : (\De,J) &\to (1-J,1-\De).
\ee
In particular, $\wL[\cO_{\De,J}]$ can have non-integer spin. The average null energy operator $\cE=\wL[T]$ (the light transform of the stress tensor) is a special case, having dimension $-1$ and spin $1-d$. We will see that $\wL$ is part of a dihedral group ($D_8$) of intrinsically Lorentzian integral transforms that generalize the Euclidean shadow transform \cite{Ferrara:1972uq,SimmonsDuffin:2012uy}. These Lorentzian transforms implement affine Weyl reflections that preserve the Casimirs of the conformal group. For example, the quadratic Casimir eigenvalue is given by
\be
C_2(\De,J) &= \De(\De-d) + J(J+d-2),
\ee
and this is indeed invariant under (\ref{eq:wLquantumnumbers}). The transformation (\ref{eq:chtransformation}) appearing in Caron-Huot's formula is another affine Weyl reflection. The Lorentzian transforms do not give precisely a representation of $D_8$, but instead satisfy an interesting ``anomalous" algebra that we derive in section~\ref{sec:algebraoftransforms}. Mack's classification implies that $\wL[\cO_{\De,J}]$ must annihilate the vacuum whenever $\cO_{\De,J}$ is a local operator. This is also easy to see directly by deforming the $\a$ contour into the complex plane, as we show in section~\ref{sec:lightproperties}.

We claim that the operators $\wL[\cO_{\De,J}]$ can be analytically continued in $J$, and their continuations are light-ray operators.\footnote{Note that $\wL[\cO_{\De,J}](x,z)$ has dimension $1-J$ and spin $1-\De$. Thus, analytic continuation in $J$ is really analytic continuation in the {\it dimension\/} of $\wL[\cO_{\De,J}]$ away from negative integer values. We will continue to refer to it as analytic continuation in spin, since $J$ labels the spin of local operators.} As an example, consider Mean Field Theory (a.k.a.\ Generalized Free Fields) in $d=2$ with a scalar primary $\f$. This theory contains ``double-trace" operators
\be
[\f\f]_J(u,v)\equiv:\!\f(u,v) \ptl_v^J \f(u,v)\!:+\,\ptl_v(\ldots)
\ee
 with dimension $2\De_\f+J$ and even spin $J$.
Here, $:\,:$ denotes normal ordering and we have written out the definition up to total derivatives (which are required to ensure that this is a primary operator).  We are using lightcone coordinates $u=x-t,v=x+t$, and for simplicity focusing on operators with $\ptl_v$ derivatives only.  The corresponding analytically-continued light-ray operators are
\be
\label{eq:examplecontinuousspin}
\mathbb{O}_J(0,-\oo) &=\frac{i\Gamma(J+1)}{2^{J}}\int_{-\oo}^\oo dv \int_{-\oo}^\oo \frac{ds}{2\pi}\, \p{\frac{1}{(s+i\e)^{J+1}} + \frac{1}{(-s+i\e)^{J+1}}} :\!\f(0,v+s)\f(0,v-s)\!:.
\ee
When $J$ is an even integer, we have
\be
\frac{i\Gamma(J+1)}{2\pi}\, \p{\frac{1}{(s+i\e)^{J+1}} - \frac{1}{(s-i\e)^{J+1}}} = \frac{\ptl^J\de(s)}{\ptl s^J} \qquad (J\in 2\Z_{\geq 0}).
\ee
Thus, when $J$ is an even integer, $\mathbb{O}_J$ becomes
\be
\mathbb{O}_J(0,-\oo) &=
2^{-J}\int_{-\oo}^\oo dv \int_{-\oo}^\oo ds\, \frac{\ptl^J\de(s)}{\ptl s^J}:\!\f(0,v+s)\f(0,v-s)\!: \nn\\
&= \int_{-\oo}^\oo dv :\!\f\ptl_v^J \f\!:\!(0,v)
=\wL[[\f\f]_J](0,-\oo)\qquad (J\in 2\Z_{\geq 0}).
\ee
By contrast, when $J$ is not an even integer, $\mathbb{O}_J$ is a legitimately nonlocal light-ray operator whose correlators are analytic continuations of the correlators of $\wL[[\f\f]_J]$. In particular, three-point functions $\<\cO_1 \cO_2 \mathbb{O}_J\>$ give an analytic continuation of the three-point coefficients of $\<\cO_1 \cO_2 [\f\f]_J\>$.

Similar light-ray operators have a long history in the gauge-theory literature \cite{BALITSKY1989541,Braun:2003rp} (see \cite{Caron-Huot:2013fea,Balitsky:2013npa,Balitsky:2015oux,Balitsky:2015tca} for recent discussions). There, one often considers a bilocal integral of operators inserted along a null Wilson line. Such operators were discussed in \cite{Hofman:2008ar}, where they were argued to control OPEs of the average null energy operator $\cE$.  In perturbation theory, it is reasonable to imagine constructing more operators like (\ref{eq:examplecontinuousspin}). However, it is less clear how to define them in a nonperturbative context where normal ordering is not well-defined, and there can be  complicated singularities when two operators become lightlike-separated. It is also not clear what a null Wilson line means in an abstract CFT.

Our tool for constructing analogs of $\mathbb{O}_J$ in general CFTs will be harmonic analysis \cite{Dobrev:1977qv}.  Given primary operators $\cO_1,\cO_2$, we find in section~\ref{sec:lightray} an integration kernel $K_{\De,J}(x_1,x_2,x,z)$ such that
\be\label{eq:lightraykernelschematic}
\mathbb{O}_{\De,J}(x,z) &= \int d^dx_1 d^dx_2 K_{\De,J}(x_1,x_2,x,z) \cO_1(x_1) \cO_2(x_2)
\ee
transforms like a primary with dimension $1-J$ and spin $1-\De$ (when inserted in a time-ordered correlator). The object $\mathbb{O}_{\De,J}$ is meromorphic in $\De$ and $J$ and has poles of the form
\be
\mathbb{O}_{\De,J}(x,z) &\sim \frac{1}{\De-\De_i(J)} \mathbb{O}_{i,J}(x,z).
\ee
We conjecture based on examples that poles must come from the region where $x_1,x_2$ are close to the light ray $x+\R_{\geq 0} z$ (we have not established this rigorously in a general CFT). The residues of the poles can thus be interpreted as light-ray operators $\mathbb{O}_{i,J}(x,z)$ that make sense in arbitrary correlators. Furthermore, when $J$ is an integer, the residues are light-transforms of local operators $\wL[\cO]$. Thus the $\mathbb{O}_{i,J}$ give analytic continuations of $\wL[\cO]$ for all $\cO\in\cO_1\x\cO_2$.

In section~\ref{sec:inversionformulae}, we show that $\<\cO_3 \cO_4 \mathbb{O}_{\De,J}\>$ can be computed via the integral of a double-commutator $\<[\cO_4,\cO_1][\cO_2,\cO_3]\>$ over a Lorentzian region of spacetime. This leads to a simple proof of Caron-Huot's Lorentzian inversion formula. The contour manipulation from \cite{Simmons-Duffin:2017nub} is crucial for this computation. However, the light-ray perspective makes our proof simpler than the one in \cite{Simmons-Duffin:2017nub}. In particular, it makes it clearer why the unusual conformal block $G_{J+d-1,\De-d+1}$ appears.  The reason is that the quantum numbers $(J+d-1,\De-d+1)$ are dual to those of the light-transform $(1-J,1-\De)$ in the sense that the product
\be
d^d x\, d^d z\, \de(z^2)\, \cO_{1-J,1-\De}(x,z)\cO_{J+d-1,\De-d+1}(x,z)  
\ee
has dimension zero and spin zero.  Our perspective also leads to a natural generalization of Caron-Huot's formula to the case of arbitrary operator representations, which we describe in section~\ref{sec:generalizedinversiontwo}. Subsequently in section~\ref{sec:conformalregge}, we generalize conformal Regge theory to arbitrary operator representations as well, along the way showing that light-ray operators describe part of the Regge limit of four-point functions as conjectured in \cite{Banks:2009bj}. 

As mentioned above, the average null energy operator $\cE=\wL[T]$ is an example of a light-ray operator. The average null energy condition (ANEC) states that $\cE$ is positive-semidefinite, i.e.\ its expectation value in any state is nonnegative. Some implications of the ANEC in CFTs are discussed in \cite{Hofman:2008ar,Cordova:2017zej,Cordova:2017dhq}. The ANEC was recently proven in \cite{Faulkner:2016mzt} using techniques from information theory and in \cite{Hartman:2016lgu} using causality.  By expressing $\cE$ as the residue of an integral of a pair of real operators $\f(x_1) \f(x_2)$, we find  a new proof of the ANEC in section~\ref{sec:positivityandtheanec}.\footnote{Our proof is conceptually very similar to the one in \cite{Hartman:2016lgu}, but it has a technical advantage that it does not require any assumptions about the behavior of correlators outside the regime of OPE convergence. A disadvantage is that we require the dimension $\De_\f$ to be sufficiently low, though we expect it should be possible to relax this restriction.} Furthermore, $\cE$ is part of a family of light-ray operators $\cE_J$ labeled by continuous spin $J$, and our construction of light-ray operators applies to this entire family. This lets us derive a novel generalization of the ANEC to continuous spin. More precisely, we show that
\be
\<\Psi|\cE_J|\Psi\> &\geq 0,\qquad (J\in \R_{\geq J_\mathrm{min}}),
\ee
where $\cE_J$ is the family of light-ray operators whose values at even integer $J$ are given by
\be
\cE_J &= \wL[\cO_{\De_\mathrm{min}(J),J}] \qquad (J\in 2\Z,\ J\geq 2),
\ee
where $\cO_{\De_\mathrm{min}(J),J}$ is the operator with spin $J$ of minimal dimension. Here, $J_\mathrm{min}\leq 1$ is the smallest value of $J$ for which the Lorentzian inversion formula holds \cite{Caron-Huot:2017vep}.

We conclude in section~\ref{sec:discussion} with discussion and numerous questions for the future. The appendices contain useful mathematical background, further technical details, and some computations needed in the main text. In particular, appendix~\ref{sec:contspincorrelators} includes a general discussion of continuous-spin tensor structures and their analyticity properties, appendix~\ref{sec:euclideanharmonicanalysis} contains a lightning review of harmonic analysis for the Euclidean conformal group, and appendix~\ref{sec:conformalblockscontspin} gives details on conformal blocks with continuous spin.

\subsection*{Notation}

In this work, we use the convention that correlators in the state $|\Omega\>$ represent physical correlators in a CFT. For example,
\be
\<\Omega|\cO_1\cdots\cO_n|\Omega\>
\ee
is a physical Wightman function, and
\be
\<\cO_1\cdots\cO_n\>_\O &\equiv \<\O|T\{\cO_1\cdots\cO_n\}|\O\>
\ee
is a physical time-ordered correlator.

Often, we discuss two- and three-point structures that are fixed by conformal invariance up to a constant. These structures do not represent physical correlators --- they are simply known functions of spacetime points. We write them as correlators in the ficticious state $|0\>$. For example, if $\f_i$ are scalar primaries with dimensions $\De_i$, then
\be
\label{eq:structurenotation}
\<0|\f_1(x_1)\f_2(x_2)\f_3(x_3)|0\> &= \frac{1}{(x_{12}^2 + i \e t_{12})^{\frac{\De_1+\De_2-\De_3}{2}} (x_{23}^2 + i \e t_{23})^{\frac{\De_2+\De_3-\De_1}{2}} (x_{13}^2 + i \e t_{13})^{\frac{\De_1+\De_3-\De_2}{2}}}
\ee
denotes the unique conformally-invariant three-point structure for scalars with dimensions $\De_i$, with the  $i\e$-prescription appropriate for the given Wightman ordering. Similarly,
\be
\label{eq:structurenotationtimeordered}
\<\f_1(x_1)\f_2(x_2)\f_3(x_3)\> &= \frac{1}{(x_{12}^2 + i \e)^{\frac{\De_1+\De_2-\De_3}{2}} (x_{23}^2 + i \e)^{\frac{\De_2+\De_3-\De_1}{2}} (x_{13}^2 + i \e)^{\frac{\De_1+\De_3-\De_2}{2}}}
\ee
denotes the unique conformally-invariant structure with the $i\e$-prescription for a time-ordered correlator.
In particular, (\ref{eq:structurenotation}) and (\ref{eq:structurenotationtimeordered}) do not include OPE coefficients.

\section{The light transform}

This section is devoted to mathematical background and results that will be needed for constructing and studying light-ray operators.
We first review some basic facts about the Lorentzian conformal group and its representation theory, with an emphasis on continuous spin operators. We then introduce a set of intrinsically Lorentzian integral transforms, which generalize the well-known Euclidean shadow transform, and study their properties. One of these transforms is the ``light transform" mentioned in the introduction. It will play a key role in the sections that follow.

\subsection{Review: Lorentzian cylinder}
\label{sec:lorentziancylinder}

Similarly to Euclidean space $\R^{d}$, Minkowski space $\cM_d=\R^{d-1,1}$ is not invariant under finite conformal transformations. In Euclidean space, this problem is easily solved by studying CFTs on $S^d$, the conformal compactification of $\R^d$. In Lorentzian signature, the problem is more subtle.

The simplest extension of Minkowski space $\cM_d=\R^{d-1,1}$  that is invariant under the Lorentzian conformal group $\SO(d,2)$ is its conformal compactification $\cM_d^c$. The space $\cM_d^c$ can be easily described by the embedding space construction~\cite{Dirac:1936fq,Mack:1969rr,Boulware:1970ty,Ferrara:1973eg,Ferrara:1973yt,Cornalba:2009ax,Weinberg:2010fx}: it is the projectivization of the null cone in $\R^{d,2}$ on which $\SO(d,2)$ acts by its vector representation. If we choose coordinates on $\R^{d,2}$ to be $X^{-1}, X^{0}, \ldots X^{d}$ with the metric 
\be
	X^2=-(X^{-1})^2-(X^{0})^2+(X^1)^2+\ldots+(X^d)^2,
\ee
then the null cone is defined by
\be
	(X^{-1})^2+(X^{0})^2=(X^1)^2+\ldots+(X^d)^2.
\ee
If we mod out by positive rescalings (i.e.\ by $\R_+$), we can set both sides of this equation to $1$, identifying the space of solutions with $S^1\times S^{d-1}$, where the $S^1$ is timelike. To get $\cM_d^c$, we mod out by $\R$ rescalings,\footnote{In the Euclidean embedding space construction based on $\R^{d+1,1}$ we usually just take the future null cone instead of considering negative rescalings, but in $\R^{d,2}$ the null cone is connected and this is not possible.} obtaining $\cM_d^c=S^1\times S^{d-1}/\Z_2$, where $\Z_2$ identifies antipodal points in both $S^1$ and $S^{d-1}$.  Minkowski space $\cM_d\subset \cM_d^c$ can be obtained by introducing lightcone coordinates in $\R^{d,2}$,
\be
X^\pm = X^{-1}\pm X^d,
\ee
and considering points with $X^+\neq 0$. Using $\R$ rescalings we can set $X^+=1$ for such points, and the null cone equation becomes
\be
	X^-=-(X^0)^2+(X^1)^2+\ldots+(X^{d-1})^2.
\ee
If we set $x^\mu = X^\mu$ for $\mu=0,\ldots d-1$, this gives the standard embedding of $\R^{d-1,1}$, 
\be
(X^+,X^-,X^\mu)=(1,x^2,x^\mu).
\label{eq:MinkowskiEmbedding}
\ee
One can check that the action of $\SO(d,2)$ on $X$ induces the usual conformal group action on $x^\mu$. The points that lie in $\cM_d^c\backslash\cM_d$ have $X^+=0$ and thus $X^\mu X_\mu=0$ with arbitrary $X^-$. They correspond to space-time infinity\footnote{In $\cM_d^c$ the infinite future, the infinite past and the spatial infinity of Minkowski space are identified. The past neighborhood of the future infinity, the future neighborhood of the past infinity and the spacelike neighborhood of the spatial infinity together form a complete neighbourhood of the space-time infinity in $\cM_d^c$.} ($X^\mu= 0$) and null infinity ($X^\mu\neq 0$).

By construction, $\cM_d^c$ has an action of $\SO(d,2)$ and is thus a natural candidate for the space on which a conformally-invariant QFT can live. However, it is unsuitable for this purpose due to the existence of closed timelike curves that are evident from its description as $S^1\times S^{d-1}/\Z_2$ with timelike $S^1$. This problem can be fixed by instead considering the universal cover $\widetilde\cM_d=\R\times S^{d-1}$,\footnote{For $d=2$ this is not the universal cover.} which is simply the Lorentzian cylinder. It was shown in~\cite{Luscher:1974ez} that Wightman functions of a CFT on $\R^{d-1,1}$ can be analytically continued to $\widetilde{\cM}_d$. Indeed, one can first Wick-rotate the CFT to $\R^d$, map it conformally to the Euclidean cylinder $\R\times S^{d-1}$, and then Wick-rotate to $\widetilde{\cM}_d$ (of course the actual proof in~\cite{Luscher:1974ez} is more involved). 

\begin{figure}[t!]
	\centering
	\begin{tikzpicture}

	\draw[fill=yellow, opacity = 0.05,blue] (3,0) -- (0,3) -- (-3,0) -- (0,-3) -- cycle;
	
	\draw[] (-3,0) -- (0,3) -- (3,0) -- (0,-3) -- cycle;
	\draw[] (0,3) -- (0.3,3.3);
	\draw[] (0,3) -- (-0.3,3.3);
	\draw[] (0,-3) -- (0.3,-3.3);
	\draw[] (0,-3) -- (-0.3,-3.3);
	
	\draw[fill=black] (-3,0) circle (0.07); 
	\draw[fill=black] (3,0) circle (0.07); 
	
	\node[left] at (-3,0) {$\infty$};
	\node[right] at (3,0) {$\infty$};
	
	\node at (0,0) {$\cM_d$};
	
	\node at (2,2) {$\tl\cM_d$};
	
	\draw[dashed] (-3,-3.3) -- (-3,3.3);				
	\draw[dashed] (3,-3.3) -- (3,3.3);				
	
	\end{tikzpicture}
	\caption{Poincare patch $\cM_d$ (blue, shaded) inside the Lorentzian cylinder $\tl\cM_d$ in the case of 2 dimensions. The spacelike infinity of $\cM_d$ is marked by $\infty$. The dashed lines should be identified.}
	\label{fig:lorentziancylinder}
\end{figure}
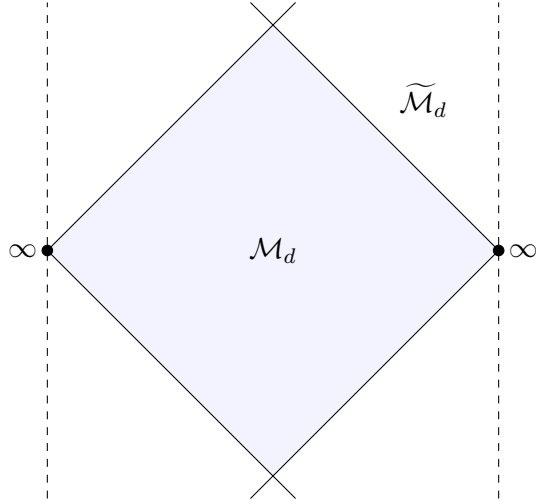

To describe coordinates on $\widetilde{\cM}_d$, it is convenient to first consider the null cone in $\R^{d,2}$ mod $\R_+$. It is equivalent to $S^1\times S^{d-1}$ defined by
\be
	(X^{-1})^2+(X^{0})^2=(X^1)^2+\ldots+(X^d)^2=1,
\ee
and we can use the parametrization
\be
	X^{-1}&=\cos\tau,\nn\\
	X^{0}&=\sin\tau,\nn\\
	X^i&=e^i,\quad i=1\ldots d,\label{eq:CylinderCovering}
\ee
where $\vec e$ is a unit vector in $\R^d$. Here $\tau$ is the coordinate on $S^1$ with identification $\tau\sim \tau+2\pi$, and taking the universal cover is equivalent to removing this identification. The coordinates $(\tau,\vec e)$ with $\tau\in\R$ then cover $\widetilde\cM_d$ completely. Minkowski space $\cM_d$ can be conformally identified with a particular region in $\widetilde{\cM}_d$ by using the embedding~\eqref{eq:MinkowskiEmbedding}. This gives
\be
	x^0&=\frac{\sin\tau}{\cos\tau+e^d},\nn\\
	x^i&=\frac{e^i}{\cos\tau+e^d}, \quad i=1,\ldots d-1,
\ee
in the region where $\cos\tau+e^d>0$ and $-\pi<\tau<\pi$. This region consists of points spacelike separated from $\tau=0,\, \vec e=(0,\ldots,0,-1)$, which is the spatial infinity of $\cM_d$ (see figure~\ref{fig:lorentziancylinder}). We will refer to this particular region as the (first) Poincare patch. Note that the null cone in $\R^{d,2}$ modulo $\R_+$ contains two Poincare patches -- one with $X^+>0$ and one with $X^+<0$. The relation between Wightman functions on $\cM_d$ and $\widetilde\cM_d$ (in their natural metrics) for operators reads as\footnote{When applied to operators with spin, this identity does not produce a nice function on $\tl\cM_d$, because in typical bases of spin indices on Minkowski space translations in $\tau$ act by matrices which have singularities. Therefore, in order to have nice functions on~$\tl\cM_d$ one has to perform a redefinition of spin indices~\cite{Luscher:1974ez}.}
\be
	\<\Omega|\cO_1(x_1)\cdots \cO_n(x_n)|\Omega\>_{\cM_d}=\prod_{i=1}^{n}(\cos\tau_i+e^d_i)^{\Delta_i}\<\Omega|\cO_1(\tau_1,\vec e_1)\cdots \cO_n(\tau_n,\vec e_n)|\Omega\>_{\widetilde\cM_d}.
\ee

Let us discuss the action of the conformal group on $\widetilde\cM_d$. First of all, because we have taken the universal cover of $\cM_d^c$, it is no longer true that $\SO(d,2)$ acts on $\widetilde\cM_d$. Instead, the universal covering group $\widetilde\SO(d,2)$ acts on $\widetilde\cM_d$. Indeed, the rotation generator $M_{-1,0}$ generates shifts in $\tau$ and in $\SO(d,2)$ we have $e^{2\pi M_{-1,0}}=1$, whereas this is definitely not true on $\widetilde\cM_d$ because $\tau\nsim\tau+2\pi$. In the universal cover $\widetilde{\SO}(d,2)$, this direction gets decompactified so that the action becomes consistent. 

\subsubsection{Symmetry between different Poincare patches}

There exists an important symmetry $\tsym$ of $\widetilde{\cM}_d$ that commutes with the action of $\widetilde\SO(d,2)$. Namely, if we take a point with coordinates $p=(\tau,\vec e)$ and send light rays in all future directions, they will all converge at the point $\tsym p\equiv (\tau+\pi,-\vec e)$. The points $p$ and $\tsym p$ in $\widetilde\cM_d$ correspond to the same point in $\cM_d^c$ and thus $\tsym$ commutes with infinitesimal conformal generators and therefore also with the full $\widetilde\SO(d,2)$.

When $d$ is even, $\tsym$ lies in the center of $\widetilde{\SO}(d,2)$ and we can take
\be\label{eq:evendt}
	\tsym=e^{\pi M_{-1,0}}e^{\pi M_{1,2}+\pi M_{3,4}+\ldots+\pi M_{d-1,d}}.
\ee
For odd $d$ only $\tsym^2$ lies in $\tl\SO(d,2)$. But if the theory preserves parity, i.e.\ we have an operator $P$ that maps $x^1\to -x^1$ in the first Poincare patch, then we can take
\be
	\tsym = e^{\pi M_{0,-1}+\pi M_{23}+\ldots +\pi M_{d-1,d}}P.
\ee
If the theory doesn't preserve parity, $\tsym$ can still be defined as an operation on correlation functions in the sense specified below.

If $\tsym$ exists as a unitary operator on the Hilbert space ($d$ even or parity-preserving theory in odd $d$), then we can consider its action on local operators. For scalars we clearly have
\be
	\tsym \phi(x) \tsym^{-1}=\phi(\tsym x),
\ee
up to intrinsic parity in odd $d$. To understand the action of $\tsym$ on operators with spin, it is convenient to work in the embedding space, where we have for tensor operators
\be
	\tsym\cO(X,Z_1,Z_2,\ldots Z_n)\tsym^{-1}=\cO(-X,-Z_1,-Z_2,\ldots, -Z_n).
\ee
Here the point $-X$ is interpreted as the point in the Poincare patch which is in immediate future of the first Poincare patch, and $Z_i$ are null polarizations corresponding to the various rows of the Young diagram of $\cO$. Again, in odd dimensions we might need to add a factor of intrinsic parity.

Note that the above action on tensor operators can be defined regardless of the dimension $d$ or whether or not the theory preserves parity. We will thus define $\tsym$ as an operator which can act on functions on $\tl\cM_d$ according to
\be
	(\tsym\.\cO)(X,Z_1,Z_2,\ldots Z_n)\equiv\cO(-X,-Z_1,-Z_2,\ldots, -Z_n),
\ee
where again $-X$ is interpreted as corresponding to $\tsym x$. As discussed above, in even dimensions this always comes from a unitary symmetry of the theory defined by~\eqref{eq:evendt}, but in odd dimensions it may not be a symmetry (even if the theory preserves parity). In such cases we can still use $\tsym$ thus defined to study conformally-invariant objects, similarly to how we can separate tensor structures into parity-odd and parity-even regardless of whether the theory preserves parity. To have a uniform discussion, we will use this definition of $\tsym$ action in the rest of the paper.

Finally, let us note that in even dimensions for tensor operators
\be\label{eq:tsymeigenvalue}
	\tsym \cO(x)|\O\>&=e^{i\pi(\De+N)}\cO(x)|\O\>,\nn\\
	\<\O|\cO(x)\tsym&=e^{i\pi(\De+N)}\<\O|\cO(x),
\ee
where $N$ is the total number of boxes in the $\SO(d-1,1)$ Young diagram of $\cO$. This follows from the fact that the representation generated by $\cO$ acting on the vacuum is irreducible. One can check the eigenvalue by considering this identity inside a Wightman two-point function. The same relation holds in parity-even structures in odd dimensions (in particular, in two-point functions) and with a minus sign in parity-odd structures.

\subsubsection{Causal structure}

The action of $\widetilde\SO(d,2)$ on $\widetilde\cM_d$ preserves the causal structure of the Lorentzian cylinder~\cite{Luscher:1974ez}. This property will allow us to define conformally-invariant integration regions. We usually label points in $\widetilde\cM_d$ by natural numbers and we write $1<2$ when point $1$ is inside the past lightcone of $2$ and $1\approx 2$ when $1$ is spacelike from $2$. Furthermore, we write $1^\pm$ for $\tsym^{\pm 1}1$ (more generally, $1^{\pm k}$ for $\tsym^{\pm k}1$). That is, $1^+$ is the point in the ``next" Poincare patch with the same Minkowski coordinates as $1$. Similarly, $1^-$ is the point in the ``previous" Poincare patch with the same Minkowski coordinates as 1. Some causal relationships between points can be written in different ways, for example $1\approx 2$ if and only if $2^-<1<2^+$ (figure~\ref{fig:causalrelationships}).

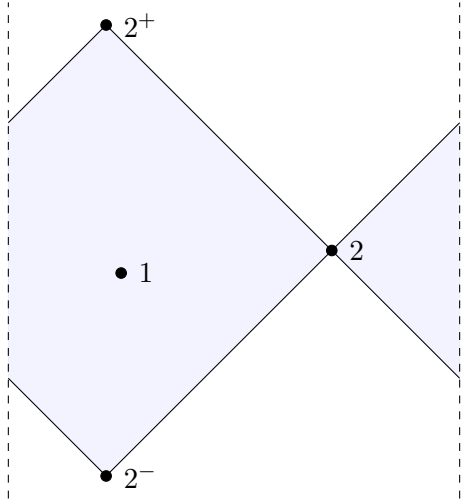
\begin{figure}[ht!]
	\centering
		\begin{tikzpicture}
		\draw[fill=yellow, opacity = 0.05,blue] (1.3,0) -- (-1.7,3) -- (-3,1.7) -- (-3,-1.7) -- (-1.7,-3) -- cycle;
		\draw[fill=yellow, opacity = 0.05,blue] (1.3,0) -- (3,1.7) -- (3,-1.7) -- cycle;
		
		\draw[fill=black] (-1.5,-0.3) circle (0.07); 
		\draw[fill=black] (1.3,0) circle (0.07);  
		\draw[fill=black] (-1.7,-3) circle (0.07);  
		\draw[fill=black] (-1.7,3) circle (0.07);  

        \draw[dashed] (-3,-3.3) -- (-3,3.3);				
        \draw[dashed] (3,-3.3) -- (3,3.3);				
		\draw[] (1.3,0) -- (-1.7,3);
		\draw[] (1.3,0) -- (-1.7,-3);
		\draw[] (1.3,0) -- (3,1.7);
		\draw[] (1.3,0) -- (3,-1.7);
		\draw[] (-1.7,3) -- (-3,1.7);
		\draw[] (-1.7,-3) -- (-3,-1.7);
		
		\node[right] at (-1.4,-0.3) {$1$};
		\node[right] at (1.4,0) {$2$};
		\node[right] at (-1.6,3) {$2^+$};
		\node[right] at (-1.6,-3) {$2^-$};
		
		\end{tikzpicture}
		\caption{$1$ is spacelike from $2$ ($1\approx 2$) if and only if $1$ is in the future of $2^-$ and the past of $2^+$ ($2^-<1<2^+$). The figure shows the Lorentzian cylinder in 2-dimensions. The dashed lines should be identified.}
		\label{fig:causalrelationships}
\end{figure}

\subsection{Review: Representation theory of the conformal group}
\label{sec:reptheoryreview}

We will also need some facts from unitary representation theory of the conformal groups $\SO(d+1,1)$ and $\SO(d,2)$. These groups are non-compact and their unitary representations are infinite-dimensional. We will mostly be interested in a particular class of unitary representations known as principal series representations, and also their non-unitary analytic continuations.

Unitary principal series representations of $\SO(d+1,1)$ are the easiest to describe. In this case, a principal series representation $\cE_{\De,\rho}$ is labeled by a pair $(\De,\rho)$, where $\De$ is a scaling dimension of the form $\De=\frac{d}{2}+i s$ with $s\in \R$ and an $\rho$ is an irreducible $\SO(d)$ representation. The elements of $\cE_{\De,\rho}$ are functions on $\R^d$ (more precisely, on the conformal sphere $S^d$) that transform under $\SO(d+1,1)$ as primary operators with scaling dimension $\De$ and $\SO(d)$ representation $\rho$. The inner product between two functions $f^a(x)$ and $g^a(x)$ (where $a$ is an index for $\rho$) is defined by
\be
	(f,g)\equiv \int d^d x (f^a(x))^* g^a(x).
\ee
This is positive-definite by construction. It is conformally-invariant because while $g$ transforms with scaling dimension $\De=\frac d 2 + is$ in $\rho$ of $\SO(d)$, $f^*$ transforms with scaling dimension $\De^*=\frac d 2 - is$ in $\rho^*$ of $\SO(d)$, and thus the integrand is a scalar of scaling dimension $\De+\De^*=d$, as required for conformal invariance. The representations $\cE_{\De,\rho}$ are important because the representations of primary operators that appear in CFTs are their analytic continuations to real $\De$.\footnote{It will not be important to give a precise meaning to this ``analytic continuation''; in most of the paper we only use $\cE_{\De,\rho}$ as a guide for writing conformally-invariant formulas. The same remark concerns representations of $\tl\SO(d,2)$ below.} Also, $\cE_{\De,\rho}$ appear in partial wave analysis of Euclidean correlators~\cite{Dobrev:1977qv}.

The pair $(\De,\rho)$ can be thought of as a weight of the algebra $\mathfrak{so}_\C(d+2)$ if we define $-\De$ to be the length of the first row of a Young diagram, and use the Young diagram of $\rho$ for the remaining rows. Through this identification, the unitary representations of $\SO(d+2)$ have non-positive (half-)integer $\De$. For $\SO(d+1,1)$, we instead have continuous  $\De$ because the corresponding Cartan generator $D\propto M_{-1,d+1}$ of $\SO(d+1,1)$ is noncompact (i.e.\ it must be multiplied by $i$ in order to relate the Lie algebra $\mathfrak{so}(d+1,1)$ to the compact form $\mathfrak{so}(d+2)$).

In $\SO(d,2)$ there are two noncompact Cartan generators ($D$ and $M_{01}$), and both of their weights become continuous. Thus, the unitary principal series representations $\cP_{\De,J,\l}$  for $\SO(d,2)$ are parametrized by a triplet $(\De, J, \l)$, where $\De\in\frac d 2+i\R$, $J\in-\frac{d-2}{2}+i\R$ and $\l$ is an irrep of $\SO(d-2)$. Here the pair $(J,\l)$ can be thought of as a weight of $\SO(d)$, where $J$ is the component corresponding to the length of the first row of a Young diagram. In this sense we have a continuous-spin generalization of $\SO(d)$ irreps.

To make sense of functions with continuous spin, we follow the logic described in the introduction. Let us first review the case of integer spin, and take $\l$ to be trivial for simplicity. The elements of integer spin representations are tensors that are traceless and symmetric in their indices
\be
\label{eq:indexf}
f^{\mu_1\cdots\mu_J}(x).
\ee
We can always contract $f$ with a null polarization vector $z^\mu$ to obtain a homogeneous polynomial of degree $J$ in $z$,
\be
\label{eq:polarizationf}
f(x,z) &\equiv f^{\mu_1\cdots\mu_J}(x)z_{\mu_1}\cdots z_{\mu_J}.
\ee
The tensor $f^{\mu_1\cdots \mu_J}(x)$ can be recovered from $f(x,z)$ via
\be
\label{eq:recovertensor}
f_{\mu_1\cdots\mu_J}(x)= \frac{1}{J!(\frac{d-2}{2})_J}D_{\mu_1}\cdots D_{\mu_J}f(x,z),
\ee
where
\be
D_\mu=\p{\frac{d-2}{2}+z\.\pdr{}{z}}\pdr{}{z^\mu} - \frac 1 2 z_\mu \frac{\ptl^2}{\ptl z^2}
\ee
is the Thomas/Todorov operator \cite{Thomas1926,Dobrev:1975ru,bailey1994}. Thus, the two ways (\ref{eq:indexf}) and (\ref{eq:polarizationf}) of representing $f$  are equivalent.

The generalization to continuous spin is now as stated in the introduction: we can consider functions $f(x,z)$ that are homogeneous of degree $J$ in $z$, where $J$ is no longer an integer and $f(x,z)$ is no longer a polynomial in $z$. More precisely, the elements of $\cP_{\De,J}$ are functions $f(x,z)$ with $x\in \cM_{d}^c$ and $z\in \R^{d-1,1}_+$ a future-pointing null vector that are constrained to satisfy
\be\label{eq:conthomogeneity}
	f(x,\a z)=\a^J f(x,z),\quad \alpha>0.
\ee
The object $f(x,z)$ transforms under conformal transformations in the same way as functions of the form (\ref{eq:polarizationf}) would. The operation of recovering the underlying tensor (\ref{eq:recovertensor}) only makes sense when $J$ is a nonnegative integer.\footnote{Also, $f(x,z)$ should satisfy a differential equation in $z$. This differential equation is conformally invariant and is essentially a generalization of the $(d-2)$-dimensional conformal Killing equation, similarly to the equations discussed in~\cite{Karateev:2017jgd}. Such equations only exist for nonnegative integer $J$ and express the fact that $f(x,z)$ is actually polynomial in $z$.}

To describe representations $\cP_{\De,J,\l}$ with non-trivial $\l$, we can make use of an analogy between the space of polarization vectors $z$ and the embedding space. The embedding space lets us lift functions on $\R^d$ with indices for an $\SO(d)$ representation to functions on the null cone in $d+2$ dimensions with indices for an $\SO(d+1,1)$ representation. In the present case, $\l$ is a representation of $\SO(d-2)$, so we can lift it to a representation of $\SO(d-1,1)$ defined on the null cone $z^2=0$ in a similar way. For example, if $\l$ is a rank-$k$ tensor representation of $\SO(d-2)$, then we consider functions
\be
	f^{a_1\ldots a_k}(x,z),
\ee
with $a_i$ being $\SO(d-1,1)$-indices, where $f$ obeys gauge redundancies and transverseness constraints~\cite{Costa:2011mg}
\be
	f^{a_1\ldots a_k}(x,z)&\sim f^{a_1\ldots a_k}(x,z)+z^{a_i}h^{a_1\ldots a_{i-1}a_{i+1}\ldots a_k}(x,z),\label{eq:contspingaugeinvariance}\\
	z_{a_i}f^{a_1\ldots a_k}(x,z)&=0.
\ee
Additionally, $f$ should be homogeneous (\ref{eq:conthomogeneity}) and satisfy the same tracelessness and symmetry conditions in $a_i$ as $\l$-tensors of $\SO(d-2)$.\footnote{To make more direct contact with integer spin, instead of~\eqref{eq:contspingaugeinvariance} one can use 
\be
	D_{a_i}f^{a_1\ldots a_k}(x,z)&=0,
\ee
where $D$ is the Todorov operator acting on $z$. In this case, for integer spin tensors the function $f^{a_1\ldots a_k}(x,z)$ is given simply by contracting $z_\mu$ with the first-row indices of the tensor.} Other types of representations can be described by adapting other embedding space formalisms. In most of this paper we focus on trivial $\l$ for simplicity.

We can define an inner product for Lorentzian principal series representations by
\be
	(f,g)&\equiv \int d^dx D^{d-2} z f^*(x,z) g(x,z),\\
	D^{d-2} z&\equiv \frac{2 d^d z \theta(z^0) \delta(z^2)}{\mathrm{vol}\, \R_+}.\label{eq:zmeasure}
\ee
Here the integral over $z$ replaces the index contraction that we would use for integer $J$. The measure for $z$ is manifestly Lorentz-invariant and supported on the null cone. Together with the measure, the integrand is invariant under rescaling of $z$. Thus, we obtain a finite result by dividing by the volume of the group of positive rescalings, $\vol\, \R_+$. The $z$-integral is exactly the kind of integral considered in~\cite{SimmonsDuffin:2012uy} in the context of the embedding space formalism. Here, we have adapted it to describe $\SO(d-1,1)$-invariant integration on the null cone $z^2=0$. 

In section~\ref{sec:weylandintegral} we will use analytic continuations of $\cP_{\De,J,\l}$ to find interesting relations for primary operators in Lorentzian CFTs. But before we can do this, we should note that these representations are constructed on $\cM^c_d$, which is unsatisfactory from the physical point of view. We can construct similar representations of $\widetilde\SO(d,2)$ consisting of functions on $\tl\cM_d$, which we call $\widetilde\cP_{\De,J,\l}$. These representations behave very similarly to $\cP_{\De,J,\l}$ but there is an important distinction. While the representations $\cP_{\De,J,\l}$ are generically irreducible, their analogues $\widetilde\cP_{\De,J,\l}$ are not. Indeed, the action of $\tsym$ on $\widetilde\cM_d$ commutes with the action of $\widetilde\SO(d,2)$ and thus $\widetilde\cP_{\De,J,\l}$ decompose into a direct integral of irreducible subrepresentations in which $\tsym$ acts by a constant phase. 

\subsection{Weyl reflections and integral transforms}
\label{sec:weylandintegral}

Given the principal series representations described in section~\ref{sec:reptheoryreview}, we can ask whether there exist equivalences between them. Equivalent representations must have the same eigenvalues of the Casimir operators,\footnote{Here we mean all Casimir operators, not just the quadratic Casimir.} and these eigenvalues are polynomials in the weights $(\De,\rho)$ (for $\SO(d+1,1)$) and $(\De,J,\l)$ (for $\SO(d,2)$). For example, the quadratic and quartic Casimir eigenvalues for $\cP_{\De,J}$ (with trivial $\l$) are
\be
\label{eq:examplecas}
C_2(\cP_{\De,J}) &= \De(\De-d) + J(J+d-2), \nn\\
C_4(\cP_{\De,J}) &= (\De-1)(d-\De-1)J(2-d-J).
\ee
The ``restricted Weyl group" $W'$ is a finite group that acts on these weights, doesn't mix discrete and continuous labels, and leaves the Casimir eigenvalues invariant. Conversely, if two principal series weights have the same Casimirs, they can be related by an element of $W'$.

For example, in the case of $\SO(d+1,1)$, the restricted Weyl group is $W'=\Z_2$. Its non-trivial element $\mathrm{S}_E\in W'$ acts by
\be
	\mathrm{S}_E(\De,\rho)=(d-\De,\rho^R),
\ee 
where $\rho^R$ is the reflection of $\rho$. Other transformations exist that leave all Casimir eigenvalues invariant, but $\mathrm{S}_E$ is the only one that does not mix the integral weights of $\rho$ with the continuous weight $\De$.

In the case of $\SO(d,2)$, there are two continuous parameters that can mix, and thus the restricted Weyl group $W'$ is larger. It is isomorphic to a dihedral group of order $8$, $W'=\mathrm{D}_8$.\footnote{This also turns out to be the Weyl group of $BC_2$ root system, which was recently studied in the context of conformal blocks in~\cite{Isachenkov:2016gim,Isachenkov:2017qgn}. It would be interesting to better understand the connection of the present discussion with that work.} This group has a faithful representation on $\R^2$ where it acts as symmetries of the square. Its action on $\De=\frac d 2 + is$ and $J=-\frac{d-2}{2}+iq$ can be described by taking $s$ and $q$ to be Cartesian coordinates in this $\R^2$. It is easy to see that this action preserves the eigenvalues (\ref{eq:examplecas}). Altogether, the elements of $W'$ are given in table~\ref{tab:wprime}.\footnote{To check that the action on $\l$ is as in the table, one can consider the 4d case. The eigenvalues of all 3 Casimirs of $\tl\SO(2,4)$ are written out, for example, in appendix F of~\cite{Cuomo:2017wme} with $\ell=J+\l, \bar\ell=J-\l$ and $\l^R=-\l$. More generally, by solving the system of polynomial equations expressing invariance of these explicit Casimir eigenvalues, one can check that $W'$ is indeed isomorphic to $\mathrm{D}_8$.}

\begin{table}[ht!]
\begin{center}
	\begin{tabular}{ll|c|cccc}
		$w$ && order & $\De'$ & $J'$ & $\l'$ \\
		\hline
		1 && 1 & $\De$ & $J$ & $\l$ \\
		$\mathrm{S}_\De$\!\!\!\!\!\! &= $\mathrm{L} \mathrm{S}_J \mathrm{L}$ & 2& $d-\De$ & $J$ & $\l^R$ \\
		$\mathrm{S}_J$\!\!\!\!\!\! &&2& $\De$ & $2-d-J$ & $\l^R$ \\
		$\mathrm{S}$\!\!\!\!\!\! &= $(\mathrm{S}_J\mathrm{L})^2$ & 2& $d-\De$ & $2-d-J$ & $\l$ \\
		$\mathrm{L}$\!\!\!\!\!\! && 2 &$1-J$ & $1-\De$ & $\l$ \\
		$\mathrm{F}$\!\!\!\!\!\! &= $\mathrm{S}_J \mathrm{L} \mathrm{S}_J$ &2& $J+d-1$ & $\De-d+1$ & $\l$\\
		$\mathrm{R}$\!\!\!\!\!\! &= $\mathrm{S}_J \mathrm{L}$ &4& $1-J$ & $\De-d+1$ & $\l^R$\\
		$\mathrm{\bar R}$\!\!\!\!\!\! &= $\mathrm{L} \mathrm{S}_J$ &4& $J+d-1$ & $1-\De$ & $\l^R$
	\end{tabular}
\end{center}
\caption{
\label{tab:wprime}
The elements of the restricted Weyl group $W'=\mathrm{D}_8$ of $\SO(d,2)$. Each element $w$ takes the weights $(\De,J,\l)$ to $(\De',J',\l')$. The order $2$ elements other than $\mathrm{S}$ are the four reflection symmetries of the rectangle, while $\mathrm{S}$ is the rotation by $\pi$. The center of the group is $Z\mathrm{D}_8=\{1,\mathrm{S}\}$. Finally, the element $\mathrm{R}$ is a $\pi/2$ rotation.  The group is generated by $\mathrm{L}$ and $\mathrm{S}_J$, with the relations $\mathrm{L}^2=\mathrm{S}_J^2=(\mathrm{L}\mathrm{S}_J)^4=1$.}
\end{table}

As mentioned above, the representations defined by weights in an orbit of $W'$ have equal Casimir eigenvalues, which means that potentially they can be equivalent. This indeed turns out to be true~\cite{KnappStein1,KnappStein2}.
Equivalence of representations means that there exist intertwining maps between $\cE_{(\De,\rho)}$ and $\cE_{w(\De,\rho)}$, as well as between $\cP_{(\De,J,\l)}$ and $\cP_{w(\De,J,\l)}$ for all $w\in W'$.

The intertwining map between $\SO(d+1,1)$ representations $\cE_{\De,\rho}$ and $\cE_{d-\De,\rho^R}$ is well-known~\cite{Ferrara:1972uq,Dobrev:1977qv,SimmonsDuffin:2012uy}: it is given by the so-called shadow transform
\be\label{eq:euclidshadow}
	\widetilde\cO^a(x)=\wSE[\cO]^a(x')\equiv\int d^d x' \<\widetilde\cO^a(x)\widetilde\cO^\dagger_b(x')\> \cO^b(x').
\ee
Here $\widetilde\cO\in \cE_{d-\De,\rho^R}$, $\cO\in \cE_{\De,\rho}$, we use dagger to denote taking the dual reflected representation of $\SO(d)$, and $\<\widetilde\cO^a(x)\widetilde\cO^\dagger_b(x')\>$ is a standard choice of two-point function for the operators in their respective representations. 
The integration region is the full $\R^d$ (more precisely, the conformal sphere $S^d$).

According to our discussion above, in Lorentzian signature there should exist $6$ new integral transforms, corresponding to the other non-trivial elements of $W'$. There in fact exists a general formula for these transforms, valid for any element of $W'$~\cite{KnappStein1,KnappStein2}.\footnote{In the mathematical literature, these transforms are known as Knapp-Stein intertwining operators.} However, it is most naturally written using a different construction of $\cP_{\De,J,\l}$, and the conversion to the form appropriate for our purposes is cumbersome.\footnote{See~\cite{Dobrev:1977qv} for an example of this conversion in the case of the shadow transform~\eqref{eq:euclidshadow}.} Thus instead of deriving these transforms from the general result we will simply give the final expressions and check that they are indeed conformally-invariant. Furthermore, we will lift these transforms to representations $\widetilde\cP_{\De,J,\l}$ of $\widetilde\SO(d,2)$.

Although the Lorentzian transforms we define are only necessarily isomorphisms when acting on principal series representations $\cP_{\De,J,\l}$, it is still interesting to consider the analytic continuation of their action on other representations, like those associated to physical CFT operators. For example the action of $\mathrm{L}$ will be well-defined on physical local operators. The result of this action will generically be a primary operator with non-integer spin. One can then ask how such operators make sense in a CFT and what properties do they have. In this and the following sections we will be able to answer these question by studying the examples provided by integral transforms. In appendix~\ref{sec:contspincorrelators} we study the same questions on more general grounds (by using unitarity, positivity of energy, and conformal symmetry) and reach similar conclusions.

\subsubsection{Transforms for $\mathrm{S}_\De,\mathrm{S}_J,\mathrm{S}$}

Let us start with the Lorentzian analogue of~\eqref{eq:euclidshadow}. The idea is to essentially keep the form~\eqref{eq:euclidshadow} while generalizing to continuous spin,
\be\label{eq:lorentzshadow}
	\wSD[\cO](x,z)&\equiv i\int_{x'\approx x} d^d x' \frac{1}{(x-x')^{2(d-\De)}} \cO(x',I(x-x')z),\\
	I^\mu_\nu(x)&=\delta^\mu_\nu - 2\frac{x^\mu x_\nu}{x^2}.
\ee
The integrand is conformally-invariant because $I(x-x')$ performs a conformally-invariant translation of a vector at $x$ to a vector at $x'$. The factor of $i$ is to match a Wick-rotated version of the Euclidean shadow transform, although we still have $\wSE=(-2)^J\wSD$ after Wick rotation because of our convention for two-point functions~(\ref{eq:standardtwoptconvention}).

We must specify a conformally-invariant integration region for $x'$. The essentially unique choice is to integrate over the region spacelike separated from $x$. If $x$ is at spatial infinity of $\cM_d$, then this region is the full Poincare patch $\cM_d\subset \widetilde\cM_d$, and for integer $J$ the integral is simply the Wick rotation of the Euclidean shadow integral~\eqref{eq:euclidshadow}. If, however, $x$ is inside the first Poincare patch, then the integral extends beyond the first Poincare patch on the Lorentzian cylinder $\widetilde{\cM}_d$. All other conformally-invariant regions defined by $x$ are translations of the spacelike region by powers of $\tsym$ or unions thereof. The two-point function in these regions differs from the two-point function in the spacelike region only by a constant phase, and thus the most general choice of $\wSD$ differs from the above by multiplication by a function of $\tsym$.\footnote{In particular, there is no ambiguity in representations $\cP_{\De,J,\l}$ of $\SO(d,2)$.}  The possibility of multiplying by a function of $\tsym$ is present for all the transforms we consider and we just make the simplest choice. The choice~\eqref{eq:lorentzshadow} is natural because of its relation to~\eqref{eq:euclidshadow}. 

For $\mathrm{S}_J$, the integral transform is
\be
	\wSJ[\cO](x,z)&\equiv \int D^{d-2} z' (-2\,z\cdot z')^{2-d-J} \cO(x,z'),
\ee
where the measure $D^{d-2}z$ is defined in~\eqref{eq:zmeasure}. We call this the ``spin shadow transform."
Note that this is essentially the same as the shadow transform in the embedding space~\cite{SimmonsDuffin:2012uy}, with $X$ replaced by $z$ and $d$ replaced by $d-2$.

The transform for $\mathrm{S}$, which we call the ``full shadow transform," is simply the composition of the commuting transforms for $\mathrm{S}_\De$ and $\mathrm{S}_J$,
\be
	\wS[\cO](x,z)&\equiv (\wSJ \wSD)[\cO](x,z)= i\int_{x'\approx x} d^d x' D^{d-2} z' \frac{(-2\,z\cdot z')^{2-d-J}}{(x-x')^{2(d-\De)}} \cO(x',I(x-x')z')\nn\\
	&= (\wSD\wSJ)[\cO](x,z)=i\int_{x'\approx x} d^d x' D^{d-2} z' \frac{(-2\,z\cdot I(x-x')z')^{2-d-J}}{(x-x')^{2(d-\De)}} \cO(x',z').
	\label{eq:integralforfullshadow}
\ee
These two forms of $\wS$ are equivalent because $I(x-x')^2=1$, for spacelike $x-x'$ $I(x-x')$ is an element of the orthochronous Lorentz group $\mathrm{O}^+(d-1,1)$, and the measure of the $z$-integration is invariant under $\mathrm{O}^+(d-1,1)$.

The second line of (\ref{eq:integralforfullshadow}) can also be written as
\be
\label{eq:fullshadowastwoptintegral}
\wS[\cO](x,z) &= i \int_{x'\approx x} d^d x' D^{d-2} z' \<\cO^\mathrm{S}(x,z) \cO^\mathrm{S}(x',z')\> \cO(x',z')
\ee
where $\cO^\mathrm{S}$ denotes the representation with dimension $d-\De$ and spin $2-d-J$. Here, we are using the following convention for a two-point structure
\be
\<\cO(x_1,z_1)\cO(x_2,z_2)\> &= \frac{(-2 z_1\.I(x_{12})\.z_2)^J}{x_{12}^{2\De}},
\ee
which differs by a factor of $(-2)^J$ from some more traditional conventions. Our conventions for two- and three-point structures are summarized in appendix~\ref{app:23conventions}

\subsubsection{Transform for $\mathrm{L}$}
\label{sec:lighttransformintro}
 The integral transform corresponding to $\mathrm{L}$ is
\be\label{eq:lightdefinition}
	\wL[\cO](x,z)=\int_{-\infty}^{+\infty} d\a\, (-\a)^{-\De-J} \cO\left(x-\frac{z}{\a},z\right).
\ee
Because it involves integration along a null direction, we call $\wL$ the ``light transform." Although most of the transforms in this section are only well-defined on nonphysical representations like Lorentzian principal series representations, the light transform is significant because it can be applied to physical operators as well.  Note that it converges near $\a=\pm\infty$ only for $\De+J>1$.\footnote{For Lorentzian principal series $\Re(\De+J)=1$ but for non-zero $\Im(\De+J)$ the integral still makes sense.} In unitary theories it can therefore be applied to all non-scalar operators and to scalars with dimension $\De>1$ (which includes all non-trivial scalars in $d\geq 4$).

Before discussing conformal invariance, let us describe the contour of integration in more detail. The integral starts at $\alpha=-\infty$, in which case the argument of $\cO$ is simply $x$. It then increases to $\alpha=-0$, and in the process $\cO$ moves along $z$ to future null infinity in $\cM_d$. As $\alpha$ crosses $0$, the integration contour leaves the first Poincare patch $\cM_d$ and enters the second Poincare patch $\tsym\cM_d\subset \widetilde\cM_d$. Finally, at $\alpha=+\infty$ it ends at $\tsym x\in \tsym\cM_d$. In other words, the integration contour is a null geodesic in $\widetilde{\cM}_d$ from $x$ to $\tsym x$ with direction defined by $z$ (figure~\ref{fig:lighttransformcontour}). This is obviously a conformally-invariant contour.

\begin{figure}[ht!]
	\centering
		\begin{tikzpicture}
		
		\draw[] (-3,0) -- (0,3) -- (3,0) -- (0,-3) -- cycle;
		\draw[] (0,3) -- (0.3,3.3);
		\draw[] (0,3) -- (-0.3,3.3);
		\draw[] (0,-3) -- (0.3,-3.3);
		\draw[] (0,-3) -- (-0.3,-3.3);
		\draw[color=blue] (0,1.2) -- (1.5,2.7);				
		\draw[->,color=blue] (-1.5,-0.3) -- (0,1.2);				
				
		\draw[fill=black] (-1.5,-0.3) circle (0.07); 
		\draw[fill=black] (1.5,2.7) circle (0.07); 

        \draw[dashed] (-3,-3.3) -- (-3,3.3);				
        \draw[dashed] (3,-3.3) -- (3,3.3);				
		
		\node[right] at (-1.4,-0.3) {$x$};
		\node[right] at (1.6,2.8) {$\tsym x$};
		
		\end{tikzpicture}
		\caption{The contour prescription for the light-transform. The contour starts at $x\in \cM_d$ and moves along the $z$ direction to the point $x^+=\tsym x$ in the next Poincare patch $\tsym \cM_d$.}
		\label{fig:lighttransformcontour}
\end{figure}
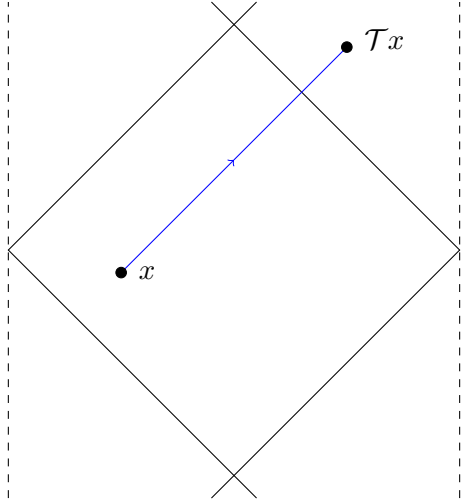

It turns out that no phase prescription is necessary to define $(-\a)^{-\De-J}$ for $\a>0$, because the naive singularity at $\a=0$ is cancelled in correlators of $\cO$. To see this, note that (\ref{eq:lightdefinition}) is equivalent to the following integral in the embedding formalism of~\cite{Costa:2011mg},
\be\label{eq:lightdefinitionEF}
	\wL[\cO](X,Z) &=\int_{-\infty}^{+\infty} d\a\, (-\a)^{-\De-J} \cO\left(X-\frac{Z}{\a},Z\right)
	\nn\\
	&=\int_{-\infty}^{+\infty} d\a\, \cO(Z-\a X,-X),
\ee
where in the second equality we used the homogeneity properties of $\cO(X,Z)$ in the region $\alpha<0$, together with  gauge invariance $\cO(X,Z+\b X)=\cO(X,Z)$. In (\ref{eq:lightdefinitionEF}) it is clear that the point $\a=0$ is not special (see also appendix~\ref{app:squarelight} for yet another explanation). 

The embedding space integral (\ref{eq:lightdefinitionEF}) makes conformal invariance of the light-transform manifest: it is $\SO(d,2)$ invariant, and gauge invariance  
\be
	\wL[\cO](X,Z+\beta X)=\wL[\cO](X,Z)
\ee
can be proved by shifting $\alpha$ by $\beta$ in the integral. It is also clear from homogeneity in $X$ and $Z$ that the dimension and spin of $\wL[\cO](X,Z)$ are $1-J$ and $1-\De$, respectively. (Note that the parameter $\a$ carries homogeneity $1$ in $Z$ and $-1$ in $X$.) Finally, (\ref{eq:lightdefinitionEF}) confirms the prescription that the integral goes between $x$ and $\tsym x$. Indeed, according to the discussion in section~\ref{sec:lorentziancylinder} the embedding space covers two Poincare patches and $\tsym X$ is simply $-X$. The integral in~\eqref{eq:lightdefinitionEF} starts at the argument $Z+\infty X$ which is the same as $X$ modulo $\R_+$  and ends at $Z-\infty X$ which is $-X=\tsym X$ modulo $\R_+$ .

Let us describe another way of writing $\wL$ that will be useful.
Equation (\ref{eq:lightdefinition}) expresses $\wL$ in a conformal frame where $x$ is in the interior of a Poincare patch. In this case, the integration contour extends from one patch into the next. However, if we place $x$ at past null infinity, the integration contour fits entirely within a single Poincare patch. Specifically, in the integral (\ref{eq:lightdefinitionEF}), let us set\footnote{This choice reverses the role of $X,Z$ relative to the usual Poincare section gauge fixing. However, it still satisfies the required conditions $X^2=Z^2=X\.Z=0$. To obtain these expressions, consider the usual Poicare coordinates for a point shifted by $-L z$ for large $L$, 
	\be X&=(1,(x-Lz)^2, x-Lz)\simeq L\times (0,-2x\.z,-z),\nn\\
		Z&=(0,2z\.x,z)= L^{-1}\times\p{(1,x^2,x)-X},\nn
	\ee
from where the new gauge-fixing follows.}

\be
Z &= (1,y^2,y), \nn\\
X &= (0,-2y\.z,-z)
\ee
to obtain
\be
\label{eq:nullrayintegral}
\wL[\cO](x,z) &= \int_{-\oo}^\oo d\a\, \cO(y+\a z,z).
\ee
Here, $x=y-\oo z$.
Equation~(\ref{eq:nullrayintegral}) is simply the integral of $\cO$ along a null ray from past null infinity to future null infinity, contracted with a tangent vector to the ray. As an example, the ``average null energy" operator is given by
\be
\mathcal{E} &= \int_{-\oo}^\oo d\a T_{\mu\nu}(\a z) z^{\mu} z^{\nu} = \wL[T](-\oo z,z),
\ee
where $T_{\mu\nu}$ is the stress tensor.
It follows from our discussion that $\mathcal{E}$ transforms like a primary with dimension $-1$ and spin $1-d$, centered at $-\oo z$.

\subsubsection{Transforms for $\mathrm{F}, \mathrm{R}, \mathrm{\bar R}$}
The transforms for the remaining elements $\mathrm{F}, \mathrm{R}, \mathrm{\bar R} \in \mathrm{D}_8$ are compositions
\be
	\wF&\equiv \wSJ \wL\wSJ,\nn\\
	\wR &\equiv \wSJ \wL,\nn\\
	\wRb &\equiv \wL\wSJ.
\ee
For example, 
\be
	\wF[\cO](x,z)\equiv \int d^d \zeta D^{d-2} z'\delta(\zeta^2)\theta(\zeta^0)(-2\,\zeta\cdot z')^{-J-d+2}
	(-2\,\zeta\cdot z)^{\De-d+1} \cO(x+\zeta,z')\nn\\
	+\int d^d \zeta D^{d-2} z'\delta(\zeta^2)\theta(\zeta^0)(-2\,\zeta\cdot z')^{-J-d+2}
	(-2\,\zeta\cdot z)^{\De-d+1} (\tsym\cO)(x-\zeta,z').
\ee
Note that here the second term involves an integral over the second Poincare patch $\tsym\cM_d$. Similarly to the light transform, here we integrate over all future-directed null geodesics from $x$ to $\tsym x$. Because we integrate over all null directions, we call $\wF$ the ``floodlight transform." 

Similarly, we have
\be
	\wR[\cO](x,z)&=\int d^d \zeta \delta(\zeta^2)\theta(\zeta^0)(-2z\.\zeta)^{1-d+\Delta}\cO(x+\zeta,\zeta)\nn\\
	&\quad+\int d^d \zeta \delta(\zeta^2)\theta(\zeta^0)(-2z\.\zeta)^{1-d+\Delta}(\tsym\cO)(x-\zeta,\zeta),\\
	\wRb[\cO](x,z)&=\int d\a
	D^{d-2}z'(-\a)^{-\De-2+d+J} (-2\,z\.z')^{2-d-J}\cO\p{x-\frac{z}{\a},z'},
\ee

As an example, $\wR[T]=\wSJ[\wL[T]]$ is given by integrating the average null energy operator $\cE=\wL[T]$ over null directions. This is equivalent to integrating the stress tensor over a complete null surface, which produces a conformal charge. We can understand this more formally as follows. Note that the dimension and spin of $\wR[T]$ are given by
\be
\mathrm{R}(d,2) &= (-1,1).
\ee
These are exactly the weights of the adjoint representation of the conformal group. Conservation of $T^{\mu\nu}$ ensures that $\wR[T]$ transforms irreducibly, so that it transforms precisely in the adjoint representation. In other words, conservation equation for $T$ becomes the conformal Killing equation for $\wR[T]$. It can thus be written as a linear combination of conformal Killing vectors (CKVs):\footnote{See \cite{Karateev:2017jgd} for more discussion of writing finite-dimensional representations of the conformal group in terms of fields on spacetime.}
\be
\wR[T](x,z) &= Q^A w_A^\mu(x) z_\mu \nn\\
&= K\.z - 2(x\.z) D + (x_\rho z_\nu - x_\nu z_\rho)M^{\nu\rho} + 2(x\.z) (x\.P) - x^2 (z\.P).
\label{eq:adjointcharges}
\ee
Here, $A$ is an index for the adjoint representation of the conformal group, $w_A^\mu(x)$ are CKVs, and the $Q^A$ are the associated charges. On the second line, we've given the charges their usual names. We can see from (\ref{eq:adjointcharges}) that inserting $\wR[T]$ at spatial infinity $x=\oo$ gives the momentum charge. This is a familiar fact from ``conformal collider physics" \cite{Hofman:2008ar}. Similarly, when $J$ is a conserved spin-1 current, $\wR[J]$ has dimension-0 and spin-0, which are the correct quantum numbers for a conserved charge.

\subsection{Some properties of the light transform}
\label{sec:lightproperties}

As noted above, the light transform of the stress-energy tensor is the average null energy operator $\wL[T]=\cE$. The average null energy condition (ANEC) states that $\cE$ is non-negative,
\be
\label{eq:anec}
	\<\Psi|\cE|\Psi\>\geq 0.
\ee
Non-negative operators with vanishing vacuum expectation value $\<\Omega|\cE|\Omega\>=0$ must necessarily annihilate the vacuum $|\Omega\>$~\cite{Epstein:1965zza}.\footnote{We thank Clay C\'ordova for discussion on this point.}${}^,$\footnote{Intuitively, the vacuum must contain the same amount of positive-$\cE$ states and negative-$\cE$ states in order for $\<\O|\cE|\O\>$ to vanish. Since there are no negative-$\cE$ states, the vacuum only contains vanishing-$\cE$ states and is thus annihilated by $\cE$.} Indeed, using the Cauchy-Schwarz inequality for the inner product defined by $\cE$, we find
\be
	|\<\Psi|\cE|\Omega\>|^2 \leq \<\Psi|\cE|\Psi\>\<\Omega|\cE|\Omega\>=0
\ee
for any state $|\Psi\>$.
Thus $\cE|\Omega\>=0$.

In fact, we know that $\wL[\cO]|\Omega\>=0$ for any local primary operator $\cO$ --- not just the stress tensor. Indeed, if $\cO$ has scaling dimension $\De$, then $\wL[\cO]$ has spin $1-\De$, which in a unitary theory is a non-negative integer only if $\De=0$ or $\De=1$. However, in these cases $J=0$ and the light transform diverges. For all other scaling dimensions $\wL[\cO]$ is a continuous-spin operator and thus must annihilate the vacuum. This makes it possible for other null positivity conditions (like those proved in \cite{Hartman:2016lgu} and section~\ref{sec:positivityandtheanec}) to hold as well. In the rest of this subsection we check explicitly that $\wL[\cO]|\Omega\>=0$ for all $\De+J>1$ and make some general comments about properties of $\wL$.

\begin{lemma}\label{lemma:light}
	The light transform of a local primary operator, when it exists (i.e.\ $\De+J>1$), annihilates the vacuum,\footnote{For general spin representations $J$ must be replaced by the sum of all Dynkin labels with spinor labels taken with weight $\half$.\label{ft:generalJ}}
	\be
		\wL[\cO]|\O\>=0.
	\ee
\end{lemma}
\begin{proof}
	We will show that for any local operators $V_i$,
	\be
		\<\O|V_n(x_n)\cdots V_1(x_1) \wL[\cO](y,z)|\O\>=0,
	\ee
	which implies the result. Let us work in a Poincare patch where $y$ is at past null infinity and for simplicity assume that the $x_i$ fit in this patch; other configurations can be obtained by analytic continuation. Using a Lorentz transformation we can set $z=(1,1,0,\ldots,0)$ and parameterize the light transform contour as $x_0=\left(\frac{v-u}{2},\frac{v+u}{2},0,0,\ldots\right)$ for $v\in(-\oo,\oo)$.  We are then computing 
	\be\label{eq:lemmaeplsilonprescription}
		&\int_{-\oo}^{\oo} dv \<\O| V_n(x_n)\cdots V_1(x_1) \cO(x_0,z)|\O\>=\nn\\
		&=\lim_{\e\to +0}\int_{-\oo}^{\oo} dv \<\O| V_n(x_n-i n\e\hat e_0)\cdots V_1(x_1-i\e\hat e_0) \cO(x_0,z)|\O\>,
	\ee
	where $\hat e_0$ is the future-pointing unit vector in the time direction. The above $i\epsilon$ prescription arranges the operators so that they are time-ordered in Euclidean time, and this is precisely how the Wightman function should be defined as a distribution. Let us now write 
	\be
		x_k-i k\e\hat e_0=y_k+i \z_k,\quad k=0,1,\ldots n,
	\ee
	where both $y_k$ and $\z_k$ are real vectors. Positivity of energy implies that Wightman functions are analytic if $\z_k$ is in the absolute future of $\z_{k+1}$ for all $k$~\cite{streater2016pct}:\footnote{For example, it is easy to check that under this condition $(y_{ik}+i\z_{ik})^2\neq 0$ for all $y_{ik}$, and thus there are no obvious null cone singularities. More generally, see appendix~\ref{sec:contspincorrelators}.}
\be
\z_0 > \z_1 > \cdots > \z_n.
\ee
	 This condition clearly holds when the $x_k$ are real. If we then give an arbitrary positive imaginary part to $v$ while keeping $u$ and other components of $x_0$ fixed, $\z_0=\Im(v) z$ will remain in the future of $\z_1=-\e \hat e_0$ (see figure~\ref{fig:zetarelations}). Therefore, the integrand is an analytic function of $v$ in the upper half plane. If we can close the $v$ contour in the upper half plane, that would imply the required result.
	
	\begin{figure}[t]
		\centering
		\begin{tikzpicture}
		\draw[fill=black] (0,0) circle (0.05);
		\draw (0,0) -- (3,3);
		\draw (0,0) -- (-3,3);
		
		\draw[fill=black] (0,-1) circle (0.05);
		\draw (0,-1) -- (3,2);
		\draw (0,-1) -- (-3,2);
		
		\draw[dotted] (0, -1.2) -- (0, -1.7);
		\draw[dotted] (3, 3-1.2) -- (3, 3-1.7);
		\draw[dotted] (-3, 3-1.2) -- (-3, 3-1.7);
		
		\draw[->,color=blue] (0,1) -- (2,3);
		\draw[fill=black] (0,1) circle (0.05);
		\draw[<->, dashed] (0,0.05) -- (0,1-0.05);

		\node[below] at (0.3,1+0.1) {$\zeta_0$};
		\node[below] at (0.3,0+0.1) {$\zeta_1$};
		\node[below] at (0.3,-1+0.1) {$\zeta_2$};
		
		\node[left] at (0,0.5) {$\e$};
		
		\node[color=blue] at (0,2) {$\mathrm{Im}\, v>0$};
		\end{tikzpicture}
		\caption{Relationships between the imaginary parts $\zeta_k$. A deformation of $v$ in the positive imaginary direction is shown in blue.}
		\label{fig:zetarelations}
	\end{figure}
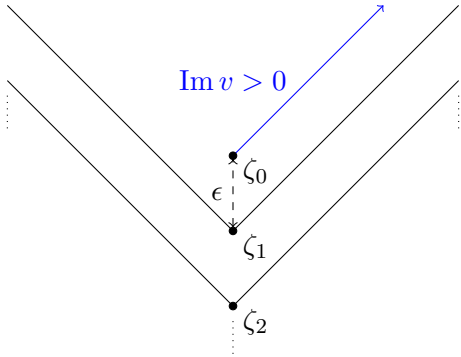
	
	According to the discussion around~\eqref{eq:lightdefinition}, conformal invariance implies that the integral~\eqref{eq:lightdefinition} is regular as $\alpha\to -0$, which in turn implies that the integrand of~\eqref{eq:lemmaeplsilonprescription} decays as $|v|^{-\De-J}$ for real $v$. We will now show that this is also true for complex $v$ in the upper half-plane, so we can close the contour as long as $\De+J>1$.
	
	To compute the rate of decay in $v$, we can use the OPE for the operators $V_i$, which converges acting on the left vacuum.\footnote{For this argument it is important that $i\e$-prescriptions and positive imaginary part of $v$ smear the operators so that we are working with normalizable states. An argument from the Euclidean OPE is that the $i\e$ shifts separate the operators on the Euclidean cylinder, and Lorentzian times do not affect  convergence of the OPE. The operators in the right hand side of the OPE can be placed anywhere in Euclidean future of $\cO$. Alternatively to (but not logically independently from) the OPE argument, we could have just started with $\<\O|\cO\wL[\cO]|\O\>$ in the first place, since states of the form $\int d^d x f(x)\<\O|\cO(x)$ are dense in the space of states which can have a non-zero overlap with $\wL[\cO]|\O\>$.} The leading contribution at large $v$ will be from $\cO$ in this OPE, leading to a two-point function of $\cO$. Because $v$ is moving in the direction of its polarization $z$, the decay of this two-point function is governed not by $\De$ but by $\De+J$. Indeed, we need to consider the two-point function
	\be
		\<\cO(0,z')\cO(u,v;z)\>.
	\ee
	The problem is then essentially two-dimensional: the statement that $v$ is along $z$ means that $\cO$ has definite left and right-moving weights of the 2d conformal subgroup. Invariance under the 2d conformal subgroup then selects the component of $z'$ with the same weights, so the two-point function is proportional to
	\be\label{eq:twoptexpectation}
	\<\cO(0,z')\cO(u,v;z)\>\propto \frac{(z'^1-z'^0)^J}{u^{\De-J}v^{\De+J}}.
	\ee
	Let us see this explicitly in the case of traceless-symmetric tensor $\cO$,
	\be
		\<\cO(0,z')\cO(u,v;z)\>\propto \frac{(z'_\mu I^{\mu\nu}(x_0) z_\nu)^J}{(uv)^\De},
	\ee
	where we have $x_0=\half v z+\half u z^\perp$. Here $z^\perp=(-1,1,0,\ldots)$ is the basis vector for the $u$ coordinate and we have $(z\cdot z^\perp)=2$. The numerator is then
	\be
		z'_\mu I^{\mu\nu}(x_0) z_\nu&=(z'\cdot z)-\frac{2u}{uv}\left(\half(z'\cdot z)v+\half(z'\cdot z^\perp)u\right)=(z'^1-z'^0)\frac{u}{v}.
	\ee
	This indeed leads to the expected form~\eqref{eq:twoptexpectation}.
	
In summary, we can close the $v$ contour in the upper half plane to give zero whenever $\De+J>1$.
\end{proof}

Recall that the condition $\De+J>1$ is true for all non-scalar operators in unitary CFTs, and for all non-identity scalar operators in $d\geq 4$ dimensions. 

As as simple corollary of lemma~\ref{lemma:light}, light transforms of local operators not acting on the vacuum can be expressed in terms of commutators. For example,
\be
\label{eq:commutatorexamples}
\<\O|\cO_1 \wL[\cO_3] \cO_2|\O\> &= \<\O|[\cO_1, \wL[\cO_3]] \cO_2|\O\> = \<\O|\cO_1[\wL[\cO_3], \cO_2]|\O\>.
\ee
Note that these commutators vanish at spacelike separations, so the integral in the light transforms only receives contributions from timelike separations. More explicitly, we can understand the commutators (\ref{eq:commutatorexamples}) as follows. In the integral
\be
\int_{-\oo}^\oo d\a (-\a)^{-\De-J}\<\O|\cO_1 \cO_3(x-z/\a,z) \cO_2|\O\>,
\ee
there is one singularity in the lower half-plane where $3$ becomes lightlike from $1$ and another in the upper half-plane where $3$ becomes lightlike from $2$ (figure~\ref{fig:alphacontours}). If we deform the contour to wrap around the first singularity ($3\sim 1$), we obtain the commutator $[\cO_1,\cO_3]$; if we deform the contour around the second singularity ($3\sim 2)$, we obtain $[\cO_3,\cO_2]$.

\begin{figure}[ht!]
	\centering
	\begin{tikzpicture}

\draw[fill=black] (0.98,0.3) circle (0.07);
\draw[fill=black] (-0.73,-0.3) circle (0.07);

\draw[] (-3.5,1.3) -- (-3.5,0.8) -- (-4,0.8);

\draw[->,line width=0.7] (-4,0) -- (0,0);
\draw[line width=0.7] (0,0) -- (4,0);

\draw[blue,->,line width=0.7] (-4,-0.15) -- (-2.35,-0.15);
\draw[blue,line width=0.7] (-2.35,-0.15) -- (-0.7,-0.15) to[out=0,in=0,distance=0.2cm] (-0.7,-0.45);
\draw[blue,->,line width=0.7] (-0.7,-0.45) -- (-2.35,-0.45);
\draw[blue,line width=0.7] (-2.35,-0.45) -- (-4,-0.45);
\draw[red,line width=0.7] (2.35,0.15) -- (4,0.15);
\draw[red,->,line width=0.7]  (0.96,0.45) to[out=180,in=180,distance=0.2cm] (0.96,0.15) -- (2.35,0.15);
\draw[red,line width=0.7] (0.96,0.45)-- (2.35,0.45);
\draw[red,->,line width=0.7]  (4,0.45)-- (2.35,0.45);

	\node[above] at (-3.75,0.8) {$\a$}; 
	\node[below] at (-0.73,-0.4) {$3\sim 1$};
	\node[above] at (0.96,0.4) {$3\sim 2$};
	\end{tikzpicture}
	\caption{Contour prescriptions for the $\a$ integral in the light transform of a three-point function (\ref{eq:commutatorexamples}). The black contour corresponds to $\<\Omega|\cO_1 \wL[\cO_3]\cO_2|\Omega\>$, the blue contour corresponds to $\<\O|[\cO_1, \wL[\cO_3]] \cO_2|\O\>$, and the red contour corresponds to $\<\O|\cO_1[\wL[\cO_3], \cO_2]|\O\>$.}
	\label{fig:alphacontours}
\end{figure}
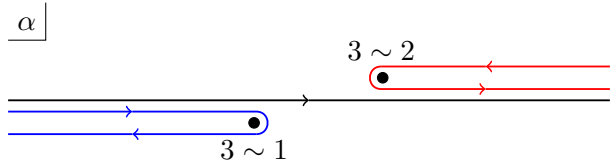

Lemma~\ref{lemma:light} has the following simple consequence for time-ordered correlators:
\begin{lemma}
\label{lemma:timeorderedlighttransform}
	Let $\cO$ be a local primary operator with $\De+J>1$. In a time-ordered correlator
	\be\label{eq:lemmatimeordered}
		\<V_1\ldots V_n \wL[\cO]\>_\O,
	\ee
	if the integration contour of $\wL[\cO]$ crosses only past or only future null cones, the transform is zero. Note that on the Lorentzian cylinder, generically, the contour crosses the null cone of each $V_i$ exactly once.
\end{lemma}
Note that here the notation~\eqref{eq:lemmatimeordered} means that $\wL$ is applied to a physical time-ordered correlation function, as opposed to time-ordering acting on the continuous spin operator $\wL[\cO]$. (Since continuous spin operators are necessarily non-local, it is unclear how to define the latter time-ordering in a Lorentz-invariant way, see appendix~\ref{sec:contspincorrelators}.) We also use the subscript $\O$ to stress that we mean a physical correlation function, as opposed to a conformally-invariant tensor structure.

Finally, let us note that if we use the usual Wightman $i\e$-prescription,\footnote{In other words, add small Euclidean times to the operators to make the expectation value time-ordered in Euclidean time.} the light transform of a Wightman function is an analytic function of its arguments, including the polarizations. This follows simply from the fact that it is an integral of an analytic function. This is consistent with our statements concerning analyticity of Wightman functions of continuous-spin operators in appendix~\ref{sec:contspincorrelators}.

\subsection{Light transform of a Wightman function}

As a concrete example, and because it will play an important role later, let us compute the light-transform of the Wightman function
\be
\label{eq:threeptfnexample}
	\<0|\phi_1(x_1)\cO(x_3,z)\phi_2(x_2)|0\>=\frac{\p{2z\.x_{23}\, x_{13}^2 - 2z\.x_{13}\, x_{23}^2}^J}{x_{12}^{\De_1+\De_2-\De+J} x_{13}^{\De_1+\De-\De_2+J} x_{23}^{\De_2+\De-\De_1+J}},
\ee
where $\phi_i$ are scalar operators with dimensions $\De_i$, and $\cO$ has dimension $\De$ and spin $J$. (Our three-point structure normalization differs by a factor of $2^J$ from some more conventional normalizations. Our conventions are summarized in appendix~\ref{app:23conventions}.) In the above expression, the Wightman $i\e$ prescription is implicit. As discussed at the end of the introduction, we use the convention that expectation values in the state $|\O\>$ denote physical correlation functions, whereas the expectation values in the state $|0\>$ denote two- or three-point tensor structures fixed by conformal invariance. The same comment applies to time-ordered correlation functions $\<\cdots\>_\O$ and $\<\cdots\>$ respectively.

Because the light-transform of a local operator annihilates the vacuum (lemma~\ref{lemma:light}), it is equivalent to the commutators
\be
\label{eq:lighttransformexample}
	\<0|\phi_1 \wL[\cO] \phi_2|0\>=\<0|\phi_1 [\wL[\cO],\phi_2]|0\>=\<0|[\phi_1,\wL[\cO]] \phi_2|0\>.
\ee
Specifically, let us compute the third expression above,
\be\label{eq:3ptlight}
	\<0|\big[\phi_1(x_1),\wL[\cO](x_3,z)\big] \phi_2(x_2)|0\>=\int_{-\infty}^{+\infty} d\alpha (-\alpha)^{-\De-J} 
	\<0|\left[\phi_1(x_1),\cO\left(x_3-\frac{z}{\alpha},z\right)\right] \phi_2(x_2)|0\>.
\ee

\begin{figure}[t]
	\centering
	\begin{tikzpicture}

	\draw[fill=black] (0,0) circle (0.05);
	\draw (-2,-2) -- (2,2);
	\draw (2,-2) -- (-2,2);

	\draw[fill=black] (2,0) circle (0.05);
	\draw (-2+2,-2) -- (2+2,2);
	\draw (2+2,-2) -- (-2+2,2);
	
	\draw[fill=black] (-.5,-1.5) circle (0.05);
	\draw[->,dashed] (-.5,-1.5) -- (-.5+3.5,-1.5+3.5);
	
	\draw[line width=2, color=blue] (-.5,-1.5) -- (-.5+1,-1.5+1);

	\node[above] at (0,0.1) {$1$}; 
	\node[above] at (2,0.1) {$2$}; 
	\node[left] at (-0.5,-1.5) {$3$}; 
	\end{tikzpicture}
	\caption{Causal relationships between points in the light transform (\ref{eq:lighttransformexample}). The original integration contour is the union of the solid blue line and the dashed line. The solid blue line shows the region where the commutator $[\f_1,\cO]$ is non-zero.}
	\label{fig:light3ptrelations}
\end{figure}
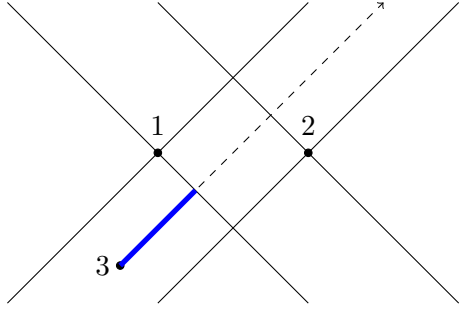

Since the light transform of a Wightman function is analytic (see section~\ref{sec:lightproperties} and appendix~\ref{sec:contspincorrelators}), we can compute it for any choice of causal relationships, and obtain the answer for other configurations by analytic continuation. We will work with the configuration in figure~\ref{fig:light3ptrelations}. All points lie in a single Poincare patch. The points $1$ and $2$ are spacelike separated, and the integration contour starts at $3<1$ and ends at $3^+>2$. The commutator $[\f_1,\cO]$ vanishes at spacelike separation, so the upper limit of the integral (\ref{eq:3ptlight}) gets restricted to the value of $\a$ when $3$ crosses the past null cone of $1$.

In our configuration, we have
\be
	(z\cdot x_{13}) &< 0,\label{eq:ineq1} \\ 
	-2\frac{(z\cdot x_{13})}{x_{13}^2} &< -2\frac{(z\cdot x_{23})}{x_{23}^2}.\label{eq:ineq2}
\ee
The first inequality follows because $z$ and $x_{13}$ are future-pointing and $x_{13}$ is not null. The second inequality expresses the fact that the null cone of $1$ is crossed before the null cone of $2$.

Taking into account that $x_{13}^2=e^{i\pi}|x_{13}^2|$ for the ordering $\f_1\cO$ and $x_{13}^2=e^{-i\pi}|x_{13}^2|$ for the ordering $\cO\f_1$, and restricting the range of integration to the past lightcone of $1$, we find 
\be
	&\<0|[\phi_1(x_1),\wL[\cO]](x_3,z) \phi_2(x_2)|0\>=\nn\\
	&=-2i\sin\pi\tfrac{\De_1+\De-\De_2+J}{2}\int_{-\infty}^{-\frac{2(z\cdot x_{13})}{x_{13}^2}} d\alpha (-\alpha)^{-\De-J} 
	\frac{\p{2z\.x_{23}\, x_{13}^2 - 2z\.x_{13}\, x_{23}^2}^J}{x_{12}^{\De_1+\De_2-\De+J} |x_{13'}|^{\De_1+\De-\De_2+J} x_{23'}^{\De_2+\De-\De_1+J}},
\ee
where $x_3'=x_3-z/\alpha$. Note that the factor $(\ldots)^J$ in the numerator is independent of $\alpha$ because $z$ is null. We thus need to compute
\be
	&\int_{-\infty}^{-\frac{2(z\cdot x_{13})}{x_{13}^2}} d\alpha (-\alpha)^{-\De-J} 
	\frac{1}{|x_{13'}|^{\De_1+\De-\De_2+J} x_{23'}^{\De_2+\De-\De_1+J}}\nn\\
	&=\int_{\frac{2(z\cdot x_{13})}{x_{13}^2}}^{+\infty} d\alpha 
	\frac{1}{|\alpha x_{13}^2-2(z\. x_{13})|^{\frac{\De_1+\De-\De_2+J}{2}} (\a x_{23}^2-2(z\. x_{23}))^{\frac{\De_2+\De-\De_1+J}{2}}}\nn\\	
	&=\frac{\Gamma(\De+J-1)\Gamma\left(1-\frac{\De+\De_1-\De_2+J}{2}\right)}{\Gamma\left(\frac{\De-\De_1+\De_2+J}{2}\right)}\frac{1}{|x_{13}|^{\De_1+\De-\De_2+J}x_{23}^{\De_2+\De-\De_1+J}}\left(\frac{2(z\cdot x_{13})}{x_{13}^2} -\frac{2(z\. x_{23})}{x_{23}^2}\right)^{1-\De-J}.
\ee
By~\eqref{eq:ineq1}, $\alpha$ has constant sign, which allows us to go to the second line. Because of~\eqref{eq:ineq2}, the function of $z$ which enters $(\ldots)^{1-\De-J}$ is positive, so the result is well-defined. 

Putting everything together, we find
\be
\label{eq:finalresultlighttransform}
&\<0|\phi_1(x_1)\,\wL[\cO](x_3,z)\, \phi_2(x_2)|0\> \nn\\
	&=L(\f_1\f_2[\cO])
	\frac{
		\p{2z\.x_{23}\, x_{13}^2 - 2z\.x_{13}\, x_{23}^2}^{1-\De}
	}{(x_{12}^2)^{\frac{\De_1+\De_2-(1-J)+(1-\De)}{2}}(-x_{13}^2)^{\frac{\De_1+(1-J)-\De_2+(1-\De)}{2}}(x_{23}^2)^{\frac{\De_2+(1-J)-\De_1+(1-\De)}{2}}},
	\ee
where
	\be
	\label{eq:lighttransformcoefficient}
	L(\f_1\f_2[\cO])&\equiv
-2\pi i \frac{\Gamma(\De+J-1)}{\Gamma(\frac{\De+\De_1-\De_2+J}{2})\Gamma(\frac{\De-\De_1+\De_2+J}{2})}.
\ee
The result (\ref{eq:finalresultlighttransform}) indeed takes the form of a conformally-invariant correlation function of $\f_1$ and $\f_2$ with an operator of dimension $1-J$ and spin $1-\De$. Note how continuous spin structures arise in a natural way from the light transform. Note also that (\ref{eq:finalresultlighttransform}) is pure negative-imaginary in the configuration of figure~\ref{fig:light3ptrelations}, where all quantities in the denominator are real. This is related to Rindler positivity as we discuss in section~\ref{sec:rindlerpositivity}.

Although we did the computation in a specific configuration, we have expressed the result in terms of an analytic function of the positions. Because the result should be analytic, the resulting expression (\ref{eq:finalresultlighttransform}) is valid for any configuration.  The $i\e$-prescription in (\ref{eq:finalresultlighttransform}) is the same as for the original Wightman function. In particular, if we move $x_3$ back into a configuration where all the points are spacelike separated, we obtain a phase
\be
\label{eq:spacelikephase}
e^{i\pi\frac{\De_1+(1-J)-\De_2+(1-\De)}{2}}
\ee
coming from $-x_{13}^2$ becoming negative. This phase will play a role in section~\ref{sec:algebraoftransforms}.

\subsection{Light transform of a time-ordered correlator}

Finally, let us discuss the light-transform of a time-ordered correlator $\<\cO_1\cO_2\wL[\cO_3]\>$. By lemma (\ref{lemma:timeorderedlighttransform}), this is nonzero only if $2^- < 3 < 1$ (as in figure~\ref{fig:light3ptrelations}) or $1^- < 3 < 2$. In the first nonzero configuration $2^- < 3 <1$, the time-ordered correlator is equivalent to the Wightman function $\<0|\cO_1 \cO_3 \cO_2|0\>$ along the entire integration contour of the light transform. The other nonzero configuration differs by $1\leftrightarrow 2$. Thus, we have
\be
\label{eq:lighttransformtimeordered1}
\<\cO_1\cO_2 \wL[\cO_3]\> &= \<0|\cO_1 \wL[\cO_3] \cO_2|0\> \th(2^- < 3 < 1) + \<0|\cO_2 \wL[\cO_3] \cO_1|0\> \th(1^- < 3 < 2).
\ee
Note that here the standard Wightman functions $\<0|\cO_1\cO_3\cO_2|0\>$ and $\<0|\cO_2\cO_3\cO_1|0\>$ (on which the light transforms act) are related to each other by analytic continuation and not by merely by relabeling the operators in the standard tensor structures $\<0|\ldots|0\>$.

For example, consider the three-point structure (\ref{eq:threeptfnexample}), now assumed to have $i\e$ prescriptions appropriate for a time-ordered correlator.  From (\ref{eq:lighttransformtimeordered1}) and our computation for the Wightman function (\ref{eq:finalresultlighttransform}), the light-transform is
\be
&\<\f_1 \f_2 \wL[\cO](x_3,z)\> = 
L(\f_1\f_2[\cO]) \nn\\
&\quad \x\left[
	\frac{
		\p{2z\.x_{23}\, x_{13}^2 - 2z\.x_{13}\, x_{23}^2}^{1-\De}
	}{(x_{12}^2)^{\frac{\De_1+\De_2-(1-J)+(1-\De)}{2}}(-x_{13}^2)^{\frac{\De_1+(1-J)-\De_2+(1-\De)}{2}}(x_{23}^2)^{\frac{\De_2+(1-J)-\De_1+(1-\De)}{2}}}\th(2^- < 3 < 1)\right.\nn\\
&\quad \left.
\quad+\ 
\frac{
		(-1)^J\p{2z\.x_{13}\, x_{23}^2 - 2z\.x_{23}\, x_{13}^2}^{1-\De}
	}{(x_{12}^2)^{\frac{\De_1+\De_2-(1-J)+(1-\De)}{2}}(x_{13}^2)^{\frac{\De_1+(1-J)-\De_2+(1-\De)}{2}}(-x_{23}^2)^{\frac{\De_2+(1-J)-\De_1+(1-\De)}{2}}}\th(1^- < 3 < 2).
\right]
\label{eq:lighttransformtimeordered}
\ee
The factor of $(-1)^J$ in the second term comes from the fact that the original structure $\<\f_1\f_2\cO\>$ picks up $(-1)^J$ when we swap $1\leftrightarrow 2$.\footnote{As we explain in appendix~\ref{sec:contspincorrelators}, time-ordered correlators with continuous spin do not make sense, so we must assume $J$ is an integer in this computation. This means that the factor $(-1)^J$ is unambiguous. The light transform $\<\f_1 \f_2 \wL[\cO]\>$ still gives a sensible continuous-spin structure because the result (\ref{eq:lighttransformtimeordered}) is no longer a time-ordered correlator, e.g.\ it has $\theta$-functions.}

\subsection{Algebra of integral transforms}

\label{sec:algebraoftransforms}

The $\wL$-transformation in (\ref{eq:finalresultlighttransform}) has the curious property that $\wL^2$ is a nontrivial function of $\De_1,\De_2,\De$ and $J$, even though it originates from a Weyl reflection $(\De,J)\leftrightarrow(1-J,1-\De)$ that squares to $1$. Specifically, its square acting on a three-point Wightman function is given by
\be
	\<0|\phi_1(x_1)\,\wL^2[\cO](x_3,z)\, \phi_2(x_2)|0\>&=\a_{\De_1,\De_2,\De,J}
	\<0|\phi_1(x_1)\cO(x_3,z)\phi_2(x_2)|0\>,
\ee
where
\be
\a_{\De_1,\De_2,\De,J}
&=
 e^{i\pi\frac{\De_1+\De-\De_2+J}{2}}L(\f_1\f_2[\cO^\mathrm{L}])\x e^{i\pi\frac{\De_1+(1-J)-\De_2+(1-\De)}{2}}
 L(\f_1\f_2[\cO]) \nn\\
	&=
	\frac{\pi}{(\De+J-1)\sin\pi(\De+J)}(e^{i\pi(\De_1-\De_2)}-e^{i\pi(\De+J)})(e^{i\pi(\De_1-\De_2)}-e^{-i\pi(\De+J)}).	
	\label{eq:lightsquareaction}
\ee
The phases in the first line of (\ref{eq:lightsquareaction}) are from (\ref{eq:spacelikephase}).

Note that the square of the light transform {\it does\/} give back a three-point function of the same functional form as the original. However, the coefficient $\a_{\De_1,\De_2,\De,J}$ depends on $\De_1,\De_2$ in a non-trivial way that cannot be removed by redefining $\wL$ by some function of $\De,J$ alone. This is in contrast to the Euclidean shadow transform, which squares to a coefficient $\cN(\De,J)$ that is independent of the correlation function it acts on  (appendix~\ref{app:euclideanintegrals}).

This ``anomaly" in the group relation $\mathrm{L}^2 = 1$ occurs for the following reason. The group-theoretic origin of $\wL$ only guarantees that it squares to a multiple of the identity when acting on principal series representations $\cP_{\De,J}$ defined on the conformal compactification of Minkowski space $\cM_d^c$. However, here we are applying it to the space $\tl \cP_{\De,J}$ defined on the universal cover $\tl \cM_d$. The squared transformation $\wL^2$ still commutes with $\tl {\SO}(d,2)$, so it becomes a non-trivial automorphism of the representation $\widetilde\cP_{\De,J}$.

By Schur's lemma, nontrivial automorphisms can only occur in reducible representations. Indeed, as discussed in section~\ref{sec:reptheoryreview}, $\widetilde\cP_{\De,J}$ is reducible and its irreducible components are the eigenspaces of $\tsym$. Within these irreducible components $\wL^2$ must act by a constant, and thus we should have
\be
	\wL^2=f_\mathrm{L}(\De,J,\tsym).
\ee
Furthermore, note that $\wL^2[\cO](x,z)$ only depends on the values of $\cO$ between $x$ and $\tsym^2x$. This means that $f_\mathrm{L}(\De,J,\tsym)$ must be at most a quadratic polynomial in $\tsym$. Finally, because $\wL^2[\cO]$ vanishes when acting on the past or future vacuum, $f_\mathrm{L}(\De,J,\tsym)$ should have roots at the eigenvalues of $\tsym$ in $\cO|\O\>$ and $\<\O|\cO$ inside a correlation function,\footnote{Here we need the adjoint action as $\cO\to\tsym \cO\tsym^{-1}$, c.f.\ equation~\eqref{eq:tsymeigenvalue}.} which are $e^{\pm i\pi(\De+J)}$. In fact, as we show explicitly in appendix~\ref{app:squarelight},
\be\label{eq:lightsquare}
	\wL^2=f_\mathrm{L}(\De,J,\tsym)=\frac{\pi}{(\De+J-1)\sin\pi(\De+J)}(\tsym-e^{i\pi(\De+J)})(\tsym-e^{-i\pi(\De+J)}).
\ee
This immediately implies~\eqref{eq:lightsquareaction} because $e^{i\pi(\De_1-\De_2)}$ is the eigenvalue of $\tsym$ acting on $\cO$ in the Wightman function $\<0|\phi_1(x_1)\cO(x_3,z)\phi_2(x_2)|0\>$. To see this, write the action of $\tsym$ on $\cO$ as
\be
	\<0|\phi_1(x_1)\tsym\cO(x_3,z)\tsym^{-1}\phi_2(x_2)|0\>
\ee
and use~\eqref{eq:tsymeigenvalue}. 

In fact, we can also turn this reasoning around and use the relatively simple computation \eqref{eq:lightsquareaction} to fix the polynomial $f_\mathrm{L}(\De,J,\tsym)$ in general. This will be helpful in appendix~\ref{app:proof} where we will need the statement that for general Lorentz irreps $\rho$ the ratio
\be\label{eq:tsymdep}
	\frac{f_L(\De,\rho,\tsym)}{(\tsym-\gamma)(\tsym-\gamma^{-1})},
\ee
where $\gamma$ is the eigenvalue in~\eqref{eq:tsymeigenvalue} corresponding to $(\De,\rho)$, is independent of $\tsym$.

More generally, this reasoning implies that relations between restricted Weyl reflections $w\in \mathrm{D}_8$ also hold for the corresponding integral transforms, but only up to multiplication by polynomials in $\tsym$ with coefficients depending on $\De$ and $J$. In the remainder of this section we derive these modified relations between integral transforms.

First of all, some relations hold by construction given the definitions in section~\ref{sec:weylandintegral},
\be
	\wS &= \wSJ \wSD = \wSD \wSJ, \nn\\
	\wF &= \wSJ\wL \wSJ, \nn\\
	\wR &= \wSJ\wL, \nn\\
	\wRb &= \wL \wSJ.
\ee
Furthermore, we already know that (for simplicity, we consider only $\tl\cP_{\De,J,\l}$ with trivial $\l$)
\be
	\wL^2&=f_\mathrm{L}(\De,J,\tsym),\label{eq:wLsquare}\\
	\wSJ^2&=f_J(J).\label{eq:wSJsquare}
\ee
Here $f_\mathrm{L}$ is a quadratic polynomial in $\tsym$ defined in~\eqref{eq:lightsquare}, while $f_J(J)$ depends only on $J$ and is equal to the square of Euclidean shadow transform in $d-2$ dimensions:
\be
	f_J(J) = \frac{\pi^{d-1}}{(J+\frac{d-2}{2})\sin\pi(J+\frac{d}{2})}\frac{1}{\Gamma(-J)\Gamma(J+d-2)}.
\ee
That is, $f_J(J)=\cN(-J,0)$ in $d-2$ dimensions, where $\cN(\De,J)$ in $d$ dimensions is given in~\eqref{eq:plancherelexample}.
These equations allow us to compute
\be
	\wR \wRb &=f_\mathrm{L}(\De,2-d-J,\tsym)f_J(J),\\
	\wRb \wR&=f_\mathrm{L}(\De,J,\tsym)f_J(1-\De).
\ee

As we show in appendix~\ref{app:shadowLSLrelation}, there is another relation,
\be\label{eq:wsdwlwsjrelation}
	2\wSD=i\tsym^{-1}\, \wL \wSJ\wL.
\ee
Together with $\wS= \wSJ \wSD = \wSD \wSJ$ this implies
\be
	2\wS=i\tsym^{-1}\, \wR^2=i\tsym^{-1}\, \wRb^2,
\ee
and thus we find
\be
	4\wS^2=-\tsym^{-2} \wR^2 \wRb^2 =-\tsym^{-2} f_\mathrm{L}(\De,2-d-J,\tsym)f_J(J)f_\mathrm{L}(J+d-1,1-d+\De,\tsym)f_J(1-\De).
\ee
Due to $4\wS^2=-\tsym^{-2} (\wSJ \wL)^4=-\tsym^{-2} (\wL\wSJ)^4$, we also have
\be\label{eq:wSJLfourth}
	(\wL\wSJ)^4=(\wSJ \wL)^4=f_\mathrm{L}(\De,2-d-J,\tsym)f_J(J)f_\mathrm{L}(J+d-1,1-d+\De,\tsym)f_J(1-\De).
\ee
At this point it is obvious that $f_J$ and $f_\mathrm{L}$ completely determine the relations between all integral transforms, since $\mathrm{D}_8$ is generated by $\mathrm{L}$ and $\mathrm{S}_J$ modulo $\mathrm{L}^2=\mathrm{S}_J^2=(\mathrm{S}_J\mathrm{L})^4=1$ and we have already found the generalization of these relations to the integral transforms $\wL$ and $\wSJ$ in~\eqref{eq:wLsquare},~\eqref{eq:wSJsquare}, and~\eqref{eq:wSJLfourth}. 

A convenient way to summarize these results is by using normalized versions of $\wL$ and $\wSJ$. Specifically, we define
\be
	\hat\wL&\equiv \wL \frac{1}{\Gamma(\De+J-1)(\tsym-e^{i\pi(\De+J)})},\\
	\hat{\mathbf{S}}_J&\equiv \wSJ\frac{\Gamma(-J)}{\pi^{\frac{d-2}{2}}\Gamma(J+\frac{d-2}{2})},
\ee
where $\De$ and $J$ in there right hand side should be understood as operators reading off the dimension and spin of the functions they act upon. One can then check the following relations
\be
	\hat\wL^2=1,\qquad \hat{\mathbf{S}}^2_J=1,\qquad (\hat\wL \hat\wS_J)^4=(\hat\wS_J\hat\wL)^4=1.
\ee
These normalized transforms therefore generate the dihedral group $\mathrm{D}_8$ without any extra coefficients.  Note that $\hat\wL$ is very non-local because it has $\tsym$ in the denominator. In particular, by doing a Taylor expansion in $\tsym$ we see that it involves a sum over an infinite number of different Poincare patches. Thus, even though $\hat\wL$ satisfies a simpler algebra, we mostly prefer to work with $\wL$.

\section{Light-ray operators}
\label{sec:lightray}

In this section, we explain how to fuse a pair of local operators $\cO_1,\cO_2$ into a light-ray operator $\mathbb{O}_{i,J}$ which gives an analytic continuation in spin $J$ of the light-transform of local operators in the $\cO_1\x\cO_2$ OPE. This amounts to defining correlation functions
\be
	\<\O|V_1\ldots V_k \mathbb{O}_{i,J} V_{k+1}\ldots V_n|\O\>
\ee
in terms of those of $\cO_1$ and $\cO_2$,
\be
	\<\O|V_1\ldots V_k \cO_1\cO_2 V_{k+1}\ldots V_n|\O\>.
\ee
When $J$ is an integer, $\mathbb{O}_{i,J}$ is related to a local operator in the $\cO_1\cO_2$ OPE, and these correlation functions are linked by Euclidean harmonic analysis~\cite{Dobrev:1977qv}. Our strategy will be to start with this relation, rephrase it in Lorentzian signature, and then analytically continue in $J$. By the operator-state correspondence, it suffices to consider just two insertions $V_i$, and for simplicity we will also restrict to scalars $\cO_1=\f_1$ and $\cO_2=\f_2$. (The generalization to arbitrary spin of $\cO_1,\cO_2$ will be straightforward.)

\subsection{Euclidean partial waves}
\label{sec:euclideanpartialwaves}

Consider a Euclidean correlation function $\<\f_1 \f_2 V_3 V_4 \>_\Omega$, where the $V_3$ and $V_4$ are local operators of any spin (not necessarily primary) and $\f_1,\f_2$ are local primary scalars. By the Plancherel theorem for $\SO(d+1,1)$ (due to Harish-Chandra \cite{harish-chandra1970}), such a correlation function can be expanded in partial waves $P_{\De,J}$ that diagonalize the action of the conformal Casimirs acting simultaneously on points $1$ and $2$~\cite{Dobrev:1977qv,Fitzpatrick:2011dm},\footnote{For general spin operators we should also include contributions from a discrete series of partial waves.}${}^,$\footnote{In \cite{Fitzpatrick:2011dm}, the process of forming the Euclidean partial wave $P_{\De,J}$ is called ``conglomeration."}
\be
\label{eq:multipointcompleteness}
\<V_3 V_4 \f_1 \f_2\>_\Omega &= \sum_{J=0}^\oo \int_{\frac d 2}^{\frac d 2 + i\oo} \frac{d\De}{2\pi i} \mu(\De,J) \int d^d x P_{\De,J}^{\mu_1\cdots \mu_J} (x_3,x_4,x) \<\tl \cO^\dagger_{\mu_1\cdots \mu_J}(x) \f_1 \f_2\>.
\ee
Here, $\cO$ has spin $J$ and dimension $\De\in \frac d 2 + i\R^+$ on the principal series. The factor $\mu(\De,J)$ is the Plancherel measure (\ref{eq:plancherelexample}), which we have inserted in order to simplify later expressions. For traceless-symmetric $\cO$ there is no difference between representations $\tl\cO^\dagger$ and $\tl\cO$, but we will keep the daggers in what follows with the view towards the more general case.

Let us make two technical comments about the applicability of this formula. It follows directly from $L^2(G)$ harmonic analysis on $\SO(d+1,1)$ if $\De_1-\De_2$ is pure imaginary (possibly $0$) and $\<V_3 V_4 \f_1 \f_2\>_\Omega$ is square-integrable in the sense that
\be
	\int d^dx_1 d^dx_2 \,x_{12}^{-2d+4\Re \De_1} \<V_3 V_4 \f_1 \f_2\>_\Omega (\<V_3 V_4 \f_1 \f_2\>_\Omega)^* < \oo.
\ee
This is precisely the situation when the conformal Casimir operators acting on points 1 and 2 are self-adjoint and we can perform their spectral analysis.\footnote{The reason why it is important to have $\De_1-\De_2\in i\R$ is that the adjoint of a Casimir operator acts on functions with conjugate shadow scaling dimensions $\tl\De_i^*$. This is a different space of functions than the one $\<V_3 V_4 \f_1 \f_2\>_\Omega$ lives in unless $\tl\De_i^*=\De_i$, which is the case when $\De_i\in\frac{d}{2}+i\R$ are principal series representations. It furthermore turns out that only $\De_1-\De_2$ is important for the argument, since $\De_1+\De_2$ can be changed by multiplying $\<V_3 V_4 \f_1 \f_2\>_\Omega$ by a two-point function $x_{12}^{\delta}$ for some $\delta$, and such two-point functions cancel out in equations.} Neither of these conditions is satisfied by a typical correlator in a physically-relevant CFT. Lifting the restriction of square integrability is conceptually easy and is similar to the usual Fourier transform: non-square integrable correlation functions can be interpreted as distributions (of some kind) and their partial waves also become distributions.\footnote{The distributional contribution to the partial wave can be analyzed by subtracting a finite number of contributions of low dimensional operators to make the function better behaved. This analysis was essentially performed in~\cite{Caron-Huot:2017vep} and in generic cases amounts to a deformation of $\De$-contour in~\eqref{eq:multipointcompleteness}.} 

Relaxing the restriction $\De_1-\De_2\in i\R$, on the other hand, seems to be hard to do from first principles, since the Casimir operators are not self-adjoint anymore. We will thus not attempt to do this here and instead adopt the following pedestrian approach: we will imagine multiplying correlation functions by products of scalar two-point functions $x_{ij}^{\kappa\delta_{ij}}$ with $\kappa=1$ so that the scaling dimensions of external operators will formally become principal series (this will of course modify the conformal block decomposition of these functions).\footnote{Note that such two-point functions have the right Wightman analyticity properties, and thus do not spoil the analyticity of physical correlators which we use in the arguments below.} We perform harmonic analysis for these modified functions and then remove the auxiliary two-point functions by sending $\kappa\to 0$. For this to make sense we have to assume that the final expressions can be analytically continued to $\kappa = 0$.

With these comments in mind, we may proceed with~\eqref{eq:multipointcompleteness}. Using the bubble integral (\ref{eq:bubbleintegral}), we find that $P_{\De,J}$ is given by
\be
\label{eq:euclideanintegralforpartialwave}
&P^{\mu_1\cdots \mu_J}_{\De,J}(x_3,x_4,x) =\p{\<\f_1 \f_2\tl \cO^\dagger \>,\<\tl \f_1^\dagger \tl \f_2^\dagger \cO\>}^{-1}_E\int d^d x_1 d^d x_2 \<V_3 V_4 \f_1 \f_2\>_\Omega \<\tl \f_1^\dagger \tl \f_2^\dagger \cO^{\mu_1\cdots \mu_J}(x)\>,
\ee
where
\be
\p{\<\f_1 \f_2\tl \cO^\dagger \>,\<\tl \f_1^\dagger \tl \f_2^\dagger \cO\>}_E=\frac{2^{2J}\hat C_J(1)}{2^d\vol(\SO(d-1))}
\ee 
is the three-point pairing defined in appendix~\ref{eq:euclideanpairings}.
In anticipation of performing the light-transform, let us contract spin indices of $\cO$ with a null polarization vector $z^\mu$ to give
\be
\label{eq:euclideanintegralforpartialwavez}
&P_{\De,J}(x_3,x_4,x,z) =\p{\<\f_1 \f_2\tl \cO^\dagger \>,\<\tl \f_1^\dagger \tl \f_2^\dagger \cO\>}^{-1}_E\int d^d x_1 d^d x_2 \<V_3 V_4 \f_1 \f_2\>_\Omega \<\tl \f_1^\dagger \tl \f_2^\dagger \cO(x,z)\>,
\ee
where $\cO(x,z) = \cO^{\mu_1\cdots\mu_J}(x)z_{\mu_1}\cdots z_{\mu_J}$.

Physical correlation functions $\<V_3 V_4 \cO_*\>_\Omega$ of operators $\cO_*$ in the $\f_1 \x \f_2$ OPE are residues of the partial waves, 
\be
\label{eq:residueequation}
f_{12*} \<V_3 V_4 \cO_*(x,z)\>_\Omega &= -\Res_{\De= \De_*} \left.\mu(\De,J) S_E(\f_1\f_2 [\tl \cO^\dagger]) P_{\De,J}(x_3,x_4,x,z)\right|_{J=J_*}.
\ee
Here, $S_E(\f_1\f_2 [\tl \cO^\dagger])$ is the shadow transform coefficient (\ref{eq:scalarshadowfactor}), and $f_{12*}$ is the OPE coefficient of $\cO_*\in \f_1\x \f_2$.  Equation (\ref{eq:residueequation}) is a simple generalization of the standard result for primary four-point functions.
We derive it in appendix~\ref{sec:argumentforresidues}.

\subsection{Wick-rotation to Lorentzian signature}

To obtain the promised analytic continuation of $\wL[\cO]$, we need to first go to Lorentzian signature, and then apply the light transform.

We thus Wick-rotate all the operators $\f_1,\f_2,V_3,V_4,\cO$ to Lorentzian signature by setting
\be
\tau &= (i+\e) t,
\ee
where $\tau$ and $t$ are Euclidean and Lorentzian time, respectively. In more detail, we simultaneously rotate the time coordinates of each of the operators $\f_1,\f_2,V_3,V_4,\cO$. For the operators $V_3,V_4,\cO$, this means we analytically continue in the coordinates $x_3,x_4,x$. The operators $\f_1,\f_2$ are being integrated over in~\eqref{eq:euclideanintegralforpartialwavez}, and we rotate their respective integration contours simultaneously with the analytic continuation of $x_3,x_4,x$. Simultaneous Wick-rotation turns Euclidean correlators into time-ordered Lorentzian correlators. The result is a double-integral of time-ordered correlators over Minkowski space
\be
\label{eq:afterwick}
&P_{\De,J}(x_3,x_4,x,z) = -\p{\<\f_1 \f_2\tl \cO^\dagger \>,\<\tl \f_1^\dagger \tl \f_2^\dagger \cO\>}^{-1}_E\int_{\oo \approx 1,2} d^d x_1 d^d x_2 \<V_3 V_4 \f_1 \f_2\>_\Omega \<\tl \f_1^\dagger \tl \f_2^\dagger \cO(x,z)\>.
\ee
Here, we have chosen a generic point $x_\oo$ on the Lorentzian cylinder $\tl \cM_d$ and written Minkowski space as the Poincare patch that is spacelike from this point.\footnote{In particular the result must be independent of which point we choose for $x_\oo$. The spurious dependence of formulas on $x_\oo$ will go away soon.}${}^,$\footnote{Note that we do not place $\cO(x,z)$ at infinity before performing the Wick rotation, in contrast to \cite{Simmons-Duffin:2017nub}. The reason is that in our case the region of integration for $1,2$ is independent of the position of $\cO$ so it is easier to analytically continue in the position of $\cO$.} All the points $1,2,3,4,x$ are constrained to lie within this patch. The minus sign in (\ref{eq:afterwick}) comes from two Wick rotations in the measure $d\tau_1 d\tau_2 = -dt_1 dt_2$.

\subsection{The light transform and analytic continuation in spin}

\begin{figure}[t]
	\centering
	\begin{subfigure}[t]{0.4\textwidth}
		\begin{tikzpicture}
		\draw[fill=yellow, opacity = 0.2] (-1.5,-1.5) -- (-3,0) -- (0,3) -- (1.5,1.5) -- cycle;
		\draw[fill=red, opacity = 0.1] (-1.5,-1.5) -- (0,-3) -- (3,0) -- (1.5,1.5) -- cycle;
		
		\draw (-3,0) -- (0,3) -- (3,0) -- (0,-3) -- cycle;
		
		\draw[fill=black] (-3,0) circle (0.07); 
		\draw[fill=black] (3,0) circle (0.07);  
		
		\draw[fill=black] (-1,0) circle (0.07); 
		\draw[fill=black] (1,0) circle (0.07);  
		
		\draw[fill=black] (-1.5,-1.5) circle (0.07);  
		\draw[fill=black] (1.5,1.5) circle (0.07);  
		
		\draw[opacity=.3] (-2,1) -- (1,-2); 
		\draw[opacity=.3] (-2,-1) -- (1,2); 
		
		\draw[opacity=.3] (2,1) -- (-1,-2); 
		\draw[opacity=.3] (2,-1) -- (-1,2); 
		
		\draw[dashed, color=blue] (-1.5,-1.5) -- (1.5,1.5);
		
		\node[right] at (3,0) {$\oo$};
		\node[left] at (-3,0) {$\oo$};
		
		\node[left] at (-1.05,0) {$4$};
		\node[right] at (1.05,0) {$3$};	
		
		\node[left] at (-1.55,-1.55) {$x$};
		\node[right] at (1.55,1.55) {$x^+$};
		
		\node at (-1,1) {$1$};
		\node at (1,-1) {$2$};
		\end{tikzpicture}
		\caption{Integration region in~\eqref{eq:bilocaldefinition}}
		\label{fig:nullrayoperator}
	\end{subfigure}
	\hspace{1.5cm}
	\begin{subfigure}[t]{0.4\textwidth}
		\begin{tikzpicture}
		\draw[fill=yellow, opacity = 0.2,yellow] (-1.5,-1.5) -- (-1.65,-1.35) -- (1.35,1.65) -- (1.5,1.5) -- cycle;
		\draw[fill=red, opacity = 0.1,red] (-1.5,-1.5) -- (-1.35,-1.65) -- (1.65,1.35) -- (1.5,1.5) -- cycle;
		
		\draw (-3,0) -- (0,3) -- (3,0) -- (0,-3) -- cycle;
		
		\draw[fill=black] (-3,0) circle (0.07); 
		\draw[fill=black] (3,0) circle (0.07);  
		
		\draw[fill=black] (-1,0) circle (0.07); 
		\draw[fill=black] (1,0) circle (0.07);  
		
		\draw[fill=black] (-1.5,-1.5) circle (0.07);  
		\draw[fill=black] (1.5,1.5) circle (0.07);  
		
		\draw[opacity=.3] (-2,1) -- (1,-2); 
		\draw[opacity=.3] (-2,-1) -- (1,2); 
		
		\draw[opacity=.3] (2,1) -- (-1,-2); 
		\draw[opacity=.3] (2,-1) -- (-1,2); 
		
		\draw[dashed, color=blue] (-1.5,-1.5) -- (1.5,1.5);
		
		\node[right] at (3,0) {$\oo$};
		\node[left] at (-3,0) {$\oo$};
		
		\node[left] at (-1.05,0) {$4$};
		\node[right] at (1.05,0) {$3$};	
		
		\node[left] at (-1.55,-1.55) {$x$};
		\node[right] at (1.55,1.55) {$x^+$};
		
		\node at (-1/4,1/4) {$1$};
		\node at (1/4,-1/4) {$2$};
		\end{tikzpicture}
		\caption{Region $S$ in~\eqref{eq:bilocaldefinition} which contributes to the residue.}
		\label{fig:regionS}
	\end{subfigure}
	\caption{The configuration of points within the Poincare patch of $\oo$. Point $4$ is in the future of $x$ and $3$ is in the past of $x^+$, while $x$ is null separated and in the past of $\oo$. The shaded yellow (red) region is the region of integration for 1 (2) after taking the light transform, in the first term in equations~\eqref{eq:afterlighttransform} and~\eqref{eq:bilocaldefinition}. The dashed null line is spanned by $z$. Note that in (b), for $d>2$ the region $S$ extends in and out of the picture, while the dashed null line doesn't.}
\end{figure}
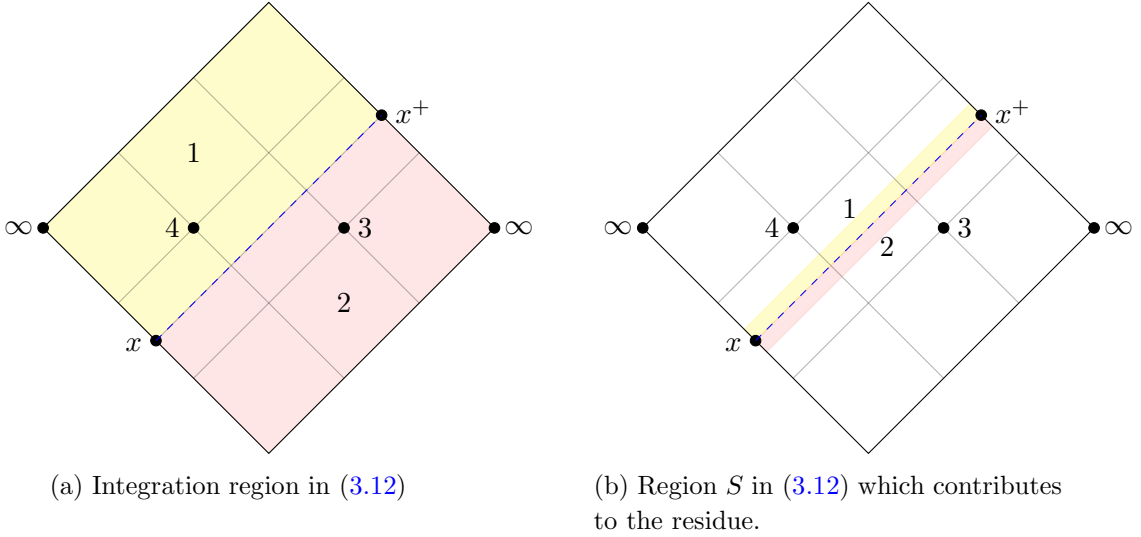

Let us now move $\cO(x,z)$ to  past null infinity and perform the light transform. We choose $3,4$ such that $3^-<x<4$, so that the left-hand side is nonzero, see figure~\ref{fig:nullrayoperator}. Since $\cO$ is on the Euclidean principal series, the condition $\Re(\De+J)>1$ is satisfied and we can plug in~(\ref{eq:lighttransformtimeordered1}) to find
\be
\wL[P_{\De,J}](x_3,x_4,x,z) &= -\p{\<\f_1 \f_2\tl \cO^\dagger \>,\<\tl \f_1^\dagger \tl \f_2^\dagger \cO\>}^{-1}_E\int_{\substack{2^- < x < 1 \\ \oo \approx 1,2}} d^d x_1 d^d x_2 \<V_3 V_4 \f_1 \f_2\>_\Omega \<0|\tl \f_1^\dagger \wL[\cO](x,z) \tl \f_2^\dagger|0\> \nn\\
& \qquad + (1\leftrightarrow 2).
\label{eq:afterlighttransform}
\ee
See the discussion below~\eqref{eq:lighttransformtimeordered1} for the precise meaning of the $(1\leftrightarrow 2)$ term.

Let us now define
\be\label{eq:bilocaldefinition}
	\mathbb{O}_{\De,J}(x,z) &\equiv \frac{\mu(\De,J)S_E(\f_1\f_2[\tl\cO^\dagger])}{\p{\<\f_1 \f_2\tl \cO^\dagger \>,\<\tl \f_1^\dagger \tl \f_2^\dagger \cO\>}_E}\int_{\substack{2^- < x < 1 \\ \oo \approx 1,2}} d^d x_1 d^d x_2\<0|\tl \f_1^\dagger \wL[\cO](x,z) \tl \f_2^\dagger|0\> \f_1 \f_2+
	(1\leftrightarrow 2).
\ee
It is implicit here that $x$ is null separated from $\oo$. This expression makes sense (at least formally) for continuous $J$. The euclidean three-point structure $\<\tl \f_1^\dag \tl \f_2^\dag \cO\>$ that we started with is single-valued only for integer $J$. However, due to the particular Wightman ordering the structures in (\ref{eq:bilocaldefinition}) are well-defined for any $J$, as discussed in appendix~\ref{sec:contspincorrelators}. In order to continue to non-integer $J$, we must also choose an analytic continuation of the prefactors in (\ref{eq:bilocaldefinition}), which we discuss in more detail below. One consequence is that we have two different analytic continuations: one from even values of $J$ that we denote $\mathbb{O}_{\De,J}^+$, and one from odd values of $J$ that we denote $\mathbb{O}_{\De,J}^-$.

For integer $J$,~\eqref{eq:afterlighttransform} and~\eqref{eq:residueequation} imply that the residues $\mathbb{O}_{i,J}^\pm$, defined by
\be\label{eq:lightraydefinition}
	\mathbb{O}^\pm_{\De,J}(x,z)\sim \frac{1}{\De-\De^\pm_i(J)}\mathbb{O}^\pm_{i,J}(x,z),
\ee
have the same three-point functions as light-transforms of local operators in the $\f_1\times\f_2$ OPE. (We include a $\pm$ subscript on $\De_i^\pm(J)$ because the positions of poles in the $(\De,J)$ plane are in general different for the even/odd cases.) To be precise, when $J$ is an integer, the residue of a \textit{time-ordered} correlator, where time-ordering acts on $\f_1$ and $\f_2$ inside the definition of $\mathbb{O}^\pm_{\De,J}$,
\be\label{eq:TOofODeJ}
	\<V_3 V_4 \mathbb{O}^\pm_{\De,J}(x,z)\>_\O,
\ee
agrees with
\be\label{eq:TOofLLocal}
	f_{12\cO}\<V_3 V_4 \wL[\cO_{i,J}]\>_\O.
\ee
for a local operator $\cO_{i,J}$, where $\pm$ is determined by $(-1)^J=\pm 1$.

We now claim that, for any $J$, the residue in~\eqref{eq:TOofODeJ} comes from a region $S$ where $\f_1$ and $\f_2$ are simultaneously almost null-separated from $x$ and from each other, see figure~\ref{fig:regionS}. Indeed, we always expect singularities in correlators when points are null-separated. In integrated correlators,  such singularities can be removed by $i\e$-prescriptions. However, lightlike singularities in the region $S$ are not removed because they coincide with boundaries in the integration regions for $x_1,x_2$. In a time-ordered correlator, we can also have singularities at coincident points. However, we expect singularities related to the $\f_1\times\f_2$ OPE to come from $1$ being lightlike to $2$ and not from other coincident limits. 

Let us focus on the first term of~\eqref{eq:bilocaldefinition}. For this term, it is guaranteed that $1\geq 3$, $2\leq 4$, and $1\geq 2$. In the region $S$ we furthermore have $1\leq 4$ and $2\geq 3$, i.e.\ we have the ordering $4\geq 1\geq 2 \geq 3$, and the contribution of the first term of $\eqref{eq:bilocaldefinition}$ to the time-ordered correlator~\eqref{eq:TOofODeJ} agrees with its contribution to the Wightman function
\be
	\<\O|V_4 \mathbb{O}^\pm_{\De,J} V_3|\O\>.
\ee
The same obviously holds for the second term, and, moreover,~\eqref{eq:TOofLLocal} agrees with the Wightman function
\be
	f_{12\cO}\<\O|V_4 \wL[\cO_{i,J}]V_3|\O\>.
\ee
Since any state in CFT can be approximated by local operators $V_i$ acting on the vacuum in an arbitrarily small region, this implies that we can interpret~\eqref{eq:bilocaldefinition} and~\eqref{eq:lightraydefinition} as operator equations. Furthermore, by construction, for non-negative integer $J$ we must have, as an operator equation,
\be
	\mathbb{O}^\pm_{i,J}=f_{12\cO}\wL[\cO_{i,J}]\qquad(J\in \Z_{\geq 0},\ (-1)^J=\pm 1)
\ee
for some local operator $\cO_{i,J}$. 

For non-integer $J$ the definition~\eqref{eq:bilocaldefinition} with~\eqref{eq:lightraydefinition} provides an analytic continuation in $J$ of $\wL[\cO_{i,J}]$. As we will show in section~\ref{sec:inversionformulae}, it is precisely the matrix elements of $\mathbb{O}^\pm_{\De,J}$ and $\mathbb{O}^\pm_{i,J}$ which are computed by Caron-Huot's Lorentzian inversion formula. As discussed above, the residues $\mathbb{O}^\pm_{i,J}$ should only depend on the region of the integral where $\f_1$ and $\f_2$ are almost null-separated. In fact, it is natural to expect that the residue is further localized onto the null line defined by $z$. Hence we refer to them as light-ray operators. In the next subsection we show this explicitly in the case of mean field theory (MFT).

In our argument for the existence of light-ray operators, it is not necessary that $\mathbb{O}^\pm_{\De,J}$ be a meromorphic function with simple poles. We expect that any non-analyticity in $\mathbb{O}^\pm_{\De,J}$ in the $(\De,J)$ plane should come from the region where $\f_1$ and $\f_2$ are lightlike-separated. Thus, for example, it should be possible to define light-ray operators by taking discontinuities across branch cuts of $\mathbb{O}^\pm_{\De,J}$ (if they exist). Determining the analyticity structure of $\mathbb{O}^\pm_{\De,J}$ in the $(\De,J)$ plane is an important problem for the future.

As mentioned above, to analytically continue $\mathbb{O}^\pm_{\De,J}$ in spin, we must choose an analytic continuation in $J$ of the prefactors
\be
\label{eq:prefactorweneedtocontinue}
&\frac{\mu(\De,J)S_E(\f_1\f_2[\tl\cO^\dagger])}{\p{\<\f_1 \f_2\tl \cO^\dagger \>,\<\tl \f_1^\dagger \tl \f_2^\dagger \cO\>}_E}\nn\\
&= (-1)^J \frac{\G(J+\tfrac d 2)\G(d+J-\De)\G(\De-1)}{2\pi^d \G(J+1)\G(\De-\tfrac d 2)\G(\De+J-1)} \frac{\G(\tfrac{\De+J+\De_1-\De_2}{2})\G(\tfrac{\De+J-\De_1+\De_2}{2})}{\G(\tfrac{d-\De+J+\De_1-\De_2}{2})\G(\tfrac{d-\De+J-\De_1+\De_2}{2})}.
\ee
Additionally, the term in (\ref{eq:bilocaldefinition}) with $(1\leftrightarrow 2)$ has a prefactor differing by $(-1)^J$.
Because of the $(-1)^J$'s, we must make two separate analytic continuations from even and odd $J$, leading to $\mathbb{O}^\pm_{\De,J}$. In general, we expect the even and odd spectrum of light-ray operators to be different because they are distinguished by an eigenvalue $s_{\mathbb{O}}=\pm 1$, as explained in section~\ref{eq:evenvsodd}. For example, in MFT with a real scalar $\f$, the analytic continuation of even-$J$ two-$\f$ operators is nontrivial, but there are no odd-$J$ two-$\f$ operators.

The analytic continuation of the remaining $\G$-function factors in (\ref{eq:prefactorweneedtocontinue}) is determined by requiring that they be meromorphic and polynomially bounded at infinity in the right half-plane. This is important for the Sommerfeld-Watson resummation discussed in section~\ref{sec:sommerfeld}. The expression (\ref{eq:prefactorweneedtocontinue}) satisfies these conditions, so provides a good analytic continuation. When $\f_1,\f_2$ are not scalars, then we can relate the prefactor to a rational function of $J$ times (\ref{eq:prefactorweneedtocontinue}) using weight-shifting operators~\cite{Karateev:2017jgd,ShadowFuture}, and this provides a good analytic continuation in that case as well.

Although we have assumed scalar $\f_1,\f_2$ in this section for simplicity, the generalization to arbitrary representations $\cO_1,\cO_2$ is straightforward. We discuss some aspects of the general case in section~\ref{sec:generalizedinversiontwo}.

\subsubsection{More on even vs.\ odd spin}
\label{eq:evenvsodd}

There is a natural operation that distinguishes even-spin and odd-spin light-ray operators. First recall that every quantum field theory has an anti-unitary symmetry $J_\Omega = \mathsf{CRT}$ which acts on local operators by \cite{streater2016pct}
\be
\label{eq:modularconjugation}
J_\Omega \cO^a(x) J_\Omega^{-1} = \p{i^F (e^{-i\pi \cM^{01}})^{a}{}_b \cO^{b}(\bar x)}^\dag.
\ee
Here, $\bar x = (-x^0, -x^1,x^2,\dots, x^{d-1})$ is the Rindler reflection of $x$, $(\cM^{\mu\nu})^a{}_b$ are  representation matrices associated to the Lorentz generators $M^{\mu\nu}$,\footnote{The $M^{\mu\nu}$ are antihermitian in our conventions.}
\be
[M^{\mu\nu},\cO^a(0)] &= -(\cM^{\mu\nu})^a{}_b \cO^b (0),
\ee
and $F$ is the fermion number of $\cO$. We call $J_\Omega$ ``Rindler conjugation" because it is  the modular conjugation operator for the Rindler wedge in the vacuum state \cite{doi:10.1063/1.522898}.\footnote{The alternative notation $\mathsf{CRT}$ comes from the fact that this operator reverses charges and implements a reflection in both time and a single spatial direction. By contrast the operator $\mathsf{CPT}$ implements a reflection in all spatial directions simultaneously.} It is useful to introduce the notation
\be
\label{eq:defofbar}
\bar{\cO} &\equiv J_\Omega \cO J_\Omega^{-1}.
\ee
Note that $J_\Omega^2 = 1$. Furthermore, using (\ref{eq:defofbar}), we clearly have
\be
\bar{\cO_1 \cO_2} = \bar \cO_1 \bar \cO_2,
\ee
so Rindler conjugation preserves operator ordering.

Rindler conjugation is an anti-unitary symmetry. If we combine it with Hermitian conjugation, we obtain a linear map of operators
\be
\label{eq:linearbutnotsymmetry}
\cO &\to \bar{\cO}^\dag.
\ee
This is no longer a symmetry on Hilbert space because it reverses operator ordering.
Nevertheless, it makes sense to classify operators into eigenspaces of (\ref{eq:linearbutnotsymmetry}). Consider first a local operator $\cO(x,z)$ with spin $J$, and let us set $z=z_0=(1,1,0,\dots,0)$. We have
\be
\bar{\cO(x,z_0)}^\dag &= \cO(\bar x,\bar z_0) = (-1)^J \cO(\bar x,z_0).
\ee
Integrating $x$ along the $z_0$ direction, we obtain
\be
\bar{\wL[\cO](-\oo z_0,z_0)}^\dag &= (-1)^J \wL[\cO](-\oo z_0,z_0).
\ee
For a more general light-ray operator, we have
\be
\bar{\mathbb{O}(-\oo z_0, z_0)}^\dag &= s_\mathbb{O} \mathbb{O}(-\oo z_0, z_0),
\ee
where now the eigenvalue $s_\mathbb{O} = \pm 1$ is not necessarily related to the quantum number $J$. If we obtain $\mathbb{O}$ by analytically continuing $\wL[\cO]$ from the case where $J$ is even (odd), we will obtain $s_\mathbb{O}=+1$ ($-1$). In this work, we abuse terminology and refer to light-ray operators with $s_{\mathbb{O}}=+1$ as ``even-spin" and operators with $s_{\mathbb{O}}=-1$ as ``odd-spin."

\subsection{Light-ray operators in Mean Field Theory}
\label{sec:lightraymft}

In this section we explicitly show that $\mathbb{O}^\pm_{i,J}$ are light-ray operators in Mean Field Theory (MFT). For simplicity, we assume that the scalar operators in~\eqref{eq:bilocaldefinition} are distinct fundamental MFT scalars. More generally, we can imagine that they belong to two decoupled CFTs.

The kernel in~\eqref{eq:bilocaldefinition} is obtained from~\eqref{eq:finalresultlighttransform} by sending $x_3$ to past null infinity according to the rule
\be
	\cO(-z\oo,z)=\lim_{L\to +\oo}L^{\De+J}\cO(-L z,z),
\ee
i.e.
\be
	&\<0|\tl\phi_1^\dagger\wL[\cO]\tl\phi_2^\dagger|0\>=\nn\\
		&\quad=L(\tl\f_1^\dagger\tl\f_2^\dagger|\cO)
	\frac{2^{J-1}
	\p{z\.x_{2}\, x_{1}^2 - z\.x_{1}\, x_{2}^2}^{1-\De}
	}{
		(x_{12}^2)^{\frac{\tl\De_1+\tl\De_2+J-\De}{2}}
		(-z\.x_{1})^{\frac{\tl\De_1-\tl\De_2+2-\De-J}{2}}
		(z\.x_{2})^{\frac{\tl\De_2-\tl\De_1+2-\De-J}{2}}
	}.
\ee
The expression~\eqref{eq:finalresultlighttransform} was written for $1>3, 3\approx 2, 1\approx 2$. With these conditions, the ratio above is positive. In the integral we need to relax $1\approx 2$, which is done by adding $i\e$ to $x_2^0$ and $-i\e$ to $x_1^0$, according to the Wightman ordering above. We now introduce lightcone coordinates by writing
\be
	x_i=\half z v_i+\half z' u_i + \bx_i
\ee
with $z'^2=0, z'\. z=2$ and $\bx_i\.z=\bx_i\. z'=0$. Since this requires $z'$ to be past-directed, the $i\e$-prescription is equivalent to adding a positive imaginary part to $u_1$ and $v_2$ and negative to $u_2$ and $v_1$. We then find for the integral in the first term of~\eqref{eq:bilocaldefinition}
\be
	\frac{1}{4}\int du_1du_2dv_1dv_2 d^{d-2}\bx_1d^{d-2}\bx_2
	\frac{2^{J-1}
		\p{u_1u_2v_{12}+u_2\bx_1^2-u_1\bx_2^2}^{1-\De}
		\phi_1(x_1)\phi_2(x_2)
	}{
		(u_{12}v_{12}+\bx_{12}^2)^{\frac{\tl\De_1+\tl\De_2+J-\De}{2}}
		(-u_1)^{\frac{\tl\De_1-\tl\De_2+2-\De-J}{2}}
		u_2^{\frac{\tl\De_2-\tl\De_1+2-\De-J}{2}}
	}.
\ee
We have temporarily suppressed the light transform coefficient $L(\tl\f_1^\dagger\tl\f_2^\dagger[\cO])$.

The integration region has $u_1<0$ and $u_2>0$. Let us assume for now that $v_2>v_1$ and make the change of variables
\be
	u_1&=-r\a,\nn\\
	u_2&=r(1-\a),\nn\\
	\bx_i&=(r v_{21})^{\half}\bw_i.
\ee
The integral becomes
\be\label{eq:integralwithr}
	&\frac{1}{4}\int_0^1 d\a \int dv_1dv_2  d^{d-2}\bw_1d^{d-2}\bw_2\frac{2^{J-1}
		v_{21}^{-1-\frac{\De-\De_1-\De_2+J}{2}}\p{\a(1-\a)+(1-\a)\bw_1^2+\a\bw_2^2}^{1-\De}
	}{
		(1+\bw_{12}^2)^{\frac{\tl\De_1+\tl\De_2+J-\De}{2}}
		\a^{\frac{\tl\De_1-\tl\De_2+2-\De-J}{2}}
		(1-\a)^{\frac{\tl\De_2-\tl\De_1+2-\De-J}{2}}
	}\times\nn\\
&\times
	\int_0^{\infty}\frac{dr}{r} r^{-\frac{\De-\De_1-\De_2-J}{2}}\phi_1(-r\a,v_1,(r v_{21})^{\half}\bw_1)\phi_2(r(1-\a),v_2,(r v_{21})^{\half}\bw_2).
\ee
In the second line, we have isolated the integral
\be\label{eq:nullmellin}
	\int_0^{\infty}\frac{dr}{r} r^{-\frac{\De-\De_1-\De_2-J}{2}}\phi_1(-r\a,v_1,(r v_{21})^{\half}\bw_1)\phi_2(r(1-\a),v_2,(r v_{21})^{\half}\bw_2).
\ee
The region $r\sim 0$ corresponds to $\phi_1$ and $\phi_2$ being localized near the light ray defined by $z$.

Now imagine expanding the product of field operators in a power series in $r$. This is possible since we have assumed that $\phi_1$ and $\phi_2$ do not interact and thus there is no lightcone singularity between them.\footnote{If we consider $\f_1=\f_2=\f$, then in MFT we have $\f(x_1)\f(x_2)=:\!\f(x_1)\f(x_2)\!:\!+\<\O|\f(x_1)\f(x_2)|\O\>$. The singular term is positive-energy in $x_2$ and negative-energy in $x_1$. But in~\eqref{eq:bilocaldefinition} we are integrating against $\<0|\tl\f_1\wL[\cO]\tl\f_2|0\>$, which has the same energy conditions on $x_1$ and $x_2$. Since the integrals pick out the term with vanishing total energy in the direction of $z$ for both $x_1$ and $x_2$, the singular piece does not contribute to~\eqref{eq:bilocaldefinition}. See also the discussion in section~\ref{sec:inversionformulae}.} We find terms of the form
\be
	r^{n+m+\half(a+b)}(-\a)^n(1-\a)^m v_{21}^{\half(a+b)} \bw_1^a \bw_2^b.
\ee
Only terms with even values of $a+b$ contribute, since the $\bw_i$ integral is invariant under $\bw_i\to -\bw_i$. Therefore, $N=n+m+\half(a+b)\geq 0$ is an integer and the integral over $r$ takes the form
\be
	\int_0^{\infty}\frac{dr}{r} r^{-\frac{\De-\De_1-\De_2-J-2N}{2}}\sim -\frac{2}{\De-\De_1-\De_2-J-2N}.
\ee
The pole comes from the region of small $r$.  We can see this by imposing an upper cutoff on $r$: the residue will be independent of it. (In particular, we can make the cutoff depend on $\a$ and $\bw_i$ thereby cutting out arbitrary regions around the null ray and the residue won't change.) The pole is at
\be
	\De=\De_1+\De_2+J+2N,
\ee
which for integer $J$ are precisely the locations of double-trace operators $[\f_1\f_2]_{N,J}$. For every $N$, the residue of~\eqref{eq:nullmellin} only depends on a finite number of derivatives of $\f_i$ on the null ray, and thus is localized on it, as promised in the introduction. 

For simplicity, let us focus on the leading twist trajectory with $N=0$. The residue of~\eqref{eq:nullmellin} is then 
\be
	-2\f_1(0,v_1,0)\f_2(0,v_2,0)
\ee
and the residue of the integral~\eqref{eq:integralwithr} becomes
\be\label{eq:Rappearance}
	&\frac{-1}{2}\int_0^1 d\a \int  d^{d-2}\bw_1d^{d-2}\bw_2\frac{2^{J-1}
	\p{\a(1-\a)+(1-\a)\bw_1^2+\a\bw_2^2}^{1-\De_1-\De_2-J}
	}{
		(1+\bw_{12}^2)^{d-\De_1-\De_2}
		\a^{{-\De_1+1-J}}
		(1-\a)^{{-\De_2+1-J}}
	}\times\nn\\
	&\times
	\int dv_1 dv_2 (v_{21}+i\e)^{-1-J}
	\phi_1(0,v_1,0)\phi_2(0,v_2,0).
\ee
The first line is an overall coefficient which we compute in appendix~\ref{app:Rcoefficient} and here simply denote by $\cR(\De_1,\De_2,J)$. In the second line, we have restored the $i\e$ prescription for $v_i$, which allows us to relax the assumption $v_2>v_1$. (The factor $(v_{21}+i\e)^{-1-J}$ is understood to be positive for positive $v_{21}$ and real $J$.)

Combining everything together, we conclude that the leading twist operators $\mathbb{O}_{0,J}$ are given by
\be
	\mathbb{O}_{0,J}(-z\oo,z)
		=&i\frac{(-1)^J}{4\pi}\int ds dt \p{(t+i\e)^{-1-J}+(-1)^J(-t+i\e)^{-1-J}}\phi_1(0,s-t,0)\phi_2(0,s+t,0),	
\ee
where we have included the contribution of the second term in~\eqref{eq:bilocaldefinition}, performed the change of variables $v_1=s-t,\,v_2=s+t$, and used the identity
\be
L(\tl\f_1^\dagger\tl\f_2^\dagger[\cO])\cR(\De_1,\De_2,J)\frac{\mu(\De,J)S_E(\f_1\f_2[\tl\cO^\dagger])}{\p{\<\f_1 \f_2\tl \cO^\dagger \>,\<\tl \f_1^\dagger \tl \f_2^\dagger \cO\>}_E}=i\frac{(-1)^J2^{J-2}}{\pi}.
\ee
The analytic continuations from even and odd $J$ are\footnote{It is straightforward to check that $\overline{\mathbb{O}_{0,J}^\pm}^\dagger=\pm \mathbb{O}_{0,J}^\pm$.}
\be
	\mathbb{O}_{0,J}^+(-z\oo,z)&=
	+\frac{i}{4\pi}\int ds dt \p{(t+i\e)^{-1-J}+(-t+i\e)^{-1-J}}\phi_1(0,s-t,0)\phi_2(0,s+t,0),\label{eq:mftlightray+}\nn\\
	\mathbb{O}_{0,J}^-(-z\oo,z)&=
	-\frac{i}{4\pi}\int ds dt \p{(t+i\e)^{-1-J}-(-t+i\e)^{-1-J}}\phi_1(0,s-t,0)\phi_2(0,s+t,0).
\ee
These are exactly the null-ray operators advertised in the introduction. We can check that they are indeed primary by lifting their definitions to the embedding space, where they are variants of
\be
	\sim\int_{-\oo}^{+\oo} d\a d\b\, \f_1(Z-\a X)\f_2(Z-\b X) (\a-\b)^{-J-1}.
\ee
We discuss conformal invariance of this embedding-space integral in the next subsection.

For integer $J$ both kernels for the $t$-integral are equal to
\be
	&(t+i\e)^{-1-J}+(-1)^J(-t+i\e)^{-1-J}=\nn\\&\quad=\frac{(-1)^J}{\Gamma(J+1)}\frac{\ptl^J}{\ptl t^J}\p{(t+i\e)^{-1}-(t-i\e)^{-1}}=-2\pi i\frac{(-1)^J}{\Gamma(J+1)}\frac{\ptl^J}{\ptl t^J}\delta(t).
\ee
Thus, for integer $J$ we find
\be
	\mathbb{O}_{0,J}(-z\oo,z)&=\frac{(-1)^J}{\Gamma(J+1)}\int \frac{ds}{2}\, \phi_1(0,s,0)(\overset{\leftrightarrow}{\ptl_s})^J\phi_2(0,s,0)=\wL[[\f_1\f_2]_{0,J}](-z\oo,z).
\ee
Since total derivatives vanish in the integral over $s$, it follows that for integer spin $\mathbb{O}_{0,J}$ is given by the light transform of a primary double-twist operator of the form
\be
	[\f_1\f_2]_{0,J}(x,z)\equiv \frac{(-1)^J}{\Gamma(J+1)}\phi_1(x)(z\.{\ptl})^J\phi_2(x)+(z\.\ptl)(\ldots).
\ee

Let us check that these operators are correctly normalized. It was found in~\cite{Penedones:2010ue} that the full expression for the primary $[\f_1\f_2]_{0,J}$ is 
\be\label{eq:doubletraceexplicit}
	[\f_1\f_2]_{0,J}(x,z)=c_J\sum_{k=0}^J\frac{(-1)^k}{k!(J-k)!\Gamma(\De_1+k)\Gamma(\De_2+J-k)}(z\.\ptl)^k\phi_1(x)(z\.\ptl)^{J-k}\phi_2(x)
\ee
and in our case $c_J$ is given by
\be
	c_J=\frac{(-1)^J}{\Gamma(J+1)}\p{\sum_{k=0}^J\frac{1}{k!(J-k)!\Gamma(\De_1+k)\Gamma(\De_2+J-k)}}^{-1}.
\ee
If we write now
\be
	\<\f_1\f_2[\f_1\f_2]_{0,J}\>_\O=f_{12J}\<\f_1\f_2\cO_J\>,
\ee
and
\be
	\<[\f_1\f_2]_{0,J}[\f_1\f_2]_{0,J}\>_\O=\cC_J\<\cO_J\cO_J\>,\label{eq:doubletracetwopt}
\ee
where in the right hand side we use the standard structures defined in appendix~\ref{app:23conventions}, then our normalization conventions are such that $\cC_J/f_{12J}=1$.\footnote{To be more precise, if $\cO$ is an operator in $\f_1\times\f_2$ OPE, we are computing $[\f_1\f_2]_J=f_{12\cO}\cO/\cC_\cO$, which is independent of the normalization of $\cO$. Using $[\f_1\f_2]_J$ instead of $\cO$ then yields the claimed normalization condition.} It is a straightforward exercise to show using~\eqref{eq:doubletraceexplicit} that
\be
	\frac{\cC_J}{f_{12J}}&=(-1)^J\Gamma(J+1)c_J\sum_{k=0}^J\frac{1}{k!(J-k)!\Gamma(\De_1+k)\Gamma(\De_2+J-k)}=1.
\ee
In doing the calculation it is convenient to use the same null polarization vector for both operators in~\eqref{eq:doubletracetwopt}.

\subsubsection{Subleading families and multi-twist operators}

Although we will not compute the residue of $\mathbb{O}_{\De,J}$ for $N>0$, let us comment on the form of the light-ray operators that we expect to obtain, as well as on some further interesting generalizations. For simplicity, in this section we ignore $i\e$-prescriptions, the difference between even and odd $J$, and normalization factors. As mentioned above, the leading double-twist operators are essentially the primaries
\be
	\mathbb{O}_{0,J}(X,Z)\equiv \int d\a\, d\b\, \f_1(Z-\a X)\f_2(Z-\b X) (\a-\b)^{-J-1}.
\ee
The fact that $\mathbb{O}$ is a primary follows from conformal invariance of the integral on the right-hand side. According to the usual rules of the embedding space formalism~\cite{Costa:2011mg}, conformal invariance is equivalent to
\begin{enumerate}
	\item homogeneity in $X$ and $Z$ with degrees $-\De_\mathbb{O}$ and $J_\mathbb{O}$, and
	\item invariance under $Z\to Z+\l X$.
\end{enumerate}
The former requirement is fulfilled due to homogeneity of the measure $d\a\, d\b$, the ``wavefunction'' $(\a-\b)^{-J-1}$, and the original primaries $\f_i$, which leads to 
\be\label{eq:mftquantumnumbers0}
	\De_\mathbb{O}&=1-J,\nn\\
	J_\mathbb{O}&=1-\De_1-\De_2-J.
\ee
The latter requirement is satisfied due to translational invariance of the measure $d\a\, d\b$ and the wavefunction $(\a-\b)^{-J-1}$.

This leads to two simple observations. The first is that since the only requirement on $\f_i$ is that of being a primary, we can dress them with weight-shifting operators~\cite{Karateev:2017jgd}. For example, let $D_m$ be the Thomas/Todorov differential operator which increases the scaling dimension of a primary by $1$ and carries a vector embedding space index $m$. Then we can define
\be
	\mathbb{O}_{N,J}(X,Z)=\int d\a d\b (D_{m_1}\cdots D_{m_N}\f_1)(Z-\a X) (D^{m_1}\cdots D^{m_N}\f_2)(Z-\b X)(\a-\b)^{-J-1}.
\ee
By construction, we now have
\be
	\De_\mathbb{O}&=1-J,\nn\\
	J_\mathbb{O}&=1-\De_1-\De_2-J-2N.
\ee
With appropriate $i\e$-prescriptions for $\a$- and $\b$-contours, for integer $J$ these operators reduce to light transforms of the local family $[\f_1\f_2]_{N,J}$. It is clear how (at least in principle) this construction generalizes to non-scalar $\f_i$.

The second observation is that this construction straightforwardly generalizes to multi-twist operators. In particular, define
\be
	\mathbb{O}_\psi(X,Z)=\int d\a_1\cdots d\a_n \f_1(Z-\a_1 X)\cdots \f_n(Z-\a_n X) \psi(\a_1,\ldots,\a_n),
\ee
where $\psi$ is a wavefunction which is translationally-invariant and homogeneous,
\be
	\psi(\a_1+\b,\ldots,\a_n+\b)&=\psi(\a_1,\ldots,\a_n),\nn\\
	\psi(\l\a_1,\ldots,\l\a_n)&=\l^{-J-1} \psi(\a_1,\ldots,\a_n).
\ee
We can easily check that $\mathbb{O}_\psi$ is a primary with scaling dimension and spin given by
\be
	\De_\mathbb{O}&=1-J,\nn\\
	J_\mathbb{O}&=1-J+\sum_{i=1}^n \De_n.
\ee
Subleading families can be obtained as above, by dressing with weight-shifting operators. The generalization to non-scalar $\f_i$ is also clear.

\section{Lorentzian inversion formulae}
\label{sec:inversionformulae}

In this section we show that matrix elements of $\mathbb{O}_{\De,J}$ are computed by a Lorentzian inversion formula of the type discussed by Caron-Huot \cite{Caron-Huot:2017vep}. Our derivation will borrow some key steps from \cite{Simmons-Duffin:2017nub}. However the light transform will simplify the derivation to the point where its generalization to external spinning operators is obvious. In particular, after using the light transform in the appropriate way, it will be immediately clear why the conformal block $G_{J+d-1,\De-d+1}$ and its generalizations appear. For simplicity, we will present most of the derivation with scalar operators and generalize to spinning operators at the end.

\subsection{Inversion for the scalar-scalar OPE}

\subsubsection{The double commutator}

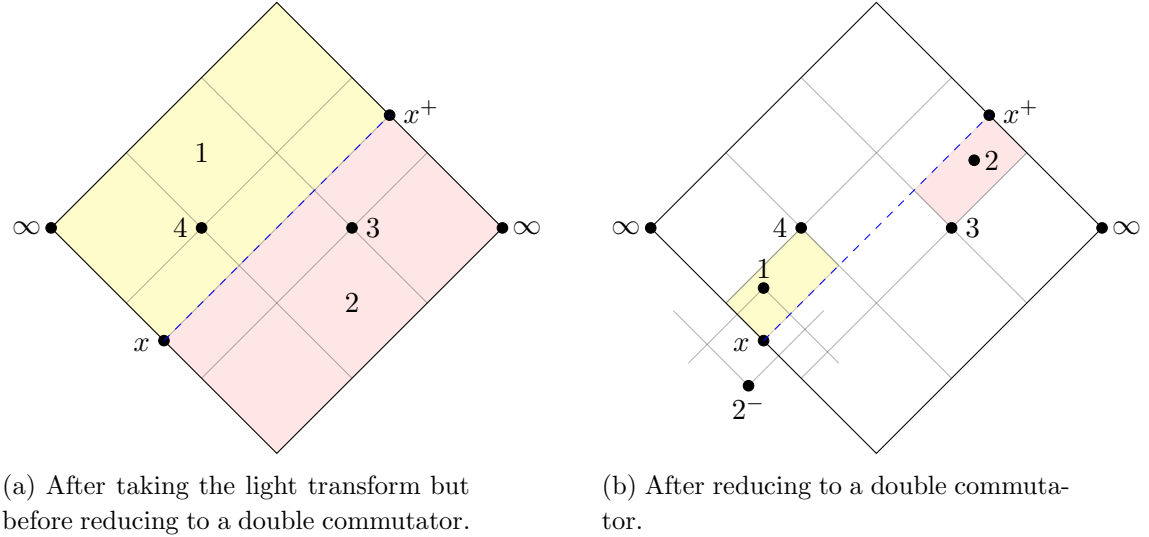
\begin{figure}
	\centering
	\begin{subfigure}[t]{0.4\textwidth}
		\begin{tikzpicture}
			\draw[fill=yellow, opacity = 0.2] (-1.5,-1.5) -- (-3,0) -- (0,3) -- (1.5,1.5) -- cycle;
			\draw[fill=red, opacity = 0.1] (-1.5,-1.5) -- (0,-3) -- (3,0) -- (1.5,1.5) -- cycle;
	
			\draw (-3,0) -- (0,3) -- (3,0) -- (0,-3) -- cycle;
	
			\draw[fill=black] (-3,0) circle (0.07); 
			\draw[fill=black] (3,0) circle (0.07);  
	
			\draw[fill=black] (-1,0) circle (0.07); 
			\draw[fill=black] (1,0) circle (0.07);  
			
			\draw[fill=black] (-1.5,-1.5) circle (0.07);  
			\draw[fill=black] (1.5,1.5) circle (0.07);  
			
			\draw[opacity=.3] (-2,1) -- (1,-2); 
			\draw[opacity=.3] (-2,-1) -- (1,2); 
	
			\draw[opacity=.3] (2,1) -- (-1,-2); 
			\draw[opacity=.3] (2,-1) -- (-1,2); 
			
			\draw[dashed, color=blue] (-1.5,-1.5) -- (1.5,1.5);
			
			\node[right] at (3,0) {$\oo$};
			\node[left] at (-3,0) {$\oo$};
			
			\node[left] at (-1.05,0) {$4$};
			\node[right] at (1.05,0) {$3$};	
			
			\node[left] at (-1.55,-1.55) {$x$};
			\node[right] at (1.55,1.55) {$x^+$};
			
			\node at (-1,1) {$1$};
			\node at (1,-1) {$2$};
		\end{tikzpicture}
		\caption{After taking the light transform but before reducing to a double commutator.}
		\label{fig:poincarebeforecommutator}
	\end{subfigure}
	\hspace{1.5cm}
	\begin{subfigure}[t]{0.4\textwidth}
		\begin{tikzpicture}
		\draw[fill=yellow, opacity = 0.2] (-1.5,-1.5) -- (-2,-1) -- (-1,0) -- (-0.5,-0.5) -- cycle;
		\draw[fill=red, opacity = 0.1] (1.5,1.5) -- (2,1) -- (1,0) -- (0.5,0.5) -- cycle;
		
		\draw (-3,0) -- (0,3) -- (3,0) -- (0,-3) -- cycle;
		
		\draw[fill=black] (-3,0) circle (0.07); 
		\draw[fill=black] (3,0) circle (0.07);  
		
		\draw[fill=black] (-1,0) circle (0.07); 
		\draw[fill=black] (1,0) circle (0.07);  
		
		\draw[fill=black] (-1.5,-0.8) circle (0.07); 
		\draw[fill=black] (1.3,0.9) circle (0.07);  
		\draw[fill=black] (1.3-3,0.9-3) circle (0.07);  
		
		\draw[fill=black] (-1.5,-1.5) circle (0.07);  
		\draw[fill=black] (1.5,1.5) circle (0.07);  
		
		\draw[opacity=.3] (-2,1) -- (1,-2); 
		\draw[opacity=.3] (-2,-1) -- (1,2); 
		
		\draw[opacity=.3] (2,1) -- (-1,-2); 
		\draw[opacity=.3] (2,-1) -- (-1,2); 

		\draw[opacity=.3] (-1.5,-0.8) -- (-2.5,-1.8); 
		\draw[opacity=.3] (-1.5,-0.8) -- (-0.5,-1.8); 
		
		\draw[opacity=.3] (1.3-3,0.9-3) -- (1.3-2,0.9-2); 
		\draw[opacity=.3] (1.3-3,0.9-3) -- (1.3-4,0.9-2); 
		
		\draw[dashed, color=blue] (-1.5,-1.5) -- (1.5,1.5);
		
		\node[right] at (3,0) {$\oo$};
		\node[left] at (-3,0) {$\oo$};
		
		\node[left] at (-1.05,0) {$4$};
		\node[right] at (1.05,0) {$3$};	
		
		\node[left] at (-1.55,-1.55) {$x$};
		\node[right] at (1.55,1.55) {$x^+$};
		
		\node[above] at (-1.5,-0.8) {$1$};
		\node[below] at (1.3-3,0.9-3) {$2^-$};
		\node[right] at (1.3,0.9) {$2$};
		\end{tikzpicture}
		\caption{After reducing to a double commutator.}
		\label{fig:poincareaftercommutator}
	\end{subfigure}
	\caption{The configuration of points within the Poincare patch of $x_\oo$ at various stages of the derivation. The blue dashed line shows the support of light transform of $\cO(x,z)$. The yellow (red) shaded region shows the allowed region for $1$ ($2$). In the right-hand figure, we indicate that $x$ is constrained to satisfy $2^- < x < 1$. Note that after reducing to a double-commutator, the yellow and red regions are independent of $x_\oo$ (as long as $x$ is lightlike from $x_\oo$).}
\end{figure}

Our starting point is the light-transformed expression~\eqref{eq:afterlighttransform}. Let us concentrate on the first term in (\ref{eq:afterlighttransform}). Because of the restrictions $3^- < x < 4$ and $2^- < x < 1$, the lightcone of $x$ splits Minkowski space into two regions, with $2,3$ in the lower region and $1,4$ in the upper, see figure~\ref{fig:poincarebeforecommutator}. Thus, we can write the integrand as
\be
\label{eq:integrandtimeordered}
\<\O|T\{V_4 \f_1\} T\{\f_2 V_3\}|\O\>\<0|\tl \f_1^\dagger \wL[\cO](x,z) \tl \f_2^\dagger|0\>.
\ee
Recall that in our notation, expectation values in the state $|\Omega\>$ denote physical correlation functions, whereas expectation values in the state $|0\>$ denote two- or three-point structures that are fixed by conformal invariance. (For instance, three-point structures $\<0|\cdots|0\>$ don't include OPE coefficients.)

We can now use the reasoning in lemma~\ref{lemma:light} to obtain a double commutator.\footnote{This argument is the same as the contour manipulation in \cite{Simmons-Duffin:2017nub}.} Consider a modified integrand where $\f_1$ acts on the future vacuum,
\be
\label{eq:modifiedintegrand}
\<\O|\f_1 V_4 T\{\f_2 V_3\}|\O\>\<0|\tl \f_1^\dagger \wL[\cO](x,z) \tl \f_2^\dagger|0\>.
\ee
Imagine integrating $\f_1$ over a lightlike line in the direction of $z$, with coordinate $v_1$ along the line. Because $\f_1$ acts on the future vacuum, the correlator is analytic in the lower half $v_1$-plane. Furthermore, at large $v_1$, the product of correlators goes like
\be
\frac{1}{v_1^{\De_1}} \x \frac{1}{v_1^{\frac{\tl \De_1 + \tl \De_2 + \De + J - 2}{2}}}.
\ee
Here, the first factor comes from the estimate (\ref{eq:twoptexpectation}) of $\<\O|\f_1 \cdots|\O\>$ using the OPE and the second factor comes from direct computation using the three-point function (\ref{eq:finalresultlighttransform}). Thus, we can deform the $v_1$ contour in the lower half-plane to give zero whenever
\be
\Re(2(d-2)+\De_1-\De_2+\De+J) > 0.
\ee
This condition is certainly true for $\De\in \frac d 2 + i\oo$ and $J\geq 0$, assuming (for now) that $\Re(\De_2-\De_1)=0$ (see section~\ref{sec:euclideanpartialwaves}).

Consequently, the $x_1$ integral vanishes if we replace (\ref{eq:integrandtimeordered}) with (\ref{eq:modifiedintegrand}), so we can freely replace
\be
T\{V_4 \f_1\} \  &\to\  T\{V_4 \f_1\}-\f_1 V_4 \ =\ [V_4,\f_1] \th(1<4).
\ee
By similar reasoning, we can replace 
\be
T\{\f_2 V_3\} \to [\f_2,V_3]\th(3<2).
\ee
Overall, we find a double commutator in the integrand, together with some extra restrictions on the region of integration
\be
\label{eq:afterdoublecommutator}
\int_{\substack{x<1<4 \\ 3<2<x^+}} d^d x_1 d^d x_2 \<\O|[V_4, \f_1] [\f_2,V_3]|\O\> \<0|\tl \f_1^\dagger \wL[\cO](x,z) \tl \f_2^\dagger|0\> + (1\leftrightarrow 2)
\ee
Note that the spurious dependence on the point at infinity $x_\oo$ has disappeared because the commutators are only nonzero if $x<1<4$ and $3<2<x^+$, and these restrictions imply that $1,2$ lie in the same Poincare patch as $3,4,x$.

In terms of $\mathbb{O}_{\De,J}$ we have
\be\label{eq:ODeJdoublecommutatorformula}
&\<V_4\mathbb{O}_{\De,J}(x,z)V_3\>_\Omega=\nn\\
&=\frac{\mu(\De,J)S_E(\f_1\f_2[\tl\cO^\dagger])}{\p{\<\f_1 \f_2\tl \cO^\dagger \>,\<\tl \f_1^\dagger \tl \f_2^\dagger \cO\>}_E}\int_{\substack{x<1<4 \\ 3<2<x^+}} d^d x_1 d^d x_2 \<\O|[V_4, \f_1] [\f_2,V_3]|\O\> \<0|\tl \f_1^\dagger \wL[\cO](x,z) \tl \f_2^\dagger|0\> + (1\leftrightarrow 2).
\ee
This gives a Lorentzian inversion formula analogous to the Euclidean inversion formula~\eqref{eq:euclideanintegralforpartialwave}. It is different from Caron-Huot's formula~\cite{Caron-Huot:2017vep} in that it is not formulated in terms of cross-ratio integrals and it is valid for non-primary or non-scalar $V_i$. The form of the inversion formula above will be useful in section~\ref{sec:positivityandtheanec} where we discuss the average null energy condition and its generalizations. Note also that the generalization to  operators $\cO_1$ and $\cO_2$ with nonzero spin is straightforward. In the rest of this subsection we show how to reduce~\eqref{eq:ODeJdoublecommutatorformula} to a cross-ratio integral in the form of~\cite{Caron-Huot:2017vep}.

\subsubsection{Inversion for a four-point function of primaries}
\label{sec:inversionforfourptofprimaries}

To obtain an integral over cross-ratios, let us specialize to the case where $V_3=\f_3$ and $V_4=\f_4$ are primary scalars. The partial wave $P_{\De,J}$ in this case is fixed by conformal invariance up to a coefficient:
\be
\label{eq:simplematrixelement}
\mu(\De,J)S_E(\f_1\f_2[\tl \cO^\dagger]) P_{\De,J}(x_3,x_4,x,z) &= C(\De,J) \<\f_3 \f_4 \cO(x,z)\>.
\ee
OPE data is encoded in the resiudes of $C(\De,J)$ by (\ref{eq:residueequation}),
\be
f_{12\cO_*}f_{34\cO_*} &= -\Res_{\De= \De_*} C(\De,J_*).
\ee
The matrix element $\<\f_4 \mathbb{O}_{\De,J}(x,z) \f_3\>_\Omega$ is the light-transform of~\eqref{eq:simplematrixelement}, so~\eqref{eq:ODeJdoublecommutatorformula} becomes
\be
& C(\De,J) \<0|\f_4 \wL[\cO](x,z) \f_3|0\>
\nn\\
&=-\frac {\mu(\De,J)S_E (\f_1 \f_2 [\tl \cO^\dagger])}{\p{\<\f_1 \f_2\tl \cO^\dagger \>,\<\tl \f_1^\dagger \tl \f_2^\dagger \cO\>}}\int_{\substack{x<1<4 \\ 3<2<x^+}} d^d x_1 d^d x_2 \<\O|[\f_4, \f_1] [\f_2,\f_3]|\O\> \<0|\tl \f_1^\dagger \wL[\cO](x,z) \tl \f_2^\dagger|0\>\nn\\&\qquad + (1\leftrightarrow 2).
\ee
For reasons that will become clear in a moment, let us replace $x_4\to x_4^+$ (equivalently act with $\tsym_4$ on both sides). This converts the condition $3^-<x<4$ into $3^-<x<4^+$. At the same time, let us make the change of variables $x_2\to x_2^+$ in the integral. We obtain
\be
\label{eq:formulawithprimaries}
& C(\De,J) \<0|\f_{4^+} \wL[\cO](x,z) \f_3|0\>
\nn\\
&=-\frac {\mu(\De,J)S_E(\f_1 \f_2 [\tl \cO])}{\p{\<\f_1 \f_2\tl \cO^\dagger \>,\<\tl \f_1^\dagger \tl \f_2^\dagger \cO\>}_E}\int_{3^-<2<x<1<4^+} d^d x_1 d^d x_2\, \<\O|[\f_{4^+}, \f_1] [\f_{2^+},\f_3]|\O\> \<0|\tl \f_1^\dagger \wL[\cO](x,z) \tl \f_{2^+}^\dagger|0\>\nn\\&\qquad + (1\leftrightarrow 2).
\ee
Explicitly, the structure on the left-hand side is (under the additional constraint $3>4$)
\be
&\<0|\f_{4^+} \wL[\cO](x_0,z) \f_3|0\>\nn\\
 &= L(\f_3\f_4[\cO])
	\frac{
		(-1)^J\p{2 z\.x_{40}\, x_{30}^2 -2z\.x_{30}\, x_{40}^2}^{1-\De}
	}{(-x_{43}^2)^{\frac{\De_4+\De_3+J-\De}{2}}(x_{40}^2)^{\frac{\De_4-\De_3+2-\De-J}{2}}(x_{30}^2)^{\frac{\De_3-\De_4+2-\De-J}{2}}},\phantom{>0}
	\label{eq:structlhs}
\ee
where $L(\f_3\f_4[\cO])$ is given by (\ref{eq:lighttransformcoefficient}).
This expression comes from making the replacements $1,2,3\to 3,4^+,0$ in the second line of (\ref{eq:lighttransformtimeordered}) and using $x_{i4^+}^2 = - x_{i4}^2$ and $z\.x_{4^+0}=-z\.x_{40}$.\footnote{These relations follow from the embedding space representation of these quantities as inner products with $X_4$. An alternative way to obtain this result is to use $\<0|\phi_{4^+}\wL[\cO]\f_3|0\>=\<0|\tsym\f_4\tsym^{-1}\wL[\cO]\f_3|0\>=e^{-i\pi\De_4}\<0|\f_4\wL[\cO]\f_3|0\>$ and then~\eqref{eq:finalresultlighttransform} with replacements $1\to4,2\to3,3\to 0$, analytically continued. The factor $(-1)^J$ comes from the fact that the standard structure~\eqref{eq:standardthreeptconvention} depends on a formal ordering of operators and we need $\<\f_3\f_4\cO\>$ by convention.}
Similarly, the structure in the right hand side is
\be\label{eq:structure12forH}
&\<0|\tl\phi_1^\dag\,\wL[\cO](x_0,z)\, \tl\phi_{2^+}^\dag|0\> \nn\\
&=L(\tl\f_1^\dag\tl\f_2^\dag[\cO])
\frac{
	\p{2z\.x_{10}\, x_{20}^2 - 2z\.x_{20}\, x_{10}^2}^{1-\De}
}{(-x_{12}^2)^{\frac{\tl \De_1+\tl \De_2+J-\De}{2}}(-x_{10}^2)^{\frac{\tl \De_1-\tl\De_2+2-\De-J}{2}}(-x_{20}^2)^{\frac{\tl\De_2-\tl\De_1+2-\De-J}{2}}}>0,
\ee
which follows from~\eqref{eq:finalresultlighttransform} by using the same rules.

We would now like to express the coefficient $C(\De,J)$ as an integral of the double-commutator $\<\O|[\f_{4^+}, \f_1] [\f_{2^+},\f_3]|\O\>$ against a conformal block. Both sides of the above equation transform like conformal three-point functions. We can pick out the coefficient $C(\De,J)$ by taking a conformally-invariant pairing of both sides with a three-point structure that is ``dual" to the one on the left-hand side.

In other words, in order to isolate $C(\De,J)$, we should find a structure $T$ such that
\be
\label{eq:pairingisone}
\big(T, \<0|\f_{4^+} \wL[\cO](x,z) \f_3|0\>\big)_L &= 1,
\ee
with the pairing $(\cdot,\cdot)_L$ defined in equation~\eqref{eq:lorentzianpairing} as
\be
\label{eq:lorentzianpairingmaintext}
&\p{\<\cO_3\cO_4\cO\>, \<\tl\cO_3^\dag\tl \cO_4^\dag\cO^{\mathrm{S}\dag}\>}_L \nn\\&\equiv \int_{\substack{4<3 \\ x\approx 3,4}} 
\frac{d^dx_3 d^dx_4 d^d x D^{d-2} z}{\vol(\tl{\SO}(d,2))} \<\cO_3(x_3)\cO_4(x_4) \cO(x,z)\> \<\tl\cO_3^\dag(x_3) \tl \cO_4^\dag(x_4) \cO^{\mathrm{S}\dag}(x,z)\>.
\ee
(Note the causal restrictions in the integral.) It will be convenient to write (\ref{eq:pairingisone}) using the shorthand notation
\be
T &= \<0|\f_{4^+} \wL[\cO](x,z) \f_3|0\>^{-1}.
\ee
For the pairing (\ref{eq:pairingisone}) to be well-defined, $\<0|\f_{4+} \wL[\cO]\f_3|0\>^{-1}$ must transform like a three-point function with representations $\<\tl \f_4^\dagger  \cO^{\mathrm{F}\dagger} \tl \f_3^\dagger \>$, where $\cO^\mathrm{F}$ has dimension and spin
\be
\De_{\cO^\mathrm{F}} \ &=\ J+d-1,\nn\\
J_{\cO^\mathrm{F}} \  &=\  \De-d+1.
\ee
The quantum numbers of $\cO^\mathrm{F}$ are precisely those appearing in Caron-Huot's block. We will see shortly that this is not a coincidence. Explicitly, the dual structure $\<0|\f_{4^+} \wL[\cO] \f_3|0\>^{-1}$ is given by (again for $3>4$)
\be
&\<0|\f_{4^+} \wL[\cO](x_0,z) \f_3|0\>^{-1}\nn\\
 &\quad=\frac{2^{2d-2}\vol(\SO(d-2))}{L(\f_3\f_4[\cO])}
	\frac{
		(-1)^J \p{2z\.x_{40}\, x_{30}^2 -2z\.x_{30}\, x_{40}^2}^{\De-d+1}
	}{(-x_{43}^2)^{\frac{\tl\De_4+\tl\De_3-J+\De-2d+2}{2}}(x_{40}^2)^{\frac{\tl\De_4-J-\tl\De_3-\De+2}{2}}(x_{30}^2)^{\frac{\tl\De_3-J-\tl\De_4-\De+2}{2}}}.
\label{eq:dualconfig}
\ee
This follows easily from the alternative characterization of the paring~\eqref{eq:lorentzianpairingmaintext} given in appendix~\ref{app:contspinpairings}.

Finally, pairing both sides of (\ref{eq:formulawithprimaries}) with $\<0|\f_{4^+} \wL[\cO] \f_3|0\>^{-1}$, we obtain
\be
\label{eq:integralforcspacetime}
C(\De,J)
&=\int_{\substack{1>2 \\ 3>4}} \frac{d^d x_1 \cdots d^d x_4}{\vol(\tl{\SO}(d,2))} \<\O|[\f_{4^+}, \f_1] [\f_{2^+},\f_3]|\O\> H_{\De,J}(x_i) + (1\leftrightarrow 2),
\ee
where
\be
H_{\De,J}(x_i)
&=-\frac {\mu(\De,J)S_E(\f_1 \f_2 [\tl \cO])}{\p{\<\f_1 \f_2\tl \cO^\dagger \>,\<\tl \f_1^\dagger \tl \f_2^\dagger \cO\>}_E}
\int_{2<x<1} d^d x D^{d-2} z \<0|\tl \f_1^\dagger \wL[\cO](x,z) \tl \f_{2^+}^\dagger|0\>\<0|\f_{4^+} \wL[\cO](x,z) \f_3|0\>^{-1}.
\label{eq:integralforh}
\ee
In the integral for $C(\De,J)$, all the pairs of points $x_i$ are spacelike separated except for $1>2$ and $3>4$. The causal relations in~\eqref{eq:integralforcspacetime} and~\eqref{eq:integralforh} come from the causal relations in~\eqref{eq:formulawithprimaries} and~\eqref{eq:lorentzianpairingmaintext} which are, together,
\be
	4^-<3^-<2<x<1<4^+<3^+.
\ee
Recalling that $a\approx b$ is equivalent to $a^-<b<a^+$ (figure~\ref{fig:causalrelationships}), we easily find that the above relations are the same as
\be
	&1>x>2,\quad 3>4,\nn\\
	&1\approx 3,\quad 1\approx 4,\quad 2\approx 3,\quad 2\approx 4.
\ee

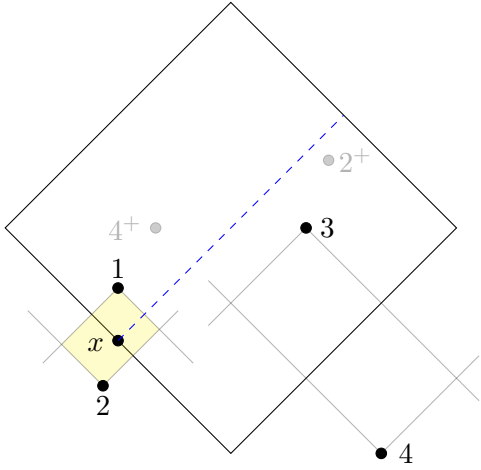
\begin{figure}[t!]
\centering
\begin{tikzpicture}
		
		\draw (-3,0) -- (0,3) -- (3,0) -- (0,-3) -- cycle;
		
		\draw[fill=black,opacity=0.2] (-1,0) circle (0.07);
		\draw[fill=black] (-1+3,0-3) circle (0.07);
		\draw[fill=black] (1,0) circle (0.07);
		\fill[yellow, opacity = 0.2] (-1.5,-0.8) -- (-0.95,-1.35) -- (1.3-3,0.9-3) -- (1.3-3-0.55,0.9-3+0.55) -- cycle;
		
		\draw[fill=black] (-1.5,-0.8) circle (0.07);
		\draw[fill=black,opacity=0.2] (1.3,0.9) circle (0.07);
		\draw[fill=black] (1.3-3,0.9-3) circle (0.07);
		
		\draw[fill=black] (-1.5,-1.5) circle (0.07);
		\draw[opacity=.3] (-0.3,-1.3) -- (1,0) -- (3.3,-2.3);
		\draw[opacity=.3] (-0.3,-0.7) -- (2,-3) -- (3.3,-1.7);

		\draw[opacity=.3] (-1.5,-0.8) -- (-2.5,-1.8);
		\draw[opacity=.3] (-1.5,-0.8) -- (-0.5,-1.8);
		
		\draw[opacity=.3] (1.3-3,0.9-3) -- (1.3-2,0.9-2);
		\draw[opacity=.3] (1.3-3,0.9-3) -- (1.3-4,0.9-2);
		
		\draw[dashed, color=blue] (-1.5,-1.5) -- (1.5,1.5);
		
		\node[left,opacity=0.3] at (-1.05,0) {$4^+$};
		\node[right] at (2.1,-3) {$4$};
		\node[right] at (1.05,0) {$3$};	
		
		\node[left] at (-1.55,-1.55) {$x$};
		
		\node[above] at (-1.5,-0.8) {$1$};
		\node[below] at (1.3-3,0.9-3) {$2$};
		\node[right,opacity=0.3] at (1.3,0.9) {$2^+$};
		\end{tikzpicture}
\caption{After (temporarily) relabeling the points $2^- \to 2$ and $4^- \to 4$, we have a configuration where $1>x>2$ and $3>4$, with all other pairs of points spacelike separated. This is the same configuration as in figure~\ref{fig:configcontinuousspinblocktwo} of appendix~\ref{sec:lorentzianblock}, where we compute the Lorentzian integral for a conformal block. The integration region for $x$ is shaded yellow. Importantly, it stays away from $3,4$, so the $3\to 4$ limit can be computed inside the integrand.}
\label{eq:lorentzianblockafterrelabeling}
\end{figure}

Now the benefit of performing the light-transform becomes clear. The integral~\eqref{eq:integralforh} over the diamond $2<x<1$ precisely takes the form of a well-known Lorentzian integral for a conformal block. Note that the integral (\ref{eq:integralforh}) is conformally-invariant and is an eigenfunction of the conformal Casimir operators acting on points $1,2$ (equivalently $3,4$) by construction. Importantly, the integral over $x$ stays away from the region near $3,4$, see figure~\ref{eq:lorentzianblockafterrelabeling}. Thus, we can determine its behavior in the OPE limit $3\to 4$ by simply taking the limit inside the integrand. (This limit corresponds to the Regge limit of the physical operators at $1,2^+,3,4^+$.) Any eigenfunction of the conformal Casimirs is fixed by its OPE limit, so this determines the full function. Thus, it's clear that $H_{\De,J}$ is proportional to a conformal block, with external operators $\tl \f_1^\dagger,\dots,\tl \f_4^\dagger$, and an exchanged operator with the quantum numbers of $\cO^\mathrm{F\dagger}$.

We perform this analysis in detail in appendix~\ref{sec:lorentzianblock}. Using the result (\ref{eq:resultforlorentzianintegralforblock}), we find
\be
H_{\De,J}(x_i) &= \frac{q_{\De,J}}{(-x_{12}^2)^{\frac{\tl\De_1+\tl\De_2}{2}} (-x_{34}^2)^{\frac{\tl\De_3+\tl\De_4}{2}}} \p{\frac{x_{14}^2}{x_{24}^2}}^{\frac{\tl\De_2-\tl\De_1}{2}} \p{\frac{x_{14}^2}{x_{13}^2}}^{\frac{\tl\De_3-\tl\De_4}{2}} G_{J+d-1,\De-d+1}^{\tl\De_i}(\chi,\bar\chi),
\label{eq:eqnforh}
\ee
where
\be
q_{\De,J} &=-(-1)^J\frac {2^{2d-2}\vol(\SO(d-2))}{\p{\<\f_1 \f_2\tl \cO^\dagger \>,\<\tl \f_1^\dagger \tl \f_2^\dagger \cO\>}_E}\frac{\mu(\De,J)S_E(\f_1 \f_2 [ \tl \cO])L(\tl \f_1 \tl \f_2[\cO])}{L(\f_3 \f_4[\cO])} b^{\tl \De_1,\tl \De_2}_{J+d-1,\De-d+1}\nn\\
&= - 2^{2d} \vol(\SO(d-2)) \frac{\G(\frac{\De+J+\De_1-\De_2}{2})\G(\frac{\De+J-\De_1+\De_2}{2})\G(\frac{\De+J+\De_3-\De_4}{2}) \G(\frac{\De+J-\De_3+\De_4}{2})}{16\pi^2 \G(\De+J)\G(\De+J-1)}.
\ee
(The quantity $b_{\De,J}^{\De_1,\De_2}$ is defined in (\ref{eq:bcoeff}) and the conformal block $G$ is defined in appendix~\ref{app:3ptgluing}.) Factors other than $b_{\De,J}^{\De_1,\De_2}$ come from~\eqref{eq:integralforh} and the structures~\eqref{eq:structure12forH} and~\eqref{eq:dualconfig}.
In the proof of the Lorentzian inversion formula in \cite{Simmons-Duffin:2017nub}, performed without using the light transform, one obtains an expression for $H_{\De,J}$ as an integral over a region totally spacelike from $1,2^+,3,4^+$, which is harder to understand. 

\subsubsection{Writing in terms of cross-ratios}

Finally, let us replace $2^+\to 2$ and $4^+\to 4$ so that the physical operators are again at the points $1,2,3,4$. The inversion formula reads
\be
C(\De,J)
&=\int_{\substack{4>1 \\ 2>3}} \frac{d^d x_1 \cdots d^d x_4}{\vol(\tl{\SO}(d,2))} \<\O|[\f_{4}, \f_1] [\f_{2},\f_3]|\O\> (\tsym_2^{-1} \tsym_4^{-1} H_{\De,J}(x_i)) + (1\leftrightarrow 2).
\label{eq:nicerintegralforC}
\ee
Here, $\tsym_i^{-1}$ denotes a shift $x_i\to x_i^-$ or, more generally, application of the $\tsym^{-1}$ to the operator at $i$-th position.
In the integrand, we can isolate quantities that depend only on cross-ratios, times a universal dimensionful factor $|x_{12}|^{-2d} |x_{34}|^{-2d}$,
\be
\<\Omega|[\f_4,\f_1][\f_2,\f_3]|\Omega\> (\tsym_2^{-1} \tsym_4^{-1} H_{\De,J}(x_i)) &= \frac{1}{|x_{12}|^{2d}|x_{34}|^{2d}}\frac{\<\Omega|[\f_4,\f_1][\f_2,\f_3]|\Omega\>}{T^{\De_i}(x_i)} G_{J+d-1,\De-d+1}^{\tl\De_i}(\chi,\bar\chi),
\ee
where
\be
\label{eq:absolutevaluestructure}
T^{\De_i}(x_i) &\equiv \frac{1}{|x_{12}|^{\De_1+\De_2}|x_{34}|^{\De_3+\De_4}} \p{\frac{|x_{14}|}{|x_{24}|}}^{\De_2-\De_1} \p{\frac{|x_{14}|}{|x_{13}|}}^{\De_3-\De_4}.
\ee
Since we now have a fixed causal ordering of the points, we do not have to worry about an $i\e$ prescription in these expressions and we can simply take absolute values of spacetime intervals.

We can gauge-fix (\ref{eq:nicerintegralforC}) to obtain an integral over cross-ratios alone. As explained in \cite{Simmons-Duffin:2017nub},\footnote{We use a definition of the measure on $\tl\SO(d,2)$ which differs from the one~\cite{Simmons-Duffin:2017nub} by a factor of $2^d$.} the measure becomes
\be
\int \frac{d^d x_1\cdots d^d x_4}{\vol(\tl\SO(d,2))} \frac{1}{|x_{12}|^{2d}|x_{34}|^{2d}}
&\to
\frac{1}{2^{2d} \vol(\SO(d-2))} \int_0^1 \int_0^1 \frac{d\chi d\bar\chi}{\chi^2 \bar\chi^2} \left|\frac{\bar\chi-\chi}{\chi\bar\chi}\right|^{d-2}.
\ee
Putting everything together, we find
\be
\label{eq:simonsformula}
C(\De,J) &= \frac{q_{\De,J}}{2^{2d} \vol(\SO(d-2))}\left[ \int_0^1 \int_0^1 \frac{d\chi d\bar\chi}{\chi^2 \bar\chi^2} \left|\frac{\bar\chi-\chi}{\chi\bar\chi}\right|^{d-2} \frac{\<\Omega|[\f_4,\f_1][\f_2,\f_3]|\Omega\>}{T^{\De_i}(x_i)} G_{J+d-1,\De-d+1}^{\tl\De_i}(\chi,\bar\chi)\right.\nn\\
&\left.\qquad\quad\quad +(-1)^J \int_{-\oo}^0\int_{-\oo}^0 \frac{d\chi d\bar\chi}{\chi^2 \bar\chi^2} \left|\frac{\bar\chi-\chi}{\chi\bar\chi}\right|^{d-2} \frac{\<\Omega|[\f_4,\f_2][\f_1,\f_3]|\Omega\>}{T^{\De_i}(x_i)} \hat G_{J+d-1,\De-d+1}^{\tl\De_i}(\chi,\bar\chi)\right].
\ee
Here, $\hat G_{\De,J}(\chi,\bar\chi)$ denotes the solution to the Casimir equation that behaves as $(-\chi)^{\frac{\De-J}{2}}(-\bar\chi)^{\frac{\De+J}{2}}$ for negative cross-ratios satisfying $|\chi|\ll|\bar\chi|\ll 1$. This precisely coincides with Caron-Huot's Lorentzian inversion formula.

\subsubsection{A natural formula for the Lorentzian block}
\label{sec:naturalformula}

To make it easy to generalize the above result to arbitrary representations, let us write it in a more transparent way. First we need to introduce more flexible notation for a conformal block. Let
\be
\frac{\<\cO_1\cO_2\cO\>\<\cO_3\cO_4\cO\>}{\<\cO\cO\>}
\ee
denote the conformal block formed by gluing the three-point structures in the numerator using the two-point structure in the denominator. We describe the gluing procedure in more detail in appendix~\ref{app:3ptgluing}. In particular, the gluing procedure is well-defined (for a restricted causal configuration) even if $\cO$ is a continuous-spin operator.
Using this notation, the coefficient function $C(\De,J)$ is defined by
\be
\label{eq:defofC}
\<\f_1\f_2\f_3\f_4\>_\Omega &= \sum_{J=0}^\oo \int_{\frac d 2 - i\oo}^{\frac d 2 + i\oo}\frac{d\De}{2\pi i} C(\De,J) \frac{\<\f_1\f_2 \cO\>\<\f_3\f_4\cO\>}{\<\cO\cO\>},
\ee
where $\cO$ has dimension $\De$ and spin $J$.

Using the same notation, we claim that the function $H_{\De,J}(x_i)$ in (\ref{eq:nicerintegralforC}) is given by
\be
\label{eq:claimabouth}
H_{\De,J}(x_i) &= -\frac{1}{2\pi i} \frac{(\tsym_2\<\f_1\f_2 \wL[\cO]\>)^{-1} (\tsym_4\<\f_3\f_4\wL[\cO]\>)^{-1}}{\<\wL[\cO]\wL[\cO]\>^{-1}},\qquad (1>2,3>4).
\ee
In the numerator, $(\tsym_2\<\f_1\f_2 \wL[\cO]\>)^{-1}$ is the dual structure to $\tsym_2\<\f_1\f_2 \wL[\cO]\>$ via the three-point pairing (\ref{eq:lorentzianpairing}). It is given by (\ref{eq:dualconfig}), with the replacement $3,4\to 1,2$. Note that while we have written the structures in the numerators in terms of light transforms of time ordered products, they can alternatively be written in terms of Wightman functions for the kinematics we are considering, since
\be
\tsym_2\<\f_1 \f_2 \wL[\cO]\> &= \tsym_2\<0|\f_2 \wL[\cO] \f_1|0\>\qquad (\textrm{when}\ 1>2,\quad 1,2\approx 0), \nn\\
\tsym_4\<\f_3 \f_4 \wL[\cO]\> &= \tsym_4\<0|\f_4 \wL[\cO] \f_3|0\>\qquad (\textrm{when}\ 3>4,\quad 3,4\approx 0).
\ee

The structure $\<\wL[\cO]\wL[\cO]\>^{-1}$ in the denominator is dual to the double light-transform of the time-ordered two-point function $\<\cO\cO\>$ via the conformally-invariant two-point pairing,
\be
\label{eq:twoptdenominatorpairing}
\p{\<\wL[\cO]\wL[\cO]\>^{-1}, \<\wL[\cO]\wL[\cO]\>}_L &= 1.
\ee
Here the pairing $(\cdot,\cdot)_L$ for two-point functions is defined in~\eqref{eq:2ptpairingL}. In order for the pairing in (\ref{eq:twoptdenominatorpairing}) to be conformally-invariant, $\<\wL[\cO]\wL[\cO]\>^{-1}$ must transform like a two-point function of $\cO^\mathrm{F}$.

We have already computed the three-point structures in the numerator, so to verify (\ref{eq:claimabouth}), we  need to compute $\<\wL[\cO]\wL[\cO]\>$. Here, it is important to treat two-point structures as distributions. By lemma~\ref{lemma:timeorderedlighttransform}, $\<\cO(x_1,z_1)\wL[\cO](x_2,z_2)\>$ vanishes if $x_2 > x_1$ or $x_2 < x_1$ --- i.e.\ it vanishes almost everywhere. However, it is nonzero if $x_1$ is precisely lightlike from $x_2$. Specifically, $\<\cO(x_1,z_1)\wL[\cO](x_2,z_2)\>$ is a distribution localized where $x_2$ is on the past lightcone of $x_1$.\footnote{Note that this is different from treating two-point functions as physical Wightman functions, so there is no contradiction with previous discussion.} In fact, it is proportional to the integral kernel for the ``floodlight transform" $\wF$.

Let us now actually compute $\<\wL[\cO]\wL[\cO]\>$. It is useful to think of this structure as an integral kernel $K$, defined by
\be
\label{eq:kerneltwopt}
(Kf)(x,z) &\equiv \int d^d x' D^{d-2} z'\,\<\wL[\cO](x,z)\wL[\cO](x',z')\> f(x',z').
\ee
In (\ref{eq:kerneltwopt}), we can integrate one of the $\wL$-transforms by parts, giving
\be
\label{eq:integratebyparts}
(Kf)(x,z) &= \int d^d x' D^{d-2} z'\,\<\wL[\cO](x,z)\cO(x',z')\> (\tsym^{-1} \wL[f])(x',z').
\ee

To simplify (\ref{eq:integratebyparts}) further, we can express the time-ordered two-point function $\<\cO\cO\>$ in terms of integral transforms and use the algebra derived in section~\ref{sec:algebraoftransforms}. When $x,x'$ are spacelike, $\<\cO(x,z)\cO(x',z')\>$ is precisely the kernel for $\wS$. However, $\wS$ is supported only in the region $x\approx x'$, whereas the time-ordered two-point function has support everywhere. More precisely, keeping track of the phases as we move $x,x'$ into different Poincare patches, we have
\be
\<\cO(x,z)\cO(x',z')\> &= \frac{(-2z\.z' (x-x')^2 + 4 z\.(x-x') z'\.(x-x'))^J}{((x-x')^{2}+i\e)^{\De+J}} \nn\\
&= \wS\p{1 + \sum_{n=1}^\oo e^{-in\pi(\De+J)} \tsym^n + \sum_{n=1}^\oo e^{-in\pi(\De+J)} \tsym^{-n}} \nn\\
&= \wS \frac{-2i \tsym\sin\pi(\De+J)}{(\tsym-e^{i\pi(\De+J)})(\tsym-e^{-i\pi(\De+J)})}.
\ee
Plugging this into (\ref{eq:integratebyparts}), we find
\be
K &= \wL \wS \frac{-2i \tsym\sin\pi(\De+J)}{(\tsym-e^{i\pi(\De+J)})(\tsym-e^{-i\pi(\De+J)})} \tsym^{-1} \wL \nn\\
&= \wS \frac{-2i\sin\pi(\De+J)}{(\tsym-e^{i\pi(\De+J)})(\tsym-e^{-i\pi(\De+J)})} \wL^2 \nn\\
&= \frac{-2\pi i}{\De+J-1} \wS,
\ee
where in the second line we used that $\wL,\wS,\tsym$ commute with each other, together with the formula  $\wL^2 = f_L(J+d-1,\De-d+1,\tsym)$, where $f_L$ is given in equation~(\ref{eq:lightsquare}). The arguments of $f_L$ come from the fact that $K$ acts on a representation with dimension $J+d-1$ and spin $\De-d+1$.

The kernel of $\wS$ in the last line is the two-point function of an operator with spin $1-\De$ and dimension $1-J$. Thus, using our two-point pairing (\ref{eq:2ptpairingL}), we find
\be
\label{eq:dualtwopt}
\<\wL[\cO]\wL[\cO]\>^{-1} &= -\frac{\De+J-1}{2\pi i} 2^{2d-2}\vol(\SO(d-2))\<\cO^\mathrm{F}\cO^\mathrm{F}\>,
\ee
where $\<\cO^\mathrm{F}\cO^\mathrm{F}\>$ is the standard two-point structure (\ref{eq:standardtwoptconvention}) for an operator with dimension $J+d-1$ and spin $\De-d+1$. Combining this with the three-point structures in the numerator, and comparing with the result (\ref{eq:eqnforh}) for $H_{\De,J}(x_i)$, we verify (\ref{eq:claimabouth}).

Note that (\ref{eq:claimabouth}) is independent of a choice of normalization of the integral transform $\wL$. In fact, it depends only on the three-point structures $\<\f_1\f_2\cO\>,\<\f_3\f_4\cO\>$, the two-point structure $\<\cO\cO\>$, and the existence of a conformally-invariant map between representations $\cP_{\De,J,\l}$ and $\cP_{1-J,1-\De,\l}$ (which $\wL$ implements). The formula would still be true if we chose different normalization conventions for two and three-point functions, because this would change the definition of $C(\De,J)$ in a compatible way, via (\ref{eq:defofC}). Because it is essentially independent of conventions, we call (\ref{eq:claimabouth}) a ``natural" formula.

\subsection{Generalization to arbitrary representations}
\label{sec:generalizedinversiontwo}

\subsubsection{The light transform of a partial wave}

The derivation in the previous section is straightforward to generalize to the case of arbitrary conformal representations $\f_i\to\cO_i$.  In this case, three-point functions admit multiple conformally-invariant structures $\<\cO_1\cO_2\cO\>^{(a)}$, so partial waves $P_{\cO,(a)}$ carry an additional structure label.\footnote{The possible structures in a three-point function of spinning operators are classified in \cite{Kravchuk:2016qvl}.} They are defined by
\be
\label{eq:multipointcompletenessgeneralop}
\<V_3 V_4 \cO_1 \cO_2\>_\Omega &= \sum_{\rho,a} \int_{\frac d 2}^{\frac d 2 + i\oo} \frac{d\De}{2\pi i} \mu(\De,J) \int d^d x P_{\cO,(a)}(x_3,x_4,x) \<\tl \cO^\dag(x) \cO_1 \cO_2\>^{(a)}.
\ee
(Here, we implicitly contract the $\SO(d)$ indices of $P_{\cO,(a)}$ and the operator $\tl\cO^\dag$.)

The logic leading to the double-commutator integral (\ref{eq:afterdoublecommutator}) is essentially unchanged. 
We find
\be
&\wL[P_{\cO,(a)}](x_3,x_4,x,z)\nn\\
&=-(\<\cO_1 \cO_2  \tl \cO^\dagger\>^{(a)},\<\tl \cO_1^\dagger \tl \cO_{2}^\dagger\cO\>^{(b)})^{-1}_E \int_{\substack{x<1<4 \\ 3<2<x^+}} d^d x_1 d^d x_2 \<\O|[V_4, \cO_1] [\cO_2,V_3]|\O\> \<0|\tl \cO_1^\dag \wL[\cO](x,z) \tl \cO_2^\dag|0\>^{(b)} \nn\\
&\quad+ (1\leftrightarrow 2),
\ee
where $(\<\cO_1 \cO_2  \tl \cO^\dagger\>^{(a)},\<\tl \cO_1^\dagger \tl \cO_{2}^\dagger\cO\>^{(b)})^{-1}_E$ is the inverse of the three-point pairing (\ref{eq:euclideanthreeptpairing}) defined by
\be
	(\<\cO_1 \cO_2  \tl \cO^\dagger\>^{(a)},\<\tl \cO_1^\dagger \tl \cO_{2}^\dagger\cO\>^{(b)})^{-1}_E(\<\cO_1 \cO_2  \tl \cO^\dagger\>^{(c)},\<\tl \cO_1^\dagger \tl \cO_{2}^\dagger\cO\>^{(b)})_E=\delta_{a}^c
\ee

\subsubsection{The generalized Lorentzian inversion formula}
\label{sec:generalizedinversion}

To generalize the remaining steps leading to the Lorentzian inversion formula, we seemingly need to understand  of all the factors entering the expression for $H_{\De,J}(x_i)$  (\ref{eq:eqnforh}). However, this is unnecessary because the generalization is obvious from the natural formula (\ref{eq:claimabouth}).

The coefficient function $C_{ab}(\De,\rho)$ we would like to compute is defined by 
\be
\<\cO_1\cdots\cO_4\>_\Omega &= \sum_{\rho,a,b} \int_{\frac d 2 -i\oo}^{\frac d 2 + i\oo} \frac{d\De}{2\pi i} C_{ab}(\De,\rho) \frac{\<\cO_1\cO_2\cO^\dagger\>^{(a)} \<\cO_3\cO_4\cO\>^{(b)}}{\<\cO\cO^\dagger\>},
\ee
where $\cO$ has dimension $\De$ and $\SO(d)$-representation $\rho$. Here, we sum over principal series representations $\cE_{\De,\rho}$, as well as three-point structures $a,b$. The obvious generalization of~\eqref{eq:integralforcspacetime} and~(\ref{eq:claimabouth}) is
\be
C_{ab}(\De,\rho) &= -\frac{1}{2\pi i}  \int_{\substack{4>1\\2>3}} \frac{d^d x_1\cdots d^d x_4}{\vol(\tl \SO(d,2))} \<\O|[\cO_4, \cO_1] [\cO_2,\cO_3]|\O\>\nn\\
&\qquad\qquad \qquad\qquad  \x \tsym_2^{-1} \tsym_4^{-1}\frac{\p{\tsym_2\<\cO_1 \cO_2 \wL[\cO^\dagger]\>^{(a)}}^{-1}\p{\tsym_4\<\cO_4 \cO_3 \wL[\cO]\>^{(b)}}^{-1}}{\<\wL[\cO]\wL[\cO^\dagger]\>^{-1}} +(1\leftrightarrow 2).
\label{eq:obviousgeneralization}
\ee
The dual structures in the numerator are defined by
\be
\p{\p{\tsym_2\<\cO_1 \cO_2 \wL[\cO^\dagger]\>^{(a)}}^{-1}, \tsym_2\<\cO_1\cO_2\wL[\cO^\dagger]\>^{(c)}}_L &= \de_a^c,\nn\\
\p{\p{\tsym_4\<\cO_4 \cO_3 \wL[\cO]\>^{(b)}}^{-1}, \tsym_4\<\cO_4\cO_3\wL[\cO]\>^{(d)}}_L &= \de_b^d,
\label{eq:threeptpairingsgeneralized}
\ee
where $(\cdot,\cdot)_L$ is the three-point pairing defined in (\ref{eq:lorentzianpairing}). 
The two-point structure in the denominator is the dual of $\<\wL[\cO]\wL[\cO^\dagger]\>$ via the two-point pairing (\ref{eq:2ptpairingL}). 

Note that the structure $\p{\tsym_2\<\cO_1 \cO_2 \wL[\cO^\dagger]\>^{(a)}}^{-1}$ transforms like a three-point function of representations $\<\tl \cO_1^\dag \tl \cO_2^\dag \cO^{\dag\mathrm{F}}\>$ and similarly for the operators 3 and 4. In (\ref{eq:obviousgeneralization}), we are implicitly contracting Lorentz indices of $\cO_i$ with their dual indices in these structures.

\subsubsection{Proof using weight-shifting operators}

Equation~\eqref{eq:obviousgeneralization} follows if we prove the generalization of the expression~\eqref{eq:claimabouth} for $H$, with $H$ defined using the appropriate generalization of~\eqref{eq:integralforh}. Specifically, the definition of $H$ becomes
\be
	H_{\De,\rho,(ab)}(x_i)
	&=-\mu(\De,\rho^\dagger)S_E(\cO_1 \cO_2 [\tl \cO^\dagger])^c{}_a(\<\cO_1 \cO_2  \tl \cO^\dagger\>^{(c)},\<\tl \cO_1^\dagger \tl \cO_{2}^\dagger\cO\>^{(d)})^{-1}\times\nn\\
	&\quad\times\int_{2<x<1} d^d x D^{d-2} z \<0|\tl \cO_1^\dagger \wL[\cO](x,z) \tl \cO_{2^+}^\dagger|0\>^{(d)}(\<0|\cO_{4^+} \wL[\cO](x,z) \cO_3|0\>^{(b)})^{-1}.
	\label{eq:integralforhgeneralized}
\ee
We want to prove that
\be
H_{\De,\rho,(ab)}(x_i)
&=-\frac{1}{2\pi i}  \frac{\p{\tsym_2\<\cO_1 \cO_2 \wL[\cO^\dagger]\>^{(a)}}^{-1}\p{\tsym_4\<\cO_4 \cO_3 \wL[\cO]\>^{(b)}}^{-1}}{\<\wL[\cO]\wL[\cO^\dagger]\>^{-1}}.\label{eq:claimforhgeneralized}
\ee
Our proof will proceed in two steps. Here we are going to show that if for a given $\rho$~\eqref{eq:claimforhgeneralized} is valid for some ``seed'' choice of $\SO(d)$ irreps of external operators, it is then valid for all choices of external irreps. In appendix~\ref{app:proof} using methods of~\cite{Karateev:2017jgd} we show that validity of~\eqref{eq:claimforhgeneralized} for traceless-symmetric $\rho$ implies its validity for seed blocks for all $\rho$. Together these statements imply~\eqref{eq:claimforhgeneralized} in full generality.

\paragraph{Generalizing the external representations}

It is convenient to consider the structure defined by
\be
	T_a\equiv \mu(\De,\rho^\dagger)S_E(\cO_1 \cO_2  [\tl \cO^\dagger])^c{}_a(\<\cO_1 \cO_2  \tl \cO^\dagger\>^{(c)},\<\tl \cO_1^\dagger \tl \cO_{2}^\dagger\cO\>^{(d)})^{-1}_E\<\tl \cO_1^\dagger \tl \cO_{2^+}^\dagger\cO\>^{(d)}.
\ee
We can check that
\be
	T_a={(\<\cO^\dagger\cO\>,\<\tl\cO^\dagger\tl\cO\>)_E}(\<\cO_1\cO_2\wSE[\cO^\dagger]\>^{(a)})^{-1}_E,
\ee
where all pairings and inverses are Euclidean. Indeed, we can compute the Euclidean paring
\be
	(T_d,\<\cO_1\cO_2\wSE[\cO^\dagger]\>^{(a)})_E=&S_E(\cO_1\cO_2[\cO^\dagger])^a{}_b(T_a,\<\cO_1\cO_2\tl\cO^\dagger\>^{(b)})_E\nn\\=&\mu(\De,\rho^\dagger)S_E(\cO_1\cO_2[\cO^\dagger])^a{}_bS_E(\cO_1 \cO_2  [\tl \cO^\dagger])^b{}_d\nn\\
	=&\mu(\De,\rho^\dagger)\cN(\De,\rho^\dagger)\delta^a_d={(\<\cO^\dagger\cO\>,\<\tl\cO^\dagger\tl\cO\>)}_E\delta^a_d.
\ee
Here we used the relation~\eqref{eq:plancherelvsshadowsq} between the Plancherel measure and the square of the Euclidean shadow transform. Importance of the structures $T_a$ comes from the fact that it is the light transform of their Wick rotation which enters~\eqref{eq:integralforhgeneralized}. 

We now choose some other $\SO(d)$ irreps $\rho'_1$ and $\rho'_2$ for operators $\cO'_1$ and $\cO'_2$ such that there is a unique tensor structure\footnote{In odd dimensions and for fermionic $\rho$ the number of tensor structures is always even, and so it is not possible to make this choice. However, there we can make a choice such that there is only one parity-even structure, which will be good enough.}
\be
	\<\cO'_1\cO'_2\tl\cO^\dagger\>.
\ee
We then can write
\be
	 T_a={(\<\cO^\dagger\cO\>,\<\tl\cO^\dagger\tl\cO\>)_E}\tsym^{-1}_2\cD_{12,a}\tsym_2 (\<\cO'_1\cO'_2\wSE[\cO^\dagger]\>)^{-1}_E,
\ee
where $\cD_{12,a}$ are contractions of weight-shifting operators acting on points $1$ and $2$~\cite{Costa:2011dw,Karateev:2017jgd}.\footnote{Note that $\tsym_2^{-1}\cD_{12,d}\tsym_2$ are differential operators which can be interpreted in Euclidean signature. In particular, if $\cD_{12,d}=D_{1,A}D_2^A$ for $A$ transforming in an irreducible representation $W$ of the conformal group then $\tsym_2^{-1}\cD_{12,d}\tsym_2$ is proportional to $\cD_{12,d}$ with coefficient equal to the eigenvalue of $\tsym$ in $W$.} We can use this to write
\be
	H_{\De,\rho,(ab)}(x_i)=\cD_{12,a} H'_{\De,\rho,(b)}(x_i),
\ee
where $H'$ is given by~\eqref{eq:integralforhgeneralized} with $\cO'_1$ and $\cO'_2$ instead of $\cO_1$ and $\cO_2$, and using the unique tensor structure on the left of $H'$. 

On the other hand, we can write
\be
\delta^a_d&=\frac{1}{(\<\cO^\dagger\cO\>,\<\tl\cO^\dagger\tl\cO\>)_E}(T_d,\<\cO_1\cO_2\wSE[\cO^\dagger]\>^{(a)})_E\nn\\
&=(\tsym_2^{-1}\cD_{12,d}\tsym_2(\<\cO'_1\cO'_2\wSE[\cO^\dagger]\>)^{-1},\<\cO_1\cO_2\wSE[\cO^\dagger]\>^{(a)})_E\nn\\
&=((\<\cO'_1\cO'_2\wSE[\cO^\dagger]\>)^{-1},(\tsym_2^{-1}\cD_{12,d}\tsym_2)^*\<\cO_1\cO_2\wSE[\cO^\dagger]\>^{(a)})_E,
\ee
where we integrated the differential operators $\tsym_2^{-1}\cD_{12,d}\tsym_2$ by parts inside the Euclidean pairing. This produces new operators $\cD_{12,d}^*$, which are again contractions of weight-shifting operators.\footnote{For details see appendix~\ref{app:operations} and~\cite{Karateev:2017jgd,ShadowFuture}.} We thus conclude that
\be
	(\tsym_2^{-1}\cD_{12,d}\tsym_2)^*\<\cO_1\cO_2\wSE[\cO^\dagger]\>^{(a)}=\delta^a_d\<\cO'_1\cO'_2\wSE[\cO^\dagger]\>.
\ee
Canceling $\wSE$ on both sides (it is invertible on generic tensor structures) we find
\be\label{eq:D12definitionwithoutinverses}
	(\tsym_2^{-1}\cD_{12,d}\tsym_2)^*\<\cO_1\cO_2\cO^\dagger\>^{(a)}=\delta^a_d\<\cO'_1\cO'_2\cO^\dagger\>.
\ee
We now want to show that 
\be\label{eq:D12target}
	\cD_{12,a}(\tsym_2\<\cO'_1\cO'_2 \wL[\cO^\dagger]\>)^{-1}_L=(\tsym_2\<\cO_1\cO_2 \wL[\cO^\dagger]\>^{(a)})^{-1}_L
\ee
where the inverse structure is understood with respect to Lorentzian pairing. This follows by doing the above calculation in reverse and in Lorentzian signature. First, we apply $\wL$ to both sides of~\eqref{eq:D12definitionwithoutinverses} and use $\tsym^*=\tsym^{-1}$,
\be
	\tsym_2^{-1}\cD_{12,d}^*\tsym_2\<\cO_1\cO_2\wL[\cO^\dagger]\>^{(a)}=\delta^a_d\<\cO'_1\cO'_2\wL[\cO^\dagger]\>,
\ee
Then, we apply $\tsym_2$ to both sides and take Lorentzian contraction with $(\tsym_2\<\cO'_1\cO'_2\wL[\cO^\dagger]\>)^{-1}_L$
\be
	((\tsym_2\<\cO'_1\cO'_2\wL[\cO^\dagger]\>)^{-1}_L,\cD_{12,d}^*\tsym_2\<\cO_1\cO_2\wL[\cO^\dagger]\>^{(a)})_L=\delta^a_d,
\ee
and finally integrate by parts,
\be
	(\cD_{12,d}(\tsym_2\<\cO'_1\cO'_2\wL[\cO^\dagger]\>)^{-1}_L,\tsym_2\<\cO_1\cO_2\wL[\cO^\dagger]\>^{(a)})_L=\delta^a_d.
\ee
This is equivalent to~\eqref{eq:D12target} The crucial point here is that integration by parts leads to the same operation on the weight-shifting operators both in Euclidean and Lorentzian signature (on integer-spin operators). A way to summarize this calculation is by saying that 
\be
	(\tsym_2\<\cO_1\cO_2\wL[\cO^\dagger]\>)^{-1}_L\quad\text{and}\quad\tsym_2(\<\cO_1\cO_2\wSE[\cO^\dagger]\>)^{-1}_E
\ee
have the same transformation properties under weight-shifting operators acting on 1 and 2.

This implies that if~\eqref{eq:claimforhgeneralized} is true for $\cO'_1$ and $\cO'_2$, it is also true for $\cO_1$ and $\cO_2$, since we can simply apply $\cD_{12,a}$ in both~\eqref{eq:integralforhgeneralized} and~\eqref{eq:claimforhgeneralized}. Since exactly the same tensor structure appears for the operators $\cO_3,\cO_4$ in~\eqref{eq:integralforhgeneralized} and~\eqref{eq:claimforhgeneralized}, an analogous (even simpler) argument works for this tensor structure as well. In conclusion, if~\eqref{eq:claimforhgeneralized} holds for a seed conformal block, it holds for all conformal blocks with the same $\rho$.

\section{Conformal Regge theory}
\label{sec:conformalregge}

\subsection{Review: Regge kinematics}

Consider a time-ordered four-point function of scalar operators $\<\f_1\cdots \f_4\>$. Its conformal block expansion in the $12\to 34$ channel takes the form
\be
\label{eq:usualOPE}
\<\f_1(x_1)\cdots \f_4(x_4)\> &=  \sum_{\De,J} p_{\De,J} G^{\De_i}_{\De,J}(x_i) \nn\\
&= \frac{1}{(x_{12}^2)^{\frac{\De_1+\De_2}{2}} (x_{34}^2)^{\frac{\De_3+\De_4}{2}}} \p{\frac{x_{14}^2}{x_{24}^2}}^{\frac{\De_2-\De_1}{2}} \p{\frac{x_{14}^2}{x_{13}^2}}^{\frac{\De_3-\De_4}{2}}
\sum_{\De,J} p_{\De,J} G^{\De_i}_{\De,J}(\chi,\bar\chi),
\ee
where $p_{\De,J}$ are products of OPE coefficients.
This expansion is convergent whenever $\chi,\bar\chi \in \C\backslash[1,\oo)$ \cite{Pappadopulo:2012jk}.
However, it fails to converge in the Regge limit.\footnote{The other OPE channels $14\to 23$ and $13\to 24$ are still convergent, though they are approaching the boundaries of their regimes of validity, as discussed in the introduction.}

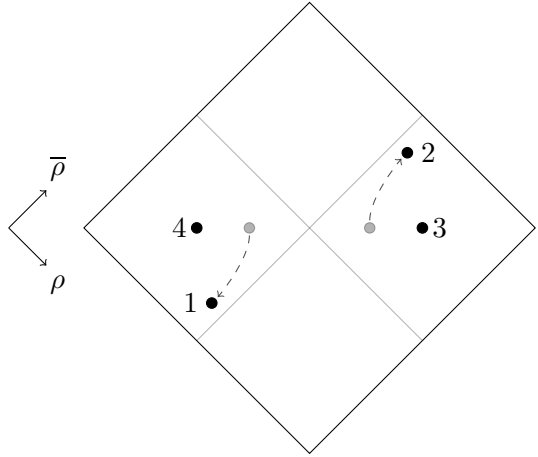
\begin{figure}[ht!]
	\centering
		\begin{tikzpicture}

		\draw[->] (-4,0) -- (-3.5,0.5);		
		\draw[->] (-4,0) -- (-3.5,-0.5);
		\draw[dashed,->,opacity=0.7] (0.8,0.1) to[out=90,in=-130,distance=0.25cm] (1.22,0.92);	
		\draw[dashed,->,opacity=0.7] (-0.8,-0.1) to[out=-90,in=50,distance=0.25cm] (-1.22,-0.92);	
		
		\draw (-3,0) -- (0,3) -- (3,0) -- (0,-3) -- cycle;
				
		\draw[fill=black] (-1.3,-1) circle (0.07); 
		\draw[fill=black] (1.3,1) circle (0.07);  
		\draw[fill=black] (-1.5,0) circle (0.07); 
		\draw[fill=black] (1.5,0) circle (0.07);  
		\draw[fill=black,opacity=0.3] (0.8,0) circle (0.07);  
		\draw[fill=black,opacity=0.3] (-0.8,0) circle (0.07);  
				
		\draw[opacity=.3] (-1.5,-1.5) -- (1.5,1.5);
		\draw[opacity=.3] (-1.5,1.5) -- (1.5,-1.5);
		
		\node[left] at (-1.35,-1) {$1$};
		\node[right] at (1.35,1) {$2$};	
		\node[left] at (-1.5,0) {$4$};
		\node[right] at (1.5,0) {$3$};
		\node[below] at (-3.35,-0.5) {$\rho$};
		\node[above] at (-3.35,0.5) {$\bar\rho$};
		
		\end{tikzpicture}
		\caption{The Regge limit in the configuration (\ref{eq:lightconecoordinatesforRegge}). We boost points $1$ and $2$  while keeping points $3$ and $4$ fixed. This configuration is related by an overall boost to the one in figure~\ref{fig:reggelimit}.}
		\label{fig:reggelimittwo}
\end{figure}

To reach the Regge regime, which was originally described for CFT correlators in \cite{Cornalba:2007fs}, let us place the operators in a 2d Lorentzian plane with lightcone coordinates
\be
\label{eq:lightconecoordinatesforRegge}
x_1 &= (-\rho,-\bar\rho),\nn\\
x_2 &= (\rho,\bar\rho), \nn\\
x_3 &= (1,1), \nn\\
x_4 &= (-1,-1).
\ee
The usual cross-ratios are given by
\be
\chi = \frac{4\rho}{(1+\rho)^2},\qquad \bar\chi &= \frac{4\bar\rho}{(1+\bar\rho)^2}.
\ee
It is also useful to introduce polar coordinates
\be
\rho = r e^{i\th}=rw, \qquad\bar\rho = re^{-i\th}=rw^{-1}.
\ee
In Euclidean signature, $r$ and $\th$ are real. By contrast in Lorentzian signature, $r$ is real, $\th$ becomes pure-imaginary (it is conjugate to a boost), and $\rho,\bar\rho$ become independent real variables. To reach the Regge regime, we apply a large boost to operators $1$ and $2$, while keeping $3$ and $4$ fixed (figure~\ref{fig:reggelimittwo}). More precisely, we take
\be
\th &= it + \e,\qquad (t\to \oo),
\ee
so that
\be
\label{eq:boostkinematics}
\rho = r e^{-t+i\e},\quad\bar\rho = re^{t-i\e},\qquad (t\to \oo).
\ee
Here, we use the correct $i\e$ prescription to compute a time-ordered Lorentzian correlator when $t>0$. With this prescription, the cross-ratios behave as follows. As $t$-increases, $\chi$ moves toward zero. Meanwhile, $\bar\chi$ initially increases, then goes counterclockwise around $1$, and finally decreases back to zero (figure~\ref{fig:crossratiosregge}).

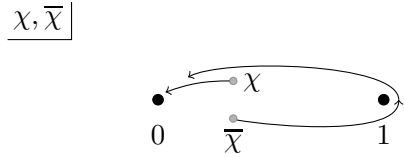
\begin{figure}[ht!]
	\centering
		\begin{tikzpicture}

\draw[] (-1.15,1.3) -- (-1.15,0.8) -- (-2,0.8);

\draw[fill=black] (0,0) circle (0.07);
\draw[fill=black] (3,0) circle (0.07);
\draw[fill=black,opacity=0.3] (1,-0.25) circle (0.05);
\draw[fill=black,opacity=0.3] (1,0.25) circle (0.05);
\draw[->] (0.95,0.25) to[out=-180,in=20] (0.1,0.1);
\draw[->] (1.05,-0.27) to[out=-10,in=-90,distance=0.5cm] (3.2,0);
\draw[->] (3.2,0) to[out=90,in=30,distance=0.5cm] (0.4,0.3);

	\node[above] at (-1.6,0.8) {$\chi,\bar\chi$}; 
	\node[below] at (0,-0.2) {$0$};
	\node[below] at (3,-0.2) {$1$};
	\node[below] at (1,-0.25) {$\bar\chi$};
	\node[above] at (1.25,-0.03) {$\chi$};
		\end{tikzpicture}
		\caption{The paths of the cross ratios $\chi,\bar\chi$ when moving from the Euclidean regime to the Regge regime. In the Euclidean regime, $\chi,\bar\chi$ are complex conjugates (gray points). As we boost $x_1,x_2$, the cross ratio $\chi$ decreases towards zero, while $\bar\chi$ moves counterclockwise around $1$ before decreasing towards zero. For sufficiently large $t$, $\bar\chi$ follows the same path as $\chi$, but we have separated the paths to clarify the figure.}
		\label{fig:crossratiosregge}
\end{figure}

The only difference between the Regge and $1\to 2$ OPE limits from the perspective of the cross-ratios $\chi,\bar\chi$ is the continuation of $\bar\chi$ around $1$. In both cases, we take $\chi,\bar\chi\to 0$. This is because the Regge limit resembles an OPE limit between points in different Poincare patches. This observation was made in \cite{Caron-Huot:2013fea}. Specifically, the configuration in figure~\ref{fig:reggelimittwo} is related by a boost to the one in figure~\ref{fig:reggelimitagain}. The Regge limit can thus be described as $1\to2^-$ and $3\to 4^-$. The cross-ratios $\chi,\bar\chi$ are unchanged when we apply $\tsym$ to any of the points, which is why they still go to zero in this limit.

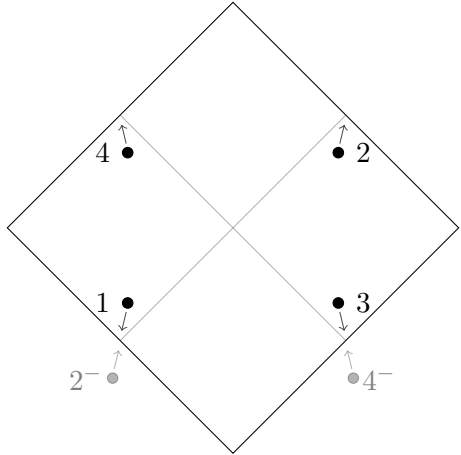
\begin{figure}[ht!]
	\centering
		\begin{tikzpicture}
		
		\draw (-3,0) -- (0,3) -- (3,0) -- (0,-3) -- cycle;
				
		\draw[fill=black] (-1.4,-1) circle (0.07); 
		\draw[fill=black] (1.4,1) circle (0.07);  
		\draw[fill=black] (-1.4,1) circle (0.07); 
		\draw[fill=black] (1.4,-1) circle (0.07);  
		\draw[fill=black,opacity=0.3] (-1.6,-2) circle (0.07); 
		\draw[fill=black,opacity=0.3] (1.6,-2) circle (0.07); 
				
		\draw[opacity=.3] (-1.5,-1.5) -- (1.5,1.5);
		\draw[opacity=.3] (-1.5,1.5) -- (1.5,-1.5);

		\draw[->,opacity=0.7] (-1.42,1.12) -- (-1.48,1.37);
		\draw[->,opacity=0.7] (-1.42,-1.12) -- (-1.48,-1.37);
		\draw[->,opacity=0.7] (1.42,1.12) -- (1.48,1.37);
		\draw[->,opacity=0.7] (1.42,-1.12) -- (1.48,-1.37);
		\draw[->,opacity=0.3] (-1.58,-1.88) -- (-1.52,-1.63);
		\draw[->,opacity=0.3] (1.58,-1.88) -- (1.52,-1.63);

		\node[left] at (-1.5,-1) {$1$};
		\node[right] at (1.5,1) {$2$};	
		\node[left] at (-1.5,1) {$4$};
		\node[right] at (1.5,-1) {$3$};	
		\node[left,opacity=0.5] at (-1.6,-2) {$2^-$};
		\node[left,opacity=0.5] at (2.3,-2) {$4^-$};
		
		\end{tikzpicture}
		\caption{Another description of the Regge limit is $x_1\to x_2^-$ and $x_3\to x_4^-$. The points $x_2^-,x_4^-$ are shown in gray. The cross-ratios $\chi,\bar\chi$ associated with the points $1,2,3,4$ are the same as those associated with $1,2^-,3,4^-$.}
		\label{fig:reggelimitagain}
\end{figure}

To understand what happens to the conformal block expansion~(\ref{eq:usualOPE}) in the Regge regime, we must compute the monodromy of $G^{\De_i}_{\De,J}(\chi,\bar\chi)$ from taking $\bar\chi$ counterclockwise around $1$. This was described in \cite{Caron-Huot:2017vep}. Firstly, we have the decomposition
\be
G^{\De_i}_{\De,J}(\chi,\bar\chi) &= g_{\De,J}^\mathrm{pure}(\chi,\bar\chi) + \frac{\G(J+d-2)\G(-J-\frac{d-2}{2})}{\G(J+\frac{d-2}{2})\G(-J)} g^\mathrm{pure}_{\De,2-d-J}(\chi,\bar\chi),
\ee
where $g^\mathrm{pure}_{\De,J}$ is the solution to the conformal Casimir equation defined by
\be
g^\mathrm{pure}_{\De,J}(\chi,\bar\chi) &= \chi^{\frac{\De-J}{2}} \bar\chi^{\frac{\De+J}{2}} \x (1 + \textrm{integer powers of }\chi/\bar\chi,\bar\chi) \qquad (\chi\ll \bar\chi \ll 1).
\ee
For small $\chi$, $g^\mathrm{pure}_{\De,J}$ has a simple form in terms of a hypergeometric function \cite{DO3},
\be
\label{eq:collinearblock}
g^\mathrm{pure}_{\De,J}(\chi,\bar\chi) &= \chi^{\frac{\De-J}{2}} k_{\De+J}(\bar\chi)\x(1+O(\chi)) \qquad\qquad (\chi\ll 1),\\
k_{2\bar h}(\bar\chi) &= \bar\chi^{\bar h} {}_2F_1\p{\bar h - \frac{\De_{12}}{2},\bar h+\frac{\De_{34}}{2},2\bar h,\bar\chi},
\ee
where $\De_{ij}\equiv\De_i-\De_j$. The monodromy of $g^\mathrm{pure}_{\De,J}$ as $\bar\chi$ goes around $1$ can then be determined from (\ref{eq:collinearblock}) using elementary hypergeometric function identities, keeping $\chi$ small so that the approximation (\ref{eq:collinearblock}) remains valid.

Let us defer discussing the precise form of the monodromy until section~\ref{sec:relationtolight}, and focus on one important feature. Note that $k_{2\bar h}(\bar\chi)$ is a conformal block for $\SL(2,\R)$. In particular, it is a solution to the conformal Casimir equation (a second-order differential equation) with eigenvalue $\bar h(\bar h - 1)$. Under monodromy, it will mix with the other solution, which differs by $\bar h \to 1-\bar h$. In terms of $\De,J$, this becomes
\be
\label{eq:reggeweyl}
(\De,J) &\to (1-J,1-\De),
\ee
i.e.\ it is the affine Weyl reflection associated to the light transform. After monodromy, in the limit $\chi,\bar\chi\to 0$ each block contains a term
\be
\chi^{\frac{\De-J}{2}} \bar\chi^{\frac{1-\De+1-J}{2}} &\sim e^{(J-1) t} \qquad(t\gg 1).
\ee
In other words, the monodromy of each block grows as $e^{(J-1)t}$ in the Regge limit.
Because the sum (\ref{eq:usualOPE}) includes arbitrarily large $J$, the OPE expansion formally diverges as $t\to \oo$.

In what follows, it will be important to understand the large-$J$ limit of conformal blocks in slightly more detail. We compute this in  appendix~\ref{app:largeJ}:
\be
\label{eq:largeJlimitofgpure}
g^\mathrm{pure}_{\De,J}(\chi,\bar\chi)&\sim 
\frac{4^\De f_{1-\De}(\tfrac 1 2(r+\tfrac 1 r)) w^{-J}}{(1-w^2)^{\frac{d-2}{2}}(r^2 + \tfrac 1 {r^2} - w^2 - \tfrac 1 {w^2})^{\frac 1 2}}\p{\frac{(1-\tfrac r w)(1-rw)}{(1+\tfrac r w)(1+rw)}}^{\frac{\De_{12}-\De_{34}}{2}}
\qquad(|J|\gg 1),
\ee
where $w= e^{i\th}$ and $f_{1-\De}(x)$ is given in (\ref{eq:radialsolngegenbauer}). For us, the most important feature of (\ref{eq:largeJlimitofgpure})  is that its $J$-dependence is $w^{-J}$. Note that the small-$w$ limit of (\ref{eq:largeJlimitofgpure}) is consistent with the claim that $g^\mathrm{pure}_{\De,J}$ grows as $w^{1-J} = e^{(J-1)t}$ in the limit $t\to\oo$.

\subsection{Review: Sommerfeld-Watson resummation}
\label{sec:sommerfeld}

Taking the monodromy of $\bar\chi$ around 1 requires leaving the region $|\bar\rho|<1$ where the sum over $\De$ in the conformal block expansion converges. The conformal partial wave expansion gives a way to avoid this problem: we replace a sum of the form $\sum_\De|\rho\bar\rho|^{\De/2}$ with an integral over $\De\in \frac d 2 + i\R$. This integral is better-behaved when $|\bar\rho|>1$.

In the Regge limit we still have the problem that each individual block grows like $e^{(J-1)t}$. This can be dealt with in a similar way: by replacing the sum over $J$ with an integral in the imaginary direction. This trick is called the Sommerfeld-Watson transform.

Let us begin with the conformal partial wave expansion
\be
\<\f_1(x_1)\cdots\f_4(x_4)\> &= \sum_{J=0}^\oo \int_{\frac d 2 - i\oo}^{\frac d 2 + i\oo} \frac{d\De}{2\pi i} C(\De,J) F^{\De_i}_{\De,J}(x_i), \nn\\
F^{\De_i}_{\De,J}(x_i) &\equiv \frac 1 2 \p{G^{\De_i}_{\De,J}(x_i) + \frac{S_E(\f_1\f_2[\cO])}{S_E(\f_3\f_4[\tl\cO])} G^{\De_i}_{d-\De,J}(x_i)}.
\ee
For integer $J$, the coefficient function $C(\De,J)$ can be written
\be
C(\De,J) &= C^t(\De,J) + (-1)^J C^u(\De,J), \qquad(J\in \Z),
\ee
where $C^t$ comes from the first term in the Lorentzian inversion formula (\ref{eq:simonsformula}), and $C^u$ comes from the second term with $1\leftrightarrow 2$. (The superscripts $t$ and $u$ stand for ``$t$-channel" and ``$u$-channel.") Each of the functions $C^{t,u}(\De,J)$ has a natural analytic continuation in $J$ that is bounded in the right half-plane. This follows from (\ref{eq:simonsformula}), since the conformal block $G^{\tl \De_i}_{J+d-1,\De-d+1}(\chi,\bar\chi)$ is well-behaved in the square $\chi,\bar\chi\in [0,1]$ when $J$ is in the right half-plane.

Let us split the partial wave $F_{\De,J}^{\De_i}$ into two pieces
\be
\label{eq:splitofpartialwave}
F_{\De,J}^{\De_i}(x_i) &= \cF_{\De,J}(x_i) + \cH_{\De,J}(x_i),
\ee
where $\cF_{\De,J}$ behaves like $w^{-J}$ at large $J$,
\be
\label{eq:explicitfpure}
& \cF_{\De,J}(x_i) \equiv \nn\\
&\frac{1}{(x_{12}^2)^{\frac{\De_1+\De_2}{2}} (x_{34}^2)^{\frac{\De_3+\De_4}{2}}} \p{\frac{x_{14}^2}{x_{24}^2}}^{\frac{\De_2-\De_1}{2}} \p{\frac{x_{14}^2}{x_{13}^2}}^{\frac{\De_3-\De_4}{2}} \frac 1 2 \p{g^\mathrm{pure}_{\De,J}(\chi,\bar\chi) + \frac{S_E(\f_1\f_2[\cO])}{S_E(\f_3\f_4[\tl\cO])} g^{\mathrm{pure}}_{d-\De,J}(\chi,\bar\chi)},
\ee
and $\cH_{\De,J}(x_i)$ represents the remaining terms, which behave like $w^{J+d-2}$ at large $J$.
We must treat the two terms in (\ref{eq:splitofpartialwave}) differently in the Sommerfeld-Watson transform. Let us focus on the first term. The sum over integer spins can be written as a contour integral
\be
\label{eq:firstpuresum}
\sum_{J=0}^\oo C(\De,J) \cF_{\De,J}(x_i) &= -\oint_\Gamma dJ \frac{C^t(\De,J) + e^{-i\pi J} C^u(\De,J)}{1-e^{-2\pi i J}} \cF_{\De,J}(x_i)\nn\\
&\qquad\qquad (\Re(\th)\in(0,\pi), \Im(\th)=0),
\ee
where the contour $\G$ encircles all the nonnegative integers clockwise. Here, we have carefully chosen the analytic continuation of $C(\De,J)$ so that the integrand is bounded at large $J$ in the right half-plane  whenever $\th$ satisfies the given conditions. For this, we use the fact that $\cF_{\De,J}(x_i)$ behaves as $w^{-J}$ at large $J$.
Because the other term in (\ref{eq:splitofpartialwave}) behaves as $w^{J+d-2}$ at large $J$, we must replace $e^{-i\pi J}\to e^{i\pi J}$ to get an integral for that term that is valid in the same range of $\th$.

The contour integral (\ref{eq:firstpuresum}) is more suitable than a na\"ive sum over spins for continuing to the Regge regime. Recall that the issue with a sum over $J$ was that a conformal block with spin $J$ grows as $e^{(J-1)t}$ in the Regge limit. Because the integrand in (\ref{eq:firstpuresum}) is well-behaved at large $J$, we can deform the contour $\G$ to a region where $\Re(J)<1$, so that its contributions die as $t\to \oo$.\footnote{A natural choice is the Lorentzian principal series $\Re(J) = -\frac{d-2}{2}$.} In doing so, we may pick up new poles in $C^{u,t}(\De,J)$ with real part $\Re(J)>1$. The rightmost such pole will dominate the correlator in the Regge limit. Denote the deformed contour, including these new poles, by $\G'$  (figure~\ref{fig:gammaprimecontour}).

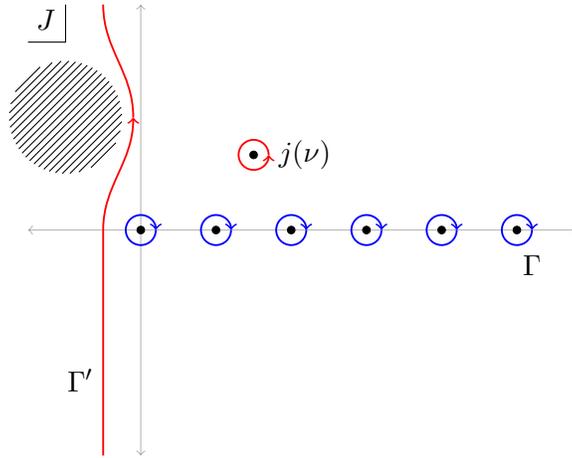
\begin{figure}[ht!]
\centering
\begin{tikzpicture}
\draw[] (-1,3) -- (-1,2.5) -- (-1.5,2.5);

\draw[->,opacity=0.3] (0,0) -- (-1.5,0);
\draw[->,opacity=0.3] (0,0) -- (0,3);
\draw[->,opacity=0.3] (0,0) -- (0,-3);
\draw[->,opacity=0.3] (0,0) -- (5.8,0);
\draw[fill=black] (0,0) circle (0.05);
\draw[fill=black] (1,0) circle (0.05);
\draw[fill=black] (2,0) circle (0.05);
\draw[fill=black] (3,0) circle (0.05);
\draw[fill=black] (4,0) circle (0.05);
\draw[fill=black] (5,0) circle (0.05);
\draw[fill=black] (1.5,1) circle (0.05);

\draw[] (-0.662,0.830591) -- (-0.331,1.16178);
\draw[] (-0.822,0.771524) -- (-0.272,1.32161);
\draw[] (-0.9416,0.752275) -- (-0.252,1.44162);
\draw[] (-1.0434,0.751257) -- (-0.251,1.5434);
\draw[] (-1.133,0.761969) -- (-0.262,1.63345);
\draw[] (-1.215,0.78141) -- (-0.281,1.71478);
\draw[] (-1.289,0.807926) -- (-0.308,1.78902);
\draw[] (-1.357,0.84052) -- (-0.341,1.85719);
\draw[] (-1.420,0.878568) -- (-0.379,1.91991);
\draw[] (-1.478,0.921687) -- (-0.422,1.97755);
\draw[] (-1.530,0.96967) -- (-0.470,2.03033);
\draw[] (-1.578,1.02245) -- (-0.522,2.07831);
\draw[] (-1.621,1.08009) -- (-0.580,2.12143);
\draw[] (-1.659,1.14281) -- (-0.643,2.15948);
\draw[] (-1.692,1.21098) -- (-0.711,2.19207);
\draw[] (-1.719,1.28522) -- (-0.785,2.21859);
\draw[] (-1.738,1.36655) -- (-0.867,2.23803);
\draw[] (-1.749,1.4566) -- (-0.9566,2.24874);
\draw[] (-1.748,1.55838) -- (-1.0584,2.24772);
\draw[] (-1.728,1.67839) -- (-1.178,2.22848);
\draw[] (-1.669,1.83822) -- (-1.338,2.16941);

\draw[<-,line width=0.7,blue] (0.2,0) to[out=90,in=0] (0,0.2) to[out=180,in=90] (-0.2,0) to[out=-90,in=180] (0,-0.2) to[out=0,in=-90] (0.2,0);
\draw[<-,line width=0.7,blue] (1.2,0) to[out=90,in=0] (1,0.2) to[out=180,in=90] (0.8,0) to[out=-90,in=180] (1,-0.2) to[out=0,in=-90] (1.2,0);
\draw[<-,line width=0.7,blue] (2.2,0) to[out=90,in=0] (2,0.2) to[out=180,in=90] (1.8,0) to[out=-90,in=180] (2,-0.2) to[out=0,in=-90] (2.2,0);
\draw[<-,line width=0.7,blue] (3.2,0) to[out=90,in=0] (3,0.2) to[out=180,in=90] (2.8,0) to[out=-90,in=180] (3,-0.2) to[out=0,in=-90] (3.2,0);
\draw[<-,line width=0.7,blue] (4.2,0) to[out=90,in=0] (4,0.2) to[out=180,in=90] (3.8,0) to[out=-90,in=180] (4,-0.2) to[out=0,in=-90] (4.2,0);
\draw[<-,line width=0.7,blue] (5.2,0) to[out=90,in=0] (5,0.2) to[out=180,in=90] (4.8,0) to[out=-90,in=180] (5,-0.2) to[out=0,in=-90] (5.2,0);
\draw[->,line width=0.7,red] (1.7,1) to[out=90,in=0] (1.5,1.2) to[out=180,in=90] (1.3,1) to[out=-90,in=180] (1.5,0.8) to[out=0,in=-90] (1.7,1);
\draw[->,line width=0.7,red] (-0.5,-3) -- (-0.5,0) to[out=90,in=-90] (-0.1,1.5);
\draw[line width=0.7,red] (-0.1,1.5) to[out=90,in=-90] (-0.5,3) -- (-0.5,3);

\node[left] at (-1,2.8) {$J$};
\node[below] at (5.2,-0.2) {$\Gamma$};
\node[left] at (-0.5,-2) {$\Gamma'$};
\node[right] at (1.7,1) {$j(\nu)$};
\end{tikzpicture}
\caption{Integration contours in the $J$ plane. The contour $\G$ (blue) encircles all the integers clockwise. The deformed contour $\G'$ runs parallel to the imaginary axis, asymptotically approaching $\Re(J)=-\frac{d-2}{2}$ at large imaginary $J$. In deforming the contour, we must ensure that $\G'$ avoids non-analyticities, like a pole at non-integer $J$, branch cuts, or other singularities. Here, we show a single non-integer pole at $J=j(\nu)$ and possible non-analyticities in the shaded region. However, this is only an example --- we don't know the  structure of the $J$-plane in general.}
\label{fig:gammaprimecontour}
\end{figure}

After deforming the contour, we now have a representation of the correlator that is valid in the strip 
\be
\Re(\th)\in (0,\pi),\quad \Im(\th)>0,
\ee
which includes the angle $\th=it+\e$ required for a time-ordered Lorentzian correlator. Thus, we can continue to the Regge regime. The continuation of $\cH_{\De,J}(x_i)$ does not give a growing contribution in the Regge limit, so let us ignore it for the moment. We find that the four-point function behaves as
\be
\<\f_1(x_1)\cdots \f_4(x_4)\> &\sim -\oint_{\Gamma'} dJ \int_{\frac d 2 - i\oo}^{\frac d 2 + i\oo} \frac{d\De}{2\pi i} \frac{C^t(\De,J) + e^{-i\pi J} C^u(\De,J)}{1-e^{-2\pi i J}} \cF_{\De,J}(x_i)^\circlearrowleft,
\label{eq:doubleintegralrep}
\ee
where $\cF_{\De,J}(x_i)^\circlearrowleft$ denotes the continuation to Regge kinematics, including the monodromy of $\bar\chi$ around 1 and phases arising from the prefactor in (\ref{eq:explicitfpure}).\footnote{Representing the correlator as an integral over both $\De$ and $J$ is natural from the point of view of Lorentzian harmonic analysis, where principal series representations are labeled by continuous $\De=\frac d 2 + is$ and $J=-\frac{d-2}{2}+it$. However, it is not immediately obvious how the representation (\ref{eq:doubleintegralrep}) is related to the Plancherel theorem for $\tl\SO(d,2)$. We leave this question for future work.}

In planar large-$N$ theories, the rightmost feature of $\G'$ is conjectured to be an isolated pole $J=j(\nu)$ where $\De=\frac d 2 + i\nu$. Assuming this is the case, we obtain
\be
\<\f_1(x_1)\cdots \f_4(x_4)\> &\sim - 2\pi i \int_{-\oo}^{\oo} \frac{d\nu}{2\pi} \Res_{J=j(\nu)} \frac{C^t(\tfrac d 2 + i\nu,J) + e^{-i\pi J} C^u(\tfrac d 2 + i\nu,J)}{1-e^{-2\pi i J}} \cF_{\tfrac d 2 + i\nu,J}(x_i)^\circlearrowleft.
\label{eq:reggepoleassumption}
\ee

\subsection{Relation to light-ray operators}
\label{sec:relationtolight}

The appearance of the affine Weyl transform (\ref{eq:reggeweyl}) is suggestive that Regge kinematics should be related to the light transform and light-ray operators. To see how, let us finally compute $\cF_{\De,J}(x_i)^\circlearrowleft$ using (\ref{eq:collinearblock}). We find
\be
\label{eq:monodromyoffpure}
\cF_{\De,J}(x_i)^\circlearrowleft &= -\frac{i \pi \G(\De+J)\G(\De+J-1)}{\G(\tfrac{\De+J+\De_{12}}{2})\G(\tfrac{\De+J-\De_{12}}{2})\G(\tfrac{\De+J+\De_{34}}{2})\G(\tfrac{\De+J-\De_{34}}{2})} T^{\De_i}(x_i) G_{1-J,1-\De}(\chi,\bar\chi) \nn\\
&\qquad + \dots,
\ee
where $T^{\De_i}(x_i)$ is the product of $|x_{ij}|$'s given in (\ref{eq:absolutevaluestructure}). Here, we have explicitly written the term that is growing in the Regge limit. The ``$\dots$"'s represent other solutions of the Casimir equations that do not grow in the Regge limit, coming from both $\mathcal{F}_{\De,J}$ and $\mathcal{H}_{\De,J}$. The above expression is valid in the configuration $4>1,2>3$, with other points spacelike-separated.  

Comparing with (\ref{eq:eqnforh}) and (\ref{eq:claimabouth}), we immediately recognize
\be
\cF_{\De,J}(x_i)^\circlearrowleft &= \frac{1}{2} \tsym_2^{-1} \tsym_4^{-1} \frac{(\tsym_2\<\f_1\f_2\wL[\cO^\dagger]\>)(\tsym_4\<\f_3\f_4\wL[\cO]\>)}{\<\wL[\cO]\wL[\cO^\dagger]\>} + \dots,
\label{eq:growinglight}
\ee
where we use the notation for a conformal block introduced in section~\ref{sec:naturalformula}.
Equation~(\ref{eq:growinglight}) is the main observation of this section. In the case where Regge kinematics is dominated by an isolated pole (\ref{eq:reggepoleassumption}), the residue $\Res_{J=j(\nu)}$ means that coefficients in the integrand can be interpreted as products of OPE coefficients for light-ray operators. This is because a nontrivial residue comes from the neighborhood of the light ray.\footnote{The same is true if the Regge limit is dominated by a cut instead of a pole, though now we have a doubly-continuous family of light-ray operators, parameterized by $\nu$ and $J$ along the cut.} Plugging (\ref{eq:growinglight}) into (\ref{eq:reggepoleassumption}), we find a sum/integral of conformal blocks for these light-ray operators. 

In the gauge-theory literature, the object that controls the Regge limit of a planar amplitude is called the ``Pomeron" \cite{Chew:1961ev,Gribov:1961ex}. Here, we see that for planar CFT correlation functions, the Pomeron is a light-ray operator: it is proportional to the rightmost residue in $J$ of $\mathbb{O}_{\De,J}$, for $\De\in \frac d 2 + i\R$.

The observation (\ref{eq:growinglight}) also lets us immediately generalize conformal Regge theory to arbitrary operator representations. In the Regge limit, we have 
\be
\<\cO_1(x_1)\cdots \cO_4(x_4)\> &\sim - \half\sum_{\l,a,b} \oint_{\Gamma'} dJ \int_{\frac d 2 - i\oo}^{\frac d 2 + i\oo} \frac{d\De}{2\pi i} \frac{C_{ab}(\De,J,\l)}{1-e^{-2\pi i J}} \nn\\
&\qquad\qquad\qquad \x \tsym_2^{-1} \tsym_4^{-1} \frac{(\tsym_2\<\cO_1\cO_2\wL[\cO^\dagger]\>^{(a)})(\tsym_4\<\cO_3\cO_4\wL[\cO]\>^{(b)})}{\<\wL[\cO]\wL[\cO^\dagger]\>}.
\label{eq:doubleintegralrepgeneral}
\ee
Here, $C_{ab}(\De,J,\l)$ is the unique analytic continuation of $C_{ab}(\De,\rho)$ such that $\frac{C_{ab}(\De,J,\l)}{1-e^{-2\pi i J}}e^{-i\th J}$  is bounded for large $J$ in the right-half plane and $\th\in(0,\pi)$. The weight $J$ is the length of the first row of the Young diagram of $\rho$, and $\l$ represents the remaining weights of $\rho$, as discussed in section~\ref{sec:reptheoryreview}. The indices $a,b$ run over three-point structures.

As before, it is straightforward to argue that  (\ref{eq:doubleintegralrepgeneral}) is the only possibility consistent with the scalar case and with weight-shifting operators. It would be interesting to verify it more directly, and in general to characterize all monodromies of blocks in terms of the integral transforms in section~\ref{sec:weylandintegral}. Note that (\ref{eq:doubleintegralrepgeneral}) displays a beautiful duality with the generalized Lorentzian inversion formula (\ref{eq:obviousgeneralization}).

We can try to interpret (\ref{eq:growinglight}) as a contribution to the non-vacuum OPE of $\f_1\f_2$ in the following way. We construct light-ray operators as an integral of the form~\eqref{eq:lightraykernelschematic}, which together with conformal symmetry implies that we should be able to write, schematically,
\be\label{eq:nonvacuumOPE}
	\f_1\f_2 = \int d\nu\,  \cB_{\nu, j(\nu)}[\mathbb{O}_{0,j(\nu)}] +\text{other contributions}.
\ee
Here $\cB$ is a kind of OPE kernel which is fixed by conformal symmetry, and the equation should be interpreted in an operator sense. The representation~\eqref{eq:growinglight} suggests that~\eqref{eq:nonvacuumOPE} is a good version of the OPE in non-vacuum states, with the first term giving the only possibly-growing contribution in the Regge limit.

The ``other contributions" can perhaps be understood by studying the terms that we ignored above, coming form $\mathcal{H}_{\De,J}$ and part of $\mathcal{F}_{\De,J}^\circlearrowleft$. We expect that they can be understood more systematically using harmonic analysis on the Lorentzian conformal group $\tl \SO(d,2)$. (We hope to address this in future work.) In a finite-$N$ CFT, the correlator saturates in the Regge limit --- i.e.\ it eventually stops growing. Thus, the details of these terms will presumably be important for determining the actual behavior of the correlator in the Regge limit.\footnote{We thank Sasha Zhiboedov for discussions on this point.}

\section{Positivity and the ANEC}
\label{sec:positivityandtheanec}

The average null energy condition (ANEC) states that $\cE=\wL[T]$ is a positive-semidefinite operator. The ANEC was proven in \cite{Faulkner:2016mzt} using information theory and in \cite{Hartman:2016lgu} using causality. The causality-based proof \cite{Hartman:2016lgu} proceeds by isolating the contribution of $\cE$ in a correlation function and using Rindler positivity to show that the contribution is positive. Isolating $\cE$ requires using the OPE outside its na\"ive regime of validity. However, the authors of \cite{Hartman:2016lgu} give an argument that one can still trust the leading term in the OPE in an asymptotic expansion in the lightcone limit.

From our work in section~\ref{sec:lightray}, we now have an alternative construction of $\cE$ as a special case of a light-ray operator. Using this construction, we can avoid asymptotic expansions and any technical issues associated with using the OPE outside its regime of validity. Beyond technical convenience, our approach gives extra flexibility. The authors of \cite{Hartman:2016lgu} also prove a higher-spin version of the ANEC:
\be
\label{eq:higherspinanec}
\cE_J \equiv \wL[X_J]\geq 0,\qquad (J=2,4,\dots),
\ee
where $X_J$ is the lowest-dimension operator with spin $J$.\footnote{More precisely, $X_J$ can be the lowest-dimension operator with spin $J$ in any OPE of the form $\cO^\dag \x \cO$.}${}^{,}$\footnote{The higher-spin version of the ANEC was first discussed in \cite{Komargodski:2016gci}, where it was also proven for sufficiently high spin.}${}^{,}$\footnote{The proof of the higher-spin ANEC in \cite{Hartman:2016lgu} relies on some assumptions about subleading terms when the OPE is used as an asymptotic expansion outside of its regime of convergence. We thank Tom Hartman for discussion on this point.} Our construction lets us generalize this statement to
\be
\label{eq:lightraypositivity}
\cE_J \geq 0,\qquad (J\in \R_{\geq J_\mathrm{min}}),
\ee
where $J_\mathrm{min} \leq 1$ is the smallest value of $J$ for which the Lorentzian inversion formula holds \cite{Caron-Huot:2017vep}. Here, $\cE_J(x,z)$ denotes the light-ray operator with dimension and spin $(1-J,1-\De)$, where $\De,J$ are real and $\De$ is minimal. This result follows by writing a sum rule for {\it all\/} light-ray operators, and simply observing that it is positive by Rindler positivity when $(\De,J)$ satisfy the above conditions. When $J$ is an integer, (\ref{eq:lightraypositivity}) reduces to (\ref{eq:higherspinanec}). However, when $J$ is not an integer, (\ref{eq:lightraypositivity}) is a new condition. 

A possible connection between Lorentzian inversion formulae and the ANEC was first suggested by Caron-Huot using a toy dispersion relation \cite{Caron-Huot:2017vep}. In this section, we are simply making the connection more precise.

\subsection{Rindler positivity}
\label{sec:rindlerpositivity}

Rindler positivity is a key ingredient in the causality-based proof of the ANEC \cite{Hartman:2016lgu}, so let us review it. Given $x=(t,y,\vec x)\in \R^{d+1,1}$, define the Rindler reflection
\be
\bar x &= \overline{(t,y,\vec x)} = (-t^*,-y^*,\vec x).
\ee
Rindler conjugation defined in (\ref{eq:modularconjugation}) and (\ref{eq:defofbar}) maps an operator $\cO$ in the right Rindler wedge to an operator $\bar \cO$ in the left Rindler wedge. For traceless-symmetric tensors, it is given by
\be
\bar{\cO(x,z)} &= \cO^\dag(\bar x, \bar z).
\ee

The statement of Rindler positivity is that
\be
\label{eq:rindlerpositivity}
\<\Omega|\bar{\cO}_1\cdots \bar{\cO}_n \cO_1 \cdots \cO_n|\Omega\> &\geq 0,
\ee
where $\cO_i$ are restricted to the right Rindler wedge
\be
\mathcal{W}_R &= \{(u,v,\vec{x})\ :\ uv>0,\ \mathop{\mathrm{arg}} v\in (-\tfrac \pi 2,\tfrac \pi 2),\ \vec{x}\in \R^{d-2}\}.
\ee
(Here, we use lightcone coordinates $u=y-t, v= y+t$.)

To establish (\ref{eq:rindlerpositivity}) for general causal configurations of the $\cO_i$, \cite{Casini:2010bf} appeals to Tomita-Takesaki theory. However, this is not necessary as argued in \cite{Hartman:2016lgu}. We can summarize their argument as follows. Because the operators $\cO_1\cdots \cO_n$ act on the vacuum, we can perform the OPE to replace
\be
\cO_1\cdots \cO_n|\Omega\> &= \sum_\cO C(x_i,x,\ptl_x) \cO(x)|\Omega\>,
\ee
where $C(x_i,x,\ptl_x)$ is a differential operator. We are free to choose $x$ to be any point in $\mathcal{W}_R$ (we cannot choose $x$ to be timelike from the $\bar x_i$). Truncating the sum, we approximate the right hand side by a local operator. The expectation value (\ref{eq:rindlerpositivity}) then becomes a Rindler-reflection symmetric two-point function. Positivity of this two-point function is a consequence of reflection-positivity, since the two points are spacelike-separated.

\subsection{The continuous-spin ANEC}

Following \cite{Hartman:2016lgu}, we will prove 
\be
\label{eq:thingweprovepositivityof}
i\<\Omega|\bar V \mathcal{E}'_J V|\Omega\>\geq 0,
\ee
where $V$ is any local operator located at a point $x_V=(0,\de,\mathbf{0})\in \mathcal{W}_R$ in the right Rindler wedge. Here, $\cE'_J$ is a continuous-spin light-ray operator of spin-$J$ with lowest twist, oriented along the null direction $z=(1,1,\vec{0})$. As argued in \cite{Hartman:2016lgu}, it follows that $\cE'_J$ satisfies the positivity condition
\be\label{eq:thingwewantpositivitybyrotation}
e^{i\frac \pi 2 J}\<\Omega|(R\.V)^\dag(t=-i\de)\, \mathcal{E}'_J\, (R\.V)(t=i\de)|\Omega\>\geq 0,
\ee
where $R$ rotates by $\frac \pi 2$ in the Euclidean $y\tau$-plane, with $\tau=it$, and $R\.V$ represents the action of $R$ on $V$ at the origin. States of the form $(R\.V)(t=i\de)|\Omega\>\in \mathcal{H}$ are dense in $\mathcal{H}$, by the state-operator correspondence. Thus, 
\be
\cE_J &\equiv e^{i\frac{\pi}{2} J}\cE_J'
\ee
is a positive operator.

Let $\f$ be a real scalar primary. We will produce $\cE_J'$ by smearing two $\f$ insertions. For simplicity, we will not attempt to divide by OPE coefficients in the $\f\x\f$ OPE. Thus,
when $J$ is an integer, we will actually have $\cE_J = f_{\f\f X_J}\wL[X_J]$, where $X_J$ is the lowest-twist operator of spin-$J$ in the $\f\x\f$ OPE and $f_{\f\f X_J}$ is an OPE coefficient. In particular $\cE_2$ in this section differs from the usual ANEC operator by a factor of $f_{\f\f T}$.

From (\ref{eq:ODeJdoublecommutatorformula}), we have
\be
\label{eq:integralforbbO}
i\<\bar V \mathbb{O}^+_{\De,J}(-\oo z,z) V\> &= 
 \int_{\substack{-\oo z< x_1<\bar x_V \\ x_V < x_2 < \oo z}} d^d x_1 d^d x_2 \<\Omega|[\bar {V},\f(x_1)][\f(x_2),V]|\Omega\>K_{\De,J}(x_1,x_2),  \nn\\
K_{\De,J}(x_1,x_2) &= \frac{2i\mu(\De,J) S_E(\f\f[\tl \cO])}{(\<\f\f\tl \cO\>,\<\tl \f\tl \f \cO\>)_E}\<0|\tl \f(x_1) \wL[\cO](-\oo z,z)\tl \f(x_2)|0\>.
\ee
We have included a factor of $2$ from the term $1\leftrightarrow 2$ in (\ref{eq:ODeJdoublecommutatorformula}), and we should interpret the prefactors in $K_{\De,J}$ as being analytically continued from even $J$. The matrix elements of $\mathcal{E}_J$ are defined by
\be
i\<\Omega|\bar V \cE'_J V|\Omega\> &= \Res_{\De=\De_*} i\<\bar V \mathbb{O}^+_{\De,J}(-\oo z,z) V\>,
\ee
where $\De_*$ is the location of the pole in $\mathbb{O}^+_{\De,J}$ with minimal real $\De$. The expression (\ref{eq:integralforbbO}) is guaranteed to be convergent for $\De\in \frac d 2 + i\R$ on the principal series. In particular it converges at $\De=\frac d 2$. Our strategy will be to show that $i\<\bar V \mathbb{O}^+_{\De,J}(x,z) V\>$ is strictly negative as we move rightward along the real axis starting from $\De=\frac d 2$ (figure~\ref{fig:positiveresidue}). It follows that the first pole we encounter must have positive residue.\footnote{Requiring negativity for all $\De$ between $\frac d 2$ and the first pole is stronger than necessary. It should be possible to improve our proof by establishing negativity only for $\De$ sufficiently close to the first pole.} 

\begin{figure}[ht!]
\centering
\begin{tikzpicture}
\draw[line width=0.7,blue] (0,-0.3) to[out=10,in=100] (3.4,-2);
\draw[line width=0.7,blue] (3.9,2) to[out=-80,in=170] (5,0.5);
\draw[opacity=0.5,->] (0,0) -- (-2,0);
\draw[opacity=0.5,->] (0,0) -- (5,0);
\draw[opacity=0.5,->] (0,0) -- (0,-2);
\draw[opacity=0.5,->] (0,0) -- (0,2);
\draw[dashed] (3.65,-2) -- (3.65,2);
\node[below] at (5,0) {$\Delta$};
\node[above] at (0,2) {$i\<\bar V \mathbb{O}^+_{\De,J} V\>$};
\node[left] at (0,-0.3) {$\De=\tfrac d 2$};
\node[below] at (1.52,-0.5) {$\textrm{negative}$};
\node[above] at (4,2) {positive residue};
\end{tikzpicture}
\caption{We show that $i\<\bar V \mathbb{O}^+_{\De,J} V\>$ is negative for $\De$ between $\tfrac d 2$ (the principal series) and the first pole. It follows that the first pole has positive residue.}
\label{fig:positiveresidue}
\end{figure}
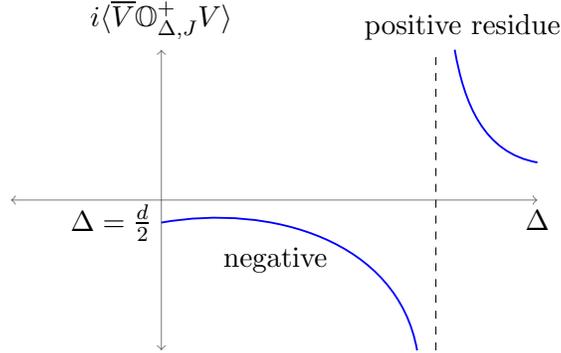

The kernel $K_{\De,J}$ is given by
\be
K_{\De,J}(x_1,x_2) &= \frac{2^J i\mu(\De,J) S_E(\f\f[\tl \cO])L(\tl\f\tl\f[\cO])}{(\<\f\f\tl \cO\>,\<\tl \f\tl \f \cO\>)_E} 
\frac{ (z\.x_2 x_1^2 - z\.x_1 x_2^2)^{1-\De}}{x_{12}^{2\tl \De_\f - \De + J} (-z\.x_1)^{\frac{2-J-\De}{2}}(z\.x_2)^{\frac{2-J-\De}{2}}}.
\ee
We would like to show that $K_{\De,J}(x_1,x_2)$ is a positive-definite kernel when integrated against Rindler-symmetric configurations of $x_1,x_2$. Note that this is a stronger condition than $K_{\De,J}(\bar x, x)\geq 0$ point-wise. 

Consider first an inversion $x\mto x'= \frac{x}{x^2}$ that places $\cE_J$ at null infinity. In this conformal frame, the three-point structure $\<0|\tl \f\wL[\cO] \tl \f|0\>$ becomes translationally invariant. Thus our kernel should be a translationally-invariant function of $x_1',x_2'$, times some scale-factors that depend independently on $x_1,x_2$. Indeed, it is easy to check
\be
&\frac{
		\p{z\.x_{2}\, x_{1}^2 - z\.x_{1}\, x_{2}^2}^{1-\De}
	}{x_{12}^{2\tl\De_\f-\De+J}(-z\.x_1)^{\frac{2-J-\De}{2}}(z\.x_2)^{\frac{2-J-\De}{2}}}
\nn\\
&\qquad= 
x_1'^{2\tl \De_\f} x_2'^{2\tl \De_\f}
\p{-z\.x_1'}^{\frac{J+\De-2}{2}}
\p{z\.x_2'}^{\frac{J+\De-2}{2}}
\frac{
	 \p{z\.(x_2'-x_1')}^{1-\De}
	}{(x_2'-x_1')^{2\tl\De_\f-\De+J}}.
\ee

Because our kernel originates from the light-transform of a three-point structure, it inherits Rindler positivity properties. These are made clear by going to a kind of complexified Fourier-space in the inverted coordinates $x_i'$. Define lightcone coordinates $x^-=u=y-t$ and $x^+=v=y+t$. One can prove the following identity which is valid in the right Rindler wedge $u,v>0$:
\be
\frac{u^{1-\De}}{(uv+\vec x^2)^{\frac {2\tl\De_\f-\De+J} 2}} &=
\frac{2^{2-2\tl\De_\f-J}}{
\pi^{\frac{d-2}{2}}
\G(\frac {2\tl\De_\f-\De+J} 2)
\G(\frac{2\tl\De_\f+J+\De-d}{2})
}\nn\\
&\qquad \x
\int_{k>0} d^d k\, (-k^2)^{\frac{2\tl\De_\f+\De+J-d-2}{2}}(-k^-)^{1-\De} f_k(x) \nn\\
f_k(x) &\equiv e^{-\frac 1 2 k^+ u+\frac 1 2 k^- v + i \vec k \. \vec x}.
\label{eq:rindlermomentum}
\ee
Here, the notation $k>0$ indicates that $k$ is restricted to the interior of the forward null cone. This ensures that $k^+ u$ is positive and $k^- v$ is negative, so that the integral is convergent. The complexified plane wave $f_k(x)$ is designed to satisfy
\be
f_k(x)^* &= f_k(-\bar x).
\ee

Putting everything together, we find
\be
K_{\De,J}(x_1,x_2) &=  \mathcal{K}_{\De,J} \int_{k>0} d^d k\, (-k^2)^{\frac{2\tl\De_\f+\De+J-d-2}{2}}(-k^-)^{1-\De} \psi_k(x_2)(\psi_k(\bar x_1))^*,
\ee
where
\be
\psi_k(x) &\equiv \frac{1}{x^{2\tl\De_\f}} \p{\frac{u}{x^2}}^{\frac{J+\De-2}{2}} \exp\p{\frac{-\frac 1 2 k^+ u+\frac 1 2 k^- v + i \vec k \. \vec x}{x^2}}, \\
\mathcal{K}_{\De,J} &= \frac{2^{1-d-\De+J+2\De_\f} \G(J+\tfrac d 2) \Gamma(\tfrac{J+d+1-\De}{2})\G(\De-1)}{\pi^{\frac{3(d-1)}{2}}\G(J+1)\G(\tfrac{d+J-\De}{2})\G(\De-\tfrac d 2)\G(\tfrac{J+\De+d-2\De_\f}{2})\G(\tfrac{J-\De+2d-2\De_\f}{2})}.
\label{eq:Kcoefficient}
\ee
Consequently, we can write
\be
i\<\bar V \mathbb{O}^+_{\De,J}(-\oo z,z) V\> 
&= 
-\mathcal{K}_{\De,J} \int_{k>0}  d^d k\, (-k^2)^{\frac{2\tl\De_\f+\De+J-d-2}{2}}(-k^-)^{1-\De}
\<\bar{\Theta_k}\Theta_k\>,\nn\\
\Theta_k &= \int_{x_V<x<\oo z} d^dx\, \psi_k(x)[\f(x),V] .
\label{eq:finalpositiveresult}
\ee
The coefficient $\mathcal{K}_{\De,J}$ is positive whenever 
\be
\label{eq:firstconditiononDelta}
\De-J < d,\quad\textrm{and}\quad
\De-J < 2(d-\De_\f).
\ee
This is also the condition for $K_{\De,J}(x_1,x_2)$ to be integrable without an $i\e$ prescription. When these conditions hold, the minus sign in (\ref{eq:finalpositiveresult}) ensures that the first nontrivial residue in $\De$ is positive. This proves the ANEC and its continuous spin generalization in this case.

Let us understand the condition $\De-J<2(d-\De_\f)$ in more detail. When this inequality fails, two things happen. Firstly, the factor
\be
\G\p{\frac{J-\De+2d-2\De_\f}{2}}
\ee
in $\cK_{\De,J}$ may no longer be positive. Secondly, the kernel $K_{\De,J}(x_1,x_2)$ develops a naively non-integrable singularity along the lightcone. To make sense of this singularity, one must take into account the appropriate $i\e$ prescription for $x_1,x_2$. This turns $K_{\De,J}(x_1,x_2)$ into a non-sign-definite distribution, and then we cannot conclude anything about the sign of (\ref{eq:finalpositiveresult}).
To get the strongest result, we should pick $\f$ to be the lowest-dimension scalar in the theory. The spin-2 ANEC then follows if $\De_\f \leq \frac{d+2}{2}$. Large-spin perturbation theory \cite{Fitzpatrick:2012yx,Komargodski:2012ek,Fitzpatrick:2015qma,Li:2015itl,Li:2015rfa,Alday:2015eya,Alday:2015ota,Alday:2015ewa,Simmons-Duffin:2016wlq,Alday:2016njk,Hofman:2016awc} and Nachtmann's theorem \cite{Nachtmann:1973mr,Komargodski:2012ek,Hartman:2015lfa,Costa:2017twz} imply that the minimum twist $\De-J$ at each spin $J$ is always less than $2\De_\f$. Thus, we can ensure $\De-J<2(d-\De_\f)$ if $\De_\f\leq \frac d 2$. This condition is also sufficient to ensure $\De-J<d$.
 Thus, the continuous-spin ANEC follows if $\De_\f \leq \frac d 2$.

\subsection{Example: Mean Field Theory}

The continuous spin version of ANEC is easy to check in MFT. (This is essentially the same calculation as in~\cite{Hartman:2016lgu,Klinkhammer:1991ki}.) We have already computed the leading twist operators $\cE_J'=\mathbb{O}_{0,J}^+$ in section~\ref{sec:lightraymft}. In this section we need the straightforward generalization of~\eqref{eq:mftlightray+} to the case of identical operators,
\be\label{eq:identicallightray}
\cE_J'=\mathbb{O}_{0,J}^+=\frac{i}{2\pi}\int ds dt (t+i\e)^{-1-J}:\!\f\p{\frac{s+t}{2}z}\f\p{\frac{s-t}{2}z}\!:,
\ee
with a future-directed null $z$. We can explicitly compute these operators in terms of creation-annihilation operators using 
\be\label{eq:mftfourier}
	\f(x)=N_{\De_\f}^{-\frac{1}{2}}\int_{p>0} \frac{d^dp}{(2\pi)^d} |p|^{\De_\f-\frac{d}{2}}\p{a^\dagger(p)e^{-ipx}+a(p)e^{ipx}},
\ee
where $\De_\f$ is the scaling dimension of $\f$ and
\be
	N_\De=\frac{2^{2\De-1}\pi^{\frac{d-2}{2}}}{(2\pi)^d}\Gamma(\De)\Gamma(\De-\tfrac{d-2}{2})>0.
\ee
The creation-annihilation operators satisfy the commutation relation
\be
	[a(p),a^\dagger(q)]=(2\pi)^d \delta^d(p-q).
\ee
Plugging~\eqref{eq:mftfourier} into~\eqref{eq:identicallightray}, we find 
\be
	\cE_J'=\frac{iN_{\De_\f}^{-\half}}{2\pi}\int_{p,q>0}\frac{d^dp}{(2\pi)^d}\frac{d^dq}{(2\pi)^d}\int ds dt (t+i\e)^{-1-J}&\left[a^\dagger(p)a^\dagger(q)e^{-\tfrac i 2 (p+q)\.z s-\tfrac i 2(p-q)\.z t}\right.\nn\\
	& +a(p)a(q)e^{\tfrac i 2(p+q)\.z s+\tfrac i 2(p-q)\.z t}\nn\\
	& +a^\dagger(p)a(q)e^{-\tfrac i 2(p-q)\.zs-\tfrac i 2(p+q)\.zt}\nn\\
	& \left.+a^\dagger(q)a(p)e^{\tfrac i 2(p-q)\.zs+\tfrac i 2(p+q)\.zt}\right].
\ee
The first two terms under the integral vanish because $s$-integration restricts $(p+q)\.z=0$, which is impossible since both $p$ and $q$ are in the forward null cone. This is consistent with the requirement that $\mathbb{O}^+_{0,J}$ should annihilate both past and future vacua. Since $(p+q)\.z<0$ we can close the $t$-contour in the upper half-plane for the third term (for $J>0$) and thus it also vanishes. We are left with the last term, where we can close the $t$-contour in the lower half-plane. Specifically, we get for $s$ and $t$ integrals
\be
	\int ds dt (t+i\e)^{-1-J}e^{\tfrac i 2(p-q)\.zs+\tfrac i 2(p+q)\.zt}
	&=\frac{2\pi^2\delta((p-q)\.z)e^{-\frac{i\pi}{2}(J+1)}}{\Gamma(J+1)}\p{\frac{-(p+q)\.z}{2}}^J.
\ee
Combining with the rest of the expression we find, using the lightcone coordinates $p=z p_v/2-z' p_u/2+\bp$ with $z\.z'=2$,
\be
	\cE'_J=\frac{\pi e^{-\frac{i\pi}{2}J}N_{\De_\f}^{-\half}}{\Gamma(J+1)}\int_0^\infty dp_u p_u^{J} A^\dagger (p_u)A(p_u),
\ee
where 
\be
	A(p_u)\equiv\int_{|\bp|<p_u p_v}\frac{dp_v d^{d-2}\bp}{(2\pi)^d} a(p_u,p_v,\bp).
\ee
For $\cE_J=e^{\frac{i\pi}{2}J}\cE_J'$ we then obtain
\be
	\cE_J=\frac{\pi N_{\De_\f}^{-\half}}{\Gamma(J+1)}\int_0^\infty dp_u p_u^{J} A^\dagger (p_u)A(p_u)\geq 0,
\ee
which is manifestly non-negative.

\subsection{Relaxing the conditions on $\De_\f$}

The conditions (\ref{eq:firstconditiononDelta}) are stronger than necessary because we have not assumed anything about the quantity that $K_{\De,J}(x_1,x_2)$ is integrated against. We can somewhat relax them as follows. Note that poles in $i\<\bar V \mathbb{O}_{\De,J}(-\oo z,z) V\>$ come from the region where $x_1,x_2$ are near the lightray $\R z$. In this region, we expect the correlator $\<\Omega|[\bar V,\f(x_1)][\f(x_2),V]|\Omega\>$ to depend most strongly on the positions $v_1,v_2$ of the operators along the light-ray and simple invariants built out of the relative position $x_1-x_2$, since $V,\bar V$ are far from the light ray.

To be more precise, consider the integral over $x_1,x_2$ in the coordinates of section~\ref{sec:lightraymft},
\be
&\frac{2^Ji\mu(\De,J) S_E(\f\f[\tl \cO])L(\tl\f\tl\f[\cO])}{(\<\f\f\tl \cO\>,\<\tl \f\tl \f \cO\>)_E}\nn\\
&\x\frac{1}{4}\int \frac {dr}{r} dv_1dv_2 d\a  d^{d-2}\bw_1d^{d-2}\bw_2\frac{2^{J-1}
		v_{21}^{-1-\frac{\De-\De_1-\De_2+J}{2}}\p{\a(1-\a)+(1-\a)\bw_1^2+\a\bw_2^2}^{1-\De}
	}{
		(1+\bw_{12}^2)^{\frac{\tl\De_1+\tl\De_2+J-\De}{2}}
		\a^{\frac{\tl\De_1-\tl\De_2+2-\De-J}{2}}
		(1-\a)^{\frac{\tl\De_2-\tl\De_1+2-\De-J}{2}}
	}\nn\\
&\qquad\times
	 r^{-\frac{\De-\De_1-\De_2-J}{2}}\phi(-r\a,v_1,(r v_{21})^{\half}\bw_1)\phi(r(1-\a),v_2,(r v_{21})^{\half}\bw_2).
\ee
The most important quantities built from $x_{12}$ are
\be
v_{21},\qquad x_{12}^2 = r v_{21}(1+\bw_-^2).
\ee
Let us make the approximation that, to leading order in $r$, the correlator $\<[\bar V, \f][\f, V]\>$ depends only on $v_1,v_2$ and $x_{12}^2$. That is, let us replace
\be
\phi(-r\a,v_1,(r v_{21})^{\half}\bw_1)\phi(r(1-\a),v_2,(r v_{21})^{\half}\bw_2) &\sim \f\p{-\frac{r}{2}(1+\bw_-^2),v_1,0} \f\p{\frac r 2 (1+\bw_-^2),v_2,0}.
\ee
This approximation would be valid, for example, if we could perform the OPE $\f(x_1)\x\f(x_2)$, since the leading terms in the OPE depend only on $v_{21}$ and $x_{12}^2$. However, our assumption is weaker than assuming that we can perform the OPE.

After rescaling $r\to r/(1+\bw_-^2)$, we can now perform the integrals over $\a$ and $\bw_\pm$, following the methods in appendix~\ref{app:Rcoefficient}. The result is
\be
&i\<\bar V \mathbb{O}^+_{\De,J}(-\oo z,z) V\>\nn\\
&\sim \frac{2^{d+J-4}}{\pi} \int \frac{dr}{r} dv_1 dv_2\, r^{\frac{2\De_\f -\De+J}{2}} v_{21}^{\frac{2\De_\f-\De-J-2}{2}} \<\Omega|[\bar V, \f(-\tfrac r 2,v_1,0)][\f(\tfrac r 2,v_2,0),V]|\Omega\> \nn\\
&= -\frac{2^{d+J-4}}{\pi\G\p{\tfrac{\De+J+2-2\De_\f}{2}}} \int \frac{dr}{r} r^{\frac{2\De_\f -\De+J}{2}}\int_0^\oo dk\, k^{\tfrac{\De+J-2\De_\f}{2}} \<\Omega|\bar{\Theta_k(r)}\Theta_k(r)|\Omega\>,
\label{eq:approximationforfirstresidue}
\ee
where
\be
\Theta_k(r) &\equiv \int_0^\oo dv\, e^{-k v} [\f(\tfrac r 2,v,0),V].
\ee
The integrand in (\ref{eq:approximationforfirstresidue}) should be correct to leading order at small $r$, which means the leading residue of $i\<\bar V \mathbb{O}_{\De,J}(-\oo z,z) V\>$ should be correct. This residue is manifestly positive whenever
\be
\De_\f < \frac{\De+J+2}{2}.
\ee
For example, this proves the continuous spin ANEC for all $J\geq 2$ if the lowest-dimension scalar in the theory has dimension $\De_\f \leq \frac{d+4}{2}$.

\section{Discussion}
\label{sec:discussion}

We have argued that every CFT contains light-ray operators that provide an analytic continuation in spin of the light-transforms of local operators. This gives a physical interpretation of Caron-Huot's Lorentzian inversion formula \cite{Caron-Huot:2017vep}. Our construction involves smearing two primary operators $\cO_1,\cO_2$ against a kernel to produce an object $\mathbb{O}_{\De,J}$, and then taking residues in $\De$ to localize the operators along a null ray. We have not shown rigorously that the integral localizes to a null ray (as opposed to a lightcone). However, we expect this is true based on the example of MFT and the fact that it's true for integer $J$. More generally, we expect that any singularity in the $(\De,J)$-plane should lead to a light-ray operator. (For instance, one could take the discontinuity across a branch cut instead of a residue.) It would be nice to understand better the structure of the $(\De,J)$-plane in general CFTs. We know that for nonnegative integer $J$, the object $\mathbb{O}_{\De,J}$ has simple poles in $\De$ at the locations of local operator dimensions. However, we do not know how it behaves for general complex $J$.\footnote{In planar $\cN=4$ SYM, beautiful pictures of the $(\De,J)$-plane have been constructed using integrability \cite{Alfimov:2014bwa,Gromov:2014caa,Gromov:2015wca,Alfimov:2018cms}.} We also have not addressed the question of whether different operators $\cO_1,\cO_2$ produce different light-ray operators. We expect that in a nonperturbative theory, the same set of light-ray operators should appear in every product $\cO_i\cO_j$, if allowed by symmetry. It would be nice to show this rigorously.

Light-ray operators have the advantage over local operators that they fit into a more rigid structure, due to analyticity in spin. However, unlike local operators, they are not included in the Hilbert space of the CFT on $S^{d-1}$ because they annihilate the vacuum. One way to realize them as states is to double the Hilbert space (with time running forwards in one copy and backwards in the other). The $\mathbb{O}_{i,J}$ then become states in the doubled Hilbert space.\footnote{$\mathbb{O}_{i,J}$ itself is a somewhat violent state. However, we can regularize it by acting on the thermofield double state with some temperature $\b$. We thank Alexei Kitaev for this suggestion.} A general message is that the doubled Hilbert space contains interesting structure that is not visible in a single copy, and it would be interesting to explore this idea further.

We have seen that light-ray operators enter the Regge limit of CFT four-point functions. It would be nice to understand the actual spectrum and OPE coefficients of continuous-spin light-ray operators in important physical theories (e.g.\ the 3d Ising model, $\cN=4$ SYM, and more), in order to determine what the Regge limit actually looks like in those theories.\footnote{Besides planar $\cN=4$ SYM, another CFT where the Regge limit of a four-point function has been computed is the 2d (supersymmetric) SYK model \cite{Murugan:2017eto}.}  Such operators have been explored in weakly-coupled gauge theories (see e.g.\ \cite{BALITSKY1989541,Braun:2003rp,Caron-Huot:2013fea,Balitsky:2013npa,Balitsky:2015oux,Balitsky:2015tca}), and it would be interesting to study other perturbative examples. For example, can one write a continuous-spin generalization of the Hamiltonian of the Wilson-Fisher theory \cite{Liendo:2017wsn}?

Another important question is the extent to which light-ray operators  form a complete basis for describing the Regge regime. Indeed, in our discussion in section~\ref{sec:conformalregge}, we ignored certain non-growing contributions in the Regge limit. It would be interesting to include them and give them operator interpretations. Perhaps lightcone operators or other types of nonlocal operators play a role. This question is also interesting in 1 dimension, where the analog of the Regge regime is the so-called ``chaos regime" of a four-point function.

In any spacetime dimension, we can ask: is there a complete basis of nonlocal operators transforming as primaries in Lorentzian signature? Identifying a complete basis could help in developing a generalization of the OPE that is valid in non-vacuum states. (The usual OPE still works as an asymptotic expansion in non-vacuum states, but we would like to find a convergent expansion.) Such a generalization would be a powerful tool for studying Lorentzian physics.

Relatedly, it would be interesting to study OPEs of light-ray operators with each other, especially the ANEC operator $\cE=\wL[T]$.\footnote{We thank Sasha Zhiboedov for discussion on this point.} In ``conformal collider physics" \cite{Hofman:2008ar} one considers ANEC operators starting at the same point $\cE(x,z_1)\cE(x,z_2)$ (usually taken to be spatial infinity $x=\oo$, so that the light-rays lie along future null infinity), and it is natural to study the limit where their polarization vectors coincide $z_1\to z_2$. This question was explored in \cite{Hofman:2008ar}, where it was argued that the leading term in the $\cE\x\cE$ OPE in $\cN=4$ SYM is a particular spin-3 light-ray operator that can be described in bulk string theory using the Pomeron vertex operator of \cite{Brower:2006ea}. It would be nice to determine a systematic expansion for this limit in a general CFT. Such an expansion could be useful for computing energy correlators and studying jet substructure in CFTs. Light-ray operators could also be useful for understanding aspects of deep inelastic scattering and PDFs.\footnote{We thank Juan Maldacena for this suggestion.}

In this work, inspired by Caron-Huot's beautiful result \cite{Caron-Huot:2017vep}, we have been led to an unusual hybrid of Euclidean and Lorentzian harmonic analysis, i.e.\ harmonic analysis with respect to the groups $\SO(d+1,1)$ and $\tl\SO(d,2)$. However, many of the resulting formulae suggest that it might be fruitful to start with $\tl \SO(d,2)$ from the beginning. For example, after applying the Sommerfeld-Watson trick, Regge correlators are written as an integral over $\De$ and $J$, which is suggestive of an expansion in Lorentzian principal series representations (this observation was also made recently in \cite{Raben:2018sjl}).
It will be important to develop this area further and explore its implications for many of the above questions.\footnote{We thank Abhijit Gadde for emphasizing this idea.}

The intrinsically Lorentzian integral transforms introduced in section~\ref{sec:weylandintegral} have been a key computational tool in this work. These transforms have a natural group-theoretic origin as Knapp-Stein intertwining operators for $\SO(d,2)$, but they can also be applied to representations of $\tl \SO(d,2)$. In this work, we have focused primarily on the light-transform, but the remaining transforms may also have interesting applications. For example, it would be interesting to compute the full monodromy matrix for spinning conformal blocks in terms of intertwining operators, generalizing (\ref{eq:growinglight}). Steps in this direction have already been taken in \cite{Isachenkov:2017qgn}.

One concrete result of this work is a generalization of Caron-Huot's Lorentzian inversion formula to four-point correlators of operators in arbitrary Lorentz representations. Caron-Huot's original formula has already proven useful in a variety of contexts \cite{Alday:2017gde,Dey:2017fab,Henriksson:2017eej,vanLoon:2017xlq,Turiaci:2018nua,Alday:2017vkk,Alday:2017zzv},\footnote{See also \cite{Lemos:2017vnx,Iliesiu:2018fao} for applications of Lorentzian inversion formulae to quantities other than vacuum four-point functions. It would be interesting to understand whether light-ray operators offer a useful perspective on these works.} and we hope that our generalization will be similarly useful. For example, one might try to determine all four-point functions in theories with weakly-broken higher spin symmetry, generalizing the results of \cite{Turiaci:2018nua}. It would also be interesting to study inversion formulae in the context of stress-tensor four-point functions, perhaps making contact with the sum rules in \cite{Cordova:2017zej,Gillioz:2018kwh}.

An important application of Lorentzian inversion formulae is to the lightcone bootstrap and large-spin perturbation theory \cite{Fitzpatrick:2012yx,Komargodski:2012ek,Fitzpatrick:2015qma,Li:2015itl,Li:2015rfa,Alday:2015eya,Alday:2015ota,Alday:2015ewa,Simmons-Duffin:2016wlq,Alday:2016njk,Hofman:2016awc}. Lorentzian inversion formulae make it particularly simple to study OPE coefficients and anomalous dimensions of ``double-twist operators" \cite{Fitzpatrick:2012yx,Komargodski:2012ek} and averaged OPE data for ``multi-twist" operators (see e.g.\ \cite{Alday:2017vkk,Alday:2017zzv}). An important problem for the future is to disentangle individual multi-twist trajectories. It is likely that this will require studying crossing symmetry for higher-point functions. We hope that light-ray operators will offer a useful perspective on this problem.

Another result of this work is a new proof of the average null energy condition (ANEC), obtained by combining the causality-based proof of \cite{Hartman:2016lgu} with the idea of an inversion formula. Our proof has some technical advantages over \cite{Hartman:2016lgu}. For example, it does not use the OPE outside its regime of validity, and it also allows one to move away from the asymptotic lightcone limit. However, it also has disadvantages. In particular, our proof requires the CFT to contain a sufficiently low-dimension operator, and this condition is absent in \cite{Hartman:2016lgu}. It would be interesting to understand whether this condition can be relaxed further while still using an inversion formula. Another technical point that is worth clarifying is the role/necessity of Rindler positivity, as opposed to the more easily-established ``wedge reflection positivity" \cite{Casini:2010bf} or the traditional positivity of norms.

The ANEC has a growing list of interesting applications in conformal field theory \cite{Hofman:2008ar,Chowdhury:2017vel,Cordova:2017zej,Cordova:2017dhq,Meltzer:2017rtf,Casini:2017vbe,Casini:2017roe}. However its higher-spin generalizations \cite{Hartman:2016lgu} have been less well-explored. We have additionally proven that the ANEC holds for continuous spin --- i.e.\ on the entire leading Regge trajectory. It would be interesting to understand the implications of this result, for example in a holographic context. (See \cite{Afkhami-Jeddi:2017rmx} for recent work on shockwave operators, which are holographically dual to light-ray operators.)  It would also be interesting to understand the information-theoretic role of continuous-spin operators. How do they behave under modular flow? Can they appear in OPEs of entangling twist defects? The ANEC can be improved to the quantum null energy condition (QNEC) \cite{Bousso:2015mna,Bousso:2015wca}, which was recently proven in \cite{Balakrishnan:2017bjg} together with a higher integer spin generalization. Is there a continuous-spin version of the QNEC?

\section*{Acknowledgements}

We thank Clay C\'ordova, Thomas Dumitrescu, Abhijit Gadde, Luca Iliesiu, Daniel Jafferis, Alexei Kitaev, Murat Kolo\u{g}lu, Raghu Mahajan, Eric Perlmutter, Matt Strassler, and Aron Wall for helpful discussions. We thank Simon Caron-Huot, Tom Hartman, Denis Karateev, Juan Maldacena, Douglas Stanford, and Sasha Zhiboedov for  discussions and comments on the draft. DSD is supported by Simons Foundation grant 488657 (Simons Collaboration on the Nonperturbative Bootstrap). This work was supported by DOE grant DE-SC0011632 and the Walter Burke Institute for Theoretical Physics.

\newpage

\appendix

\section{Correlators and tensor structures with continuous spin}
\label{sec:contspincorrelators}

In this appendix we assume that there exists a continuous-spin operator $\mathbb{O}(x,z)$ and study its Wightman functions. Note that here we are concerned with physical correlators. In other parts of this paper we discuss the existence of continuous-spin conformal invariants for fixed causal relations between the operator insertions, which is a very different problem -- Wightman functions must be well defined for arbitrary causal relationships between points.

\subsection{Analyticity properties of Wightman functions}

Recall that Wightman functions of local operators are analytic in their arguments when the appropriate $i\e$ prescription is introduced. More precisely, consider a Wightman function of local operators (suppressing polarization vectors for simplicity)
\be\label{eq:localwightman}
\<\O|\cO_n(x_n)\cdots \cO_1(x_1)|\O\>,
\ee
and let us split each $x_k$ into its real and imaginary parts,
\be\label{eq:complexx}
x_k=y_k+i\z_k,\qquad y_k,\z_k\in \R^{d-1,1}.
\ee
The Wightman function~\eqref{eq:localwightman} is analytic in the following region~\cite{streater2016pct,Haag:1992hx} (see \cite{Hartman:2015lfa} for a nice review):\footnote{In fact, these functions are analytic in an even larger region~\cite{streater2016pct,Haag:1992hx}, but we do not study consequences of this extended analyticity in this work.}
\be
\label{eq:tubecondition}
\z_1 > \z_2 >\cdots > \z_n.
\ee
Here, the notation $p>q$ means that $p-q$ is timelike and future-pointing. We will refer to this analyticity property as \textit{positive-energy analyticity}.

Positive-energy analyticity can be derived in the following way. We first represent the Wightman function~\eqref{eq:localwightman} as a Fourier transform
\be\label{eq:localwightmanfourier}
\<\O|\cO_n(x_n)\cdots \cO_1(x_1)|\O\>=\int \frac{d^dp_1}{(2\pi)^d}\cdots \frac{d^dp_n}{(2\pi)^d}e^{-ip_1x_1\ldots -ip_nx_n}\<\O|\cO_n(p_n)\cdots \cO_1(p_1)|\O\>.
\ee
The existence of the Fourier transform follows from the Wightman temperedness axiom. 
The Heisenberg equation implies
\be
\label{eq:heisenberg}
[H,\cO_i(x_i)]=-i\frac{\ptl}{\ptl x_i^0}\cO_i(x_i)
\quad\implies\quad
[H,\cO_i(p_i)]= p^0_i \cO_i(p_i),
\ee
and thus
\be
H\cO_i(p_i)\cdots \cO_1(p_1)|\O\>=(p^0_1+\ldots+ p_i^0)\cO_i(p_i)\cdots \cO_1(p_1)|\O\>.
\ee
In physical theories, all states have positive energies. Furthermore, positivity should hold in any Lorentz frame. Thus, we conclude that whenever $\<\O|\cO_n(p_n)\cdots \cO_1(p_1)|\O\>$ is nonvanishing,
\be\label{eq:spectrum}
p_1+\ldots +p_i\geq 0\qquad (i=1,\dots,n).
\ee
Here, the notation $p\geq 0$ means that $p$ is timelike or null and future-pointing.
Note that the real part of the exponential factor in~\eqref{eq:localwightmanfourier} is given by
\be
\exp\p{\z_1\.p_1+\ldots \z_n\.p_n} &=\exp[(p_n+\ldots+ p_1)\.\z_n\nn\\
&\quad\quad\quad+(p_{n-1}+\ldots+ p_1)\.(\z_{n-1}-\z_{n})\nn\\
&\quad\quad\quad+(p_{n-2}+\ldots+ p_1)\.(\z_{n-2}-\z_{n-1})\nn\\
&\quad\quad\quad+\dots\nn\\
&\quad\quad\quad+p_1\.(\z_1-\z_2)],
\ee
where $\z_k=\mathrm{Im}(x_k)$.  By translation-invariance, the first term in the exponential $(p_n+\ldots+p_1)\.\z_n$ can be replaced with zero. Suppose that the $\z_k$ satisfy (\ref{eq:tubecondition}). Due to $\eqref{eq:spectrum}$, all other terms in the exponential are non-positive and serve to damp the integral~\eqref{eq:localwightmanfourier}. Thus, we can make sense of the Wightman function as an analytic function in this region.

The above discussion in no way depends on locality properties of $\cO_i$. The only information about $\cO_i$ that we needed was the Heisenberg equation~\eqref{eq:heisenberg}. This is of course also satisfied by continuous-spin primary operators $\mathbb{O}(x,z)$, because it is simply part of the definition of being primary. This means that positive-energy analyticity also holds for Wightman functions involving continuous-spin operators. In the main text we construct examples of continuous-spin operators for which positive-energy analyticity can be checked explicitly.

This clarifies the properties of $\mathbb{O}(x,z)$ with respect to $x$. However, $\mathbb{O}(x,z)$ is also a non-trivial function of $z$, and it is interesting to study analyticity in $z$. For this, assume that we have already adopted the appropriate $i\e$-prescription. By using Lorentz and translation symmetries we can assume that we have inserted $\mathbb{O}(x,z)$ at $x=i\e\hat e_0=(i\e,0,\ldots,0)$ with $\e>0$. Then we have for $i,j=1\ldots d-1$
\be
[M_{ij},\mathbb{O}(i\e \hat e_0,z)]&=(z^j \ptl_{z^i}-z^i \ptl_{z^j})\mathbb{O}(i\e \hat e_0,z),
\ee
and so we have an $\mathrm{Spin}(d-1)\subset \tl\SO(d,2)$ subgroup which stabilizes position of $\mathbb{O}$ and allows us to change $z$. In particular, together with the homogeneity property~\eqref{eq:conthomogeneity} it allows us to relate all future-directed null $z$ to $z=\hat e_0+\hat e_1=(1,1,0,\ldots,0)$. Let $U_z\in \mathrm{Spin}(d-1)$ that takes $\a_z(\hat e_0+\hat e_1)$ with $\a_z>0$ to $z$. Then for a Wightman function with a single continuous-spin operator we can write
\be\label{eq:zanalytic}
&\<\O|\cO_n(x_n)\cdots\cO_k(x_k) \mathbb{O}(i\e \hat e_0,z) \cO_{k-1}(x_{k-1})\cdots \cO_1(x_1)|\O\>=\nn\\
&\qquad =\a_z^J\<\O|\cO_n(x_n)\cdots\cO_k(x_k) U_z\mathbb{O}(i\e \hat e_0,\hat e_0+\hat e_1) U^\dagger_z\cO_{k-1}(x_{k-1})\cdots \cO_1(x_1)|\O\>,
\ee
and compute the right hand side by acting with $U_z$ and $U^\dagger_z$ on the left and on the right. This action will act on the spin indices of local operators and also shift their positions. Change in the positions will, however, preserve the ordering of imaginary parts $\z_k$~\eqref{eq:complexx}, and thus the Wightman function will remain in the region of analyticity.\footnote{Note that in principle the stabilizer of $i\e \hat e_0$ includes a full $\mathrm{Spin}(d)\in \tl\SO(d,2)$. However, some of the transformations in $\mathrm{Spin}(d)\backslash\mathrm{Spin}(d-1)$ will change ordering of $\z_k$ and thus move Wightman function out of the region of analyticity.} Since we can take $U_z$ to depend on $z$ analytically in a neighborhood of any given $z$, this implies that in the absence of other continuous-spin operators the left hand side of~\eqref{eq:zanalytic} should be analytic in $z$. 

It would be interesting to understand the analyticity conditions in $z$ in presence of other continuous spin operators. This might depend on some extra assumptions about the nature of such operators, but it is natural to expect them to still be analytic. At least this is the case for the integral transforms defined in section~\ref{sec:weylandintegral}, since at fixed $i\e$-prescription these involve integrals of analytic functions.

\subsection{Two- and three-point functions}

Let us now study examples of Wightman functions of continuous-spin operators from the point of view of positive-energy analyticity. This is especially interesting in CFTs because the analytic structure of two- and three-point functions is fixed by conformal symmetry, and this turns out to be in strong tension with positive-energy analyticity. For simplicity, we focus on correlation functions involving the minimal number of continuous-spin operators. We also restrict to traceless-symmetric tensor operators. However, the same statements hold for general representations because the part of the tensor structure responsible for the discrete spin labels $\l$ is always positive-energy analytic.

A conformally-invariant two-point function of traceless-symmetric operators has the form
\be
\<\mathbb{O}(x_1,z_1)\mathbb{O}(x_2,z_2)\>\propto \frac{\p{2(x_{12}\.z_1)(x_{12}\.z_2)-x_{12}^2(z_1\.z_2)}^J}{x_{12}^{2(\De+J)}}.
\ee
It is easy to check that the denominator is positive-energy analytic for any choice of Wightman ordering, and we only need to study the numerator. For generic $z_1$ and $z_2$ we can write 
\be
x_{12}=\a z_1+\b z_2 + x_\perp,
\ee
where $x_\perp\.z_i=0$. Note that $x_\perp$ is spacelike, because it is orthogonal to the timelike vector $z_1+z_2$. (Recall that all polarization vectors are null and future-directed.) The numerator then takes the form
\be\label{eq:contspin2ptnum}
\p{2(x_{12}\.z_1)(x_{12}\.z_2)-x_{12}^2(z_1\.z_2)}^J=(-z_1\.z_2)^J x_\perp^{2J} > 0.
\ee
On the one hand, we see that this is positive and well-defined for all real $x_i$ and $z_i$. On the other hand, we can show that it is only positive-energy analytic for integer $J\geq 0$. Indeed, selecting a Wightman ordering and adding appropriate imaginary parts as in~\eqref{eq:complexx}, in any case we find that $\z_\perp$ is a spacelike vector (we can make it non-zero), because it is orthogonal to $z_1+z_2$. This means that by choosing an appropriate $y_{12}$ we can achieve
\be
x_\perp^2=y_\perp^2-\z^2_\perp+2i (y_\perp\. \z_\perp)=0,
\ee
and in particular wind $x_\perp^2$ around zero without leaving the region of positive-energy analyticity.\footnote{To be specific, we can wind $x_\perp^2$ around $0$ once with $y_{12}$ returning to the original position, and thus for~\eqref{eq:contspin2ptnum} to be single-valued, we need $J\in\Z$.}${}^,$\footnote{This argument doesn't work in $d=3$ because then $y_\perp$ and $\z_\perp$ are forced to lie in the same 1-dimensional subspace. In that case we are still free to change both $y_\perp$ and $\z_\perp$, and thus $x_\perp=y_\perp+i\z_\perp$, in a neighborhood of $0$. This leads to a weaker requirement that $J\in \half\Z_{\geq0}$. This has to do with the fact that for $d=3$ the null-cone is not simply-connected and it makes sense to consider multi-valued functions of $z$. In fact, fermionic operators can be described by double-valued functions of $z$. (If we write $z_\mu=\chi_\a\chi_\b \sigma^{\a\b}_\mu$ for a real spinor $\chi$, then we get polynomial functions of $\chi$.) Our argument thus shows that only single- and double-valued functions of $z$ are consistent with positive-energy analyticity. In higher dimensions we cannot describe fermionic representations by using a single null polarization and thus we do not get this subtlety.\label{foot:3dfermions}} Thus~\eqref{eq:contspin2ptnum} can not be analytic there unless $J$ is a non-negative integer.

This implies that the only way the Wightman two-point function of a generic continuous spin operator $\mathbb{O}$ can be positive-energy analytic is by being zero,\footnote{We derived this for generic $z_1$ and $z_2$, but as discussed in the previous section, we expect the Wightman functions to be continuous in polarizations.}
\be
\<\O|\mathbb{O}(x_1,z_1)\mathbb{O}(x_2,z_2)|\O\>=0.
\ee
In unitary theories vanishing of this two-point function implies
\be
\mathbb{O}(x,z)|\O\>=0.
\ee
This gives another derivation of the fact stated in the introduction: continuous-spin operators must annihilate the vacuum.

Let us now consider a three-point function with a single continuous-spin operator $\mathbb{O}$,
\be
\<\cO_1(x_1,z_1)\cO_2(x_2,z_2)\mathbb{O}(x_3,z_3)\>\propto f(x_i,z_i)\p{\frac{x_{13}\.z_3}{x_{13}^2}-\frac{x_{23}\.z_3}{x_{23}^2}}^{J_3-n_3},
\ee
where $f(x_i,z_i)$ is the part of the tensor structure which is manifestly positive-energy analytic, and is a homogeneous polynomial in $z_3$ with degree $n_3\geq 0$. The non-trivial part of the correlator can be written as 
\be
\p{\frac{x_{13}\.z_3}{x_{13}^2}-\frac{x_{23}\.z_3}{x_{23}^2}}^{J_3-n_3}=(v_{12,3}\. z_3)^{J_3-n_3},
\ee
where 
\be
v_{12,3}^2=\p{\frac{x_{13}}{x_{13}^2}-\frac{x_{23}}{x_{23}^2}}^2=\frac{x_{12}^2}{ x_{13}^2 x_{23}^2}.
\ee
We see that $v_{12,3}$ can be both spacelike and timelike, depending on the causal relationship between the three points $x_i$. This immediately implies that, for example, when all $x_{ij}$ are spacelike, the inner product $v_{12,3}\. z_3$ is not sign-definite and we need to invoke $i\e$-prescriptions to define $(v_{12,3}\. z_3)^{J_3-n_3}$, even for purely Euclidean configurations. For the $i\e$-prescriptions to make sense, the tensor structure must be positive-energy analytic. This means that in this situation, positive-energy analyticity is not only required for correlators to make physical sense, but also simply for the tensor structures to be single-valued.\footnote{This is in contrast to the two-point Wightman function case considered above, where~\eqref{eq:contspin2ptnum} is single-valued without the $i\e$-prescription.} To proceed, note that in the region of positive-energy analyticity $x_{ij}^2\neq 0$ and furthermore the map
\be
x\mapsto \frac{x}{x^2}
\ee
preserves the set of $x=y+i\z$ with future-directed (past-directed) timelike $\z$.\footnote{If $x^2=(y+i\z)^2=y^2-\z^2+2iy\.\z=0$ with timelike $\z$, then $y\.\z=0$, which implies that $y$ is spacelike and thus $y^2-\z^2>0$, leading to contradiction. Imaginary part of $\frac{x}{x^2}$ is, up to a positive factor, $\z(y^2-\z^2)-2y(y\.\z)$. For $y=0$ this is timelike and has the same direction as $\z$. For any $y$, this squares to $\z^2((y^2-\z^2)^2+4(y\.z)^2)<0$, and thus by continuity $\mathrm{Im}\frac{x}{x^2}$ remains timelike in the direction of $\z$.} Since it is also its own inverse, this implies that by varying $x_{13}$ and $x_{23}$ within the region of positive-energy analyticity, we can reproduce any pair of values for $q_1=\frac{x_{13}}{x_{13}^2}$ and $q_2=\frac{x_{23}}{x_{23}^2}$ with imaginary parts satisfying the same constraints as those of $x_{13}$ and $x_{23}$ respectively. This means that in the region of positive-energy analyticity for the orderings 
\be\label{eq:goodorderings}
\<0|\cO_2\mathbb{O}\cO_1|0\>\quad\text{and}\quad \<0|\cO_1\mathbb{O}\cO_2|0\>,
\ee
the vector $v_{12,3}=q_1-q_2$ has a timelike imaginary part restricted to be future-directed or past-directed respectively,
while for the orderings
\be\label{eq:badorderings}
\<0|\cO_i\cO_j\mathbb{O}|0\>\quad\text{and}\quad \<0|\mathbb{O}\cO_i\cO_j|0\>
\ee
this imaginary part is not restricted at all. In the former case $v_{12,3}\.z_3$ has either negative or positive imaginary part, and thus the inner product cannot vanish or wind around zero, while in the latter case this inner product can vanish or wind around zero. We thus conclude that the Wightman functions~\eqref{eq:goodorderings} are positive-energy analytic for any value of $J_3$, while the Wightman functions~\eqref{eq:badorderings} are positive-energy analytic only for integer $J_3\geq n_3$.\footnote{Recall that $n_3\leq J_3$ is the standard condition that we encounter when dealing with integer-spin tensor structures, it just means that $f(x_i,z_i)$ must be a polynomial in $z_3$ of degree at most $J_3$. The 3d subtlety we discussed in footnote~\ref{foot:3dfermions} would be visible here as well, if we allowed $f$ to be double-valued in $z$ (and polynomial in $\chi$), which would correspond to making the product $\cO_1\cO_2$ fermionic, thus forcing $J$ to be half-integer.} 

Again, recalling that the physical Wightman functions of continuous-spin operators must be positive-energy analytic, we are forced to conclude that Wightman functions~\eqref{eq:badorderings} vanish,
\be
\<\O|\cO_1\cO_2\mathbb{O}|\O\>=\<\O|\mathbb{O}\cO_1\cO_2|\O\>=0,
\ee
which of course consistent with the fact that $\mathbb{O}$ annihilates the vacuum. An interesting observation  is that the distinction we made above between the Wightman orderings~\eqref{eq:goodorderings} and~\eqref{eq:badorderings} conflicts with microcausality, because for spacelike-separated points all these Wightman functions would be equal.\footnote{Recall that as noted above, the region of spacelike separation is the problematic one, because there $v_{12,3}$ is spacelike and $v_{12,3}\.z_3$ is not sign-definite.} This means that non-trivial continuous-spin operators \textit{must} be non-local, as stated in the introduction, in the sense that they cannot satisfy microcausality.

A consequence of non-locality is that a physical correlator involving a continuous-spin operator is not well-defined without specifying an operator ordering even if all the distances are spacelike. This in particular means that time-ordered correlators are not quite well-defined in the presence of continuous-spin operators (i.e.\ how do we order $\mathbb{O}$ when it is spacelike from something?). This also makes it unclear how one would define Euclidean correlators for continuous spin (the usual Wick-rotation to Euclidean signature requires micro-causality). Another problem with attempting to describe continuous-spin operators in Euclidean signature is that under Euclidean rotation group $\SO(d)$ the orbit of a single null direction in $\R^{d-1,1}$ consists of all null directions in $\C^d$. Thus we would need to define $\mathbb{O}(x,z)$ for all complex null $z$, but above it was very important to have \textit{future-directed real} $z$ to establish positive-energy analyticity of at least some Wightman functions. 

\subsection{Conventions for two- and three-point tensor structures}
\label{app:23conventions}

When working with integer spin the simplest way to specify standard tensor structures is to give their expressions in Euclidean signature or, equivalently, in Lorentzian signature with all points are spacelike separated. With continuous spin, Euclidean signature is not an option, and as we saw above even for spacelike separations in Lorentzian signature care must be taken to define phases of three-point functions. In this section we briefly record our conventions for symmetric tensor operators.

We will choose the following convention for a two-point function in Lorentzian signature:
\be
\label{eq:standardtwoptconvention}
\<\cO(x_1,z_1)\cO(x_2,z_2)\> &= \frac{(-2 z_1 \. I(x_{12}) z_2)^J}{x_{12}^{2\De}} \nn\\
I^\mu{}_\nu(x) &= \de^\mu{}_\nu - 2\frac{x^\mu x_\nu}{x^2}.
\ee
The nonstandard numerator is so that the two-point function is positive when $1$ and $2$ are spacelike separated and $z_{1,2}$ are future-pointing null vectors. For local operators this completely defines standard Wightman two-point functions via $i\e$ prescriptions. For continuous-spin operators physical Wightman functions vanish, but we still need two-point conformal invariants in some calculations (like the definition of the $\wS$-transform), and for these purposes it suffices to specify the two-point invariant for spacelike $x_{12}$.

Now consider a three-point function $\<\phi_1(x_1)\phi_2(x_2)\cO(x_3,z)\>$, where $\f_1$ and $\f_2$ are scalars and $\cO$ has dimension $\De$ and spin $J$. We demand that the correlator (either Wightman or time-ordered) should be positive when $1,2,3$ are mutually spacelike and $z\.x_{23}\, x_{13}^2 - z\.x_{13}\, x_{23}^2 > 0$. Our precise convention is
\be
\label{eq:standardthreeptconvention}
\<\phi_1(x_1)\phi_2(x_2)\cO(x_3,z)\>=\frac{\p{2z\.x_{23}\, x_{13}^2 - 2z\.x_{13}\, x_{23}^2}^J}{x_{12}^{\De_1+\De_2-\De+J} x_{13}^{\De_1+\De-\De_2+J} x_{23}^{\De_2+\De-\De_1+J}}.
\ee
This is unambiguous for local operators since at spacelike separations there is no difference between various Wightman orderings and time-ordering.\footnote{Note however that this notation for the standard structure is somewhat abusive. For physical correlators we of course have $\<\f_1\f_2\cO\>_\O=\<\f_2\f_1\cO\>_\O$, but the standard structure~\eqref{eq:standardthreeptconvention} gains a $(-1)^J$ under this permutation. This leads to several appearances of $(-1)^J$ in our formulas which are awkward to explain.} If $J$ is continuous, we are necessarily talking about a Wightman function and we need to specify the ordering. Our choice is
\be\label{eq:standardthreeptconventioncontspin}
	\<0|\phi_1(x_1)\mathbb{O}(x_3,z)\phi_2(x_2)|0\>=\frac{\p{2z\.x_{23}\, x_{13}^2 - 2z\.x_{13}\, x_{23}^2}^J}{x_{12}^{\De_1+\De_2-\De+J} x_{13}^{\De_1+\De-\De_2+J} x_{23}^{\De_2+\De-\De_1+J}},
\ee
defined to be positive under the same conditions as~\eqref{eq:standardthreeptconvention}.

The nontraditional factors of $2$ in (\ref{eq:standardtwoptconvention}) and 
(\ref{eq:standardthreeptconvention}) are so that the associated conformal blocks have simple behavior in the limit of small cross-ratios
\be
\frac{\<\f_1\f_2\cO\>\<\cO\f_3\f_4\>}{\<\cO\cO\>} &\sim \p{\prod x_{ij}^\#}\chi^{\frac{\De-J}{2}} \bar \chi^{\frac{\De+J}{2}} \qquad \chi\ll \bar \chi \ll 1.
\ee
They also simplify several formulae in the main text.

\section{Relations between integral transforms}

\subsection{Square of light transform}
\label{app:squarelight}

In this appendix we explicitly compute the square of the light transform. In order to do this, we need to assume that the operator that the light transform acts upon belongs to the Lorentzian principal series
\be
	\De=\frac{d}{2}+is,\quad J=-\frac{d-2}{2}+iq,
\ee
so that $\De+J=1+i(s+q)=1+i\w$ and $\De^L+J^L=2-\De-J=1-i(s+q)=1-i\w$ and thus both the first and the second light transforms make sense if $w\neq 0$.

It will also be convenient to use the expression for the light transform in the coordinates $(\tau,\vec e)$ on $\widetilde\cM_d$. In these coordinates the polarization vector $z$ can be described as $(z^0, \vec z)$ where $\vec z$ is tangent to $S^{d-1}$ at $\vec e$, i.e.\ $\vec z\cdot \vec e=0$, and we have $(z^0)^2=|\vec z|^2$. We then have
\be\label{eq:lightoncylinder}
\wL[\cO](\tau,\vec e;z^0,\vec z)=\int_{0}^\pi d\kappa \,(\sin\kappa)^{\De+J-2}
(z^0)^{1-\De}\cO(\tau+\kappa,\cos\kappa\, \vec e + \sin\kappa\, \tfrac{\vec z}{z^0};
1,\cos\kappa\, \tfrac{\vec z}{z^0} - \sin\kappa\, \vec e
).
\ee
Note that this form also makes it manifest that there is no singularity associated to $\a=0$ in~\eqref{eq:lightdefinition}.

The square of light transform becomes
\be
\wL^2[\cO](\tau,\vec e; z^0, \vec z)=&\int_0^\pi\int_0^\pi d\kappa d\kappa' (z^0)^{J} (\sin \kappa')^{-\De-J}(\sin \kappa)^{\De+J-2}\nn\\ &\quad\times\cO(\tau+\kappa+\kappa',\cos(\kappa+\kappa') \vec e+\sin(\kappa+\kappa')\tfrac{\vec z}{z^0}; 1, \cos(\kappa+\kappa')\tfrac{\vec z}{z^0}-\sin(\kappa+\kappa') \vec e)\nn\\
=&\int_0^{2\pi} d\kappa K(\kappa) (z^0)^J\cO(\tau+\kappa,\cos\kappa\, \vec e + \sin\kappa\, \tfrac{\vec z}{z^0};
1,\cos\kappa\, \tfrac{\vec z}{z^0} - \sin\kappa\, \vec e
),
\ee
where
\be
K(\kappa)=\int_{\max(-\kappa/2,\kappa/2-\pi)}^{\min(\kappa/2,\pi-\kappa/2)} d\eta (\sin \frac{\kappa}{2}-\eta)^{-1-i\w}(\sin \frac{\kappa}{2}+\eta)^{-1+i\w}.
\ee
 
To compute $K(\kappa)$, for $\kappa\neq 0,\pi,2\pi$ we can use the substitution 
\be
	e^\beta = \frac{\sin\p{\frac{\kappa}{2}+\eta}}{\sin\p{\frac{\kappa}{2}-\eta}},
\ee
which turns the integral into
\be
	K(\kappa)=\frac{1}{\sin\kappa}\int_{-\oo}^{+\oo} d\beta e^{iw\beta}=0,\quad (\w\neq0).
\ee
This means that $K(\kappa)$ is supported at $\kappa=0,\pi,2\pi$. Let us thus consider first the contribution near $\kappa=0$. Near $\kappa=0$ we can expand both sines and find, introducing a regulator $\epsilon$,
\be
K(\kappa)=&\int_{-\kappa/2}^{\kappa/2} d\eta \left(\frac{\kappa}{2}-\eta\right)^{-1-i\w+\epsilon}\left(\frac{\kappa}{2}+\eta\right)^{-1+i\w+\epsilon}\nn\\
=&\kappa^{-1+2\epsilon}\int_{-1/2}^{1/2} d\eta \left(\frac{1}{2}-\eta\right)^{-1-i\w+\epsilon}\left(\frac{1}{2}+\eta\right)^{-1+i\w+\epsilon}\nn\\
=&(2\e)\kappa^{-1+2\epsilon}\frac{\Gamma(i\w+\epsilon)\Gamma(-i\w+\epsilon)}{(2\e)\Gamma(2\epsilon)}.\qquad(\kappa\ll 1)
\ee
For $\e\to 0$, using
\be
	(2\e)\kappa^{-1+2\epsilon}\to\delta(\kappa),\quad(\kappa>0)
\ee
we find
\be
	K(\kappa)=\Gamma(-i\w)\Gamma(i\w)\delta(\kappa)=\frac{\pi}{(\De+J-1)\sin\pi(\De+J)}\delta(\kappa),\quad(\kappa\ll1).
\ee
The calculation near $\kappa=2\pi$ is the same and thus we have 
\be
K(\kappa)=\frac{\pi}{(\De+J-1)\sin\pi(\De+J)}(\delta(\kappa)+\delta(\kappa-2\pi))+\<\text{contribution from }\pi\>
\ee

To find the contribution from $\kappa=\pi$, write $\kappa=\pi-r$ for small $0<r\ll 1$.\footnote{There is going a similar contribution from $r<0$.} We have now
\be
K(\kappa)=&\int_{-\frac{\pi}{2}+\frac{r}{2}}^{\frac{\pi}{2}-\frac{r}{2}} d\eta \left(\sin \frac{\pi}{2}-\frac{r}{2}-\eta\right)^{-1-i\w}\left(\sin \frac{\pi}{2}-\frac{r}{2}+\eta\right)^{-1+i\w}\nn\\
=&\int_{0}^{\pi-r} d\eta \left(\sin r+\eta\right)^{-1-i\w}\left(\sin \eta\right)^{-1+i\w}\nn\\
\approx & \int_{0}^{N r} d\eta \left(r+\eta\right)^{-1-i\w+\epsilon}\eta^{-1+i\w+\epsilon}+\int_{0}^{N r} d\eta \left(r+\eta\right)^{-1+i\w+\epsilon}\eta^{-1-i\w+\epsilon}\nn\\
=& r^{-1+2\epsilon}\left[\int_{0}^\infty d\eta \left(1+\eta\right)^{-1+i\w+\epsilon}\eta^{-1-i\w+\epsilon}+\int_{0}^\infty d\eta \left(1+\eta\right)^{-1-i\w+\epsilon}\eta^{-1+i\w+\epsilon}\right]\nn\\
=& r^{-1+2\epsilon}\frac{\pi\Gamma(1-2\epsilon)}{\Gamma(2-J-\De-\epsilon)\Gamma(J+\De-\epsilon)}\left(\csc(\pi(J+\De-\epsilon))-\csc(\pi(J+\De+\epsilon))\right).
\ee
Here $0\ll r\ll Nr\ll 1$ and the two terms come from the two sides of the integral. We can now compute for small $\L>0$
\be
\lim_{\epsilon\to 0}\int_{\pi-\Lambda}^\pi K(\kappa) d\kappa=-\frac{\pi \cos\pi(\De+J)}{(\De+J-1)\sin\pi(\De+J)}.
\ee
Recalling also that there is also a contribution from the negative values of $r$, we find the final result
\be
K(\kappa)=\frac{\pi}{(\De+J-1)\sin\pi(\De+J)}\left(\delta(\kappa)-2\cos\pi(\De+J)\delta(\kappa-\pi)+\delta(\kappa-2\pi)\right).
\ee
In terms of action on $\cO$ this immediately implies
\be
	\wL^2=&\frac{\pi}{(\De+J-1)\sin\pi(\De+J)}\left(1-2\cos\pi(\De+J)\tsym+\tsym^2\right)\nn\\=&\frac{\pi}{(\De+J-1)\sin\pi(\De+J)}\p{\tsym-e^{i\pi(\De+J)}}\p{\tsym-e^{-i\pi(\De+J)}}.
\ee

\subsection{Relation between shadow transform and light transform}
\label{app:shadowLSLrelation}

In this appendix we prove the relation~\eqref{eq:wsdwlwsjrelation}. As in the preceding part of this appendix, we must assume that~\eqref{eq:wsdwlwsjrelation} acts on an operator in the Lorentzian principal series so that this action is well-defined. We have
\be\label{eq:LSLintegral}
	\wL \wSJ\wL[\cO](x,z)=\int D^{d-2}z' d\a_1 d\a_2 (-\a_1)^{-\De-J}(-\a_2)^{d-2+J-\De}
	(-2 z\.z')^{1-d+\De}\cO(x-z'/\a_1-z/\a_2,z')
\ee
Let us write $x'=x-z'/\a_1-z/\a_2$. Then we have 
\be
	I(x-x')z=z-2\frac{(z'/\a_1+z/\a_2)(z'/\a_1+z/\a_2)\.z}{(z'/\a_1+z/\a_2)^2}=-\frac{\a_2}{\a_1}z'.
\ee
Considering the integral in the region of large negative $\a_1$ and $\a_2$ we find
\be\label{eq:intermediateintegralforLSL}
	&\int D^{d-2}z' d\a_1 d\a_2 (-\a_1)^{-\De-J}(-\a_2)^{d-2+J-\De}
	\p{-\a_1\a_2(x-x')^2}^{1-d+\De}\p{\frac{\a_1}{\a_2}}^J\cO(x',-I(x-x')z)\nn\\
	&=\int D^{d-2}z' d\a_1 d\a_2 (-\a_1)^{1-d}(-\a_2)^{-1}
	\p{-(x-x')^2}^{1-d+\De}\cO(x',-I(x-x')z)
\ee
We would now like to replace the integral $\int D^{d-2}z' d\a_1 d\a_2$ by $\int d^dx'$. For this we write
\be
	1=\int d^dx' \delta^d(x-x'-z'/\a_1-z/\a_2)
\ee
and then compute
\be
	&\int D^{d-2}z' d\a_1 d\a_2 (-\a_1)^{1-d}(-\a_2)^{-1}\delta^d(x-x'-z'/\a_1-z/\a_2)\nn\\
	&=\int \frac{2 d^d z'  d\a_1 d\a_2}{\vol\R}  \theta(z^{\prime0})\delta(z^{\prime2})(-\a_1)(-\a_2)^{-1}\delta^d(-\a_1(x-x')+z'+\a_1 z/\a_2)\nn\\
	&=\int \frac{2 d\a_1 d\a_2}{\vol\R} \delta(((x-x')-z/\a_2)^2)(-\a_1)^{-1}(-\a_2)^{-1}\nn\\
	&=-2(x-x')^{-2}.
\ee
We thus conclude that~\eqref{eq:intermediateintegralforLSL} is equal to
\be
	2\int d^d x' (-(x-x')^2)^{\De-d} \cO(x',-I(x-x')z).
\ee
More precisely, it is the contribution to~\eqref{eq:LSLintegral} from the region of large negative $\a_i$. We recognize that it has precisely the form of $\tsym$-shifted Lorentzian shadow integral~\eqref{eq:lorentzshadow}, i.e.
\be
	2\wSD=i\tsym^{-1}\wL\wSJ\wL.
\ee

\section{Harmonic analysis for the Euclidean conformal group}
\label{sec:euclideanharmonicanalysis}

\subsection{Pairings between three-point structures}
\label{eq:euclideanpairings}

The conformal representation of an operator $\cO$ is labeled by a scaling dimension $\De$ and an $\SO(d)$ representation $\rho$. The representation $\tl \cO^\dag$ has dimension $d-\De$ and $\SO(d)$ representation $\rho^*$ (the dual of $\rho$). Thus, there is a natural conformally-invariant pairing between $n$-point functions of $\cO_i$'s and $n$-point functions of $\tl \cO_i^\dag$'s, given by multiplying and integrating over all points modulo the conformal group,
\be
\p{\<\cO_1\cdots \cO_n\>,\<\tl \cO_1^\dag \cdots \tl \cO_n^\dag\>}_E &= \int \frac{d^d x_1\cdots d^d x_n}{\vol(\SO(d+1,1))}\<\cO_1\cdots \cO_n\>\<\tl \cO_1^\dag \cdots \tl \cO_n^\dag\>.
\ee
Here, we are implicitly contracting Lorentz indices between each pair $\cO_i$ and $\tl \cO^\dag_i$. The ``$E$" subscript stands for ``Euclidean."

This pairing is particularly simple for three-point structures. In that case, we can use conformal transformations to set $x_1=0,x_2=e,x_3=\oo$ (with $e$ a unit vector), and no integrations are necessary. The pairing becomes simply
\be
\label{eq:euclideanthreeptpairing}
\p{\<\cO_1 \cO_2 \cO_3\>,\<\tl \cO_1^\dag \tl \cO_2^\dag \tl \cO_3^\dag\>}_E &= \frac{1}{2^d\vol(\SO(d-1))} \<\cO_1(0)\cO_2(e)\cO_3(\oo)\>\<\tl \cO_1^\dag(0) \tl \cO_2^\dag(e) \tl \cO_3^\dag(\oo)\>.
\ee
The factor $2^{-d}$ comes from the Fadeev-Popov determinant for the above gauge-fixing.\footnote{Note that \cite{Simmons-Duffin:2017nub} used a convention where $\vol(\SO(d+1,1))$ was defined to include an extra factor of $2^{-d}$ to cancel the Fadeev-Popov determinant. Here, we prefer not to cancel this factor because it simplifies other formulae in this work.} The factor $\vol(\SO(d-1))$ is the volume of the stabilizer group of three points.

As an example, a scalar-scalar-spin-$J$ correlator has a single tensor structure $\<\f_1\f_2\cO_{3,J}\>$ given in (\ref{eq:standardthreeptconvention}). The pairing in that case is
\be
\label{eq:scalarscalarspinjpairing}
\p{\<\f_1\f_2\cO_{3,J}\>,\<\tl \f_1\tl \f_2\tl \cO_{3,J}\>}_E
&= \frac{2^{2J}}{2^d\vol(\SO(d-1))}(e^{\mu_1}\cdots e^{\mu_J}-\textrm{traces})(e_{\mu_1}\cdots e_{\mu_J} - \textrm{traces}) \nn\\
&= \frac{2^{2J}\hat C_J(1)}{2^d\vol(\SO(d-1))},
\ee
where $\hat C_J(x)$ is defined in (\ref{eq:hatC}).

\subsection{Euclidean conformal integrals}
\label{app:euclideanintegrals}

Suppose $\cO,\cO'$ are principal series representations, with dimensions $\De=\frac d 2+is,\De'=\frac d 2 + is'$ with $s,s'>0$ and $\SO(d)$ representations $\rho,\rho'$.
A ``bubble" integral of two three-point functions is proportional to their three-point pairing, 
\be
\label{eq:bubbleintegral}
\int d^dx_1 d^dx_2 \< \cO_1 \cO_2 \cO^a(x) \> \<\tl O^\dag_1 \tl O^\dag_2 \tl O^{'\dag}_b(x') \> &= \frac{\p{\< \cO_1 \cO_2 \cO\>,\<\tl O^\dag_1 \tl O^\dag_2 \tl O^{\dag}\>}_E}{ \mu(\De,\rho)} \de^a_b \de(x-x') \de_{\cO\cO'},\nn\\
\de_{\cO\cO'} &\equiv 2\pi \de(s-s') \de_{\rho \rho'}.
\ee
The right-hand side contains a term $\de_{\cO\cO'}$ restricting the representations $\cO,\cO'$ to be the same, since this is the only possibility allowed by conformal invariance.\footnote{Eq.~(\ref{eq:bubbleintegral}) is sometimes written including two terms --- one with $\de(s-s')$ and another with $\de(s+s')$. Here we have only one term because we have restricted $s,s'> 0$. The other term can be obtained by performing the shadow transform on either $\cO$ or $\tl \cO^{'\dag}$.} Here, $a,b$ are indices for the representations $\rho,\rho^*$ of $\SO(d)$, respectively. We have suppressed the $\SO(d)$ indices of the other operators, for brevity.

The factor $\mu(\De,\rho)$ in the denominator is called the Plancherel measure. It is known in great generality \cite{Dobrev:1977qv} (see \cite{ShadowFuture} for an elementary derivation). In this work, we will only need $\mu(\De,J)$ for symmetric traceless tensors:
\be
 \mu(\De,J) 
 &= \frac{\dim \rho_J }{2^d\vol(\SO(d))}\frac{\G(\De-1)\G(d-\De-1)(\De+J-1)(d-\De+J-1)}{\pi^d \G(\De-\frac d 2)\G(\frac d 2-\De)},\nn\\
 \dim\rho_J &=  \frac{\G(J+d-2)(2J+d-2)}{\G(J+1)\G(d-1)}.
 \label{eq:plancherelexample}
\ee
Here $\dim \rho_J$ is the dimension of the spin-$J$ representation of $\SO(d)$.

Another conformal integral we will need is the Euclidean shadow transform of a three-point function of two scalars and a symmetric traceless tensor
\be
\<\f_1\f_2 \wSE[\cO](y)\> &= \int d^d x \<\tl \cO(y) \tl \cO^\dag(x)\>\<\f_1 \f_2\cO(x)\>\nn\\
&= S_E(\f_1\f_2[\cO])\<\f_1\f_2\tl \cO(y)\>,
\ee
where
\be
\label{eq:scalarshadowfactor}
S_E(\f_1\f_2[\cO]) &= (-2)^J \frac{\pi^{d/2}\G(\De-\frac d 2)\G(\De+J-1)}{\G(\De-1)\G(d-\De+J)} \frac{\G(\frac{d-\De+\De_1-\De_2+J}{2})\G(\frac{d-\De+\De_2-\De_1+J}{2})}{\G(\frac{\De+\De_1-\De_2+J}{2})\G(\frac{\De+\De_2-\De_1+J}{2})}.
\ee
The factor of $(-2)^{J}$ relative to \cite{Simmons-Duffin:2017nub} is because we are using a different normalization convention for the two-point function (\ref{eq:standardtwoptconvention}).

The square of the shadow transform is related to the Plancherel measure by \cite{Dobrev:1977qv} (see \cite{ShadowFuture} for an elementary derivation)
\be
\label{eq:plancherelvsshadowsq}
\wSE^2 &= \frac{1}{\mu(\De,\rho)} \frac{\<\cO(0)\cO^\dag(\oo)\>\<\tl\cO(\oo)\tl \cO^\dag(0)\>}{2^d\vol(\SO(d))}\equiv \cN(\De,\rho),
\ee
where the indices in two-point functions are implicitly contracted. In the case of a spin-$J$ representation, we have
\be
\cN(\De,J) &=  \frac{2^{2J} \dim \rho_J}{2^d\mu(\De,J)\vol(\SO(d))},
\ee
Indeed, we can easily verify
\be
\label{eq:shadowsq}
S_E(\f_1\f_2[\cO])S_E(\f_1\f_2[\tl\cO]) &= \cN(\De,J).
\ee

\subsection{Residues of Euclidean partial waves}
\label{sec:argumentforresidues}

In this section, we prove~\ref{eq:residueequation}. The proof for primary four-point functions is standard (see e.g.\ \cite{Dobrev:1977qv,Simmons-Duffin:2017nub}). We now give a slightly more complicated argument that works for $n$-point functions. However, the key ingredients are identical to the standard argument.

Consider the integral in the completeness relation (\ref{eq:multipointcompleteness}),
\be
I &= \int d^d x P_{\De,J}(x) \< \tl \cO(x) \f_1 \f_2\>.
\ee
The partial wave $P_{\De,J}$ also depends on the coordinates $x_3,\dots,x_k$, but they don't play a role in the current discussion so we have suppressed them. We have also suppressed Lorentz indices. When we have a product of an operator and its shadow at coincident points, we will assume their Lorentz indices are contracted.

Note that $I$ is an eigenvector of the Casimirs of the conformal group acting simultaneously on points $1$ and $2$. Thus, it is completely determined by its behavior in the OPE limit $x_1\to x_2$. There are two contributions in this limit. The first comes from the regime where $x$ is sufficiently far from $x_1,x_2$ that we can use the $1\x 2$ OPE inside the integrand:
\be
\label{eq:doOPE}
\<\f_1 \f_2 \tl \cO(x)\> &= C_{12\tl \cO}(x_1,x_2,x',\ptl_{x'})\<\tl \cO(x') \tl \cO(x)\>.
\ee
Here, $C_{12\tl \cO}$ is a differential operator that encodes the sum over descendants in the $\f_1 \x \f_2$ OPE. The point $x'$ can be chosen arbitrarily inside a sphere separating $x_1,x_2$ from all other points.
We will abbreviate the right-hand side of (\ref{eq:doOPE}) as $C_{12\tl \cO}(x')\<\tl \cO(x')\tl \cO(x)\>$. Inserting (\ref{eq:doOPE}) and applying the shadow transform to the definition of $P_{\De,J}$ (\ref{eq:euclideanintegralforpartialwave}), we find
\be
I & \supset C_{12\tl \cO}(x') \int d^d x \<\tl \cO(x')\tl \cO(x)\>P_{\De,J}(x)
= S_E(\f_1\f_2[\cO]) C_{12\tl \cO}(x) P_{\tl \De,J}(x).
\label{eq:contributionone}
\ee

The second contribution to $I$ comes from the regime where $x$ is near both $x_1,x_2$ but far away from all other points. In this case, we can insert a shadow transform and then perform the OPE:
\be
I &= S_E(\f_1 \f_2[\cO])^{-1} \int d^d x d^d x' P_{\De,J}(x) \<\tl \cO(x) \tl \cO(x')\>\<\cO(x') \f_1 \f_2\> \nn\\
&\supset S_E(\f_1 \f_2[\cO])^{-1} \int d^d x d^d x' P_{\De,J}(x)  \<\tl \cO(x)\tl \cO(x')\> C_{12\cO}(x'')\<\cO(x'')\cO(x')\> \nn\\
&=S_E(\f_1 \f_2[\cO])^{-1} \cN(\De,J) C_{12\cO}(x) P_{\De,J}(x) \nn\\
&= S_E(\f_1 \f_2[\tl \cO]) C_{12\cO}(x) P_{\De,J}(x). 
\label{eq:contributiontwo}
\ee
Where we have used (\ref{eq:shadowsq}).

The two contributions (\ref{eq:contributionone}) and (\ref{eq:contributiontwo}) are already eigenvectors of the conformal Casimirs, so together they give the full answer for $I$. The two terms differ simply by the replacement $\De\leftrightarrow d-\De$. Thus, we can plug them into the completeness relation (\ref{eq:multipointcompleteness}) and use $\De\leftrightarrow d-\De$ symmetry to extend the $\De$ integral along the entire imaginary axis,
\be
\<V_3 \cdots V_k \cO_1 \cO_2\>_\Omega &= \sum_{J=0}^\oo \int_{\frac d 2 - i\oo}^{\frac d 2 + i\oo} \frac{d\De}{2\pi i} \mu(\De,J)  S_E(\f_1 \f_2[\tl \cO]) C_{12\cO} P_{\De,J}(x).
\ee
Because $C_{12\cO}$ dies exponentially at large positive $\De$, we can now close the $\De$ contour to the right and pick up poles along the positive real axis.  Comparing to the physical operator product expansion gives (\ref{eq:residueequation}).

\section{Computation of $\cR(\De_1,\De_2,J)$}
\label{app:Rcoefficient}
In this appendix we compute the coefficient $\cR$ appearing in the first line of~\eqref{eq:Rappearance}
\be
	\cR(\De_1,\De_2,J)\equiv -2^{J-2}\int d\a  d^{d-2}\bw_1d^{d-2}\bw_2\frac{2^{J-1}
		\p{\a(1-\a)+(1-\a)\bw_1^2+\a\bw_2^2}^{1-\De_1-\De_2-J}
	}{
		(1+\bw_{12}^2)^{d-\De_1-\De_2}
		\a^{{-\De_1+1-J}}
		(1-\a)^{{-\De_2+1-J}}
	}.
\ee
As the first step, we do the $\bw_i$ integrals. We define $\bw_-=\bw_{12}$ and $\bw_+=\bw_1+\bw_2$. The integral over $dw_i$ becomes (without the $-2^{J-2}$ and $\bw$-independent factors)
\be
	2^{2(\De_1+\De_2+J)-d}\int d^{d-2}\bw_+d^{d-2}\bw_-\frac{
		\p{4\a(1-\a)+\bw_+^2+\bw_-^2+2(1-2\a)\bw_+\.\bw_-}^{1-\De_1-\De_2-J}
	}{
		(1+\bw_-^2)^{d-\De_1-\De_2}
	}.
\ee
Now we shift $\bw_+\to \bw_+-(1-2\a)\bw_-$ to find
\be
	2^{2(\De_1+\De_2+J)-d}\int d^{d-2}\bw_+d^{d-2}\bw_-\frac{
	\p{4\a(1-\a)(1+\bw_-^2)+\bw_+^2}^{1-\De_1-\De_2-J}
	}{
		(1+\bw_-^2)^{d-\De_1-\De_2}
	}.
\ee
Rescaling $\bw_+$ we find
\be
	&\int d^{d-2}\bw_+d^{d-2}\bw_-\frac{
	\p{\a(1-\a)}^{1-\De_1-\De_2-J+\frac{d-2}{2}}\p{1+\bw_+^2}^{1-\De_1-\De_2-J}
}{
	(1+\bw_-^2)^{J+\frac{d}{2}}
}=\nn\\
&=(\a(1-\a))^{1-\De_1-\De_2-J+\frac{d-2}{2}}\times  \pi^{d-2}\frac{\Gamma(J+1)\Gamma{(-\frac{d}{2}+J+\De_1+\De_2)}}{\Gamma{(J+\frac{d}{2})}\Gamma{(J+\De_1+\De_2-1)}}.
\ee
The remaining $\a$-integral becomes
\be
	\int d\a \a^{-\De_2+\frac{d-2}{2}}(1-\a)^{-\De_1+\frac{d-2}{2}}=\frac{\Gamma(\frac{d}{2}-\De_1)\Gamma(\frac{d}{2}-\De_2)}{\Gamma(d-\De_1-\De_2)}.
\ee
Combining everything together we find
\be
	\cR(\De_1,\De_2,J)=-2^{J-2}\pi^{d-2}\frac{\Gamma(J+1)\Gamma{(-\frac{d}{2}+J+\De_1+\De_2)}}{\Gamma{(J+\frac{d}{2})}\Gamma{(J+\De_1+\De_2-1)}}\frac{\Gamma(\frac{d}{2}-\De_1)\Gamma(\frac{d}{2}-\De_2)}{\Gamma(d-\De_1-\De_2)}.
\ee

\section{Parings of continuous-spin structures}
\label{app:contspinpairings}

In this section we describe the natural conformally-invariant pairing between continuous spin structures. Recall that the Euclidean pairings are constructed from the basic invariant integral
\be
	\int d^d x \cO(x)\tl\cO^\dagger(x),
\ee
where contraction of $\SO(d)$ indices is implicit. This integral is conformally-invariant because if $\cO$ transforms in $(\De,\rho)$ then $\tl\cO^\dagger$ transforms in $(d-\De,\rho^*)$, where $\rho^*$ is the $\SO(d)$ irrep dual to $\rho$. We can therefore contract $\SO(d)$ indices and the dimensions in the integrand add up to $0$ (taking into account the measure $d^dx$). 

To pair continuous-spin structures in Lorentzian, we need to make use of the integral 
\be
	\int d^d x D^{d-2} z \cO(x,z)\cO^{\mathrm{S}\dagger}(x,z)
\ee
If $\cO$ transforms in $(\De,J,\l)$, then $\cO^{\mathrm{S}\dagger}$ transforms in $(d-\De,2-d-J,\l^*)$. The integrand has 0 homogeneity in $x$ and $z$, and $\l$-indices can be contracted.\footnote{Given that $\cO^{S}$ transforms in $(d-\De,2-d-J,\l)$, it is a bit non-trivial to understand why $\cO^{S\dagger}$ has $\l^*$. In odd dimensions $\l$ and $\l^*$ is the same irrep, so there is no question here. In even dimension $\dagger$ changes the sign of the last row of Young diagram of $(d-\De,2-d-J,\l)$ in the same way as it does for all $\mathfrak{so}(d)$-weights. In other words, it flips the sign if $d=4k$ and does nothing for $d=4k+2$. However, this last row is also the last row of $\l$ and $\l$ is an $\SO(d-2)$-irrep. It then turns out that from the $\SO(d-2)$ point of view, this action is equivalent to taking the dual. Another way to see this is that $\dagger$ is complex conjugation for $\SO(d-1,1)$, and thus for $\SO(d-2)$, which can be thought of as a subgroup of $\SO(d-1,1)$. But since $\SO(d-2)$ is compact, for it complex conjugation is the same as taking the dual.}

\subsection{Two-point functions}

Let us start with two-point functions. As discussed in section~\ref{sec:contspincorrelators}, two-point functions of continuous-spin operators do not make sense as Wightman functions, so in order to discuss them, we have to think about them simply as some conformal invariants defined at least for spacelike separated points.

That said, given a two-point structure for $\cO$ in representation $(\De,J,\l)$ and a two-point function for $\cO^\mathrm{S}$ in representation $\wS[(\De,J,\l)]=(d-\De,2-d-J,\l)$, we can define the two-point pairing by
\be\label{eq:2ptpairingL}
	\frac{(\<\cO\cO^\dagger\>,\<\cO^\mathrm{S}\cO^{\mathrm{S}\dagger}\>)_L}{\vol(\SO(1,1))^2}\equiv \int_{x_1\approx x_2} \frac{d^dx_1 d^dx_2 D^{d-2}z_1 D^{d-2} z_2}{\vol(\tl\SO(d,2))}\<\cO^a(x_1,z_1)\cO^{b\dagger}(x_2,z_2)\>\<\cO^\mathrm{S}_b(x_2,z_2)\cO^{\mathrm{S}\dagger}_a(x_1,z_1)\>,
\ee
where factor $\vol(\SO(1,1))^2$ is for future convenience\footnote{Similarly to the Euclidean case~\cite{ShadowFuture}, the right hand side can be alternatively computed in terms of Plancherel measure divided by $\vol(\SO(1,1))^2$. In Euclidean we get only one power of $\vol(\SO(1,1))$, which corresponds to the fact that there we have only one continuous parameter $\De$, while in Lorentzian we have both $\De$ and $J$.} and the subscript ``$L$'' stands for ``Lorentzian.'' On the right hand side, we divide by the volume of the conformal group since the integral is invariant under it. Formally, this means that we should compute the integral by gauge-fixing the action of conformal group and introducing an appropriate Faddeev-Popov determinant. To perform gauge-fixing, we can first put $x_1$ and $x_2$ into some standard configuration. A natural choice is to set $x_1=0$ and $x_2=\infty$ (spacelike infinity).\footnote{We define $\cO(\infty)=\lim_{L\to \infty}L^{2\De}\cO(Le)$, where $e$ is a conventional spacelike unit vector. We choose $e=(0,1,0,\ldots,0)$.} This configuration is still invariant under dilatation and Lorentz transformations. Thus we have
\be
	\frac{(\<\cO\cO^\dagger\>,\<\cO^\mathrm{S}\cO^{\mathrm{S}\dagger}\>)_L}{\vol(\SO(1,1))^2}=\int \frac{D^{d-2}z_1 D^{d-2} z_2}{2^d\vol(\SO(1,1)\times \SO(d-1,1))}\<\cO^a(0,z_1)\cO^{b\dagger}(\infty,z_2)\>\<\cO^\mathrm{S}_b(\infty,z_2)\cO^{\mathrm{S}\dagger}_a(0,z_1)\>,
\ee
where $2^d$ comes from the Faddeev-Popov determinant.\footnote{A fixed power of $2$ also goes into what we mean by $\vol(\SO(1,1))$.} If we define $z^R_2=(z_2^0,-z_2^1,z_2^2,\ldots,z_2^{d-1})$, so that Lorentz group transforms $z_1$ and $z_2^R$ in the same way, the integral
\be
	\int \frac{D^{d-2}z_1 D^{d-2} z^R_2}{\vol(\SO(d-1,1))}
\ee
essentially becomes the $(d-2)$-dimensional Euclidean conformal two-point integral.  It can also be computed by gauge-fixing, i.e.\ by setting $z^\mu_1=z^\mu_0\equiv(\half,\half,0,\ldots,0)$, which is the embedding-space representation of the origin of $\R^{d-2}$, $z^{R\,\mu}_2=z^\mu_\infty\equiv(\half,-\half,0,\ldots,0)$, which is the embedding-space representation of the infinity of $\R^{d-2}$. The stabilizer group of this configuration is $\SO(1,1)\times \SO(d-2)$, which consists of $(d-2)$-dimensional dilatations and rotations. We thus conclude
\be
\label{eq:resultforpairing}
	(\<\cO\cO^\dagger\>,\<\cO^\mathrm{S}\cO^{\mathrm{S}\dagger}\>)_L=\frac{1}{2^d 2^{d-2}\vol( \SO(d-2))}\<\cO^a(0,z_0)\cO^{b\dagger}(\infty,z^R_\infty)\>\<\cO^\mathrm{S}_b(\infty,z^R_\infty)\cO^{\mathrm{S}\dagger}_a(0,z_0)\>,
\ee
where we included another Faddeev-Popov determinant.  Note that the right hand side is proportional to $\dim\l$.

We can summarize this result as follows. Note that the product
\be
	\<\cO^a(x_1,z_1)\cO^{b\dagger}(x_2,z_2)\>\<\cO^\mathrm{S}_b(x_2,z_2)\cO^{\mathrm{S}\dagger}_a(x_1,z_1)\>
\ee
transforms in representation $(\De,J,\l)=(d,2-d,\bullet)$ at both $x_1$ and $x_2$. Thus we must have 
\be
	\<\cO^a(x_1,z_1)\cO^{b\dagger}(x_2,z_2)\>\<\cO^\mathrm{S}_b(x_2,z_2)\cO^{\mathrm{S}\dagger}_a(x_1,z_1)\>=A\frac{(-2 z_1 \. I(x_{12}) z_2)^{2-d}}{x_{12}^{2d}}.
\ee
For some constant $A$. Using (\ref{eq:resultforpairing}), we find
\be
	(\<\cO\cO^\dagger\>,\<\cO^\mathrm{S}\cO^{\mathrm{S}\dagger}\>)_L=\frac{A}{2^{2d-2}\vol(\SO(d-2))}.
\ee

\subsection{Three-point pairings}

We can analogously define a three-point pairing for continuous-spin structures,
\be
\label{eq:lorentzianpairing}
&\p{\<\cO_1\cO_2 \cO\>, \<\tl\cO_1^\dag \tl \cO_2^\dag \cO^{\mathrm{S}\dag}\>}_L \nn\\
&\equiv \int_{\substack{2<1 \\ x\approx 1,2}} \frac{d^dx_1 d^dx_2 d^d x D^{d-2} z}{\vol(\tl{\SO}(d,2))} \<\cO_1(x_1)\cO_2(x_2) \cO(x,z)\> \<\tl\cO_1^\dag(x_1) \tl \cO_2^\dag(x_2) \cO^{\mathrm{S}\dag}(x,z)\>.
\ee
Here, finite-dimensional Lorentz indices are implicitly contracted. Note that due to the fixed causal relationships between the points the continuous-spin structures are single-valued without $i\e$ prescriptions (see appendix~\ref{sec:contspincorrelators}). As in the Euclidean case, Lorentzian three-point pairings are simple to compute because they don't involve any actual integrals over positions. We can use the conformal group to fix all three points to a standard configuration consistent with the given causal relationships, for example
\be
x_2=0,\quad x_1=e^0,\quad x=\oo,
\ee
where $e^0$ is a unit vector in the $t$ direction. The Fadeev-Popov determinant associated with this choice is $2^{-d}$. All that remains is an integral over the polarization vector $z$,
\be
\label{eq:deflorentzianpairing}
&= \frac{1}{2^d\vol(\SO(d-1))} \int D^{d-2} z\, \<\cO_1(e^0)\cO_2(0) \cO(\oo,z)\> \<\tl\cO_1^\dag(e^0) \tl \cO_2^\dag(0) \cO^{\mathrm{S}\dag}(\oo,z)\>,
\ee
where $\vol(\SO(d-1))$ is the volume of the stabilizer group of the three points.\footnote{Note that the stabilizer group depends on the causal relationships of the points. For example, three spacelike points have stabilizer group $\SO(d-2,1)$.}
In practice, we can avoid doing the integral over $z$ as well. This is because the product in the integrand must be proportional to a three-point function of two scalars with dimension $d$ and a spinning operator with dimension $d$ and spin $2-d$. The integral of the $z$-dependent part of this product is always
\be
\frac{1}{2^d\vol(\SO(d-1))}\int D^{d-2} (-2z\.e^0)^{2-d} &= \frac{1}{2^{2d-2}\vol(\SO(d-2))}.
\ee
Thus, we can write
\be
\p{\<\cO_1\cO_2 \cO\>, \<\tl\cO_1^\dag \tl \cO_2^\dag \cO^{\mathrm{S}\dag}\>}_L &= \frac{1}{2^{2d-2}\vol(\SO(d-2))} \frac{\<\cO_1(e^0)\cO_2(0) \cO(\oo,z)\> \<\tl\cO_1^\dag(e^0) \tl \cO_2^\dag(0) \cO^{\mathrm{S}\dag}(\oo,z)\>}{(-2z\.e^0)^{2-d}}.
\ee

\section{Integral transforms, weight-shifting operators and integration by parts}
\label{app:operations}

In this appendix we elaborate on the interplay between integral transforms, weight-shifting operators, and conformally-invariant pairings, following~\cite{ShadowFuture} and generalizing the discussion to Lorentzian signature. For simplicity of discussion, we ignore possible signs coming from odd permutations of fermions.

\subsection{Euclidean signature}

In Euclidean signature we have one integral transform, $\wSE$, and a conformally-invariant pairing 
\be
(\cO,\tl\cO^\dagger) \equiv \int d^dx \cO(x)\tl\cO^\dagger(x),
\ee
where the spin indices are implicitly contracted. With respect to this paring we can define a conjugation on weight-shifting operators and on the integral transform,
\be
	(\cD\cO,\tl\cO^\dagger)&=(\cO,\cD^*\tl\cO^\dagger),\nn\\
	(\wSE\cO,\tl\cO^\dagger)&=(\cO,\wSE^*\tl\cO^\dagger).
\ee
We have $*^2=1$ and $\wSE^*=\wSE$. 

Furthermore, we can define Weyl reflection on weight-shifting operators according to
\be\label{eq:wsshadow}
	\wSE\cD=(\wSE[\cD])\wSE.
\ee
We then have
\be
	\wSE^2\cD=\wSE(\wSE[\cD])\wSE=(\wSE^2[\cD])\wSE^2,
\ee
and since $\wSE^2=\cN(\De,\rho)$, we have when acting on operators transforming in $(\De,\rho)$
\be\label{eq:wsdsquareEuclid}
	\wSE^2[\cD]=\frac{\cN(\De+\delta_\De,\rho+\delta_\rho)}{\cN(\De,\rho)}\cD,
\ee
where $(\delta_\De,\delta_\rho)$ is the weight by which $\cD$ shifts. Conjugating~\eqref{eq:wsshadow} we find
\be
	\wSE(\wSE[\cD])^*=\cD^*\wSE,
\ee
and thus
\be
	\wSE[\cD]^*=\wSE^{-1}[\cD^*].
\ee

We also note that the crossing equation for weight-shifting operators acting on a two-point function~\cite{Karateev:2017jgd} can be written in terms of shadow transform and conjugation. Namely, we can interpret $\wSE\cD^*$ as convolution with the kernel
\be
	\<\tl\cO(\cD\tl\cO^\dagger)\>,
\ee
while, on the other hand, it is equal to $\wSE[\cD^*]\wS$ which is convolution with (assume that $\cD\tl\cO^\dagger$ transforms as $\tl\cO'^\dagger$)
\be
	\<(\wSE[\cD^*]\tl\cO')\tl\cO'^\dagger\>.
\ee
We thus find the crossing equation
\be
	\<\tl\cO(\cD\tl\cO^\dagger)\>=\<(\wSE[\cD^*]\tl\cO')\tl\cO'^\dagger\>.
\ee

\subsection{Lorentzian signature}

The above discussion has an analogue in Lorentzian signature. Now we have more integral transforms, so let us denote a generic one by $\bW$. We also have a new pairing, given by
\be
	(\cO,\cO^{\mathrm{S}\dagger})_L=\int d^dx D^{d-2}z \cO(x,z)\cO^{\mathrm{S}\dagger}(x,z),
\ee
where the $\SO(d-2)$ indices are implicit and contracted. This pairing leads to a new conjugation operation on weight-shifting operators and on integral transforms,
\be
	(\cD\cO,\tl\cO^\dagger)_L&=(\cO,\overline\cD\tl\cO^\dagger)_L,\nn\\
	(\bW\cO,\tl\cO^\dagger)_L&=(\cO,\overline\bW\tl\cO^\dagger)_L.
\ee
Note that in general the Lorentzian and Euclidean conjugations do not commute (see below). Analogously to the Euclidean case, we find
\be\label{eq:lorentzianSoverlinecommutator}
\overline{\bW[\cD]}=\bW^{-1}[\overline\cD].
\ee

As in Euclidean signature, we can define the action of integral transforms on weight-shifting operators by
\be
	\bW\cD=(\bW[\cD])\bW.
\ee
In principle $\bW[\cD]$ can be a differential operator with coefficients which depend on $\tsym$. However, when acting on a function, the left hand side of this expression depends only on the values of this function in a set which fits in one Poincare patch. If $\bW[\cD]$ had non-trivial $t$ dependence, the same would not hold for the right hand side. Therefore $\bW[\cD]$ has to be a local weight-shifting differential operator.

It is easy to check that if two integral transforms commute, then their actions on weight-shifting operators also commute. Similarly to Euclidean case, relations such as $\wL^2=f_L(\De,J,\tsym)$ generalize to action on weight-shifting operators. Let us write down the square of an order two transform (any transform except $\wR$ and $\overline\wR$)
\be\label{eq:wsweylsquare}
	\bW^2[\cD]=f_W(\De,\rho,\tsym)\cD f_W^{-1}(\De,\rho,\tsym),
\ee
where $\De$ and $\rho$ are understood as operators which read off the scaling dimension and representation of whatever they act on. Let us comment on this formula in the case of $\wSD$. Modulo Wick rotation, we have the relation $\wSE=(-2)^J\wSD$ for traceless-symmetric operators. It follows that~\eqref{eq:wsdsquareEuclid} and~\eqref{eq:wsweylsquare} should be compatible. That is, we should have
\be\label{eq:lorentzeuclidconsistency}
	\frac{\cN(\De+\delta_\De,J+\delta_J)}{\cN(\De,J)}=\frac{4^{J+\delta_J}f_{\De}(\De+\delta_\De,J+\delta_J,c\tsym)}{4^Jf_{\De}(\De,J,\tsym)},
\ee
where $\delta_\De$, $\delta_J$ are the weights by which $\cD$ shifts, and $c$ is defined by
\be\label{eq:cdefn}
\tsym \cD \tsym^{-1} = c \cD.
\ee
I.e.\ $c$ is the eigenvalue of $\tsym$ in the finite-dimensional irrep of conformal group to which $\cD$ is associated. For example, for vector representation $c=-1$.
To check this relation, we can use the results of section~\ref{sec:algebraoftransforms} and in particular the relation~\eqref{eq:wsdwlwsjrelation} which implies (we consider traceless-symmetric case for simplicity)
\be
	f_{\De}(\De,J,\tsym)=-\tsym^{-2}f_L(\De,\rho,\tsym)f_{J}(1-\De)f_L(1-J,1-d+\De,\tsym).
\ee
It is then an easy exercise to verify that~\eqref{eq:lorentzeuclidconsistency} holds for vector weight-shifting operators~\cite{Karateev:2017jgd}.

Another useful result is obtained by substituting $\cD\to \bW^{-1}[\cD]$ into~\eqref{eq:wsweylsquare} to find
\be\label{eq:integraltransforminverse}
	\bW^{-1}[\cD]=f_W^{-1}(\De,\rho,\tsym)\bW[\cD] f_W(\De,\rho,\tsym).
\ee
For example,
\be
	\wL^{-1}[\cD]=\wL[\cD]\frac{f_L(\De,\rho,\tsym)}{f_L(\De+\wL[\delta_\De],\rho+\wL[\delta_\rho],c \tsym)},
\ee
where we kept explicit dependence of $f_L$ on $\tsym$, $(\wL[\delta_\De],\wL[\delta_\rho])$ is the weight by which $\wL[\cD]$ shifts. It is easy to check that $\tsym$-dependence indeed cancels out for $\cD$ in vector representation.

We can derive two-point crossing in terms of Lorentzian conjugation and $\wS$ transform,
\be\label{eq:twoptcrossing}
	\<\cO^\mathrm{S}(\cD\cO^{\mathrm{S}\dagger})\>=\<(\wS[\overline\cD]\cO'^{\mathrm{S}})\cO'^{\mathrm{S}\dagger}\>.
\ee
Comparing to the Euclidean form of two-point crossing leads to a useful relation
\be\label{eq:SDSrelation}
	\wSE[\cD^*]=\wS[\overline\cD].
\ee
We will need a version of this relation with order of integral transforms and conjugations interchanged. First,~\eqref{eq:SDSrelation} implies
\be
	(\wSE^{-1}[\cD])^*=\overline{\wS^{-1}[\cD]}.
\ee
Then we use that $\wSE$ and $\wS$ are proportional to their inverses. In particular, we find from~\eqref{eq:integraltransforminverse}
\be
	(f^{-1}_E(\De,\rho,\tsym)\wSE[\cD]f_E(\De,\rho,\tsym))^*&=\overline{(f^{-1}_{S}(\De,\rho,\tsym)\wS[\cD]f_{S}(\De,\rho,\tsym))},\nn\\
	f_E(\De,\rho,\tsym)(\wSE[\cD])^*f^{-1}_E(\De,\rho,\tsym)&=f_{S}(\De,\rho,\tsym)\overline{\wS[\cD]}f^{-1}_{S}(\De,\rho,\tsym),
\ee
where we temporarily interpret $\wSE$ as a Lorentzian transform defined by $(-2)^J\wSD$. We can now use
\be
	f_S(\De,\rho,\tsym)=\wS^2=\wSD^2 \wSJ^2=4^{-J}\wSE^2 \wSJ^2=4^{-J}f_E(\De,\rho,\tsym)f_J(\rho)
\ee
to conclude
\be
	\overline{\wS[\cD]}=4^Jf^{-1}_J(\rho)(\wSD[\cD])^*4^{-J}f_J(\rho).
\ee

\section{Proof of~\eqref{eq:claimforhgeneralized} for seed blocks}
\label{app:proof}

In this appendix we prove~\eqref{eq:claimforhgeneralized} for seed blocks by starting from the scalar case. For simplicity we consider only bosonic representations. We assume that $\cO_i$ are in $\SO(d)$ representations appropriate for the seed block for intermediate~$\rho$ which we are interested in. As discussed in section 4.4 of~\cite{Karateev:2017jgd}, we can assume that $\cO_2$ and $\cO_4$ are scalars in all seed blocks, so we don't have to change their representations. We start with the identity
\be\label{eq:seedwsdefinition}
{(\<\cO^\dagger\cO\>,\<\tl\cO^\dagger\tl\cO\>)_E}(\<\cO_1\cO_2\wSE[\cO^\dagger]\>)^{-1}_E={(\<\cO'^\dagger\cO'\>,\<\tl\cO'^\dagger\tl\cO'\>)_E}\cD_{1,A}\tl\cD^A(\<\cO_1\cO'_2\wSE[\cO^{\prime\dagger}]\>)^{-1}_E,
\ee
where $\cD$ and $\tl\cD$ are some weight-shifting operators,\footnote{Here tilde isn't related to shadow transform and $\tl\cD$ acts on the third position. The representation of index $A$ can be assumed to be vector.} while $\cO'_1$ and $\cO'$ come from a seed block for which we already know that~\eqref{eq:claimforhgeneralized} holds. A possible proportionality coefficient can be absorbed into the definition of either the weight-shifting operators or the tensor structures. Consider pairing both sides with $\<\cO_1\cO_2\wSE[\cO^\dagger]\>$ to obtain
\be
\frac{(\<\cO^\dagger\cO\>,\<\tl\cO^\dagger\tl\cO\>)_E}{(\<\cO'^\dagger\cO'\>,\<\tl\cO'^\dagger\tl\cO'\>)_E}=(\<\cO_1\cO_2\wSE[\cO^\dagger]\>,\cD_{1,A}\tl\cD^A(\<\cO_1\cO'_2\wSE[\cO^{\prime\dagger}]\>)^{-1}_E)_E.
\ee
Integrating by parts and using definitions of appendix~\ref{app:operations} we find
\be
\frac{(\<\cO^\dagger\cO\>,\<\tl\cO^\dagger\tl\cO\>)_E}{(\<\cO'^\dagger\cO'\>,\<\tl\cO'^\dagger\tl\cO'\>)_E}=(\<\cD_{1,A}^*\cO_1\cO_2\wSE[\wSE^{-1}[\tl\cD^*]^A\cO^\dagger]\>,(\<\cO_1\cO'_2\wSE[\cO^{\prime\dagger}]\>)^{-1}_E)_E,
\ee
which allows us to conclude
\be
\<\cD^*_{1,A}\cO_1\cO_2\wSE[\wSE^{-1}[\tl\cD^*]^A\cO^\dagger]\>=\frac{(\<\cO^\dagger\cO\>,\<\tl\cO^\dagger\tl\cO\>)_E}{(\<\cO'^\dagger\cO'\>,\<\tl\cO'^\dagger\tl\cO'\>)_E}\<\cO_1\cO'_2\wSE[\cO^{\prime\dagger}]\>,
\ee
or, canceling $\wSE$ on both sides,
\be\label{eq:seedwswithoutinverses}
\<\cD^*_{1,A}\cO_1\cO_2(\wSE^{-1}[\tl\cD^*]^A\cO^\dagger)\>=\frac{(\<\cO^\dagger\cO\>,\<\tl\cO^\dagger\tl\cO\>)_E}{(\<\cO'^\dagger\cO'\>,\<\tl\cO'^\dagger\tl\cO'\>)_E}\<\cO_1\cO'_2\cO^{\prime\dagger}\>.
\ee
We will use this characterization of $\cD$ and $\tl\cD$ later in the proof.

For now, let us apply~\eqref{eq:seedwsdefinition} to~\eqref{eq:integralforhgeneralized} and find that $H$ is given by
\be	
H_{\De,\rho}(x_i)
&=-\mu(\De,\rho'^\dagger)(\cO_1 \cO'_2  \wSE[\tl \cO'^\dagger])(\<\cO_1 \cO'_2  \tl \cO'^\dagger\>,\<\tl \cO_1^\dagger \tl \cO_{2}'^\dagger\cO'\>)^{-1}_E\times\nn\\
&\quad\times\int_{2<x<1} d^d x D^{d-2} z \<0|\cD_{1,A}\tl \cO_1^\dagger \wL[\tl\cD^A\cO](x,z)  \tl\cO_{2^+}'^\dagger|0\>(\<0|\cO_{4^+} \wL[\cO](x,z) \cO_3|0\>)^{-1}_L.
\ee
We now use
\be
\wL[\tl\cD^A\cO]=\wL[\tl\cD]^A\wL[\cO],
\ee
and integrate $\wL[\tl\cD]$ by parts. This gives
\be	
H_{\De,\rho}(x_i)
&=-\mu(\De,\rho'^\dagger)(\cO_1 \cO'_2  \wSD[\tl \cO'^\dagger])(\<\cO_1 \cO'_2  \tl \cO'^\dagger\>,\<\tl \cO_1^\dagger \tl \cO_{2}'^\dagger\cO'\>)^{-1}_E\times\nn\\
&\quad\times\int_{2<x<1} d^d x D^{d-2} z \<0|\cD_{1,A}\tl \cO_1^\dagger \wL[\cO](x,z)  \tl\cO_{2^+}'^\dagger|0\>\overline{\wL[\tl\cD]}^A(\<0|\cO_{4^+} \wL[\cO](x,z) \cO_3|0\>)^{-1}_L,
\ee
where $\overline{\wL[\tl\cD]}$ acts on the middle position in the right three-point structure. We can further apply a crossing transformation on the right three-point structure as in~\cite{Karateev:2017jgd} to make all differential operators act on the external operators only. We will not do this in detail, because we will anyway reverse this step in a moment. Let us denote the resulting differential operator acting on external operators by $\mathfrak{D}$. 

The conclusion of the above calculation is schematically that 
\be
H_{\rho}=\mathfrak{D}H_{\rho'},
\ee
where $H_{\rho'}$ is some conformal for which we know~\eqref{eq:claimforhgeneralized} to hold. We can thus apply $\mathfrak{D}$ to~\eqref{eq:claimforhgeneralized} written for $H_{\rho'}$. Since the right three-point structure in~\eqref{eq:claimforhgeneralized} and~\eqref{eq:integralforhgeneralized} is the same, we can unwind the steps in the derivation of~$\mathfrak{D}$ which were performed solely on the right three-point structure to conclude
\be
H_{\De,\rho}(x_i)=-\frac{1}{2\pi i}  \frac{\cD_{1,A}\p{\tsym_2\<\cO_1 \cO'_2 \wL[\cO'^\dagger]\>}^{-1}_L \overline{\wL[\tl\cD]}^A\p{\tsym_4\<\cO_4 \cO_3 \wL[\cO]\>}^{-1}_L}{(\<\wL[\cO']\wL[\cO']\>)^{-1}_L}.
\ee
We can use~\eqref{eq:blockwsrule} to write this as
\be\label{eq:intermediateH}
H_{\De,\rho}(x_i)=-\frac{1}{2\pi i}  \frac{{(\<\wL[\cO]\wL[\cO]\>)^{-1}_L}}{{(\<\wL[\cO']\wL[\cO']\>)^{-1}_L}}\frac{\wS[\wL[\tl\cD]]^A\cD_{1,A}\p{\tsym_2\<\cO_1 \cO'_2 \wL[\cO'^\dagger]\>}^{-1}_L \p{\tsym_4\<\cO_4 \cO_3 \wL[\cO]\>}^{-1}_L}{(\<\wL[\cO]\wL[\cO]\>)^{-1}_L}.
\ee

We now want to express
\be
\wS[\wL[\tl\cD]]^A\cD_{1,A}(\<\cO_1 \cO'_2 \wL[\cO'^\dagger]\>)^{-1}_L
\ee
in terms of 
\be
(\<\cO_1 \cO_2 \wL[\cO^\dagger]\>)^{-1}_L.
\ee
To do this, let us consider the Lorentzian pairing
\be\label{eq:derivationpairing}
&\p{\<\cO_1 \cO_2 \wL[\cO^\dagger]\>,\wS[\wL[\tl\cD]]^A\cD_{1,A}(\<\cO_1\cO'_2 \wL[\cO'^\dagger]\>)^{-1}_L}_L\nn\\
&=\p{\overline{\wS[\wL[\tl\cD]]}^A\cD^*_{1,A}\<\cO_1 \cO_2 \wL[\cO^\dagger]\>,(\<\cO_1 \cO'_2 \wL[\cO'^\dagger]\>)^{-1}_L}_L.
\ee
We can use the results of appendix~\ref{app:operations} and~\ref{sec:algebraoftransforms} to write
\be
\overline{\wS[\wL[\tl\cD]]}=\overline{\wL[\wS[\tl\cD]]}=\wL^{-1}[\overline{\wS[\tl\cD]}]=\frac{f_L(\wL[\De],\wL[\rho^\dagger],\tsym)}{f_L(\wL[\De]+\wL[\delta_\De],\wL[\rho^\dagger]+\wL[\delta_\rho],c\tsym)}\wL[\overline{\wS[\tl\cD]}],
\ee
where $(\delta_\De,\delta_\rho)$ is the weight by which $\overline{\wS[\tl\cD]}$ shifts and $c$ is defined by~\eqref{eq:cdefn} for $\tl\cD$. Since we consider only bosonic representations, $c=\pm 1$ ($c=-1$ for vector weight-shifting operators). We have $(\De+\delta_\De,\rho^\dagger+\delta_\rho)=(\De',\rho'^\dagger)$.
We furthermore have
\be
\wL[\overline{\wS[\tl\cD]}]\wL[\cO^\dagger]=\wL[\overline{\wS[\tl\cD]}\cO^\dagger]=\frac{4^{-J}f_J(\rho^\dagger)}{4^{-J'}f_J(\rho'^\dagger)}\wL[(\wSD[\tl\cD])^*\cO^\dagger]
\ee
and thus 
\be
\overline{\wS[\wL[\tl\cD]]}^A\cD^*_{1,A}\<\cO_1 \cO_2 \wL[\cO^\dagger]\>=\frac{4^{-J}f_J(\rho^\dagger)}{4^{-J'}f_J(\rho'^\dagger)}\frac{f_L(\wL[\De],\wL[\rho^\dagger],\tsym)}{f_L(\wL[\De'],\wL[\rho'^\dagger],c\tsym)}
\<\cO_1 \cD^*_{1,A}\cO_2 \wL[(\wSE[\tl\cD])^*\cO^\dagger]\>.
\ee
Now use $(\wSE[\tl\cD])^*=\wSE^{-1}[\tl\cD^*]$, apply $\wL$ to both sides of~\eqref{eq:seedwswithoutinverses} and conclude
\be
\overline{\wS[\wL[\tl\cD]]}^A\cD^*_{1,A}\<\cO_1 \cO_2 \wL[\cO^\dagger]\>=\frac{4^{-J}f_J(\rho^\dagger)}{4^{-J'}f_J(\rho'^\dagger)}\frac{f_L(\wL[\De],\wL[\rho^\dagger],\tsym)}{f_L(\wL[\De'],\wL[\rho'^\dagger],c\tsym)}
\frac{(\<\cO^\dagger\cO\>,\<\tl\cO^\dagger\tl\cO\>)_E}{(\<\cO'^\dagger\cO'\>,\<\tl\cO'^\dagger\tl\cO'\>)_E}\<\cO_1\cO'_2\wL[\cO^{\prime\dagger}]\>.
\ee
This implies that the pairing~\eqref{eq:derivationpairing} is equal to
\be
\frac{4^{-J}f_J(\rho^\dagger)}{4^{-J'}f_J(\rho'^\dagger)}\frac{f_L(\wL[\De],\wL[\rho^\dagger],\tsym)}{f_L(\wL[\De'],\wL[\rho'^\dagger],c\tsym)}
\frac{(\<\cO^\dagger\cO\>,\<\tl\cO^\dagger\tl\cO\>)_E}{(\<\cO'^\dagger\cO'\>,\<\tl\cO'^\dagger\tl\cO'\>)_E}
\ee
and thus
\be
&\wS[\wL[\tl\cD]]^A\cD_{1,A}\<\cO_1 \cO'_2 \wL[\cO'^\dagger]\>^{-1}\nn\\
&=	\frac{4^{-J}f_J(\rho^\dagger)}{4^{-J'}f_J(\rho'^\dagger)}\frac{f_L(\wL[\De],\wL[\rho^\dagger],\tsym)}{f_L(\wL[\De'],\wL[\rho'^\dagger],c\tsym)}
\frac{(\<\cO^\dagger\cO\>,\<\tl\cO^\dagger\tl\cO\>)_E}{(\<\cO'^\dagger\cO'\>,\<\tl\cO'^\dagger\tl\cO'\>)_E}(\<\cO_1 \cO_2 \wL[\cO^\dagger]\>)^{-1}_L.
\ee
Collecting all the pieces, we find that~\eqref{eq:intermediateH} implies~\eqref{eq:claimforhgeneralized} for the seed $H$ if
\be
C=\frac{{(\<\wL[\cO]\wL[\cO]\>)^{-1}_L}}{{(\<\wL[\cO']\wL[\cO']\>)^{-1}_L}}\frac{4^{-J}f_J(\rho^\dagger)}{4^{-J'}f_J(\rho'^\dagger)}\frac{f_L(\wL[\De],\wL[\rho^\dagger],\tsym)}{f_L(\wL[\De'],\wL[\rho'^\dagger],c\tsym)}
\frac{(\<\cO^\dagger\cO\>,\<\tl\cO^\dagger\tl\cO\>)_E}{(\<\cO'^\dagger\cO'\>,\<\tl\cO'^\dagger\tl\cO'\>)_E}=1.
\ee

\paragraph{Proof that $C=1$}

First, we note that
\be
\frac{(\<\cO^\dagger\cO\>,\<\tl\cO^\dagger\tl\cO\>)_E}{(\<\cO'^\dagger\cO'\>,\<\tl\cO'^\dagger\tl\cO'\>)_E}=\frac{4^J\dim\rho^\dagger}{4^{J'}\dim \rho'^\dagger}.
\ee
Furthermore, $f_J$ is square of shadow transform in $d-2$ dimensions. Thus if we write $\rho^\dagger=(J,\lambda)$ then (similarly to appendix~\ref{sec:euclideanharmonicanalysis})
\be
f_J(\rho^\dagger)\propto \frac{\dim\l}{\mu(\rho^\dagger)},
\ee
where $\mu$ is the Plancherel measure for $\SO(d-1,1)$. Furthermore, the ratio
\be
\frac{\mu(\rho^\dagger)}{\dim\rho^\dagger}
\ee
is independent of $\rho$~\cite{Dobrev:1977qv,ShadowFuture}. This implies that
\be\label{eq:proofEpairingsgone}
\frac{4^{-J}f_J(\rho^\dagger)}{4^{-J'}f_J(\rho'^\dagger)}
\frac{(\<\cO^\dagger\cO\>,\<\tl\cO^\dagger\tl\cO\>)_E}{(\<\cO'^\dagger\cO'\>,\<\tl\cO'^\dagger\tl\cO'\>)_E}=\frac{\dim\lambda}{\dim\lambda'}.
\ee
Furthermore, we can write
\be\label{eq:proofdimlgone}
\frac{\dim\lambda}{\dim\lambda'}=\frac{(\<\cO'\cO'^\dagger\>)^{-1}_L}{(\<\cO\cO^\dagger\>)^{-1}_L},
\ee
which is due to
\be
	(\<\cO\cO^\dagger\>,\<\cO^{S}\cO^{S\dagger}\>)_L\propto \dim\l,
\ee
and similarly for primed quantities (see appendix~\ref{app:contspinpairings}).

Now we need to recall the calculation of $\<\wL[\cO]\wL[\cO^\dagger]\>$. We have for the kernel which is represented by the time-ordered two-point function $\<\cO\cO^\dagger\>$,
\be
\<\cO\cO^\dagger\>=\wS (1+\sum_{n=1}^\infty \g^{-n}(\tsym^n+\tsym^{-n}) ),
\ee
where $\g$ is the eigenvalue of $\tsym$ corresponding to $\cO$, see~\eqref{eq:tsymeigenvalue}. The calculation in section~\ref{sec:naturalformula} then yields, in the same sense as above,
\be
\<\wL[\cO]\wL[\cO^\dagger]\>=\wS (1+\sum_{n=1}^\infty \g^{-n}(\tsym^n+\tsym^{-n}) ) \tsym^{-1} f_L(\wF[\De],\wF[\rho],\tsym).
\ee
Since $\wL$ commutes with $\wS$, we find that we can replace $f_L(\wF[\De],\wF[\rho],\tsym)$ by $f_L(\wL[\De],\wL[\rho],\tsym)$. This implies
\be
\frac{(\<\wL[\cO]\wL[\cO]\>)^{-1}_L}{(\<\wL[\cO']\wL[\cO']\>)^{-1}_L}=\frac{(1+\sum_{n=1}^\infty \g'^{-n}(\tsym^n+\tsym^{-n}) ) f_L(\wL[\De'],\wL[\rho'],\tsym)}{(1+\sum_{n=1}^\infty \g^{-n}(\tsym^n+\tsym^{-n}) ) f_L(\wL[\De],\wL[\rho],\tsym)}\frac{(\<\cO\cO^\dagger\>)^{-1}_L}{(\<\cO'\cO'^\dagger\>)^{-1}_L}.
\ee
Recall that $\overline{\wS[\tl\cD]}$ takes $\cO$ to $\cO'$ and $c\tl\cD=\tsym\tl\cD\tsym^{-1}$, which implies $\g'=c\g=\pm g$. (Recall we consider only bosonic representations.) Thus we have
\be
	&\frac{(1+\sum_{n=1}^\infty \g'^{-n}(\tsym^n+\tsym^{-n}) ) }{(1+\sum_{n=1}^\infty \g^{-n}(\tsym^n+\tsym^{-n}) )}f_L(\wL[\De'],\wL[\rho'],\tsym)\nn\\
	&=\frac{(1+\sum_{n=1}^\infty \g'^{-n}(\tsym^n+\tsym^{-n}) ) }{(1+\sum_{n=1}^\infty \g'^{-n}((c\tsym)^n+(c\tsym)^{-n}) )}f_L(\wL[\De'],\wL[\rho'],\tsym)\nn\\
	&=\frac{(c\tsym-\gamma)(c\tsym-\gamma^{-1})}{(\tsym-\gamma)(\tsym-\gamma^{-1})}f_L(\wL[\De'],\wL[\rho'],\tsym)\nn\\
	&=f_L(\wL[\De'],\wL[\rho'],c\tsym),
\ee
where we used the fact that~\eqref{eq:tsymdep} is $\tsym$-independent. We thus conclude that
\be
\frac{(\<\wL[\cO]\wL[\cO^\dagger]\>)^{-1}_L}{(\<\wL[\cO']\wL[\cO'^\dagger]\>)^{-1}_L}=\frac{f_L(\wL[\De'],\wL[\rho'],c\tsym)}{f_L(\wL[\De],\wL[\rho],\tsym)}\frac{(\<\cO\cO^\dagger\>)^{-1}_L}{(\<\cO'\cO'^\dagger\>)^{-1}_L}.
\ee
By combining this equation with~\eqref{eq:proofEpairingsgone} and~\eqref{eq:proofdimlgone} we see that indeed\footnote{Since we for simplicity restricted to bosonic representations, we haven't been very careful with distinguishing $\rho$ and $\rho^\dagger$. (There is no difference except possibly for self-dual tensors.) It would be interesting to repeat our argument in a more careful manner, accounting for fermionic representations as well.}
\be
	C=1.
\ee

\section{Conformal blocks with continuous spin}
\label{sec:conformalblockscontspin}

\subsection{Gluing three-point structures}
\label{app:3ptgluing}

Consider two three-point structures $\<\cO_1\cO_2\cO\>$ and $\<\cO \cO_3 \cO_4\>$. We can glue them into a conformal block as follows. We find a linear operator $B_{12\cO}(x_{12})$ such that in the OPE limit $1\to 2$, the first three-point structure becomes
\be
\label{eq:opelimit}
\<\cO_1 \cO_2 \cO^\dagger(x)\> &\sim B_{12\cO}(x_{12}) \<\cO(x_2)\cO^\dagger(x)\>,\qquad (|x_{12}|\ll|x_1-x|,|x_2-x|).
\ee
For example, when $\cO_1,\cO_2,\cO$ are all scalars, we have
\be
B_{12\cO}(x_{12}) &= x_{12}^{\De_\cO-\De_1-\De_2}.
\ee
($B_{12\cO}$ can be extended to a differential operator such that (\ref{eq:opelimit}) becomes an equality away from the  $1\to 2$ limit, but this is not necessary for the current discussion.) Note that to define $B_{12\cO}$ we must choose a normalization of the two-point structure $\<\cO\cO\>$.

We define a conformal block $G^{\cO_i}_\cO(x_i)$ as the conformally-invariant solution to the conformal Casimir equation~\cite{DO2} whose OPE limit is
\be\label{eq:Gopelimit}
G^{\cO_i}_\cO(x_i) &\sim B_{12\cO}(x_{12})\<\cO(x_2)\cO_3\cO_4\>,\qquad (|x_{12}|\ll|x_{ij}|).
\ee
It is very useful to introduce the following notation for a conformal block, which makes manifest the choices of two- and three-point structures needed to define it
\be
\label{eq:nicenotationforblock}
G_\cO^{\cO_i}(x_i) &= \frac{\<\cO_1 \cO_2 \cO^\dagger\>\<\cO \cO_3 \cO_4\>}{\<\cO\cO^\dagger\>}.
\ee
In our convention $\cO$ appears in the OPE $\cO_1\times\cO_2$ and $\cO^\dagger$ in the OPE $\cO_3\times\cO_4$.

\subsubsection{Example: integer spin in Euclidean signature}

As an example, let us review the case of external scalars $\f_1,\dots,\f_4$ and an exchanged operator $\cO$ with integer spin $J$,
\be
G_{\De,J}^{\De_i}(x_i) &= \frac{\<\f_1\f_2\cO\>\<\f_3\f_4\cO\>}{\<\cO\cO\>},
\ee
where $\<\f_1\f_2\cO\>$ and $\<\f_3\f_4\cO\>$ are the standard three-point structures (\ref{eq:standardthreeptconvention}) and $\<\cO\cO\>$ is the standard two-point structure (\ref{eq:standardtwoptconvention}). We will assume that all points are in Euclidean signature.

In the OPE limit $1\to 2$, we have
\be
\label{eq:onetwolimit}
\<\f_1\f_2\cO(x_0,z)\> &\sim \frac{1}{x_{12}^{\De_1+\De_2-\De+J}} \frac{(-2z\.I(x_{20})\.x_{12})^J}{x_{20}^{2\De}} \nn\\
&= \frac{1}{x_{12}^{\De_1+\De_2-\De+J}} x_{12}^{\mu_1}\cdots x_{12}^{\mu_J} \<\cO_{\mu_1\cdots\mu_J}(x_2)\cO(x_0,z)\>. 
\ee
To compute the leading behavior of the block, it suffices to take the limit $3\to 4$ in $\<\f_3\f_4\cO\>$,
\be
\<\f_3\f_4\cO_{\mu_1\cdots\mu_J}(x_2)\> &= \frac{1}{x_{34}^{\De_3+\De_4-\De+J}} \frac{(-2I(x_{42})\.x_{34})_{\mu_1}\cdots (-2I(x_{42})\.x_{34})_{\mu_J}-\textrm{traces}}{x_{42}^{2\De}}.
\ee
(This limit is identical to the first line of (\ref{eq:onetwolimit}) after replacing $1,2,0\to3,4,2$ and stripping off the polarization vector $z$.) Thus the OPE limit of the resulting block is
\be
G_{\De,J}^{\De_i}(x_i) &\sim \frac{x_{12}^{\mu_1}\cdots x_{12}^{\mu_J}}{x_{12}^{\De_1+\De_2-\De+J}x_{34}^{\De_3+\De_4-\De+J}}  \frac{(-2I(x_{42})\.x_{34})_{\mu_1}\cdots (-2I(x_{42})\.x_{34})_{\mu_J}-\textrm{traces}}{x_{42}^{2\De}}
\nn\\
&=
 \frac{1}{x_{12}^{\De_1+\De_2}x_{34}^{\De_3+\De_4}}\p{\frac{x_{12}^2x_{34}^2}{x_{42}^{4}}}^{\De/2} 2^J\hat C_J\p{\frac{-x_{12}\.I(x_{42})\.x_{34}}{|x_{12}||x_{34}|}}.
\ee
Here, we've used the identity
\be
(m^{\mu_1}\cdots m^{\mu_J})(n_{\mu_1}\cdots n_{\mu_J} - \textrm{traces}) &= |m|^J |n|^J \hat C_J\p{\frac{m\.n}{|m||n|}},
\ee
where
\be
\label{eq:hatC}
\hat C_J(\eta) &= \frac{\G(\frac{d-2}{2})\G(J+d-2)}{2^J \G(J+\frac{d-2}{2})\G(d-2)}  {}_2F_1\p{-J,J+d-2,\frac{d-1}{2},\frac{1-\eta}{2}}
\ee
is proportional to a Gegenbauer polynomial (note in particular that for $\eta=1$ the hypergeometric function reduces to 1). Factoring out some standard kinematical factors, we find
\be
\label{eq:blockwithstandardfactors}
G_{\De,J}^{\De_i}(x_i) &= \frac{1}{(x_{12}^2)^{\frac{\De_1+\De_2}{2}} (x_{34}^2)^{\frac{\De_3+\De_4}{2}}} \p{\frac{x_{14}^2}{x_{24}^2}}^{\frac{\De_2-\De_1}{2}} \p{\frac{x_{14}^2}{x_{13}^2}}^{\frac{\De_3-\De_4}{2}} G_{\De,J}^{\De_i}(\chi,\bar\chi),
\ee
where  $G_{\De,J}^{\De_i}(\chi,\bar\chi)$ is a solution to the conformal Casimir equations normalized so that
\be
G_{\De,J}^{\De_i}(\chi,\bar\chi) &\sim (\chi \bar\chi)^{\De/2} \p{\frac{\chi}{\bar\chi}}^{-J/2},\qquad (\chi\ll\bar\chi\ll1).
\ee
Here, $\chi,\bar\chi$ are conformal cross-ratios defined by $u=\chi\bar\chi$, $v=(1-\chi)(1-\bar\chi)$. This is the standard conformal block in the normalization convention of \cite{Caron-Huot:2017vep,Simmons-Duffin:2017nub}.

\subsubsection{Example: continuous spin in Lorentzian signature}

Our definition of a conformal block also works when $\cO$ has continuous spin. However, now we must allow $B_{12\cO}$ to be an integral operator in the polarization vector of $\cO$. Let us again consider external scalars $\f_1,\dots,\f_4$. For later applications, we work in a Lorentzian configuration where all four points $1,2,3,4$ are in the same Minkowski patch, with the causal relationships $1>2$, $3>4$, and all other pairs spacelike-separated, see figure~\ref{fig:configcontinuousspinblock}.

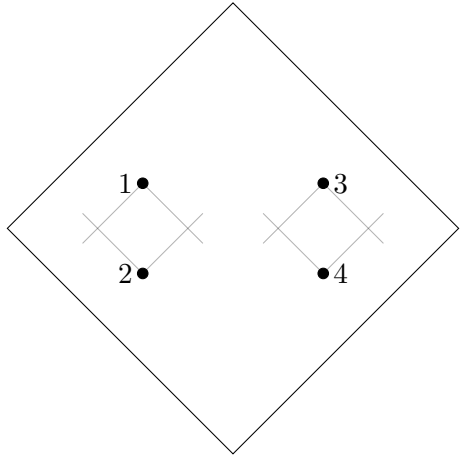
\begin{figure}[ht!]
	\centering
		\begin{tikzpicture}
		
		\draw (-3,0) -- (0,3) -- (3,0) -- (0,-3) -- cycle;
				
		\draw[fill=black] (-1.2,0.6) circle (0.07);
		\draw[fill=black] (1.2,0.6) circle (0.07);
		\draw[fill=black] (-1.2,-0.6) circle (0.07);
		\draw[fill=black] (1.2,-0.6) circle (0.07);

		\draw[opacity=0.3] (-2,-0.2) -- (-1.2,0.6) -- (-0.4,-0.2);
		\draw[opacity=0.3] (-2,0.2) -- (-1.2,-0.6) -- (-0.4,0.2);
		\draw[opacity=0.3] (2,-0.2) -- (1.2,0.6) -- (0.4,-0.2);
		\draw[opacity=0.3] (2,0.2) -- (1.2,-0.6) -- (0.4,0.2);
		
		\node[left] at (-1.2,0.6) {$1$};
		\node[left] at (-1.2,-0.6) {$2$};
		\node[right] at (1.2,0.6) {$3$};	
		\node[right] at (1.2,-0.6) {$4$};
		
		\end{tikzpicture}
		\caption{A configuration of points where $1>2$ and $3>4$, with all other pairs of points spacelike-separated. The three-point structure (\ref{eq:positivestruct}) is positive in this configuration.}
		\label{fig:configcontinuousspinblock}
\end{figure}

 We also modify the three-point structures by taking $x_{34}^2\to -x_{34}^2$ and $x_{12}^2\to -x_{12}^2$ so that they are positive when $x_0$ is spacelike from $1,2$ and $3,4$, since precisely these positive structures will appear later. Specifically, let
\be
T^{\De_1,\De_2}_{\De,J}(x_1,x_2,x_0,z) &= \frac{(2z\.x_{20} x_{10}^2 - 2z\.x_{10} x_{20}^2)^J}{(-x_{12}^2)^{\frac{\De_1+\De_2-\De+J}{2}} (x_{10}^2)^{\frac{\De_1+\De-\De_2+J}{2}} (x_{20}^2)^{\frac{\De_2+\De-\De_1+J}{2}}}.
\label{eq:positivestruct}
\ee
We will study the block
\be
\frac{T^{\De_1,\De_2}_{\De,J} T^{\De_3,\De_4}_{\De,J}}{\<\cO\cO\>},
\ee
where $\<\cO\cO\>$ is the two-point structure (\ref{eq:standardtwoptconvention}).
To define a block, our structures only need to be defined when $x_0$ is spacelike from the other points, so we do not need to give an $i\e$ prescription here. 

In the OPE limit $1\to 2$, we have
\be
T^{\De_1,\De_2}_{\De,J}(x_1,x_2,x_0,z) &\sim \frac{1}{(-x_{12}^2)^{\frac{\De_1+\De_2-\De+J}{2}}} \frac{(-2z\.I(x_{20})\.x_{12})^J}{(x_{20}^{2})^\De}\qquad (1\to 2).
\ee
The quantity on the right differs from the two-point structure $\<\cO(x_2,z')\cO(x_0,z)\>$ by the replacement $z'\to x_{12}$. We can no longer strip off $z'$ and contract indices with $x_{12}$. However, the replacement can be achieved via an integral transform:
\be\label{eq:Bdefinition}
T^{\De_1,\De_2}_{\De,J}(x_1,x_2,x_0,z) &\sim B_{12\cO}\<\cO(x_2,z') \cO(x_0,z)\> \nn\\
B_{12\cO}f(x',z') &= \frac{1}{(-x_{12}^2)^{\frac{\De_1+\De_2-\De-J-d+2}{2}}}  \frac{\G(J+d-2)}{\pi^{\frac{d-2}{2}} \G(J+\frac{d-2}{2})} \int D^{d-2} z' (-2 x_{12}\.z')^{2-d-J} f(x',z').
\ee

Now let us apply $B_{12\cO}$ to the three-point structure $T^{\De_3,\De_4}_{\De,J}(x_3,x_4,x_2,z)$, working in the limit $3\to 4$  (since this is sufficient to determine the small cross-ratio dependence of the resulting block). In doing so, we need the identity
\be
&\int D^{d-2} z'\, (-2 x_{12}\.z')^{2-d-J} (-2 z'\.I(x_{42})\.x_{34})^J\nn\\
 &= (-x_{12}^2)^{\frac{2-d-J}{2}} (-x_{34}^2)^{\frac J 2}\frac{2^{2-d} \vol(S^{d-2}) }{\hat C_J(1)}  \hat C_J\p{\frac{-x_{12} \. I(x_{42}) \. x_{34}}{(-x_{12}^2)^{1/2}(-x_{34}^2)^{1/2}}},
\label{eq:gegenbauerintegral}
\ee
where $\hat C_J(\eta)$ is given in (\ref{eq:hatC}). (Here, it is important that we use the correct definition of $\hat C_J$ for non-integer $J$.)
Using (\ref{eq:gegenbauerintegral}), we find that in the OPE limit
\be
\frac{T^{\De_1,\De_2}_{\De,J} T^{\De_3,\De_4}_{\De,J}}{\<\cO\cO\>}
&\sim
\frac{1}{(-x_{12}^2)^{\frac{\De_1+\De_2}{2}}(-x_{34}^2)^{\frac{\De_3+\De_4}{2}}}\p{\frac{x_{12}^2x_{34}^2}{x_{42}^{4}}}^{\De/2} 2^{J} \hat C_J\p{\frac{-x_{12} \. I(x_{42}) \. x_{34}}{(-x_{12}^2)^{1/2}(-x_{34}^2)^{1/2}}}, \nn\\
\ee
so that
\be
\label{eq:resultforcontinuousspipblock}
\frac{T^{\De_1,\De_2}_{\De,J} T^{\De_3,\De_4}_{\De,J}}{\<\cO\cO\>} &= \frac{1}{(-x_{12}^2)^{\frac{\De_1+\De_2}{2}} (-x_{34}^2)^{\frac{\De_3+\De_4}{2}}} \p{\frac{x_{14}^2}{x_{24}^2}}^{\frac{\De_2-\De_1}{2}} \p{\frac{x_{14}^2}{x_{13}^2}}^{\frac{\De_3-\De_4}{2}} G^{\De_i}_{\De,J}(\chi,\bar\chi).
\ee
This is the same result we would have gotten by pretending $J$ was an integer and performing the computation in the previous subsection. However, here we see that a conformal block with non-integer $J$ is well-defined and completely specified by continuous-spin two- and three-point structures.

\subsubsection{Rules for weight-shifting operators}

Let us consider how the gluing rule described in~\ref{app:3ptgluing} interacts with weight-shifting operators changing the internal representation. Suppose we can write
\be
	\<\cO_1 \cO_2 \cO^\dagger(x)\>=\<\cO_1 (\cD_A\cO_2') (\tl\cD^A\cO^{\prime\dagger})\>
\ee
for a pair of weight-shifting operators $\cD$ and $\tl\cD$. By acting with the same weight-shifting operators on~\eqref{eq:opelimit} for primed operators we find
\be
	\<\cO_1 \cO_2 \cO^\dagger(x)\>\sim (\cD_{2,A}B_{12\cO})(x_{12}) \<\cO(x_2)(\tl\cD^A\cO^{\prime\dagger})(x)\>.
\ee
Recall the crossing equation~\eqref{eq:twoptcrossing}, which holds when the two-point structures are related to the kernel of $\wS$-transform. Let us assume for now that this is the case. Then we find
\be
	\<\cO_1 \cO_2 \cO^\dagger(x)\>\sim  (\cD_{2,A}B_{12\cO})(x_{12}) \<(\wS[\overline{\tl\cD}]^A\cO)(x_2)\cO^{\prime\dagger}(x)\>.	
\ee
Substituting this into~\eqref{eq:Gopelimit}, we find
\be
G^{\cO_i}_\cO(x_i) &\sim (\cD_{2,A}B_{12\cO})(x_{12}) \<(\wS[\overline{\tl\cD}]^A\cO)(x_2)\cO_3\cO_4\>.
\ee
Using notation~\eqref{eq:nicenotationforblock} we can summarize this as\footnote{The results of~\cite{Karateev:2017jgd} concerning weight-shifting of the internal representation are recovered by further using crossing for the weight-shifting operator acting on the right three-point structure.}
\be
\label{eq:blockwsrulepre}
\frac{\<\cO_1 (\cD_A\cO_2) (\tl\cD^A\cO^{\prime\dagger})\>\<\cO \cO_3 \cO_4\>}{\<\cO\cO\>}=\frac{\<\cO_1 (\cD_A\cO_2) \cO^{\prime\dagger}\>\<(\wS[\overline{\tl\cD}]^A\cO) \cO_3 \cO_4\>}{\<\cO'\cO'\>}.
\ee
This holds if the two-point functions for $\cO$ and $\cO'$ are standard in the sense of being related to $\wS$-kernel. Generalization of this to arbitrary two-point functions is given by
\be
\label{eq:blockwsrule}
\frac{\<\cO_1 (\cD_A\cO_2) (\tl\cD^A\cO^{\prime\dagger})\>\<\cO \cO_3 \cO_4\>}{\<\cO\cO\>}=\frac{\<\cO'\cO'\>}{\<\cO\cO\>}\frac{\<\cO_1 (\cD_A\cO_2) \cO^{\prime\dagger}\>\<(\wS[\overline{\tl\cD}]^A\cO) \cO_3 \cO_4\>}{\<\cO'\cO'\>},
\ee
where the ratio of two-point functions is a scalar defined as
\be
	\frac{\<\cO'\cO'\>}{\<\cO\cO\>}\equiv \frac{\<\cO'\cO'\>}{\<\cO'\cO'\>_0}\frac{\<\cO\cO\>_0}{\<\cO\cO\>},
\ee
where the structures with subscript $0$ are standard and related to $\wS$-kernel. Note that we can reverse~\eqref{eq:blockwsrule} by replacing $\tl\cD\to \overline{\wS^{-1}[\tl\cD]}$. However, due to~\eqref{eq:lorentzianSoverlinecommutator} we have $\overline{\wS^{-1}[\tl\cD]}=\wS[\overline{\tl\cD}]$ and so we get the same rule for moving the operator from right to left.

\subsection{A Lorentzian integral for a conformal block}
\label{sec:lorentzianblock}

Conformal blocks in Euclidean signature can be computed via a ``shadow representation," where one integrates a product of three-point functions over Euclidean space \cite{Ferrara:1972uq,Dolan:2000ut,SimmonsDuffin:2012uy}. However, this integral produces a linear combination of a standard block $G^{\De_i}_{\De,J}$ and the so-called ``shadow block" $G^{\De_i}_{d-\De,J}$. The shadow block comes from regions of the integral where the OPE is not valid inside the integrand.

By contrast, there is a simple integral representation for a block alone (without its shadow) in Lorentzian signature \cite{Czech:2016xec}. The reason is that in Lorentzian signature, we can integrate over a conformally-invariant region that stays away from two of the points, say $x_{3,4}$. Thus, the $x_3\to x_4$ OPE limit can be taken inside the integrand and dictates the behavior of the result.

\begin{figure}[ht!]
	\centering
		\begin{tikzpicture}
		
		\draw (-3,0) -- (0,3) -- (3,0) -- (0,-3) -- cycle;
				
		\draw[fill=black] (-1.2,0.6) circle (0.07);
		\draw[fill=black] (1.2,0.6) circle (0.07);
		\draw[fill=black] (-1.2,-0.6) circle (0.07);
		\draw[fill=black] (1.2,-0.6) circle (0.07);
		\draw[fill=black] (-1.35,-0.05) circle (0.07);

		\draw[opacity=0.3] (-2,-0.2) -- (-1.2,0.6) -- (-0.4,-0.2);
		\draw[opacity=0.3] (-2,0.2) -- (-1.2,-0.6) -- (-0.4,0.2);
		\fill[yellow,opacity=0.3] (-1.8,0) -- (-1.2,0.6) -- (-0.6,0) -- (-1.2,-0.6) -- (-1.8,0);

		\draw[opacity=0.3] (2,-0.2) -- (1.2,0.6) -- (0.4,-0.2);
		\draw[opacity=0.3] (2,0.2) -- (1.2,-0.6) -- (0.4,0.2);
		
		\node[left] at (-1.2,0.7) {$1$};
		\node[left] at (-1.2,-0.7) {$2$};
		\node[right] at (1.2,0.6) {$3$};	
		\node[right] at (1.2,-0.6) {$4$};
		\node[right] at (-1.35,-0.05) {$0$};
		
		\end{tikzpicture}
		\caption{In the Lorentzian integral for a conformal block, the point $x_0$ is integrated over the diamond $2<0<1$ (yellow). Because the integration region is far away from points $3,4$, the $3\x 4$ OPE is valid inside the integral.}
		\label{fig:configcontinuousspinblocktwo}
\end{figure}
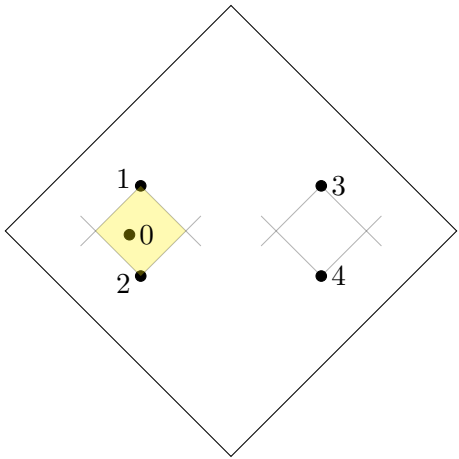

The Lorentzian integral for a conformal block plays an important role in section~\ref{sec:inversionforfourptofprimaries}, so let us compute it.  Consider the same configuration as in the previous subsection where $1,2,3,4$ are in the same Poincare patch, with $1>2$ and $3>4$, and other pairs of points spacelike separated from each other (figure~\ref{fig:configcontinuousspinblocktwo}). We can produce a conformal block in the $1,2\to 3,4$ channel by performing a shadow-like integral over the causal diamond $2<0<1$,
\be
\label{eq:lorentzianintegralforblock}
\cG_{\De,J} \equiv \int_{2<0<1} d^d x_0 D^{d-2} z \,|T^{\De_1,\De_2}_{d-\De,2-d-J}(x_1,x_2,x_0,z)| T^{\De_3,\De_4}_{\De,J}(x_3,x_4,x_0,z)
\ee
The notation $|T^{\De_1,\De_2}_{d-\De,2-d-J}|$ means that spacetime intervals $x_{ij}$ should appear with absolute values $|x_{ij}|$, so that the integrand is positive in the configuration we are considering. (This notation is somewhat imprecise, since when $\De_1,\De_2,\De,J$ are complex, we do not mean one should take the absolute value of the whole expression.)
When $J$ is an integer, there is a similar integral expression for a Lorentzian block with 
$\int D^{d-2} z$ replaced by index contractions. However (\ref{eq:lorentzianintegralforblock}) also works for continuous spin. 

The expression (\ref{eq:lorentzianintegralforblock}) is proportional to $G_{\De,J}$ because it is a conformally-invariant solution to the Casimir equation whose OPE limit agrees with the OPE limit of $T^{\De_3,\De_4}_{\De,J}$ (because the integration point stays away from $x_{3,4}$). The behavior of the integral in the limit $1\to 2$ is not immediately obvious. However, conformal invariance requires that this limit must be the same as $3\to 4$. 

More precisely, in the OPE limit $3\to 4$, we have
\be
T^{\De_3,\De_4}_{\De,J}(x_3,x_4,x_0,z) &\sim B_{34\cO}\<\cO(x_4,z') \cO(x_0,z)\> \qquad (3\to 4,\ \  0\approx3,4),
\ee
where $B_{34\cO}$ is the linear operator defined in~\eqref{eq:Bdefinition}. Plugging this in, we find
\be
\label{eq:curlyg}
\cG_{\De,J} &\sim B_{34\cO}\int_{2<0<1} d^d x_0 D^{d-2} z\, |T^{\De_1,\De_2}_{d-\De,2-d-J}(x_1,x_2,x_0,z)| \<\cO(x_4,z')\cO(x_0,z)\> \qquad (3\to 4).
\ee
The integral in the OPE limit now takes the form of an $\wS$-transform.

\subsubsection{Shadow transform in the diamond}

Let us evaluate the integral (\ref{eq:curlyg}) by splitting it into two steps: first we apply $\wSD$ and then subsequently $\wSJ$. 
For notational convenience, define
\be
\De_0 &\equiv d-\De\nn\\
J_0 &\equiv 2-d-J.
\ee
The $\wSD$ transform is fixed by conformal invariance up to a coefficient $a^{\De_1,\De_2}_{\De_0,{J_0}}$,
\be
&{\wSD}_0[|T^{\De_1,\De_2}_{d-\De,2-d-J}(x_1,x_2,x_0,z)| \th(2<0<1)]\nn\\
&= \int_{2<0<1} d^d x_0 \frac{1}{x_{04}^{2(d-\De_0)}} |T^{\De_1,\De_2}_{d-\De,2-d-J}(x_1,x_2,x_0,I(x_{04})z)|  \nn\\
 &= a^{\De_1,\De_2}_{\De_0,{J_0}} \frac{|2z\.x_{14} x_{24}^2-2z\. x_{24} x_{14}^2|^{J_0}}{|x_{12}|^{\De_1+\De_2-(d-\De_0)+{J_0}}|x_{14}|^{\De_1+(d-\De_0)-\De_2+{J_0}}|x_{24}|^{\De_2+(d-\De_0)-\De_1+{J_0}}}.
\ee
Here, we are writing expressions valid in the kinematical configuration we are considering, namely $2<0<1$ and $4\approx1,0,2$.
To find the coefficient, we choose the following configuration in lightcone coordinates
\be
x_0 &= (u,v,x_\perp),\nn\\
x_1 &= (1,0,0),\nn\\
x_2 &= (0,1,0),\nn\\
x_{4} &= (\oo,\oo,0),\nn\\
w &= I(x_{04})z = (2,0,0),
\ee
where the metric is $x^2 = uv+x_\perp^2$. Note that since $4$ is sent to infinity, $w$ is actually independent of $x_0$. Our integral becomes
\be
a^{\De_1,\De_2}_{\De_0,{J_0}} &= \frac 1 {2^{{J_0}+1}} \int du\, dv\, dx_\perp \frac{|2w\.x_{10} x_{20}^2-2w\. x_{20} x_{10}^2|^{J_0}}{|x_{12}|^{\De_1+\De_2-\De_0+{J_0}}|x_{10}|^{\De_1+\De_0-\De_2+{J_0}}|x_{20}|^{\De_2+\De_0-\De_1+{J_0}}} \nn\\
&= \frac {\vol(S^{d-3})} 2 \int du\, dv\, dr\, r^{d-3} \frac{(u(1-u)-r^2)^{J_0}}{(u(1-v)-r^2)^{\frac{\De_1-\De_2+\De_0+{J_0}}{2}}(v(1-u)-r^2)^{\frac{\De_2-\De_1+\De_0+{J_0}}{2}}}.
\ee

It is now straightforward to perform the $v$ integral over $v\in[\frac{r^2}{1-u},\frac{u-r^2}{u}]$, followed by the $r$ integral over $r\in[0,\sqrt{u(1-u)}]$, and finally the $u$ integral over $u\in[0,1]$.  The result is
\be
a^{\De_1,\De_2}_{\De_0,{J_0}} &= \frac{\pi^{\frac {d-2}{2}} \G(2-\De_0) \G(\frac{2-{J_0}-\De_0+\De_1-\De_2}{2}) \G(\frac{d+{J_0}-\De_0+\De_1-\De_2}{2}) \G(\frac{2-{J_0}-\De_0-\De_1+\De_2}{2})\G(\frac{d+{J_0}-\De_0-\De_1+\De_2}{2})}{2\G(1+\frac d 2 - \De_0)\G(2-{J_0}-\De_0)\G(d+{J_0}-\De_0)}.
\ee
Note that $a^{\De_1,\De_2}_{\De_0,{J_0}}=a^{\De_1,\De_2}_{\De_0,2-d-{J_0}}$, which is consistent with the requirement that $\wSD$ commute with $\wSJ$. We can additionally perform $\wSJ$ using
\be
\int D^{d-2} z' (-2 z\.z')^{2-d-{J_0}} (-2z'\.v)^{J_0}
&= \frac{\pi^{\frac{d-2}{2}} \G(-{J_0}-\frac{d-2}{2})}{\G(-{J_0})} (-v^2)^{\frac{d-2}{2}+{J_0}}(-2z\.v)^{2-d-{J_0}}.
\ee
Combining everything together, we find
\be
{\wS}_0[|T^{\De_1,\De_2}_{d-\De,2-d-J}(x_1,x_2,x_0,z)| \th(2<0<1)]
&= b^{\De_1,\De_2}_{\De,J} T^{\De_1,\De_2}_{\De,J}(x_1,x_2,x_4,z)\nn\\ 
b^{\De_1,\De_2}_{\De,J} &\equiv \frac{\pi^{\frac{d-2}{2}} \G(J+\frac{d-2}{2}) }{\G(J+d-2)}a^{\De_1,\De_2}_{d-\De,2-d-J}.
\label{eq:bcoeff}
\ee
Plugging this into (\ref{eq:curlyg}) and using (\ref{eq:resultforcontinuousspipblock}), we conclude 
\be
\cG_{\De,J}(x_i)
&= \frac{b^{\De_1,\De_2}_{\De,J} }{(-x_{12}^2)^{\frac{\De_1+\De_2}{2}} (-x_{34}^2)^{\frac{\De_3+\De_4}{2}}} \p{\frac{x_{14}^2}{x_{24}^2}}^{\frac{\De_2-\De_1}{2}} \p{\frac{x_{14}^2}{x_{13}^2}}^{\frac{\De_3-\De_4}{2}} G^{\De_i}_{\De,J}(\chi,\bar\chi).
\label{eq:resultforlorentzianintegralforblock}
\ee

\subsection{Conformal blocks at large $J$}
\label{app:largeJ}

In this appendix, we compute the large-$J$ behavior of a conformal block. Recall that we have the decomposition
\be
G^{\De_i}_{\De,J}(\chi,\bar\chi) &= g_{\De,J}^\mathrm{pure}(\chi,\bar\chi) + \frac{\G(J+d-2)\G(-J-\frac{d-2}{2})}{\G(J+\frac{d-2}{2})\G(-J)} g^\mathrm{pure}_{\De,2-d-J}(\chi,\bar\chi).
\ee
Thus it suffices to compute the large-$J$ behavior of $g^\mathrm{pure}_{\De,J}$.

The Casimir equation was solved in the large-$\De$ limit in \cite{Kos:2013tga,Kos:2014bka}. We can use this result together with an affine Weyl reflection to determine $g^\mathrm{pure}_{\De,J}$ at large $J$.
The solution from \cite{Kos:2014bka} is given by
\be
\label{eq:largedeltasoln}
\frac{r^{\De} f_J(\cos\th) }{(1-r^2)^{\frac{d-2}{2}}(1+r^2 + 2r\cos\th)^{\frac 1 2(1+\De_{12}-\De_{34})}(1+r^2 - 2r\cos\th)^{\frac 1 2(1+\De_{34}-\De_{12})}}\qquad(|\De|\gg 1),
\ee
where $r$ and $\th$ are defined by
\be
\rho = r e^{i\th},\quad \bar\rho = r e^{-i\th},\qquad\chi = \frac{4\rho}{(1+\rho)^2},\quad \bar\chi = \frac{4\bar\rho}{(1+\bar\rho)^2}.
\ee
From studying the regime $r \ll 1$, we find that $f_J(\cos\th)$ must obey the Gegenbauer differential equation. 

Note that the conformal Casimir equation has the following symmetries:
\be
(\De,J) &\leftrightarrow (1-J,1-\De),\nn\\
r &\leftrightarrow w=e^{i\th}.
\ee
The first is an affine Weyl reflection that preserves the Casimir eigenvalue. The second transformation is equivalent to $\bar\rho\leftrightarrow 1/\bar\rho$, which leaves $\bar\chi$ invariant, and therefore also leaves the Casimir equation invariant. Applying these transformations to (\ref{eq:largedeltasoln}), we find
\be
\frac{w^{1-J} f_{1-\De}\p{\tfrac 1 2(r+\tfrac 1 r)} }{(1-w^2)^{\frac{d-2}{2}}(1+w^2+w(r+1/r))^{\frac 1 2(1+\De_{12}-\De_{34})}(1+w^2-w(r+1/r))^{\frac 1 2(1+\De_{34}-\De_{12})}}\qquad(|J|\gg 1).
\ee
Note in particular that we have replaced large-$\De$ with large-$J$.
Demanding pure power behavior as $r\to 0$ requires us to choose the following solution to the Gegenbauer equation:
\be
\label{eq:radialsolngegenbauer}
f_J(x)&=(2x)^J {}_2F_1\p{\frac{-J}{2},\frac{1-J}{2},2-J-\frac d 2,\frac{1}{x^2}}.
\ee
Finally, fixing the constant out front and rearranging terms, we find (\ref{eq:largeJlimitofgpure}).
	
\bibliographystyle{JHEP}
\bibliography{refs}

\end{document}